\documentclass{article}

\usepackage{arxiv}

\usepackage[utf8]{inputenc} 
\usepackage[T1]{fontenc}    
\usepackage{url}            
\usepackage{booktabs}       
\usepackage{nicefrac}       
\usepackage{microtype}      
\usepackage{lipsum}
\usepackage{amsmath,amstext,amsthm,amssymb,amsfonts}
\usepackage{graphicx,psfrag}
\usepackage[linesnumbered,ruled,vlined]{algorithm2e}
\usepackage{comment}
\usepackage{multicol,multirow}
\usepackage{xcolor}
\usepackage{xspace}
\usepackage{url}
\usepackage{hyperref}
\usepackage{subcaption}
\usepackage{multirow}
\usepackage{longtable}
\usepackage{enumitem}
\graphicspath{ {./images/} }

\makeatletter
\def\hlinewd#1{%
\noalign{\ifnum0=`}\fi\hrule \@height #1 \futurelet
\reserved@a\@xhline}
\makeatother

\newtheorem{theorem}{Theorem}[section]

\newtheorem{lemma}{Lemma}[section]
\newtheorem{definition}{Definition}[section]

\newtheorem{remark}{Remark}

\title{Game of arrivals at a two queue network with heterogeneous customer routes}

\author{
 Agniv Bandyopadhyay \\
  School of Technology and Computer Science\\
  Tata Institute of Fundamental Research\\
  Mumbai, India \\
  \texttt{agniv.bandyopadhyay@tifr.res.in} \\
   \And
 Sandeep Juneja \\
Safexpress Center for Data, Learning and Decision Sciences, \\
Ashoka University,\\
Sonipat, India\\
  \texttt{sandeep.juneja@ashoka.edu.in} \\
}
\begin{document}
\maketitle
\begin{abstract}
We consider a queuing network that opens at a specified time, where  customers are non-atomic and belong to different  classes. Each class has its own route, and as is typical in the literature, the costs are a linear function of waiting  and service completion time. We restrict ourselves to a two class, two  queue network: this simplification is well motivated as the diversity in solution structure as a function of problem parameters is substantial even in this simple setting (e.g., a specific routing structure involves eight different regimes), suggesting a combinatorial blow up as the number of queues, routes and customer classes increase. We identify the unique Nash equilibrium customer arrival profile when the customer linear cost preferences are different. This profile is a function of problem parameters including the size of each class, service rates at each queue, and customer cost preferences. When customer cost  preferences match, under certain parametric settings, the equilibrium arrival profiles may not be unique and may lie in a convex set. We further make a surprising observation  that in some parametric settings, customers in one class may arrive in disjoint intervals. Further, the two classes may arrive  in contiguous intervals or in overlapping intervals, and at varying rates within an interval, depending upon the problem parameters.
\end{abstract}



\section{Introduction}\label{sec:intro}

\noindent\textbf{Motivation:}~Queueing games or games where strategic customers, 
served by a 
queue or a queuing network, decide on actions such as which queue to join, whether to join,
when to join, what level of priority to select and so on, are well studied in 
the literature (see \cite{Hassin16,Hassin03,Haviv21} for surveys). In this paper we focus 
on the queuing arrival game where customers decide on when to join a queuing network. This is
in contrast to much of the literature that considers arrival games to a single
queue (see, e.g., \cite{Breinbjerg17,Glazer83,Oz21,Haviv21,Jain11,Juneja13,Lindsey04}). Applications of such arrival games to multiple queues are many: Customers arriving 
to a health facility where they queue up to meet a doctor and then may queue 
to  get tests done and/or procure medicines; in banks where different customers may need to go to different and multiple counters; in cafeteria where customers queue up for different and multiple food items, and so on. 
In consonance with  much of the existing literature, we model customers 
as non-atomic `infinitesimal' fluid particles with costs that are
linear in waiting time and in time to service, customers are served in a first-come-first-serve manner and service facility opens
at time zero (see \cite{Honnappa15,Jain11,Lindsey04}). In \cite{Juneja13}, uniqueness of equilibrium solution was proven in a single queue setting with stochastic service rates and large no of users. Moreover, \cite{Juneja13} showed that, the equilibrium solution with a large no of users is well approximated by the corresponding fluid system and thus lending support to the fluid analysis.
This fluid setting has resulted
in   unique and  elegant, easily calculable customer equilibrium profiles
for single queues as well as certain symmetric queuing networks where customers
have homogeneous travel routes (see \cite{Honnappa15,Jain11}). \\

\noindent\textbf{Problem description:}~To keep the discussion simple we focus on a two queue, two class customer setting where 
each class of customers has a distinct  route,
and customers in each queue are served in a first-come-first-served manner.
While this set-up may be practically interesting, 
our main contributions are theoretical: Our key aim is to test whether the elegance and simplicity of customer
equilibrium behavior to a single queue extends to more general queuing networks in presence of heterogeneous routing. \\

\noindent\textbf{Brief overview of results:}~Even in the simple two queue setting, we see that, unlike for the single queue,
here the solution structure and order of arrivals in equilibrium is a function of all the problem parameters, \textit{i.e.}, linear coefficients of the cost function, the queue service rates and the population size of each class of customers. For one  set of customer travel routes  we observe that depending upon problem parameters, there exist eight 
distinct solution structures. This suggests that as the number of queues increase there may be a rapid blow-up in the solution structures. This may make the problem of identifying and learning the correct structure computationally prohibitive. In this paper, we do not address the issue of customers 
learning the equilibrium profile by repeatedly playing the game. Limiting behaviour of players repeatedly updating their action in a game using a simple rule, often called a `no-regret' policy, are studied in \cite{Blum07} and we refer the reader to  \cite{Lugosi06} for a comprehensive exposition of this literature.

Our other broad contributions/observations  are:~\textbf{1)}\hspace{0.025in} We find that similar to the single queue setting, the equilibrium profile of arriving customers is unique for a wide set of parameters. However, interestingly, when customer cost preferences across classes are identical up to a constant, there may be multiple equilibrium arrival profiles, all lying in a convex set that we identify. Although there are many  arrival profiles in equilibrium in this case,  they all have identical social cost.\\ \textbf{2)}\hspace{0.025in}In \cite{Jain11}, the equilibrium profile  is determined for the case when multiple classes of customers with linear costs are arriving at a single queue again. They find that different classes of customers arrive in non-overlapping and contiguous intervals. In our two queue setting we find that depending upon the problem parameters, in equilibrium, arrivals may come in non-overlapping and contiguous intervals,  in overlapping intervals, or under certain parametric settings, a class of customers may even arrive in disjoint intervals.  Moreover, we show that whether the classes will arrive over overlapping sets or not is independent of the population sizes and decided entirely by the queue service rates and customer preferences. \\

\noindent \textbf{Related literature:}~The arrival games to queues were first considered by \cite{Glazer83}. The concert queueing game is the fluid setting was introduced in \cite{Jain11}. The arrival game in a fluid network of bottleneck queues including tandem, Trellis, and general feed-forward networks was considered in \cite{Honnappa15}, where they characterized the  equilibrium arrival profile in each of these topologies.  

Transportation modeling community has extensively studied  arrival games. Vickrey in \cite{Vickrey69},  introduced the morning commute problem. Unlike the concert queuing game, in these transportation problems, the service facility  has no predetermined opening time. Instead, the customers have a preferred time to complete service and a  cost for arriving too early or too late (see \cite{Hendrickson81}). This led to a huge literature on arrival games to a bottleneck queue, the impact of tolls, etc. (see \cite{Li20} for an extensive list of references). 
Much of the transportation literature considers  single queue
settings.   Lindsey, in an influential work
\cite{Lindsey04}, establishes the existence of equilibrium arrival 
profile for multiple classes of customers with general non-linear cost functions arriving at a bottleneck queue with a constant service rate, through intricate fixed point arguments.  Our work differs from transportation literature in that we consider a two queue network with heterogeneous arrival routes, linear costs, and in this setting we are able  characterize the equilibrium user arrival profiles in a closed form,
and for a large class of parameters, show that these profiles are unique.\\

\noindent\textbf{Outline of the paper:}~In Section \ref{sec_prelims} we provide the background to the arrival queueing game and overview the two-class, two queue, two-route networks that we consider. We emphasize on two heterogeneous routes networks 1) where the departures of the two classes are through different queues (Heterogeneous Departure System or HDS) and 2) where the arrivals enter at different queues (Heterogeneous Arrival System or HAS). In Section \ref{sec:non_fluid_to_fluid}, we provide a characterization of the fluid model as the limit of a sequence of non-fluid instances as the no. of users increase to infinity. This establishes the practicability of our fluid model for approximating real world non-fluid systems. Furthermore, in Section \ref{sec:non_fluid_to_fluid}, we conjecture that, the equilibrium behavior of the fluid model can be approximated using the Mixed Nash Equilibrium of these non-fluid models. In Section \ref{sec_hetdepartures}, we identify the equilibrium arrival profile for all possible parameters for  arriving customers for HDS. In particular, we see that these parameters can be partitioned into four distinct regions each having a separate solution structure, when the two customer classes have unequal preferences. In Section \ref{sec_hetarrivals} we similarly analyze HAS. Here we discover that the parameter space can be partitioned into eight regions based on the solution structure, when the two customer classes have unequal preferences. Moreover, for both HDS and HAS, when the groups have identical preference, we identify a parametric regime where unique equilibrium exists, as well as a parametric regime where the equilibrium is non-unique and the set of equilibrium profiles is convex. We end with a brief conclusion in Section \ref{sec:conclusion}. In the main body, we have confined our discussion to the main proof ideas behind our results and have kept the detailed proofs in the appendix. 

\section{Preliminaries}\label{sec_prelims}


\subsection{Fluid Model} 
We consider a fluid model having two classes of customers or users. The size of each class $i=1,2 $ is given by a positive quantity $\Lambda^{(i)} >0$. In every class $i=1,2$ individual users are infinitesimal and the set of all users in class $i$ is given by the points in the interval $[0,\Lambda^{(i)}]$. 



We define functions $F^{(i)}:\mathbb{R}\to[0,\Lambda^{(i)}]$ for $i=1,2$ such that, $F^{(i)}(t)$ denotes the amount of users of class $i$ that arrive by time $t$. We call $F^{(i)}$ the arrival profile of class $i$ users. For $i=1,2$, we restrict ourselves to arrival profiles $F^{(i)}(\cdot)$ satisfying the following properties:
\begin{enumerate}[leftmargin=6em]
    \item[\textbf{Property 1:}] $F^{(i)}(\cdot)$ is non-decreasing and satisfies $F^{(i)}(-\infty)=0$ and $F^{(i)}(+\infty)= \Lambda^{(i)}$.
    \item[\textbf{Property 2:}] $F^{(i)}(\cdot)$ is right-continuous.
\end{enumerate}
Every $F^{(i)}(\cdot)$ satisfying the above two properties can be expressed as a sum of a  non-decreasing absolutely continuous function, a non-decreasing discrete function and a non-decreasing singularly continuous function (see \cite[Proposition 4.5.1]{Athreya06}). For simplicity in analysis, we further assume the following property: 
\begin{enumerate}[leftmargin=6em]
    \item[\textbf{Property 3:}]  $F^{(i)}(\cdot)$ doesn't have any singularly continuous component. 
\end{enumerate}
Thus, by property 1,2, and 3, we restrict ourselves only to  arrival profiles $F^{(1)}(\cdot), F^{(2)}(\cdot)$ having discrete and absolutely continuous components. We call the pair $\mathbf{F}=\{F^{(1)},F^{(2)}\}$ as the joint arrival profile of the two classes.  Later on, by Lemma \ref{lem_eq_has_no_jump}, we will further restrict ourselves to absolutely continuous arrival profiles for the search of equilibrium behavior.   

We consider a network comprising of two queues, both starting service at time $t=0$. Let $\mu_1$ and $\mu_2$,
respectively, denote the
deterministic fixed service rates at the two queues after they start service. 
We consider  four routes
of the two arriving classes to the two queues. These cases are displayed in Table \ref{tab:all_2queue_instances}. Note that every situation where the two classes use both the queues through their routes can be reduced to one of these four cases.

\begin{table}[h!]
    \centerline{
    \begin{tabular}{|c|c|}
    \hline  
    \parbox{2cm}{\vspace{0.05in}\textbf{Instance I}} & \parbox{2cm}{\vspace{0.05in}\textbf{Instance II}} \\
    \parbox{7cm}{\vspace{-0.05in}\centerline{\includegraphics[width=6cm]{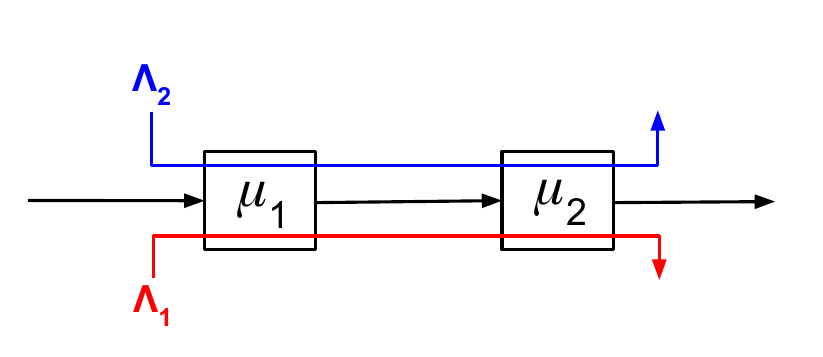}}}  &  \parbox{7cm}{\vspace{-0.05in}\centerline{\includegraphics[width=6cm]{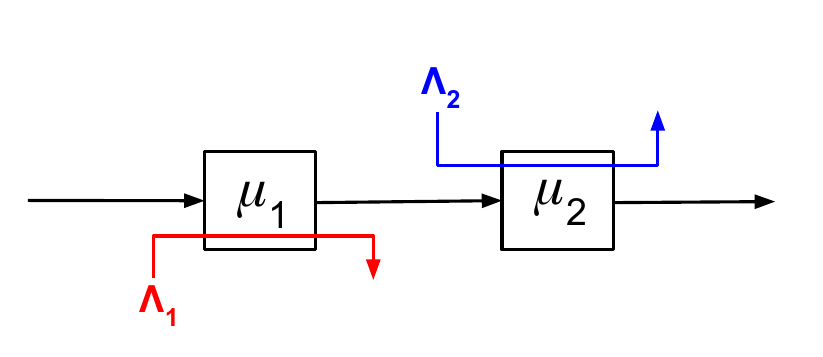}}}\\
    \hline
    \parbox{6.5cm}{\vspace{0.05in}\textbf{Instance III} (Heterogeneous Departure System or HDS)} & \parbox{6.5cm}{\vspace{0.05in}\textbf{Instance IV} (Heterogeneous Arrival System or HAS)} \\
    \parbox{7cm}{\vspace{0.05in}\centerline{\includegraphics[width=6cm]{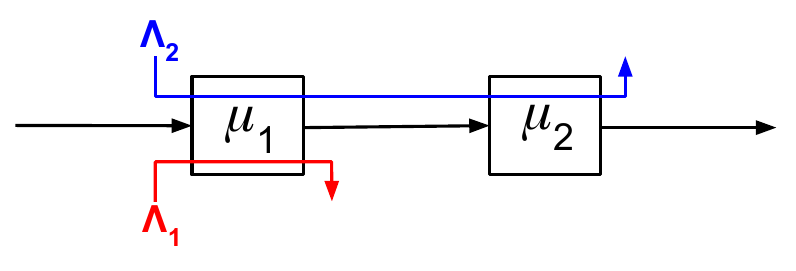}}}  & \parbox{7cm}{\vspace{0.05in}\centerline{\includegraphics[width=6cm]{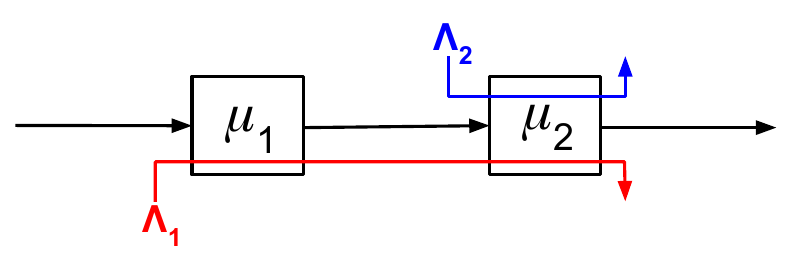}}}\\
    \hline
    \end{tabular}}
    \caption{Various instances with two groups traveling through two queues connected in tandem}
    \label{tab:all_2queue_instances}
\end{table} 

Instance I is equivalent to two groups of users arriving at a two-layer tandem network to travel by the same path. By Theorem 5 of \cite{Honnappa15}, the instance is equivalent to the case where the two groups are arriving at a single queue of capacity $\min\{\mu_1,\mu_2\}$. Instance II is equivalent to the case where the two queues independently serve the two groups and therefore is equivalent to two independent instances of a single queue with a single class customer arrivals. Hence, the first two instances are reducible to instances with just one queue studied in \cite{Jain11}. 
In this paper we study  the
 arrival equilibrium behaviour in the other two instances III and IV. We
 refer to them as  Heterogeneous Departure (HDS) and Heterogeneous Arrival Systems (HAS), respectively.\vspace{0.05in}
 
\subsection{Waiting and Departure Times}

To specify the waiting and departure times in a system,
first consider a single queue setting
where $A(t)$ denotes the total mass of users of all classes
that have arrived at the queue by time $t$. Let $\mu$ denote the service rate.
Then at time $t$, the length of the waiting queue developed in that queue will be (see Theorem 6.5 in \cite{Chen01}): 
\begin{align}\label{eq_queue_len_proc}
    Q(t)&=A(t)-\mu\cdot\max\{t,0\}+\underbrace{\sup_{s\in[0,t]}\max\left\{\mu s-A(s),0\right\}}_{U(t)}.
\end{align}
In (\ref{eq_queue_len_proc}), $U(t)/\mu$ quantifies of time the queue has remained empty during the interval $[0,t]$, and $t\mapsto U(t)$ is called the \textit{unused service process} of the queue. We assume that if there is a jump in the arrival profile $A$ at time $t$, the arrivals are positioned in the queue at uniformly random order. As a result, a user arriving at that queue at time $t$ will suffer an expected waiting time of
\begin{align}\label{eq_waiting_time_proc}
    W(t)&=\frac{Q(t+)+Q(t-)}{2\mu}+\max\{0,-t\},~\text{and departs at time}~~\tau(t)=W(t)+t,
\end{align}  
where $Q(t+)$ and $Q(t-)$ respectively denote the right and left limits of $Q$ at time $t$. Note that if the queue length process $Q(\cdot)$ is continuous (which is the case if $A(\cdot)$ is absolutely continuous), waiting time as a function of time $t$ will be $W(t)=\frac{Q(t)}{\mu}+\max\{0,-t\}$. 

If the arrival profile $A(\cdot)$ is absolutely continuous, by (\ref{eq_queue_len_proc}) and (\ref{eq_waiting_time_proc}), the departure time as a function of time $t$ will be: $\tau(t)=\frac{A(t)}{\mu}+\sup_{s\in[0,t]}\max\left\{0,s-\frac{A(s)}{\mu}\right\}$. Whenever $Q(t)>0$, the term $\sup_{s\in[0,u]}\max\{\mu s-A(s),0\}$ is independent of the choice of $u\in[t-\delta,t+\delta]$ for $\delta>0$ and sufficiently small, and when $t<0$, $\tau_1(t)=\frac{A(t)}{\mu}$. If $A(\cdot)$ is absolutely continuous,  its derivative $A^\prime(\cdot)$ will exist a.e. As a result,
\begin{align}\label{eq:derv_of_tau}
    \tau^\prime(t)&=\frac{A^\prime(t)}{\mu}~~\text{a.e. in the closure of the set of times}~\{s~\vert~s<0~\text{or}~Q(s)>0\}.
\end{align}
The above observation will be useful in our analysis of HDS and HAS in the later sections. 
\begin{definition}
    We say the queue is \emph{engaged} at time $t$ if $t$ lies in the closure of the set  $\{s~\vert~Q(s)>0\}$.
\end{definition}
By (\ref{eq:derv_of_tau}), after the queue starts serving, users depart at rate $\mu$ whenever the queue is engaged. We introduce the following notation: 

\begin{enumerate}[leftmargin=*]
    \item  $A_i(t)$ be the total mass of customers of both the groups who have arrived at queue $i=1,2$ till time $t$ (note that while $F^{(j)}(\cdot)$ denotes arrival profile
    corresponding to users of class $j$, $A_i(\cdot)$ denotes overall arrival profile to queue $i$). 
    
    \item $Q_i(t)$ and $W_i(t)$ be the length of the waiting queue, and the waiting time that a user arriving in queue $i$ at time $t$ will observe. Let $\tau_i(t)$ denote the time that user will depart the system.
    
    \item  $W_{\mathbf{F}}^{(j)}(t)$ and $\tau_{\mathbf{F}}^{(j)}(t)$ be the waiting and departure times from the network suffered by a class $j$ user arriving at time $t$ for $j\in\{1,2\}$. Note that these two functions depend on the joint arrival profile $\mathbf{F}$ and this notation will be helpful to us when we compute these functions explicitly for HAS and HDS for arbitrary $\mathbf{F}$'s.  
\end{enumerate} 

Note that we use subscripts (such as $A_i(\cdot),Q_i(\cdot),W_i(\cdot),\tau_i(\cdot)$) to denote quantities related to queue $i$ (for $i=1,2$), and use superscripts (such as $\Lambda^{(j)},F^{(j)}(\cdot),W_{\mathbf{F}}^{(j)}(\cdot),\tau_{\mathbf{F}}^{(j)}(\cdot)$) to denote quantities related to class $j$ users (for $j=1,2$). For both the queues $i=1,2$, upon defining the arrival profile $A_i(\cdot)$, using  (\ref{eq_queue_len_proc}) and (\ref{eq_waiting_time_proc}), $Q_i(\cdot)$, $W_i(\cdot)$ and $\tau_i(\cdot)$ are well-defined. Now we specify the waiting and departure times of both the queues in HDS and HAS as functions of time, under the assumption that the joint arrival profile $\mathbf{F}=\{F^{(1)},F^{(2)}\}$ is absolutely continuous (we later argue by Lemma \ref{lem_eq_has_no_jump} that considering absolutely continuous joint arrival profiles are sufficient for identifying equilibrium behavior). 
\begin{enumerate}
    \item[\textbf{HDS}]:~Arrival profiles at individual queues are $A_1(t)=F^{(1)}(t)+F^{(2)}(t)$ and $A_2(t)=F^{(2)}(\tau_1^{-1}(t))$, where $\tau_1^{-1}(t)=\sup\{s~\vert~\tau_1(s)\leq t\}$. Both $A_1(\cdot)$ and $A_2(\cdot)$ are absolutely continuous. With this, $W_{\mathbf{F}}^{(1)}(t)=W_1(t)$, $\tau_{\mathbf{F}}^{(1)}(t)=\tau_1(t)$, $W_{\mathbf{F}}^{(2)}(t)=W_1(t)+W_2(\tau_1(t))$, and $\tau_{\mathbf{F}}^{(2)}(t)=\tau_1(t)+W_2(\tau_1(t))$.
    
    \item[\textbf{HAS}]:~Arrival profile at individual queues are $A_1(t)=F^{(1)}(t)$ and  $A_2(t)=F^{(1)}(\tau_1^{-1}(t))+F^{(2)}(t)$ where $\tau_1^{-1}(t)=\sup\{s~\vert~\tau_1(s)\leq t\}$. Both $A_1(\cdot)$ and $A_2(\cdot)$ are absolutely continuous. With this, $W_{\mathbf{F}}^{(1)}(t)=W_1(t)+W_2(\tau_1(t))$, $\tau_{\mathbf{F}}^{(1)}(t)=\tau_1(t)+W_2(\tau_1(t))$, $W_{\mathbf{F}}^{(2)}(t)=W_2(t)$, and $\tau_{\mathbf{F}}^{(2)}(t)=\tau_2(t)$.
\end{enumerate}

\subsection{Solution Concept}

First we define the cost function of every user depending on the class they belong to, as well as, their arrival time. Every user, irrespective of the class, tries to simultaneously minimize her waiting and departure times. By minimizing the waiting time, every user tries to minimize the time they spend in the network. In many practical applications (such as, in a cafeteria, or, a medical facility), the quality of service (quality of the food, or, the stockpile of medicines) deteriorates with the departure time of the user.  Therefore minimizing the departure time means maximizing the quality of service.  Keeping these observations into account and following the existing literature on concert queueing games (see \cite{Juneja13}, \cite{Jain11}, \cite{Juneja18}, \cite{Honnappa15}), we model the overall objective by assuming that every user in class $i$ ($i\in\{1,2\}$) has a cost function linear in her waiting and departure times from the network given by: $C_{\mathbf{F}}^{(i)}(t)=\alpha^{(i)}\cdot W_{\mathbf{F}}^{(i)}(t)+\beta^{(i)}\cdot \tau_{\mathbf{F}}^{(i)}(t)$, where $\alpha^{(i)}$ and $\beta^{(i)}$ are positive constants quantifying the cost suffered by a class $i$ user for unit waiting time and delay in departure. Note that, for every class $i=1,2$, the relative priority over minimizing waiting time ($W_{\mathbf{F}}^{(i)}(\cdot)$) vs minimizing departure time ($\tau_{\mathbf{F}}^{(i)}(\cdot)$) is measured by the constants $\alpha^{(i)},\beta^{(i)}$, and is homogeneous across all users in class $i$.

\begin{definition}[Support of arrival profile]
Given an arrival profile $t\mapsto B(t)$ such that $B(+\infty)<\infty$, the support of $B$, denoted by $\mathcal{S}(B)$, is defined as the smallest closed set having a $B$-measure equal to $B(+\infty)$. 
\end{definition}

\begin{definition}[Equilibrium Arrival Profile (EAP)]
The joint arrival profile $\mathbf{F^\star}=\{F^{(1),\star},F^{(2),\star}\}$ is an Equilibrium Arrival Profile (EAP) of this game if for both the groups $i\in\{1,2\}$:~ $t\in \mathcal{S}(F^{(i),\star})$ and $\tilde{t}\in\mathbb{R}$ implies $C_{\mathbf{F^\star}}^{(i)}(t)\leq C_{\mathbf{F^\star}}^{(i)}(\tilde{t})$. In particular, the arrival profile is iso-cost along its support. 
\end{definition}

By definition of the EAP, users cannot strictly reduce their cost by arriving at a different time, and therefore none of the users, irrespective of the class they belong to, have any incentive to deviate. Furthermore the EAP doesn't change upon normalizing the cost function of both the classes $i=1,2$ by multiplying $1/(\alpha^{(i)}+\beta^{(i)})$. For simplicity, and without loss of generality, we assume both the classes $i=1,2$ have their normalized cost function, which are: $C_{\mathbf{F}}^{(i)}(t)=\gamma^{(i)}\cdot W_{\mathbf{F}}^{(i)}(t)+(1-\gamma^{(i)})\cdot \tau_{\mathbf{F}}^{(i)}(t)$ where $\gamma^{(i)}=\frac{\alpha^{(i)}}{\alpha^{(i)}+\beta^{(i)}}$ quantifies the preference of every class $i$ user. A value of $\gamma^{(i)}$ close to $1$ indicates users in group $i$ prefer late departure compared to waiting a long time in the network and $\gamma^{(i)}$ close to $0$ implies the opposite. So, we use $\gamma^{(i)}$ to quantify the \textit{cost preference} of every group $i$ user. 

\begin{remark}
\emph{EAP captures the aggregate equilibrium behavior of the group. We can equivalently define Nash equilibrium at individual level where under it no individual has unilateral incentive to deviate.  As is well known and discussed in more detail in \cite{Jain11}, the two concepts are equivalent.} 
\end{remark}

Recall that every arrival profile $F^{(i)}(\cdot)$ for $i=1,2$, can have discrete jumps owing to its right continuity and non-decreasing property. Lemma \ref{lem_eq_has_no_jump} below helps us restrict our search of EAP only to absolutely continuous arrival profiles without discrete jumps. 

\begin{lemma}\label{lem_eq_has_no_jump}
    In every EAP, $\mathbf{F}=\{F^{(1)},F^{(2)}\}$ of the HDS and HAS, the arrival profiles $F^{(1)}$ and $F^{(2)}$ cannot have any discontinuity (or jump increments).   
\end{lemma}

Proof of the above lemma is similar to proof of statement (ii) of Lemma 1 in \cite{Jain11}. On assuming contradiction, if any of the arrival profiles have a jump, any user arriving in that jump will be strictly better off arriving slightly early and as a result the arrival profile cannot be an EAP. 

We argued before that Instances I and II in Table \ref{tab:all_2queue_instances} are reducible to instances where one or more groups of users having distinct preferences are arriving to a single queue. \cite{Jain11} show that when two classes of customers having cost preferences $\gamma^{(1)}$ and $\gamma^{(2)}$ arrive at a single queue  with service rate $\mu$, the EAP has a simple structure. The class with smaller $\gamma^{(i)}$ comes first at arrival rate $\mu\cdot\min\{\gamma^{(1)}, \gamma^{(2)}\}$ over an interval, while the next class arrives at a contiguous but non-overlapping interval, at rate $\mu\cdot\max\{\gamma^{(1)}, \gamma^{(2)}\}$. Fig \ref{fig:one_queue_two_grps} illustrates this EAP and the resulting queue length with the assumption $\gamma^{(1)}<\gamma^{(2)}$ and is useful to contrast with the various EAP structures that we find for HDS and HAS in Sections \ref{sec_hetdepartures} and \ref{sec_hetarrivals} below. The queue length process is constructed assuming that in the EAP, class 2 users start arriving from a positive time, which is equivalent to saying, masses of the two classes satisfy $\Lambda^{(1)}>\left(\frac{1}{\gamma^{(2)}}-1\right)\Lambda^{(2)}$.  

\begin{figure}[h]
    \centering
    \includegraphics[width=12cm]{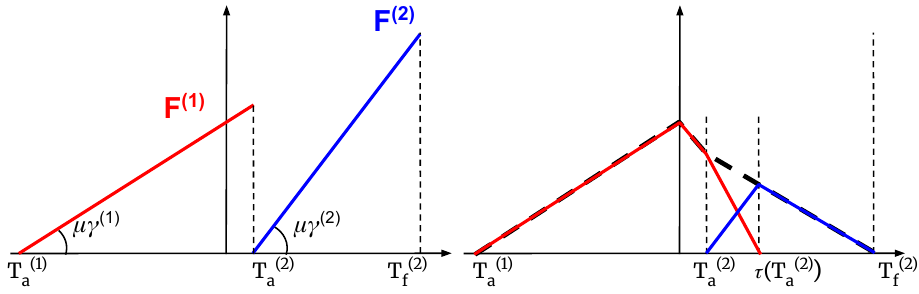}
    \caption{EAP structure (left) and resulting queue length process (right) when two classes of users with cost preferences $\gamma^{(1)}$ and $\gamma^{(2)}$ (assuming $\gamma^{(1)}<\gamma^{(2)}$) are arriving at a queue of capacity $\mu$. The support boundaries are $T^{(2)}_f=\frac{\Lambda^{(1)}+\Lambda^{(2)}}{\mu},~T^{(2)}_a=\frac{\Lambda^{(1)}}{\mu}-\left(\frac{1}{\gamma^{(2)}}-1\right)\frac{\Lambda^{(2)}}{\mu}$ and $T^{(1)}_a=-\left(\frac{1}{\gamma^{(1)}}-1\right)\frac{\Lambda^{(1)}}{\mu}-\left(\frac{1}{\gamma^{(2)}}-1\right)\frac{\Lambda^{(2)}}{\mu}$. The queue length process is illustrated only for the situation $T^{(2)}_a>0$, or equivalently $\Lambda^{(1)}>\left(\frac{1}{\gamma^{(2)}}-1\right)\Lambda^{(2)}$. \textbf{\textcolor{red}{Red}} and \textbf{\textcolor{blue}{blue}}, respectively, represents class 1 and 2 populations. The \textbf{black} dashed line represents the total waiting mass of the two classes in the plot for queue length. }
    \label{fig:one_queue_two_grps}
\end{figure}

\subsection{Relation between non-fluid queuing networks and fluid networks}\label{sec:non_fluid_to_fluid}

Following \cite{Honnappa12},  the fluid model corresponds to the limit of a series of queueing systems indexed  by $N$ where  $N\Lambda^{(1)}$ group 1 users and  $N\Lambda^{(2)}$ group 2 users arrive in the $N$-th system. In the $N$-th system, the $j$-th user of class $i$ (note that $i=1,2$ and $j\in\{1,2,\dots,N\Lambda^{(i)}\}$) chooses a distribution $F^{(i)}_{j,N}(\cdot)$ from where she samples her arrival time. We call the overall tuple $\mathbf{F}_N~:=~\left\{~F^{(i)}_{j,N}(\cdot)~:~i\in \{1,2\},~j\in\{1,2,\dots,N\Lambda^{(i)}\}\right\}$ as the strategy tuple of the $N$-th game. Now we consider HDS and HAS topologies separately and define the service time of individual users in the $N$-th queueing system: 
\begin{enumerate}[leftmargin=*]
  \item \textbf{In HDS:}~Users of both classes first arrive at queue 1 at their chosen arrival times. Service times of users arriving at queue 1 are sampled independently from an exponential distribution with mean $1/\mu_1$. After getting served at queue 1, class 1 users depart the system and class 2 users arrive at queue 2. Service times of class 2 users arriving at queue 2 are sampled independently  from an exponential distribution with mean $1/\mu_2$. Waiting time of a class 1 user is the sum of the overall time she waits in queue 1 and her service time in queue 1. Similarly, waiting time of a class 2 user is the sum of the overall time she waits in the two queues and her service times in the two queues. Departure time of every user is the sum of her arrival time and her overall waiting time in the network. 

  \item \textbf{In HAS:}~Users of class 1 first get served at queue 1 and then arrive at queue 2. Service times of class 1 users in queue 1 are sampled independently from an exponential distribution with mean $1/\mu_1$. Users of class 2 directly arrive at queue 2 alongside the class 1 users who got served from queue 1. Service times of the class 1 and 2 users at queue 2 are sampled independently from an exponential distribution with mean $1/\mu_2$. Waiting time of a class 1 user is the sum of the overall time she waits in the two queues and her service times in the two queues. Similarly, waiting time of a class 2 user is the sum of the time she waits in queue 2 and her service time in queue 2. Departure time of every user is the sum of her arrival time and her overall waiting time in the network.  
\end{enumerate}
  
We use $W_{j,N}^{(i)}$ and $\tau_{j,N}^{(i)}$, respectively, to denote the waiting and departure times of the $j$-th user of class $i$ in the $N$-th system. Both $W_{j,N}^{(i)}$ and $\tau_{j,N}^{(i)}$ are random variables. Cost function of every user is defined in the similar way as in the fluid dynamics, with the waiting time and departure times in fluid model replaced by their expectations in this non-fluid model. To be more specific, for the strategy tuple $\mathbf{F}_N$, cost of the $j$-th user of class-$i$ is defined as $C^{(i)}_{j,N}(\mathbf{F}_N)~=~\gamma^{(i)}\cdot \mathbb{E}[W^{(i)}_{j,N}]+(1-\gamma^{(i)})\cdot \mathbb{E}[\tau^{(i)}_{j,N}]$ (the expectation $\mathbb{E}$ is taken w.r.t. the joint distribution of the arrival times of users as well as the distributions of the service times at the two queues). We define Mixed Nash Equilibrium (MNE) of the $N$-th system as a strategy profile $\mathbf{F}_N$ such that, for any $i\in\{1,2\}$ and $j\in\{1,2,\dots,N\Lambda^{(i)}\}$, $j$-th user of class $i$ doesn't have any incentive to unilaterally deviate from her chosen distribution $F_{j,N}^{(i)}$, \textit{i.e.}, any such unilateral deviation will increase the expected cost $C^{(i)}_{j,N}(\cdot)$ of that user.  Let $\mathbf{F}_N$ be the MNE of the $N$-th system, assuming that a unique MNE exists for the $N$-th system. We define the limiting \textit{aggregate arrival profile}  $F^{(i)}(t)$ of class $i$ ($i=1,2$) as $F^{(i)}(t):=\lim_{N\to\infty}\frac{1}{N}\sum_{j=1}^{N\Lambda^{(i)}} F^{(i)}_{j,N}(Nt)$, with the time scaled by $N$ and space scaled by $1/N$. \cite{Juneja13} made a similar construction of a sequence of non-fluid queueing systems for the single queue fluid model. With a technically involved analysis, \cite{Juneja13} showed that the limiting aggregate arrival profile of this sequence of single queue non-fluid games is equivalent to the EAP of the single queue fluid model. Replicating the same analysis for HAS and HDS using the above model is a technically challenging task and is beyond the scope of the results presented in this paper. However, we conjecture that similar limiting results can be proved for HDS and HAS as well.

\section{Heterogeneous Departure Systems (HDS)}\label{sec_hetdepartures}

In this section, we consider the situation where the two classes arrive at the first queue and depart from different queues, as illustrated in Table \ref{tab:all_2queue_instances}. If $\mu_1\leq\mu_2$,  class 2 users arrive at queue 2 at a maximum rate of $\mu_1$ and as a result, queue 2 remains empty and the cost of class 2 is unaffected by the second queue.  Thus, if $\mu_1\leq\mu_2$, the instance becomes equivalent to both the groups arriving at a queue of capacity $\mu_1$. The problem is identical to the two-class, single queue case studied in \cite{Jain11}. Therefore, in subsequent discussion, we restrict ourselves to HDS with $\mu_1>\mu_2$. We further consider the case $\gamma^{(1)} \neq \gamma^{(2)}$ separately from $\gamma^{(1)}=\gamma^{(2)}$ since the latter displays different behaviour. 

\subsection{Unequal Preferences: $\gamma^{(1)}\neq\gamma^{(2)}$}\label{sec_inst1_uneqpref}

Theorem \ref{thm:HDS_uneqpref_brief} contains the main result for HDS with unequal preference and Theorem \ref{mainthm_inst1} is a detailed version of Theorem \ref{thm:HDS_uneqpref_brief}, where we disclose the explicit closed form of the EAP as a function of the instance parameters $\{(\Lambda^{(i)},\mu_i,\gamma^{(i)})\}_{i=1,2}$. 

\begin{theorem}\label{thm:HDS_uneqpref_brief}
  HDS with a unequal preferences, \textit{i.e.} $\gamma^{(1)}\neq \gamma^{(2)}$ and $\mu_1>\mu_2$ has a unique EAP. Furthermore, the EAP can have four different structures depending on the instance parameters.   
\end{theorem}

\noindent\textbf{Structural properties of EAP.}\hspace{0.05in} Before presenting the detailed result in Theorem \ref{mainthm_inst1}, We identify several structural properties that every EAP of HDS satisfies.  We will exploit these properties later to narrow our search of an EAP. Many of these properties are true even when the two groups have equal preference, \textit{i.e.}, $\gamma^{(1)}=\gamma^{(2)}$, and we use them later in Section \ref{sec_inst1_eqpref}.

As mentioned earlier, when $\mu_1 \leq \mu_2$, the second queue is not relevant to equilibrium behavior, and the two classes arrive in disjoint, contiguous intervals in the order of increasing $\gamma$'s. Lemma \ref{lem_threshold_behav_inst1} shows that the EAP has a similar arrival pattern as long as $\mu_1$ is below the threshold value of $\mu_2\cdot\max\left\{1,\frac{\gamma_2}{\gamma_1}\right\}$. Furthermore the threshold given by Lemma \ref{lem_threshold_behav_inst1} is tight, \textit{i.e.}, if $\mu_1>\mu_2\cdot\max\left\{1,\frac{\gamma_2}{\gamma_1}\right\}$, the two classes arrive over overlapping intervals. 

\begin{lemma}[\textbf{Threshold Behavior}]\label{lem_threshold_behav_inst1}
If $\gamma^{(1)}\neq\gamma^{(2)}$, in the EAP, the two classes arrive over contiguous intervals with disjoint interiors if and only if $\mu_1\leq\mu_2\cdot\max\left\{1,\frac{\gamma^{(2)}}{\gamma^{(1)}}\right\}$.
\end{lemma}

\noindent\textbf{Intuition behind the threshold behavior:}~One of the key observations which is essential for proving Lemma \ref{lem_threshold_behav_inst1} is:~if users of class 1 and 2 arrive together over some interval $\mathcal{I}$ of time in the EAP, then, in $\mathcal{I}$, classes 1 and 2 arrive together in queue 1 at a combined rate of $\mu_1\gamma_1$ and class 2 arrives in queue 1 at a rate of $\mu_2\gamma_2$ (we discuss them in detail in the proof sketch). Since the combined rate of the two classes must be strictly larger than the rate of class 2, we need $\mu_1\gamma_1>\mu_2\gamma_2$, which gives us $\mu_1>\mu_2\cdot \frac{\gamma_2}{\gamma_1}$. As a result, if the two classes arrive over overlapping intervals, then $\mu_1>\mu_2\cdot \frac{\gamma_2}{\gamma_1}$, which gives us the candidate threshold $\mu_2\cdot\frac{\gamma_2}{\gamma_1}$. Earlier, we observed that, if $\mu_1\leq\mu_2$, queue 2 becomes ineffective and the two classes arrive over disjoint intervals. Thus we obtain the overall threshold $\mu_2\cdot\max\left\{1,\frac{\gamma_2}{\gamma_1}\right\}$, such that, if $\mu_1$ is below this threshold, the two classes cannot arrive over overlapping intervals.  Lemma \ref{lem_threshold_behav_inst1} further shows that the obtained threshold is tight.

Below we provide a detailed sketch the proof of sufficiency of the condition stated in Lemma \ref{lem_threshold_behav_inst1}. Proving the other direction requires exploiting the behavior of the two queues in EAP and also the structure of $\mathcal{S}(F^{(1)})$ and $\mathcal{S}(F^{(2)})$. So, we prove that after stating Lemma \ref{lem_queue1_and_2_idle}. 

\textit{\textbf{Proof sketch of sufficiency in Lemma \ref{lem_threshold_behav_inst1}:}}\hspace{0.05in}~The detailed proof of sufficiency of the condition $\mu_1\leq\mu_2\cdot\max\left\{1,\frac{\gamma^{(2)}}{\gamma^{(1)}}\right\}$ is in Lemma \ref{lem_inst1_uneqpref_suff_cond_for_disj_arrival} and is similar to the proof sketch we  will present here, but instead uses some other supplementary lemmas and tools. First we show via contradiction that if $\mu_1\leq\mu_2\cdot\max\left\{1,\frac{\gamma^{(2)}}{\gamma^{(1)}}\right\}$, $\mathcal{S}(F^{(1)})$ and $(\mathcal{S}(F^{(2)}))^o$ cannot overlap. We later argue via Lemma \ref{lem_supp_are_intervals_inst1} (stated later) that in every EAP $\mathbf{F}=\{F^{(1)},F^{(2)}\}$, $\mathcal{S}(F^{(1)})$ and $\mathcal{S}(F^{(2)})$ are intervals. As a result, by the previous two statements, sufficiency of the stated condition will follow.

If $\mu_1\leq\mu_2\cdot\max\left\{1,\frac{\gamma^{(2)}}{\gamma^{(1)}}\right\}$, and $\mathcal{S}(F^{(1)}), (\mathcal{S}(F^{(2)}))^o$ overlap, we can find $t_1,t_2\in\mathcal{S}(F^{(1)})$ such that $[t_1,t_2]\subseteq\mathcal{S}(F^{(2)})$. Note that queue 1 must be engaged in $(t_1,t_2)$, otherwise the class 1 user arriving at $t_2$ is strictly better off arriving at the time when queue 1 is empty in $(t_1,t_2)$. Hence using \ref{eq:derv_of_tau}, $(C_{\mathbf{F}}^{(1)})^\prime(t)=\tau_1^\prime(t)-\gamma^{(1)}=\frac{(F^{(1)})^\prime(t)+(F^{(2)})^\prime(t)}{\mu_1}-\gamma^{(1)}$ in $[t_1,t_2]$. Now $t_1,t_2\in\mathcal{S}(F^{(1)})$ implies $C_{\mathbf{F}}^{(1)}(t_2)=C_{\mathbf{F}}^{(1)}(t_1)$, which by integrating $(C_{\mathbf{F}}^{(1)})^\prime(t)=\frac{(F^{(1)})^\prime(t)+(F^{(2)})^\prime(t)}{\mu_1}-\gamma^{(1)}$ in $[t_1,t_2]$ gives
\[F^{(1)}(t_2)+F^{(2)}(t_2)-(F^{(1)}(t_1)+F^{(2)}(t_1))=\mu_1\gamma^{(1)}\cdot(t_2-t_1).\]
Therefore $\mu_1\gamma^{(1)}\cdot(t_2-t_1)$ is the total mass of the two groups that have arrived in $[t_1,t_2]$. 

Since $[t_1,t_2]\subseteq\mathcal{S}(F^{(2)})$, we must have $(C_{\mathbf{F}}^{(2)})^\prime(t)=0$ and hence $(\tau_{\mathbf{F}}^{(2)})^\prime(t)=\gamma^{(2)}$ in $[t_1,t_2]$. As $\tau_{\mathbf{F}}^{(2)}(t)=\tau_2(\tau_1(t))$, we have $(\tau_{\mathbf{F}}^{(2)})^\prime(t)=\tau_2^\prime(\tau_1(t))\tau_1^\prime(t)$ in $[t_1,t_2]$, assuming the derivatives exist. Later we argue in Lemma \ref{lem_inst1_uneqpref_queue2_busy} that, queue 2 must remain engaged at $\tau_1(t)$ for every $t\in[t_1,t_2]$. Therefore, using (\ref{eq:derv_of_tau}), $\tau_2^\prime(\tau_1(t))=\frac{A_2^\prime(\tau_1(t))}{\mu_2}$, and hence $(\tau_{\mathbf{F}}^{(2)})^\prime(t)=\frac{A_2^\prime(\tau_1(t))\tau_1^\prime(t)}{\mu_2}$ in $[t_1,t_2]$. Since $A_2(\tau_1(t))=F^{(2)}(t)$, the previous statement implies $(\tau_{\mathbf{F}}^{(2)})^\prime(t)=\frac{(F^{(2)})^\prime(t)}{\mu_2}$ in $[t_1,t_2]$. Combining this with the observation that $(\tau_{\mathbf{F}}^{(2)})^\prime(t)=\gamma^{(2)}$ a.e. in $[t_1,t_2]$, we have $(F^{(2)})^\prime(t)=\mu_2\gamma^{(2)}$ a.e. in $[t_1,t_2]$. Therefore, 
\[F^{(2)}(t_2)-F^{(2)}(t_1)=\mu_2\gamma^{(2)}\cdot(t_2-t_1).\]

Since a positive mass of both the classes arrive in $[t_1,t_2]$, we must have $F^{(1)}(t_2)+F^{(2)}(t_2)-(F^{(1)}(t_1)+F^{(2)}(t_1))>F^{(2)}(t_2)-F^{(2)}(t_1)$, which implies $\mu_1\gamma^{(1)}>\mu_2\gamma^{(2)}$ by the previous arguments. This contradicts our assumption that $\mu_1\leq\mu_2\max\left\{1,\frac{\gamma^{(2)}}{\gamma^{(1)}}\right\}$. \hfill\qedsymbol

Lemma \ref{lem_supp_are_intervals_inst1} and \ref{lem_arrival_rates_inst1} imply EAPs can only be piece-wise linear arrival profiles with intervals as support. Proof of Lemma \ref{lem_supp_are_intervals_inst1} (in  \ref{appndx:inst1_uneqpref}) is done via contradiction by arguing if there is a gap, we can identify a user who can improve her cost by arriving at a different time. 

\begin{lemma}[{\bf Structure of supports}]\label{lem_supp_are_intervals_inst1}
    If $\mu_1>\mu_2$ and $\mathbf{F}=\{F^{(1)},F^{(2)}\}$ is an EAP, then $\mathcal{S}(F^{(1)})$ and $\mathcal{S}(F^{(1)})\cup \mathcal{S}(F^{(2)})$ must be intervals. Additionally, if $\gamma^{(1)}\neq\gamma^{(2)}$, $\mathcal{S}(F^{(2)})$ must also be an interval.
\end{lemma}
For a joint arrival profile $\mathbf{F}=\{F^{(1)},F^{(2)}\}$,  we define the quantities $T_{i,a}\overset{def.}{=}\inf \mathcal{S}(F^{(i)})$ and $T_{i,f}\overset{def.}{=}\sup \mathcal{S}(F^{(i)})$ for the two classes $i\in\{1,2\}$. For every EAP $\mathbf{F}$, $\mathcal{S}(F^{(1)}),~\mathcal{S}(F^{(2)})$ must be compact, as cost of the two classes over their support must be finite. As a result, the support  boundaries $T^{(1)}_a,T^{(1)}_f,T^{(2)}_a,T^{(2)}_f$ must be finite. By Lemma \ref{lem_supp_are_intervals_inst1}, we can further say, $\mathcal{S}(F^{(1)})=[T^{(1)}_a,T^{(1)}_f]$ and $\mathcal{S}(F^{(2)})=[T^{(2)}_a,T^{(2)}_f]$, such that union of the two intervals is also an interval. 

Lemma \ref{lem_inst1_uneqpref_queue2_busy} is about the behavior of queue 2 in equilibrium.
\begin{lemma}\label{lem_inst1_uneqpref_queue2_busy}
    If $\mu_1>\mu_2$ and $\gamma^{(1)}\neq\gamma^{(2)}$, then under EAP, for every $t\in(T^{(2)}_a,T^{(2)}_f)$, $\tau_1(t)$ belongs to the closure of the set $\{s~\vert~Q_2(s)>0\}$. 
\end{lemma}

\noindent\textbf{\textit{Proof sketch of Lemma \ref{lem_inst1_uneqpref_queue2_busy}:}}\hspace{0.05in}We consider two separate cases:\begin{enumerate}
    \item[\textbf{1.}] If $t\in(T^{(2)}_a,T^{(2)}_f)/(T^{(1)}_a,T^{(1)}_f)$, two possibilities are there. If queue 1 is engaged at $t$, class 2 users arrive from queue 1 to 2 at rate $\mu_1>\mu_2$ at $\tau_1(t)$, making queue 2 engaged. Otherwise, if queue 1 is empty at $t$, in EAP, the network cannot be empty at $t$ with a positive mass of class 2 users arriving in $(t,T^{(2)}_f)$. As a result, queue 2 must be engaged at $t=\tau_1(t)$.

    \item[\textbf{2.}] If $t\in(T^{(1)}_a,T^{(1)}_f)\cap(T^{(2)}_a,T^{(2)}_f)$, we assume a contradiction, \textit{i.e.}, queue 2 is empty in some neighbourhood of $\tau_1(t)$. As a result, by continuity of $\tau_1(\cdot)$, only queue 1 will be serving in some neighbourhood of $t$. Therefore $\tau_{\mathbf{F}}^{(1)}(\cdot)=\tau_{\mathbf{F}}^{(2)}(\cdot)=\tau_1(\cdot)$ in that neighbourhood. Also, for both the classes $i=1,2$ to be iso-cost in that neighbourhood, we need $(C_\mathbf{F}^{(i)})^\prime(s)=\tau_1^\prime(s)-\gamma^{(i)}=0$. This gives us $\tau_1^\prime(s)=\gamma^{(1)}=\gamma^{(2)}$, contradicting $\gamma^{(1)}\neq\gamma^{(2)}$.
\end{enumerate}
\hfill\qedsymbol

Lemma~\ref{lem_arrival_rates_inst1}  states conditions on the arrival rates necessary for the two classes to have constant cost over their support in any  EAP. 
\begin{lemma}[{\bf Rates of arrival}]\label{lem_arrival_rates_inst1}
    If  $\mu_1>\mu_2$, $\gamma^{(1)}\neq\gamma^{(2)}$, and $\mathbf{F}=\{F^{(1)},F^{(2)}\}$ is an EAP, the following properties must be true almost everywhere:
    \begin{align*}
        (F^{(1)})^\prime(t)&=\begin{cases}
            \mu_1\gamma^{(1)}~~&\text{if}~t\in \mathcal{S}(F^{(1)})/\mathcal{S}(F^{(2)}),\\
            \mu_1\gamma^{(1)}-\mu_2\gamma^{(2)}~~&\text{if}~t\in \mathcal{S}(F^{(1)})\cap \mathcal{S}(F^{(2)}),
        \end{cases}~~\text{and,}~(F^{(2)})^\prime(t)=\mu_2\gamma^{(2)}~~\text{if}~t\in \mathcal{S}(F^{(2)}),
    \end{align*}
    where from Lemma \ref{lem_threshold_behav_inst1}, $\mathcal{S}(F^{(1)})\cap \mathcal{S}(F^{(2)})$ has zero measure if $\mu_1\gamma^{(1)}\leq\mu_2\gamma^{(2)}$.
\end{lemma}

In the proof of Lemma \ref{lem_arrival_rates_inst1} (in \ref{appndx:inst1_uneqpref}), we first observe that both the classes $i=1,2$ have $C_{\mathbf{F}}^{(i)}(t)=\tau_{\mathbf{F}}^{(i)}(t)-\gamma^{(i)} t$ constant in $[T_{i,a},T^{(1)}_f]$, causing $(\tau_{\mathbf{F}}^{(i)})^\prime(t)=\gamma^{(i)}$. The rest of the proof relies on relating $(\tau_{\mathbf{F}}^{(i)})^\prime(t)$ with the arrival rates $(F^{(1)})^\prime(t)$ and $(F^{(2)})^\prime(t)$ of the two classes. Towards that, we use (\ref{eq:derv_of_tau}) and leverage the facts that queue 1 has positive waiting time in $(T^{(1)}_a,T^{(1)}_f)$, (otherwise, if $Q_1(t)=0$ at some $t\in(T^{(1)}_a,T^{(1)}_f)$, every class 1 user arriving in $(t,T^{(1)}_f]$ is strictly better off arriving at time $t$) and queue 2 stays engaged in $[\tau_1(T^{(2)}_a),\tau_1(T^{(2)}_f)]$  (by Lemma \ref{lem_inst1_uneqpref_queue2_busy}).   

Lemma~\ref{lem_queue1_and_2_idle} is about the state of the two queues at support boundaries $T^{(1)}_f$ and $T^{(2)}_f$. Proof of Lemma \ref{lem_queue1_and_2_idle} is done via a contradiction by showing that if the specified properties do not hold, a user can reduce her cost by arriving at a different time.

\begin{lemma}\label{lem_queue1_and_2_idle}
    If $\mu_1>\mu_2$ and $\gamma^{(1)}\neq\gamma^{(2)}$, then under EAP,
    queue length at the second queue is zero  at $\tau_1(T^{(2)}_f)$. If, in  addition, we have  $\mu_1\gamma^{(1)}>\mu_2\gamma^{(2)}$, queue length at the first queue equals zero at $T^{(1)}_f$.
\end{lemma}

\noindent\textbf{\textit{Proof of necessity of the condition in Lemma \ref{lem_threshold_behav_inst1}:}}\hspace{0.05in}We prove via contradiction that $\mu_1\leq\mu_2\cdot\max\left\{1,\frac{\gamma^{(2)}}{\gamma^{(1)}}\right\}$ is necessary for the intervals $[T^{(1)}_a,T^{(1)}_f]$ and $[T^{(2)}_a,T^{(2)}_f]$ to overlap. This forms the other direction of Lemma \ref{lem_threshold_behav_inst1}. On assuming a contradiction, we must have $\mu_1>\mu_2\cdot\max\left\{1,\frac{\gamma^{(2)}}{\gamma^{(1)}}\right\}$ and interiors of  $[T^{(1)}_a,T^{(1)}_f]$, $[T^{(2)}_a,T^{(2)}_f]$ are disjoint. Now by Lemma \ref{lem_supp_are_intervals_inst1}, union of $[T^{(1)}_a,T^{(1)}_f]$ and $[T^{(2)}_a,T^{(2)}_f]$ must be an interval. This leaves us two possibilities:
\begin{enumerate}
    \item[\textbf{1.}] If $T^{(1)}_f=T^{(2)}_a$, by Lemma \ref{lem_queue1_and_2_idle}, queue 1 will be empty at $T^{(1)}_f$, making the whole network empty at $T^{(1)}_f$. As a result, every class 2 user will be strictly better off arriving at $T^{(1)}_f$.
    \item[\textbf{2.}] If $T^{(2)}_f=T^{(1)}_a$, queue 1 must have a positive waiting time in $(T^{(2)}_a,T^{(2)}_f]$ and as a result, class 2 users arrive at queue 2 at rate $\mu_1>\mu_2$ in $[0,\tau_1(T^{(2)}_f)]$, causing $Q_2(\tau_1(T^{(2)}_f))=(\mu_1-\mu_2)\cdot\tau_1(T^{(2)}_f)>0$, contradicting Lemma \ref{lem_queue1_and_2_idle}.\hfill\qedsymbol
\end{enumerate}

\noindent\textbf{Specification of the EAP.}\hspace{0.05in}Theorem \ref{mainthm_inst1} specifies the unique EAP of this regime and we mention below the support boundaries of the unique EAP, which we will refer to later in Theorem \ref{mainthm_inst1}. 
\begin{enumerate}
    \item[\textbf{1.}] If $\gamma^{(1)}\leq\frac{\mu_2}{\mu_1}\gamma^{(2)}$, then
    \begin{align}\label{eq_bdary_inst1_uneqpref_case1}
        T^{(1)}_a&=-\left(\frac{1}{\gamma^{(1)}}-1\right)\frac{\Lambda^{(1)}}{\mu_1}-\left(\frac{1}{\gamma^{(2)}}-1\right)\frac{\Lambda^{(2)}}{\mu_2},~T^{(1)}_f=T^{(2)}_a=\frac{\Lambda^{(1)}}{\mu_1}-\left(\frac{1}{\gamma^{(2)}}-1\right)\frac{\Lambda^{(2)}}{\mu_2},\nonumber\\
        \text{and}~T^{(2)}_f&=\frac{\Lambda^{(1)}}{\mu_1}+\frac{\Lambda^{(2)}}{\mu_2}.
    \end{align}
    
    \item[\textbf{2.}] If $\gamma^{(1)}>\frac{\mu_2}{\mu_1}\gamma^{(2)}$ and $\Lambda^{(1)}\geq\min\left\{\frac{1-\gamma^{(2)}}{1-\gamma^{(1)}}\left(\frac{\mu_1\gamma^{(1)}}{\mu_2\gamma^{(2)}}-1\right),\frac{\mu_1}{\mu_2}-1\right\}\cdot\Lambda^{(2)}$, then
    \begin{enumerate}
        \item[\textbf{a)}] if $\gamma^{(1)}<\gamma^{(2)}$, then
        \begin{align}\label{eq_bdary_inst1_uneqpref_case2a}
            \vspace{-0.05in}T^{(1)}_a&=-\frac{1-\gamma^{(1)}}{\mu_1\gamma^{(1)}}\left[\Lambda^{(1)}+\frac{1-\gamma^{(2)}}{1-\gamma^{(1)}}\Lambda^{(2)}\right],~T^{(1)}_f=\frac{1}{\mu_1}\left[\Lambda^{(1)}+\frac{1-\gamma^{(2)}}{1-\gamma^{(1)}}\Lambda^{(2)}\right],\nonumber\\
            T^{(2)}_a&=\frac{1}{\mu_1}\left[\Lambda^{(1)}-\frac{\mu_1-\mu_2\gamma^{(2)}}{\mu_2\gamma^{(2)}}\frac{1-\gamma^{(2)}}{1-\gamma^{(1)}}\Lambda^{(2)}\right],\nonumber\\
            \text{and}~T^{(2)}_f&=\frac{1}{\mu_1}\left[\Lambda^{(1)}+\frac{\mu_1(\gamma^{(2)}-\gamma^{(1)})+\mu_2\gamma^{(2)}\cdot(1-\gamma^{(2)})}{\mu_2\gamma^{(2)}\cdot(1-\gamma^{(1)})}\Lambda^{(2)}\right].
        \end{align}    

        \item[\textbf{b)}] if $\gamma^{(1)}>\gamma^{(2)}$, then 
        \begin{align}\label{eq_bdary_inst1_uneqpref_case3a}
            \vspace{-0.05in}T^{(1)}_a&=\frac{\gamma^{(1)}-\gamma^{(2)}}{\gamma^{(1)}}\frac{\Lambda^{(2)}}{\mu_1\gamma^{(1)}-\mu_2\gamma^{(2)}}-\frac{1-\gamma^{(1)}}{\gamma^{(1)}}\frac{\Lambda^{(1)}+\Lambda^{(2)}}{\mu_1},~T^{(1)}_f=\frac{\Lambda^{(1)}+\Lambda^{(2)}}{\mu_1},\nonumber\\
            T^{(2)}_a&=-\frac{1-\gamma^{(1)}}{\gamma^{(1)}}\frac{\Lambda^{(1)}}{\mu_1}-\left(\frac{\gamma^{(1)}}{\gamma^{(2)}}+(1-\gamma^{(1)})\frac{\mu_2}{\mu_1}-1\right)\frac{\Lambda^{(2)}}{\mu_2\gamma^{(1)}},\nonumber\\
            \text{and}~T^{(2)}_f&=-\frac{1-\gamma^{(1)}}{\gamma^{(1)}}\frac{\Lambda^{(1)}}{\mu_1}+\left(1-(1-\gamma^{(1)})\frac{\mu_2}{\mu_1}\right)\frac{\Lambda^{(2)}}{\mu_2\gamma^{(1)}}.
        \end{align}
    \end{enumerate}
    
    \item[\textbf{3.}] If $\gamma^{(1)}>\frac{\mu_2}{\mu_1}\gamma^{(2)}$ and $\Lambda^{(1)}<\min\left\{\frac{1-\gamma^{(2)}}{1-\gamma^{(1)}}\left(\frac{\mu_1\gamma^{(1)}}{\mu_2\gamma^{(2)}}-1\right),\frac{\mu_1}{\mu_2}-1\right\}\cdot\Lambda^{(2)}$, then
    \begin{align}\label{eq_bdary_inst1_uneqpref_case2b}
         \hspace{-0.3in}~T^{(1)}_a&=\frac{1-\gamma^{(2)}}{\mu_1-\mu_2\gamma^{(2)}}\left[\Lambda^{(2)}-\frac{(1-\gamma^{(1)})\mu_1}{(1-\gamma^{(2)})(\mu_1\gamma^{(1)}-\mu_2\gamma^{(2)})}\Lambda^{(1)}\right],~T^{(1)}_f=\frac{\Lambda^{(1)}+(1-\gamma^{(2)})\Lambda^{(2)}}{\mu_1-\mu_2\gamma^{(2)}}\nonumber\\ T^{(2)}_a&=-\frac{1-\gamma^{(2)}}{\gamma^{(2)}}\frac{\Lambda^{(2)}}{\mu_2},~\text{and}~T^{(2)}_f=\frac{\Lambda^{(2)}}{\mu_2}.
    \end{align}
\end{enumerate}

\begin{theorem}\label{mainthm_inst1}
If $\mu_1>\mu_2$ and $\gamma^{(1)}\neq\gamma^{(2)}$, HDS has a unique EAP. The arrival rates in the EAP are given below along with the support boundaries: 
\begin{enumerate}
    \item[\textbf{1.}] If $\gamma^{(1)}\leq\frac{\mu_2}{\mu_1}\gamma^{(2)}$, $(F^{(1)})^\prime(t)=\mu_1\gamma^{(1)}$ for $t\in[T^{(1)}_a,T^{(1)}_f]$ and $(F^{(2)})^\prime(t)=\mu_2\gamma^{(2)}$ for $t\in[T^{(2)}_a,T^{(2)}_f]$ where $T^{(1)}_a,T^{(1)}_f,T^{(2)}_a$ and $T^{(2)}_f$ are given in (\ref{eq_bdary_inst1_uneqpref_case1}).
    
    \item[\textbf{2.}] If $\frac{\mu_2}{\mu_1}\gamma^{(2)}<\gamma^{(1)}<\gamma^{(2)}$,
    \begin{enumerate}
        \item[\textbf{2a.}] when $\Lambda^{(1)}\geq \frac{1-\gamma^{(2)}}{1-\gamma^{(1)}}\left(\frac{\mu_1\gamma^{(1)}}{\mu_2\gamma^{(2)}}-1\right)\Lambda^{(2)}$, 
            \begin{align*}
                (F^{(1)})^\prime(t)&=\begin{cases}
                    \mu_1\gamma^{(1)}~~&\text{if}~t\in[T^{(1)}_a,T^{(2)}_a],\\
                    \mu_1\gamma^{(1)}-\mu_2\gamma^{(2)}~~&\text{if}~t\in[T^{(2)}_a,T^{(1)}_f],
                \end{cases}~~\text{and}~(F^{(2)})^\prime(t)=\mu_2\gamma^{(2)}~~\text{if}~t\in[T^{(2)}_a,T^{(2)}_f] 
            \end{align*}
            where $T^{(1)}_a,T^{(1)}_f,T^{(2)}_a$ and $T^{(2)}_f$ are given in (\ref{eq_bdary_inst1_uneqpref_case2a}).

            \item[\textbf{2b.}] when $\Lambda^{(1)}<\frac{1-\gamma^{(2)}}{1-\gamma^{(1)}}\left(\frac{\mu_1\gamma^{(1)}}{\mu_2\gamma^{(2)}}-1\right)\Lambda^{(2)}$, $(F^{(1)})^\prime(t)=\mu_1\gamma^{(1)}-\mu_2\gamma^{(2)}$ for $t\in[T^{(1)}_a,T^{(1)}_f]$, and $(F^{(2)})^\prime(t)=\mu_2\gamma^{(2)}$ for $t\in[T^{(2)}_a,T^{(2)}_f]$, where $T^{(1)}_a,T^{(1)}_f,T^{(2)}_a$ and $T^{(2)}_f$ are given in (\ref{eq_bdary_inst1_uneqpref_case2b}).
    \end{enumerate}

    \item[\textbf{3.}] If $\gamma^{(2)}<\gamma^{(1)}$,
    \begin{enumerate}
        \item[\textbf{3a.}] when $\Lambda^{(1)}\geq \left(\frac{\mu_1}{\mu_2}-1\right)\Lambda^{(2)}$,
            \begin{align*}
                (F^{(1)})^\prime(t)&=\begin{cases}
                    \mu_1\gamma^{(1)}-\mu_2\gamma^{(2)}~~&\text{if}~t\in[T^{(1)}_a,T^{(2)}_f],\\
                    \mu_1\gamma^{(1)}~~&\text{if}~t\in[T^{(2)}_f,T^{(1)}_f],
                \end{cases}~~\text{and}~(F^{(2)})^\prime(t)=\mu_2\gamma^{(2)}~~\text{if}~t\in[T^{(2)}_a,T^{(2)}_f]
            \end{align*}
            where $T^{(1)}_a,T^{(1)}_f,T^{(2)}_a$ and $T^{(2)}_f$ are given in (\ref{eq_bdary_inst1_uneqpref_case3a}).
            
        \item[\textbf{3b.}] when $\Lambda^{(1)}<\left(\frac{\mu_1}{\mu_2}-1\right)\Lambda^{(2)}$, the EAP has a closed form same as case 2b.
    \end{enumerate}           
\end{enumerate}
\end{theorem}

\begin{remark}
    \emph{Figures \ref{fig:inst1case1}, \ref{fig:inst1case21}, \ref{fig:inst1case22}, and \ref{fig:inst1case31}  show the illustrative EAPs and resulting queue length processes, respectively, for cases 1, 2a, 2b and 3a of Theorem \ref{mainthm_inst1}. EAP structure in case 3b of Theorem \ref{mainthm_inst1} is similar to case 2b and is illustrated in Figure \ref{fig:inst1case22}. The structure of the queue length processes in certain regimes depend on the support boundaries and may vary accordingly. In the figures referred above, we have illustrated only one possible structure of the queue length process. Moreover, we have mentioned the conditions to be satisfied by the support boundaries for the attainment of that structure in the caption. In all these figures, \textbf{\textcolor{red}{red}} and \textbf{\textcolor{blue}{blue}}, respectively, represent class 1 and 2 populations, and the \textbf{black} dashed line represents the total mass of the two populations waiting in queue 1.}
\end{remark}

\begin{remark}
    \emph{The illustrative EAPs referred to in the previous remark also displays the threshold behavior stated in Lemma \ref{lem_threshold_behav_inst1}. Since  $\mu_1\leq\mu_2\cdot\max\left\{1,\frac{\gamma^{(2)}}{\gamma^{(1)}}\right\}$ in case 1 of Theorem \ref{mainthm_inst1}, we can see the support intervals of the two classes to have disjoint interiors in Figure \ref{fig:inst1case1}. On the other hand, since cases 2 and 3 of Theorem \ref{mainthm_inst1} have $\mu_1>\mu_2\cdot\left\{1,\frac{\gamma^{(2)}}{\gamma^{(1)}}\right\}$, two classes arrive over overlapping intervals, as can be seen in Figure \ref{fig:inst1case21}, \ref{fig:inst1case22}, and \ref{fig:inst1case31}.}
\end{remark}

The following lemma specifies the arrival order of the two classes in EAP and is necessary for proving Theorem \ref{mainthm_inst1}.
\begin{lemma}\label{lem_sign_inst1_uneqpref}
    If $\mu_1>\mu_2$ and $\gamma^{(1)}\neq\gamma^{(2)}$, the support boundaries of EAP of HDS satisfy:
    \begin{enumerate}
        \item[\textbf{1.}] $T^{(1)}_f=T^{(2)}_a$ if $\gamma^{(1)}\leq\frac{\mu_2}{\mu_1}\gamma^{(2)}$,
        \item[\textbf{2.}] $T^{(1)}_f\leq T^{(2)}_f$ if $\frac{\mu_2}{\mu_1}\gamma^{(2)}<\gamma^{(1)}<\gamma^{(2)}$, and 
        \item[\textbf{3.}] $T^{(1)}_a>T^{(2)}_a$ if $\gamma^{(1)}>\gamma^{(2)}$.
    \end{enumerate}
\end{lemma}

Proof of Lemma \ref{lem_sign_inst1_uneqpref} (in  \ref{appndx:inst1_uneqpref}) is done via contradiction by showing that if the stated property doesn't hold, a user can improve her cost by arriving at a different time.

\begin{remark}
    \emph{The arrival orders anticipated by Lemma \ref{lem_sign_inst1_uneqpref} can be observed in the illustrative EAPs referred earlier.  By Lemma \ref{lem_sign_inst1_uneqpref}, in case 2 ($\gamma^{(2)}>\gamma^{(1)}>\frac{\mu_2}{\mu_1}\gamma^{(2)}$) of Theorem \ref{mainthm_inst1}, class 2 population finishes arrival after class 1, as can be observed in Figure \ref{fig:inst1case21} and \ref{fig:inst1case22}. Similarly for case 3 ($\gamma^{(1)}>\gamma^{(2)}$) of Theorem \ref{mainthm_inst1}, class 1 population starts arrival after class 2, as can be observed in Figure \ref{fig:inst1case31} and \ref{fig:inst1case22}. The class arriving first in case 2 and finishing late in case 3 is the one having a population size significantly larger among the two classes. In case 2 (or case 3): \textbf{1)}~when $\Lambda^{(1)}>c\Lambda^{(2)}$, class 1 starts arriving before class 2 (or finishes arriving after class 2),~\textbf{2)}~when $\Lambda^{(1)}=c\Lambda^{(2)}$, both classes start arriving from the same time (or finishes arriving at the same time),~\textbf{3)}~when $\Lambda^{(1)}<c\Lambda^{(2)}$, class 1 starts arriving after class 2 starts (or finishes arriving before class 2), where $c=\min\left\{\frac{1-\gamma^{(2)}}{1-\gamma^{(1)}}\left(\frac{\mu_1\gamma^{(1)}}{\mu_2\gamma^{(2)}}-1\right),\frac{\mu_1}{\mu_2}-1\right\}$ for both case 2 and 3. }    
\end{remark}

\noindent\textit{\textbf{Key steps in the proof of Theorem \ref{mainthm_inst1}\hspace{0.05in}(details in \ref{appndx:thm_inst1_uneqpref}):}}~~By Lemma \ref{lem_supp_are_intervals_inst1} we only consider candidates $\mathbf{F}=\{F^{(1)},F^{(2)}\}$ which are absolutely continuous with supports of the form $\mathcal{S}(F^{(1)})=[T^{(1)}_a,T^{(1)}_f]$, $\mathcal{S}(F^{(2)})=[T^{(2)}_a,T^{(2)}_f]$, such that their union is an interval and the arrival rates $(F^{(1)})^\prime(\cdot),~(F^{(2)})^\prime(\cdot)$ satisfy the property in Lemma \ref{lem_arrival_rates_inst1}. We eliminate candidates not satisfying structural properties desired of an EAP and will be eventually left with only one candidate, which is the only candidate that can qualify as an EAP. We then prove that the only remaining candidate is indeed an EAP. Thus, existence and uniqueness of  EAP follows. 
 
We follow this agenda for the three cases:\hspace{0.05in}\textbf{1)}~$\gamma^{(1)}\leq\frac{\mu_2}{\mu_1}\gamma^{(2)}$,~\textbf{2)}~$\frac{\mu_2}{\mu_1}\gamma^{(2)}<\gamma^{(1)}<\gamma^{(2)}$,~and~\textbf{3)}~$\gamma^{(2)}<\gamma^{(1)}$. The final candidate we get in these cases have their arrival rates and support boundaries same as the joint arrival profiles mentioned in Theorem \ref{mainthm_inst1} under the respective cases. 

\textbf{Case 1}\hspace{0.05in}$\mathbf{\gamma^{(1)}\leq\frac{\mu_2}{\mu_1}\gamma^{(2)}}$:\hspace{0.05in}By Lemma \ref{lem_threshold_behav_inst1}, every EAP must have $[T^{(1)}_a,T^{(1)}_f]$ and $[T^{(2)}_a,T^{(2)}_f]$ disjoint.  Lemma \ref{lem:bdaryinst1case1} helps us substantially narrow our search for an EAP.

\begin{lemma}\label{lem:bdaryinst1case1}
    Under case 1 $\gamma^{(1)}\leq\frac{\mu_2}{\mu_1}\gamma^{(2)}$, every EAP has (\ref{eq_bdary_inst1_uneqpref_case1}) as support boundaries $T^{(1)}_a,T^{(1)}_f,T^{(2)}_a,T^{(2)}_f$.
\end{lemma}

\textbf{Proof Sketch:}\hspace{0.05in}By Lemma \ref{lem_sign_inst1_uneqpref}, we must have $T^{(1)}_f=T^{(2)}_a$. By Lemma \ref{lem_arrival_rates_inst1}, $(F^{(1)})^\prime(t)=\mu_1\gamma^{(1)}$ in $[T^{(1)}_a,T^{(1)}_f]$ and $(F^{(2)})^\prime(t)=\mu_2\gamma^{(2)}$ in $[T^{(1)}_f,T^{(2)}_f]$. Therefore, $T^{(1)}_a,T^{(1)}_f,T^{(2)}_f$ must satisfy:
\begin{align}\label{eq:inst1_uneqpref_case1_prfsktch1}
    T^{(1)}_f&=T^{(1)}_a+\frac{\Lambda^{(1)}}{\mu_1\gamma^{(1)}}~\text{and}~T^{(2)}_f=T^{(1)}_f+\frac{\Lambda^{(2)}}{\mu_2\gamma^{(2)}}.
\end{align} 
Now, queue 2 must be empty at $\tau_1(T^{(2)}_f)$ (by Lemma \ref{lem_queue1_and_2_idle}) and must stay engaged in the image of $[T^{(1)}_f,T^{(2)}_f]$ under $\tau_1(\cdot)$ (by Lemma \ref{lem_inst1_uneqpref_queue2_busy}), which is $[\tau_1(T^{(1)}_f),\tau_1(T^{(2)}_f)]$. Since, the first and last class 2 users arrive at queue 2, respectively at times $\tau_1(T^{(1)}_f)$ and $\tau_1(T^{(2)}_f)$, the previous statement implies,
\begin{align}\label{eq:inst1_uneqpref_case1_prfsktch2}
    \mu_2\cdot(\tau_1(T^{(2)}_f)-\tau_1(T^{(1)}_f))&=\Lambda^{(2)}   
\end{align}
For queue 2 to be empty at $\tau_1(T^{(2)}_f)$, queue 1 must also be empty at $T^{(2)}_f$. Otherwise if queue 1 has a positive waiting time at $T^{(2)}_f$, class 2 users will be arriving at queue 2 from queue 1 at rate $\mu_1>\mu_2$ in $[\tau_1(T^{(2)}_f-\delta),\tau_1(T^{(2)}_f)]$, where $\delta>0$ picked sufficiently small such that queue 1 has positive waiting time in $[T^{(2)}_f-\delta,T^{(2)}_f]$. By (\ref{eq:derv_of_tau}), $\tau_1^\prime(t)=\frac{(F^{(2)})^\prime(t)}{\mu_1}=\frac{\mu_2\gamma^{(2)}}{\mu_1}>0$ in $[T^{(2)}_f-\delta,T^{(2)}_f]$, giving us $\tau_1(T^{(2)}_f)>\tau_1(T^{(2)}_f-\delta)$. As a result, $Q_2(\tau_1(T^{(2)}_f))\geq(\mu_1-\mu_2)\cdot(\tau_1(T^{(2)}_f)-\tau_1(T^{(2)}_f-\delta))>0$, which contradicts Lemma \ref{lem_queue1_and_2_idle}. With queue 1 empty at $T^{(2)}_f$, we have $\tau_1(T^{(2)}_f)=T^{(2)}_f$.\vspace{0.05in}

Note that queue 1 must have a positive waiting time in $(T^{(1)}_a,T^{(1)}_f]$, otherwise, every class 2 user arriving after $T^{(1)}_f$ is strictly better off arriving at the time queue 1 is empty. Applying (\ref{eq:derv_of_tau}), we have $\tau_1(T^{(1)}_f)-\tau_1(T^{(1)}_a)=\frac{F^{(1)}(T^{(1)}_f)-F^{(1)}(T^{(1)}_a)}{\mu_1}=\frac{\Lambda^{(1)}}{\mu_1}$. Since the network cannot be empty at time zero, we have $T^{(1)}_a<0$. This gives us $\tau_1(T^{(1)}_a)=0$ and hence, $\tau_1(T^{(1)}_f)=\frac{\Lambda^{(1)}}{\mu_1}$. Putting $\tau_1(T^{(2)}_f)=T^{(2)}_f$ and $\tau_1(T^{(1)}_f)=\frac{\Lambda^{(1)}}{\mu_1}$ in (\ref{eq:inst1_uneqpref_case1_prfsktch2}), we get $T^{(2)}_f=\frac{\Lambda^{(1)}}{\mu_1}+\frac{\Lambda^{(2)}}{\mu_2}$. Using this in (\ref{eq:inst1_uneqpref_case1_prfsktch1}) and by the fact $T^{(2)}_a=T^{(1)}_f$, we get the values of $T^{(1)}_a,T^{(1)}_f,T^{(2)}_a,T^{(2)}_f$ mentioned in (\ref{eq_bdary_inst1_uneqpref_case1}). 

\hfill\qedsymbol

\bigskip 

The only candidate having arrival rates satisfying Lemma \ref{lem_arrival_rates_inst1} and support boundaries from Lemma \ref{lem:bdaryinst1case1} is the joint arrival profile under case 1 of Theorem \ref{mainthm_inst1}. Proving that this unique remaining candidate is an EAP requires analyzing the behavior of the two queues induced by this candidate.s Details of it are in  \ref{appndx:inst1_uneqpref}.\vspace{0.05in}

\begin{figure}[h]
    \centering
    \includegraphics[width=12cm]{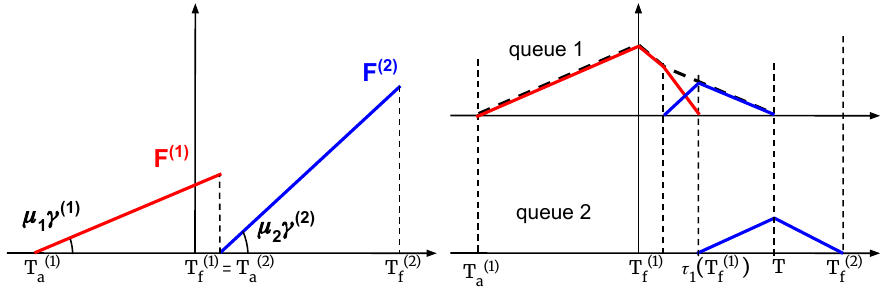}
    \caption{Illustrative EAP (left) and resulting queue length process (right) of HDS with $\mu_1>\mu_2$ under case 1: $\gamma^{(1)}\leq\frac{\mu_2}{\mu_1}\cdot\gamma^{(2)}$ and the condition $T^{(1)}_f>0$ in (\ref{eq_bdary_inst1_uneqpref_case1}). $T\in(T^{(1)}_f,T^{(2)}_f)$ is the time at which queue 1 empties after time zero. In queue 1, since class 1 users stop arriving after $T^{(1)}_f$, class 1 waiting mass decreases faster in $[T^{(1)}_f,\tau_1(T^{(1)}_f)]$ than in $[0,T^{(1)}_f]$.}
    \label{fig:inst1case1}
\end{figure}

\textbf{Case 2}\hspace{0.05in}$\frac{\mu_2}{\mu_1}\cdot\gamma^{(2)}<\gamma^{(1)}<\gamma^{(2)}$:\hspace{0.05in}By Lemma \ref{lem_threshold_behav_inst1}, $[T^{(1)}_a,T^{(1)}_f]$ and $[T^{(2)}_a,T^{(2)}_f]$ must overlap, and by Lemma \ref{lem_sign_inst1_uneqpref}, class 2 shouldn't finish arriving before class 1, \textit{i.e.}, $T^{(1)}_f\leq T^{(2)}_f$. Therefore, every EAP must have its support boundaries $T^{(1)}_a,T^{(1)}_f,T^{(2)}_a,T^{(2)}_f$ ordered by one of the following two ways: 

\begin{enumerate}
    \item[Type I:]~$T^{(1)}_a\leq T^{(2)}_a<T^{(1)}_f\leq T^{(2)}_f$ with $T^{(1)}_a<0$, and,
    \item[Type II:]~$T^{(2)}_a<T^{(1)}_a<T^{(1)}_f\leq T^{(2)}_f$ with $T^{(2)}_a<0$. 
\end{enumerate}

Note that, we need $\min\{T^{(1)}_a,T^{(2)}_a\}<0$, otherwise, any user will be strictly better off arriving at time zero when the network is empty. The following lemma now specifies the necessary and sufficient conditions for existence of an EAP under these two types.  Existence of unique EAP in case 2 of Theorem \ref{mainthm_inst1} follows trivially from this lemma.  

\begin{lemma}\label{lem:inst1case2}
    If $\gamma^{(2)}>\gamma^{(1)}>\frac{\mu_2}{\mu_1}\gamma^{(2)}$, the following statements are true:\vspace{-0.05in}\begin{itemize}
        \item There exists an EAP under Type I if and only if $\Lambda^{(1)}\geq\frac{1-\gamma^{(2)}}{1-\gamma^{(1)}}\left(\frac{\mu_1\gamma^{(1)}}{\mu_2\gamma^{(2)}}-1\right)\Lambda^{(2)}$ and if it exists, it will be unique with a closed form same as the joint arrival profile under case 2a of Theorem \ref{mainthm_inst1}.
        \item There exists an EAP under Type II if and only if $\Lambda^{(1)}<\frac{1-\gamma^{(2)}}{1-\gamma^{(1)}}\left(\frac{\mu_1\gamma^{(1)}}{\mu_2\gamma^{(2)}}-1\right)\Lambda^{(2)}$ and if it exists, it will be unique with a closed form same as the joint arrival profile under case 2b of Theorem \ref{mainthm_inst1}.
    \end{itemize}
\end{lemma}

\noindent\textit{\textbf{Proof Sketch of Lemma \ref{lem:inst1case2}:}}\hspace{0.05in}\textbf{Identifying the support boundaries:} For every EAP under Type I, we identify the following system of equations to be satisfied by its support boundaries $T^{(1)}_a,T^{(1)}_f,T^{(2)}_a,T^{(2)}_f$:\\
\textbf{1)} By Lemma \ref{lem_arrival_rates_inst1},  $(F^{(2)})^\prime(t)=\mu_2\gamma^{(2)}$ in $[T^{(2)}_a,T^{(2)}_f]$, giving us: $T^{(2)}_f=T^{(2)}_a+\frac{\Lambda^{(2)}}{\mu_2\gamma^{(2)}}$. \vspace{0.025in}

\noindent\textbf{2)} Since $T^{(1)}_a<0$, queue 1 starts serving at time zero and must have positive waiting time in $(0,T^{(1)}_f)$. By Lemma \ref{lem_queue1_and_2_idle}, queue 1 must be empty at time $T^{(1)}_f$. Therefore, $F^{(1)}(T^{(1)}_f)+F^{(2)}(T^{(1)}_f)=\mu_1 T^{(1)}_f$. Now $F^{(1)}(T^{(1)}_f)=\Lambda^{(1)}$ and $F^{(2)}(T^{(1)}_f)=\mu_2\gamma^{(2)}\cdot(T^{(1)}_f-T^{(2)}_a)$ (by Lemma \ref{lem_arrival_rates_inst1}). Plugging this in, we get: $\Lambda^{(1)}+\mu_2\gamma^{(2)}\cdot(T^{(1)}_f-T^{(2)}_a)=\mu_1 T^{(1)}_f$. \vspace{0.025in}

\noindent\textbf{3)} By definition of EAP, $C_{\mathbf{F}}^{(1)}(T^{(1)}_a)=C_{\mathbf{F}}^{(1)}(T^{(1)}_f)$. Queue 1 empties at $T^{(1)}_f$ (by Lemma \ref{lem_queue1_and_2_idle}), giving us $C_{\mathbf{F}}^{(1)}(T^{(1)}_f)=(1-\gamma^{(1)})T^{(1)}_f$. The first class 1 user arriving at time $T^{(1)}_a$ gets served by queue 1 at time zero and waits for $-T^{(1)}_a$ time, giving us $C_{\mathbf{F}}^{(1)}(T^{(1)}_a)=-\gamma^{(1)} T^{(1)}_a$. Equating these two costs, we get: $T^{(1)}_a=-\left(\frac{1}{\gamma^{(1)}}-1\right)T^{(1)}_f$.\vspace{0.025in}

\noindent\textbf{4)} By Lemma \ref{lem_inst1_uneqpref_queue2_busy}, queue 2 stays engaged in $[\tau_1(T^{(2)}_a),\tau_1(T^{(2)}_f)]$ and by Lemma \ref{lem_queue1_and_2_idle}, queue 2 empties at $\tau_1(T^{(2)}_f)$. In this way, queue 2 serves the class 2 population in $[\tau_1(T^{(2)}_a),\tau_1(T^{(2)}_f)]$, giving us $\mu_2\cdot(\tau_1(T^{(2)}_f)-\tau_1(T^{(2)}_a))=\Lambda^{(2)}$. Queue 1 empties at $T^{(1)}_f$ (by Lemma \ref{lem_queue1_and_2_idle}) and after $T^{(1)}_f$, since class 2 users arrive at rate $\mu_2\gamma^{(2)}<\mu_1\gamma^{(1)}<\mu_1$, queue 1 remains empty at $T^{(2)}_f$, giving us $\tau_1(T^{(2)}_f)=T^{(2)}_f$. Since queue 1 has positive waiting time in $(T^{(1)}_a,T^{(2)}_a]$, we have $\tau_1(T^{(2)}_a)-\tau_1(T^{(1)}_a)=\frac{F^{(1)}(T^{(2)}_a)-F^{(1)}(T^{(1)}_a)}{\mu_1}=\gamma^{(1)}\cdot(T^{(2)}_a-T^{(1)}_a)$~(by (\ref{eq:derv_of_tau}) and Lemma \ref{lem_arrival_rates_inst1}), which upon putting $\tau_1(T^{(1)}_a)=0$ gives us $\tau_1(T^{(2)}_a)=\gamma^{(1)}\cdot(T^{(2)}_a-T^{(1)}_a)$. Placing the expressions obtained for $\tau_1(T^{(2)}_a)$ and $\tau_1(T^{(2)}_f)$ into $\mu_2(\tau_1(T^{(2)}_f)-\tau_1(T^{(2)}_a))=\Lambda^{(2)}$, we get, $\mu_2(T^{(2)}_f-\gamma^{(1)}\cdot(T^{(2)}_a-T^{(1)}_a))=\Lambda^{(2)}$.\vspace{0.025in}

\noindent\textbf{Getting the necessary condition:} Solution of the identified system of equations is (\ref{eq_bdary_inst1_uneqpref_case2a}) and therefore, every EAP under Type I must have  (\ref{eq_bdary_inst1_uneqpref_case2a}) as support boundaries. Now (\ref{eq_bdary_inst1_uneqpref_case2a}) must be ordered by $T^{(1)}_a\leq T^{(2)}_a<T^{(1)}_f\leq T^{(2)}_f$ to represent an EAP of Type I. Note that (\ref{eq_bdary_inst1_uneqpref_case2a}) satisfies $T^{(2)}_a-T^{(1)}_a=\frac{1}{\mu_1\gamma^{(1)}}\left[\Lambda^{(1)}-\frac{1-\gamma^{(2)}}{1-\gamma^{(1)}}\left(\frac{\mu_1\gamma^{(1)}}{\mu_2\gamma^{(2)}}-1\right)\Lambda^{(2)}\right]$. As a result, for  $T^{(2)}_a\geq T^{(1)}_a$, we need  $\Lambda^{(1)}\geq\frac{1-\gamma^{(2)}}{1-\gamma^{(1)}}\left(\frac{\mu_1\gamma^{(1)}}{\mu_2\gamma^{(2)}}-1\right)\Lambda^{(2)}$ and therefore this is a necessary condition for existence of an EAP under Type I.\vspace{0.025in} 

\noindent\textbf{Proving sufficiency of the obtained necessary condition:} Now, if $\Lambda^{(1)}\geq\frac{1-\gamma^{(2)}}{1-\gamma^{(1)}}\left(\frac{\mu_1\gamma^{(1)}}{\mu_2\gamma^{(2)}}-1\right)\Lambda^{(2)}$, it is easy to verify that (\ref{eq_bdary_inst1_uneqpref_case2a}) satisfies $T^{(1)}_a\leq T^{(2)}_a<T^{(1)}_f\leq T^{(2)}_f$. Therefore, with the necessary condition satisfied and upon plugging in arrival rates using Lemma \ref{lem_arrival_rates_inst1}, we get a candidate having arrival rates and support boundaries same as the joint arrival profile mentioned under case 2a of Theorem \ref{mainthm_inst1}. This candidate satisfies  $F^{(1)}(T^{(1)}_f)=\Lambda^{(1)},~F^{(2)}(T^{(2)}_f)=\Lambda^{(2)}$ and will be the only candidate under Type I to qualify as an EAP. Proving that this candidate is an EAP requires analyzing the behavior of the two queues. Details of this argument is in  \ref{appndx:inst1_uneqpref}. Therefore, it follows that the obtained necessary condition is also sufficient for existence of an EAP under Type I, and once it is satisfied, there is a unique EAP under Type I which has a closed form same as the one mentioned under case 2a of Theorem \ref{mainthm_inst1}. Thus the first statement of the lemma follows.  

The second statement is proved via an argument similar to the one used for proving the first statement via identifying a system of equations satisfied by the support boundaries and getting the necessary condition by imposing $T^{(2)}_a<T^{(1)}_a$ on the solution to that system. Details of the argument in  \ref{appndx:thm_inst1_uneqpref}.\hfill\qedsymbol

\noindent\textbf{Case 3: $\gamma^{(1)}>\gamma^{(2)}$:}~By Lemma \ref{lem_threshold_behav_inst1}, $[T^{(1)}_a,T^{(1)}_f]$ and $[T^{(2)}_a,T^{(2)}_f]$ must overlap, and by Lemma \ref{lem_sign_inst1_uneqpref}, $T^{(1)}_a>T^{(2)}_a$. Therefore, the support boundaries of any EAP must be ordered by the following two orderings: 
\begin{enumerate}
    \item[Type I:] $T^{(2)}_a<T^{(1)}_a<T^{(2)}_f\leq T^{(1)}_f$ , and 
    \item[Type II:] $T^{(2)}_a<T^{(1)}_a<T^{(1)}_f<T^{(2)}_f$. 
\end{enumerate}
For both the above types, every EAP must have $T^{(2)}_a<0$. The following lemma specifies the necessary and sufficient condition for existence of an EAP under the two types mentioned above. Existence of unique EAP under case 3 of Theorem \ref{mainthm_inst1} follows trivially from this lemma. 

\begin{lemma}\label{lem:inst1case3}
    If $\gamma^{(1)}>\gamma^{(2)}$, the following statements are true: \begin{enumerate}
        \item There exists an EAP under Type I if and only if $\Lambda^{(1)}\geq\left(\frac{\mu_1}{\mu_2}-1\right)\Lambda^{(2)}$ and if it exists, it has a closed form same as the one mentioned under case 3a of Theorem \ref{mainthm_inst1}.
        \item There exists an EAP under Type II if and only if $\Lambda^{(1)}<\left(\frac{\mu_1}{\mu_2}-1\right)\Lambda^{(2)}$ and if it exists, it has a closed form same as the one mentioned under case 3b of Theorem \ref{mainthm_inst1}.
    \end{enumerate}
\end{lemma}

Proof of the above lemma follows an argument similar to the proof of Lemma \ref{lem:inst1case2} by identifying a system of linear equations to be satisfied by the support boundaries of any EAP under the two types. Details are in \ref{appndx:thm_inst1_uneqpref}. \hfill\qedsymbol

\bigskip

\begin{remark}
    \emph{The proof of the main results in Sections \ref{sec_inst1_eqpref}, \ref{sec_inst2_uneqpref} and \ref{sec_inst2_eqpref} will follow a similar sequence of arguments as used in the proof of Theorem \ref{mainthm_inst1}, through elimination of candidates not satisfying structural properties of every EAP in that regime. The skeleton of the arguments for the other sections is similar:1)~We first identify the possible ordering to be satisfied by the support boundaries $T^{(1)}_a,T^{(1)}_f,T^{(2)}_a,T^{(2)}_f$. Such statements are proved via contradiction. We show that, if the said ordering of the support boundaries do not hold, and users arrive following the arrival rates of the equilibrium arrival profile, then we can identify customers who can reduce their cost by arriving at a different time,~2)~For every possible ordering of the support boundaries, we identify a linear system to be satisfied by them. Solving this linear system gives us the unique values these support boundaries must have and the unique candidate EAP under each of these orderings,~3)~For each possible ordering, checking for whether the obtained solution satisfy the ordering gives us a necessary condition to be satisfied by the masses of the two groups for the consistency of that candidate EAP. Furthermore, we derive that, the obtained necessary conditions for the consistency of different candidate EAPs are disjoint, but their union encompasses all possible instances $\Lambda^{(1)},\Lambda^{(2)},\gamma^{(1)},\gamma^{(2)}$. Hence existence and uniqueness of EAP under the respective regimes follow naturally by our arguments.}     
\end{remark}

\begin{figure}[h]
    \centering
    \includegraphics[width=12cm]{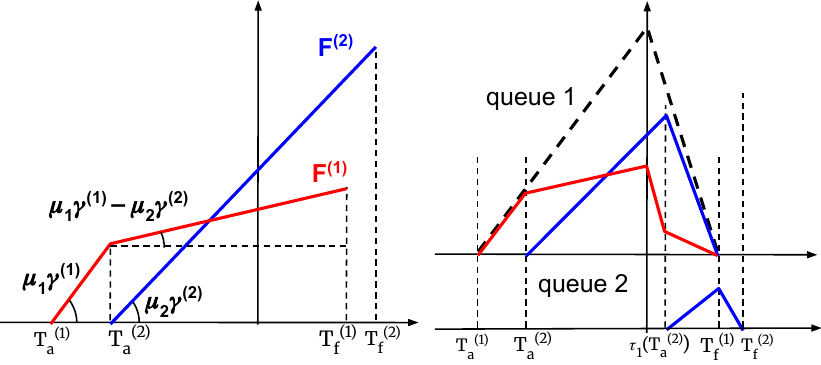}
    \caption{Illustrative EAP (left) and resulting queue length process (right) of HDS with $\mu_1>\mu_2$ under case 2a of Theorem \ref{mainthm_inst1} with the assumption $T^{(2)}_a<0$ in (\ref{eq_bdary_inst1_uneqpref_case2a}). Queue 1 divides its capacity between the two classes from $\tau_1(T^{(2)}_a)$ proportionally to their arrival rate. As a result, queue 1 serves class 1 users at a slower rate in $[\tau_1(T^{(2)}_a),T^{(1)}_f]$ than in $[0,\tau_1(T^{(2)}_a)]$, causing class 1 mass in queue 1 to decrease at a slower rate in the former interval than in the latter one.}
    \vspace{-0.1in}
    \label{fig:inst1case21}
\end{figure}

\begin{figure}[h]
    \centering
    \includegraphics[width=12cm]{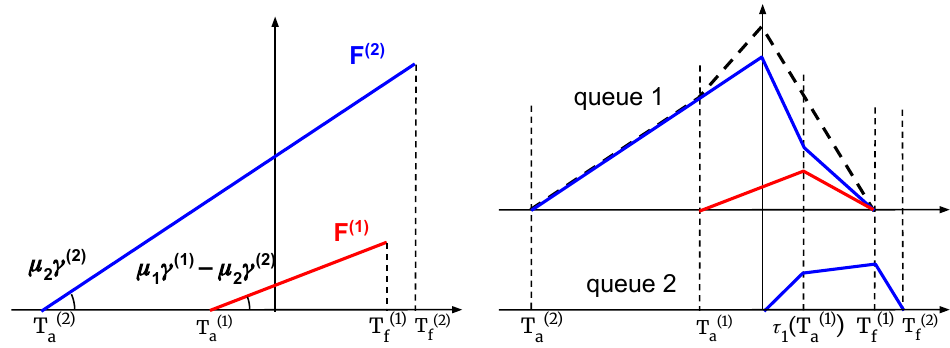}
    \caption{Illustrative EAP (left) of HDS with $\mu_1>\mu_2$ under case 2b and 3b of Theorem \ref{mainthm_inst1}. The resulting queue length process (right) is for case 2b with the assumption $T^{(2)}_a<0$ in (\ref{eq_bdary_inst1_uneqpref_case2b}). Queue 1 divides its capacity between the two classes from $\tau_1(T^{(1)}_a)$ proportionally to their service rates. As a result, queue 1 serves class 2 users at a lesser rate in $[\tau_1(T^{(1)}_a),T^{(1)}_f]$ than in $[0,\tau_1(T^{(1)}_a)]$, causing class 2 waiting mass to decrease at a lesser rate in the former interval than in the latter one. In queue 2, waiting mass increases in $[\tau_1(T^{(1)}_a),T^{(1)}_f]$ because class 2 arrival rate to queue 2 is $\mu_2\gamma^{(1)}/\gamma^{(2)} >\mu_2$. However the waiting mass increases in $[\tau_1(T^{(1)}_a),T^{(1)}_f]$ at a rate slower than in $[0,\tau_1(T^{(1)}_a)]$, since the arrival rate in the formal interval is $\mu_2\gamma^{(2)}/\gamma^{(1)}<\mu_1=$ arrival rate in the latter interval. For case 3b of Theorem \ref{mainthm_inst1}, with the assumption $T^{(1)}_a<0$ on (\ref{eq_bdary_inst1_uneqpref_case2b}), queue 1 length process has the same structure. Queue 2 length process of case 3b is also same as case 2b, except for one region: queue 2 length decreases in $[\tau_1(T^{(1)}_a),T^{(1)}_f]$, since class 2 arrival rate to queue 2 is $\mu_2\gamma^{(2)}/\gamma^{(1)}<\mu_2$ in case 3b.}
    \vspace{-0.1in}
    \label{fig:inst1case22}
\end{figure}
\begin{figure}[h]
    \centering
    \includegraphics[width=12cm]{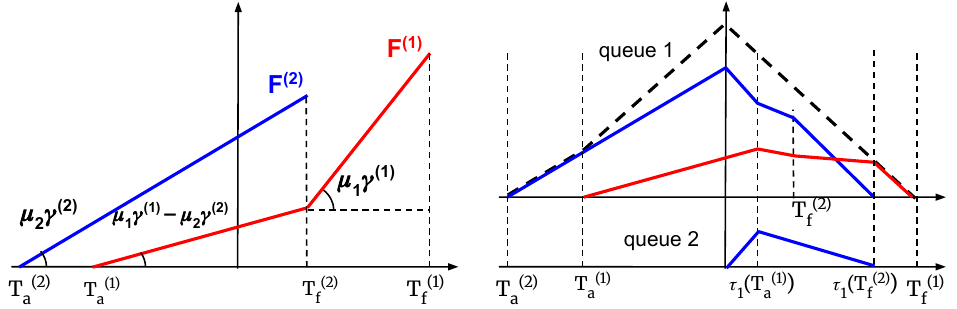}
    \caption{Illustrative EAP (left) and resulting queue length process (right) of HDS with $\mu_1>\mu_2$ under case 3a of Theorem \ref{mainthm_inst1} with the assumption $T^{(1)}_a<0$ and $\tau_1(T^{(1)}_a)=\frac{\mu_2\gamma^{(2)}}{\mu_1}(T^{(1)}_a-T^{(2)}_a)<T^{(2)}_f$ in (\ref{eq_bdary_inst1_uneqpref_case3a}). In $[\tau_1(T^{(1)}_a),T^{(1)}_f]$, queue 1 divides its capacity between the two classes proportionally to their arrival rates. As a result, class 2 mass in queue 1 decreases at a lesser rate in $[\tau_1(T^{(1)}_a),T^{(2)}_f]$ than in $[0,\tau_1(T^{(1)}_a)]$. After $T^{(2)}_f$, as class 2 users stop arriving, the class 2 mass in queue 1 decreases to zero at a higher rate in $[T^{(2)}_f,\tau_1(T^{(2)}_f)]$ than in $[\tau_1(T^{(1)}_a),T^{(2)}_f]$. Class 1 mass in queue 1 initially decreases in $[\tau_1(T^{(1)}_a),T^{(2)}_f]$. Since class 1 users increase their arrival rate after $T^{(2)}_f$, slope of the red line increases and the class 1 mass in queue 1 can even  increase in $[T^{(2)}_f,\tau_1(T^{(2)}_f)]$ if $\mu_1\gamma^{(1)}>\mu_1-\mu_2\gamma^{(2)}/\gamma^{(1)}$. After $\tau_1(T^{(2)}_f)$, queue 1 allocates the entire capacity to class 1 and the class 1 mass in queue 1 decreases to zero in $[\tau_1(T^{(2)}_f),T^{(1)}_f]$.}
    \label{fig:inst1case31}
\end{figure}

\subsection{Equal Preferences: $\gamma^{(1)}=\gamma^{(2)}=\gamma$}\label{sec_inst1_eqpref}

Below we state the main theorem of this section. 

\begin{theorem}\label{mainthm_inst1_eqpref}
    If $\mu_1>\mu_2$ and $\gamma^{(1)}=\gamma^{(2)}=\gamma$, EAP of HDS exists and:
    \begin{enumerate}
        \item[\textbf{1.}] If $\Lambda^{(1)}<\left(\frac{\mu_1}{\mu_2}-1\right)\Lambda^{(2)}$, the EAP is unique and its arrival rates and support boundaries are:
        \begin{align*}
            (F^{(1)})^\prime(t)&=(\mu_1-\mu_2)\gamma~\text{for}~t\in[T^{(1)}_a,T^{(1)}_f]~\text{and},~(F^{(2)})^\prime(t)=\mu_2\gamma~\text{for}~t\in[T_a,T_f],\\
            \text{where},~T^{(1)}_a&=\frac{1-\gamma}{\mu_1-\mu_2\gamma}\left[\Lambda^{(2)}-\frac{\mu_1}{(\mu_1-\mu_2)\gamma}\Lambda^{(1)}\right],~T^{(1)}_f=\frac{\Lambda^{(1)}+(1-\gamma)\Lambda^{(2)}}{\mu_1-\mu_2\gamma},\\
            T_{a}&=-\frac{1-\gamma}{\gamma}\frac{\Lambda^{(2)}}{\mu_2},~\text{and}~T_f=\frac{\Lambda^{(2)}}{\mu_2}.
        \end{align*}
        
        \item[\textbf{2.}] If $\Lambda^{(1)}\geq\left(\frac{\mu_1}{\mu_2}-1\right)\Lambda^{(2)}$, the set of EAPs is the convex set of joint arrival profiles $\mathbf{F}=\{F^{(1)},F^{(2)}\}$ satisfying:~~\textbf{1)}~~$\mathcal{S}(F^{(1)})=[T_a,T_f]$,~~$\mathcal{S}(F^{(2)})\subseteq [T_a,T_f]$,~~for $i=1,2$~~$F^{(i)}(T_f)=\Lambda^{(i)}$,~~\textbf{2)}~~$(F^{(1)})^\prime(t)+(F^{(2)})^\prime(t)=\mu_1\gamma$, and~~\textbf{3)}~~$(F^{(2)})^\prime(t)\leq\mu_2\gamma$  for every $t\in[T_a,T_f]$, where $T_a=-\left(\frac{1}{\gamma}-1\right)\frac{\Lambda^{(1)}+\Lambda^{(2)}}{\mu_1}$ and $T_f=\frac{\Lambda^{(1)}+\Lambda^{(2)}}{\mu_1}$.
    \end{enumerate}
\end{theorem}

In Theorem \ref{mainthm_inst1_eqpref}, note that under case 2 $\left(\Lambda^{(1)}\geq\left(\frac{1}{\gamma}-1\right)\Lambda^{(2)}\right)$, EAP is not necessarily unique and the set of EAPs is a convex set of joint arrival profiles.  

\noindent\textbf{Structural properties of any EAP.}\hspace{0.05in}Before discussing the proof of Theorem \ref{mainthm_inst1_eqpref}, we first identify the structural properties of any EAP in this regime and later we exploit them to refine our search of EAP. By Lemma \ref{lem_supp_are_intervals_inst1}, we can restrict our search to joint arrival profiles having supports satisfying $\mathcal{S}(F^{(1)})=[T^{(1)}_a,T^{(1)}_f]$ and $\mathcal{S}(F^{(1)})\cup \mathcal{S}(F^{(2)})=[T_a,T_f]$ for some $T_f>T_a$ and $T^{(1)}_f>T^{(1)}_a$. 
Lemma \ref{lem_rate_of_arrivals_inst1_eqpref} states the arrival rates of the two classes in EAP and is proved (in  \ref{appndx:inst1_eqpref}) via an argument similar to the one used for proving Lemma \ref{lem_arrival_rates_inst1}. Lemma \ref{lem_queue1_and_2_idle_inst1_eqpref} is proved via contradiction, by showing that, if the stated condition doesn't hold, users of one of the classes can improve by arriving at a different time.

\begin{lemma}[\textbf{Rate of arrival}]\label{lem_rate_of_arrivals_inst1_eqpref}
    If $\mu_1>\mu_2$ and $\gamma^{(1)}=\gamma^{(2)}=\gamma$, the following properties must be true almost everywhere for every EAP  $\mathbf{F}=\{F^{(1)},F^{(2)}\}$,
    \begin{align*}
        (F^{(1)})^\prime(t)+(F^{(2)})^\prime(t)&=\mu_1\gamma~~~\text{if $t\in \mathcal{S}(F^{(1)})$, and,}\\
        (F^{(2)})^\prime(t)&\leq\mu_2\gamma ~~~\text{if $t\in \mathcal{S}(F^{(2)})$, with equality when $Q_2(\tau_1(t))>0$}.
    \end{align*}
\end{lemma}

\begin{lemma}\label{lem_queue1_and_2_idle_inst1_eqpref}
   If $\mu_1>\mu_2$ and $\gamma^{(1)}=\gamma^{(2)}=\gamma$, in the EAP,
    \begin{enumerate}
        \item If $\mathcal{S}(F^{(2)})$ has a gap $(t,t+\delta]$ at $t\in \mathcal{S}(F^{(2)})$ for some $\delta>0$ sufficiently small, queue 2 must be empty at $\tau_1(t)$.
        \item Queue 1 must be empty time at $T^{(1)}_f$. 
    \end{enumerate}
\end{lemma}

\noindent\textbf{Proof of Theorem \ref{mainthm_inst1_eqpref}}

We only provide the key steps of the proof of Theorem \ref{mainthm_inst1_eqpref} below. The details are given in  \ref{appndx:thm_inst1_eqpref}. Following two supporting lemmas are necessary for proving Theorem \ref{mainthm_inst1_eqpref}. 

\begin{lemma}\label{lem_cost2_const_inst1_eqpref}
    In every EAP, the cost of the class 2 users remains constant over $[T_a,T_f]$.
\end{lemma}

\begin{lemma}\label{lem_inst1_eqpreF^{(2)}rate_less_than_mu2}
    In every EAP, if $t\in\mathcal{S}(F^{(1)})\cap\mathcal{S}(F^{(2)})$, then class 2 users will arrive at queue 2 at a maximum rate of $\mu_2$ at time $\tau_1(t)$. 
\end{lemma}
Proofs of Lemma \ref{lem_cost2_const_inst1_eqpref} and 
Lemma \ref{lem_inst1_eqpreF^{(2)}rate_less_than_mu2} are in  \ref{appndx:inst1_eqpref}. Lemma \ref{lem_cost2_const_inst1_eqpref} is proved  via exploiting Lemma \ref{lem_queue1_and_2_idle_inst1_eqpref} and the fact that the two classes have equal cost preferences. Proof of Lemma \ref{lem_inst1_eqpreF^{(2)}rate_less_than_mu2} is based on the observation that $A_2(\tau_1(t))=F^{(2)}(t)$. As a result, at time $\tau_1(t)$, class 2 users arrive at queue 2 at rate $A_2^\prime(\tau_1(t))=\frac{(F^{(2)})^\prime(t)}{\tau_1^\prime(t)}$, assuming that the derivatives exist. Following this, we use Lemma \ref{lem_rate_of_arrivals_inst1_eqpref} and (\ref{eq:derv_of_tau}) to upper bound the rate of arrival. \vspace{0.125in}

\noindent\textit{\textbf{Key steps of proof of Theorem \ref{mainthm_inst1_eqpref}:}}\hspace{0.05in}By Lemma \ref{lem_supp_are_intervals_inst1}, we consider candidates $\mathbf{F}=\{F^{(1)},F^{(2)}\}$ which are absolutely continuous with $\mathcal{S}(F^{(1)})=[T^{(1)}_a,T^{(1)}_f],~\mathcal{S}(F^{(1)})\cup\mathcal{S}(F^{(2)})=[T_a,T_f]$ and have arrival rates $(F^{(1)})^\prime(\cdot),~(F^{(2)})^\prime(\cdot)$ given by Lemma \ref{lem_rate_of_arrivals_inst1_eqpref}. We eliminate candidates which do not satisfy the structural properties of any EAP and narrow our search of an EAP to a smaller set of candidates. We follow this agenda separately over two subsets of candidates with:~\textbf{1)}~$T_f>T^{(1)}_f$~and~\textbf{2)}~$T_f=T^{(1)}_f$. Lemma \ref{lem:inst1eqprefcase1} and \ref{lem:inst1eqprefcase2} provides the necessary and sufficient condition for existence of EAPs of the mentioned two types. The statement of Theorem \ref{mainthm_inst1_eqpref} follows from these two lemmas. \vspace{-0.05in}

\begin{lemma}\label{lem:inst1eqprefcase1}
    If $\gamma^{(1)}=\gamma^{(2)}=\gamma$, there exists an EAP with $T_f>T^{(1)}_f$ if and only if $\Lambda^{(1)}<\left(\frac{\mu_1}{\mu_2}-1\right)\Lambda^{(2)}$, and if it exists, it will be unique with a closed form same as the joint arrival profile mentioned under case 1 of Theorem \ref{mainthm_inst1_eqpref}.
\end{lemma}

\begin{lemma}\label{lem:inst1eqprefcase2}
    If $\gamma^{(1)}=\gamma^{(2)}=\gamma$, there exists an EAP with $T_f=T^{(1)}_f$ if and only if $\Lambda^{(1)}\geq\left(\frac{\mu_1}{\mu_2}-1\right)\Lambda^{(2)}$, and if it exists, the set of all EAPs with $T^{(1)}_f=T_f$ is the set of joint arrival profiles mentioned under case 2 of Theorem \ref{mainthm_inst1_eqpref}.
\end{lemma}

Below we sketch the proof of Lemma \ref{lem:inst1eqprefcase1} and \ref{lem:inst1eqprefcase2}. The remaining details are presented in  \ref{appndx:thm_inst1_eqpref}.

\noindent\textit{\textbf{Proof sketch of Lemma \ref{lem:inst1eqprefcase1}:}}\hspace{0.05in}\textbf{Identifying the support boundaries:}~First we identify the support boundaries $T_a, T_f$ of any EAP with $T_f>T^{(1)}_f$:

\begin{enumerate}
    \item[\textbf{1)}] Exploiting Lemma \ref{lem_inst1_eqpreF^{(2)}rate_less_than_mu2} and \ref{lem_queue1_and_2_idle_inst1_eqpref}, we argue that queue 2 has a positive waiting time in $(0,T_f)$ and empties at $T_f$. As a result, queue 2 serves the whole class 2 population in $(0,T_f)$, giving us $T_f=\frac{\Lambda^{(2)}}{\mu_2}$.

    \item[\textbf{2)}] Using Lemma \ref{lem_cost2_const_inst1_eqpref} we have, $C_{\mathbf{F}}^{(2)}(T_a)=C_{\mathbf{F}}^{(2)}(T_f)$. The class 2 user arriving at $T_a$ gets served at time zero causing $C_{\mathbf{F}}^{(2)}(T_a)=-\gamma T_a$. Since the network is empty at $T_f$, $C_{\mathbf{F}}^{(2)}(T_f)=(1-\gamma)T_f$. Therefore, $C_{\mathbf{F}}^{(2)}(T_a)=C_{\mathbf{F}}^{(2)}(T_f)$ implies $T_a=-\left(\frac{1}{\gamma}-1\right)T_f=-\left(\frac{1}{\gamma}-1\right)\frac{\Lambda^{(2)}}{\mu_2}$.

    \item[\textbf{3)}] The class 2 population of mass $\Lambda^{(2)}$ has to arrive in the interval $\left[T_a,T_f\right]$ of length $\frac{\Lambda^{(2)}}{\mu_2\gamma^{(2)}}$ at maximum rate of $\mu_2\gamma^{(2)}$ (by Lemma \ref{lem_rate_of_arrivals_inst1_eqpref}). This leaves us with only one possible arrival rate for class 2 users: $(F^{(2)})^\prime(t)=\mu_2\gamma^{(2)}$ in $[T_a,T_f]$. As a result, by Lemma \ref{lem_rate_of_arrivals_inst1_eqpref}, we must have $(F^{(1)})^\prime(t)=(\mu_1-\mu_2)\gamma$ in $[T^{(1)}_a,T^{(1)}_f]$, giving us  $T^{(1)}_f=T^{(1)}_a+\frac{\Lambda^{(1)}}{(\mu_1-\mu_2)\gamma}$.
    
    \item[\textbf{4)}] Since $T_a<0$, queue 1 starts serving from time zero and has positive waiting time in $[0,T^{(1)}_f)$. By Lemma \ref{lem_queue1_and_2_idle_inst1_eqpref}, queue 1 empties at $T^{(1)}_f$. Therefore, $\mu_1 T^{(1)}_f=F^{(1)}(T^{(1)}_f)+F^{(2)}(T^{(1)}_f)=\Lambda^{(1)}+\newline\mu_2\gamma\cdot(T^{(1)}_f-T_a)$. 
\end{enumerate}

\textbf{Getting the necessary condition:} Solving the above system of equations, we get $T_a=-\left(\frac{1}{\gamma}-1\right)\frac{\Lambda^{(2)}}{\mu_2}$, $T_f=\frac{\Lambda^{(2)}}{\mu_2}$, $T^{(1)}_a=\frac{1-\gamma}{\mu_1-\mu_2\gamma}\left[\Lambda^{(2)}-\frac{\mu_1}{(\mu_1-\mu_2)\gamma}\Lambda^{(1)}\right]$, $T^{(1)}_f=\frac{\Lambda^{(1)}+(1-\gamma)\Lambda^{(2)}}{\mu_1-\mu_2\gamma}$ and as a result, every EAP with $T_f>T^{(1)}_f$ has $T_a, T_f, T^{(1)}_a,T^{(1)}_f$ same as the ones we identified. For $T_f>T^{(1)}_f$, we need $\Lambda^{(1)}<\left(\frac{\mu_1}{\mu_2}-1\right)\Lambda^{(2)}$ and therefore, this is a necessary condition for existence of an EAP with $T_f>T^{(1)}_f$. \vspace{0.05in}

\textbf{Identifying the unique EAP:} If the necessary condition $\Lambda^{(1)}<\left(\frac{\mu_1}{\mu_2}-1\right)\Lambda^{(2)}$ is satisfied, it is easy to verify that, the support boundaries obtained above follow the ordering $T_a < T^{(1)}_a<T^{(1)}_f<T_f$. With this ordering satisfied, the only candidate with $T_f>T^{(1)}_f$, which qualifies to be an EAP is: $(F^{(1)})^\prime(t)=(\mu_1-\mu_2)\gamma$ in $[T^{(1)}_a,T^{(1)}_f]$ and $(F^{(2)})^\prime(t)=\mu_2\gamma$ in $[T_a,T_f]$, with $T_a,T_f,T^{(1)}_a,T^{(1)}_f$ same as we identified before. Moreover, the candidate obtained is same as the joint arrival profile under case 1 of Theorem \ref{mainthm_inst1_eqpref}. 

\textbf{Proving sufficiency of the obtained necessary condition:} Another interesting observation is, if $\Lambda^{(1)}<\left(\frac{\mu_1}{\mu_2}-1\right)\Lambda^{(2)}$, the obtained candidate with $T_f>T^{(1)}_f$ is same as the EAPs in cases 2b and 3b of Theorem \ref{mainthm_inst1}, respectively, upon taking limits $\gamma^{(1)}=\gamma,~\gamma^{(2)}\to\gamma+$ and $\gamma^{(2)}=\gamma,~\gamma^{(1)}\to \gamma+$. Proving that this candidate is an EAP follows by the argument used for proving the unique Type II candidate in Lemma \ref{lem:inst1case2} is an EAP (in  \ref{appndx:thm_inst1_uneqpref}), except replacing $T^{(2)}_a,T^{(2)}_f$ with $T_a,T_f$ and taking $\gamma^{(1)}=\gamma^{(2)}=\gamma$. Therefore, the statement of Lemma \ref{lem:inst1eqprefcase1} stands proved. \vspace{0.05in} \hfill\qedsymbol \\

\noindent\textit{\textbf{Proof sketch of Lemma \ref{lem:inst1eqprefcase2}:}}\hspace{0.05in}\textbf{Identifying $\mathbf{T_a}$:} Since $[0,T_f]\subseteq\mathcal{S}(F^{(1)})$, queue 1 must have a positive waiting time in $[0,T_f)$. By Lemma \ref{lem_queue1_and_2_idle_inst1_eqpref}, queue 1 empties at time $T_{f}$. Therefore we get $\mu_1 T_f=F^{(1)}(T_f)+F^{(2)}(T_f)=\Lambda^{(1)}+\Lambda^{(2)}$ implying $T_f=\frac{\Lambda^{(1)}+\Lambda^{(2)}}{\mu_1}$.

\textbf{Identifying $\mathbf{T_f}$:} By Lemma \ref{lem_cost2_const_inst1_eqpref}, we have $C_{\mathbf{F}}^{(2)}(T_a)=C_{\mathbf{F}}^{(2)}(T_f)$. Using the argument used in the proof sketch of Lemma \ref{lem:inst1eqprefcase1} for finding $T_a$, we get $T_a=-\left(\frac{1}{\gamma}-1\right)T_f=-\left(\frac{1}{\gamma}-1\right)\frac{\Lambda^{(1)}+\Lambda^{(2)}}{\mu_1}$. 

\textbf{Getting the necessary condition:} By Lemma \ref{lem_rate_of_arrivals_inst1_eqpref}, all class 2 users arrive in $[T_a,T_f]$ at a maximum rate of $\mu_2\gamma$. Therefore for existence of an EAP with $T_f=T^{(1)}_f$, we need the necessary condition $\mu_2\gamma\cdot(T_f-T_a)\geq\Lambda^{(2)}$, which after some manipulation gives us $\Lambda^{(1)}\geq\left(\frac{\mu_1}{\mu_2}-1\right)\Lambda^{(2)}$.

\textbf{Identifying the set of EAPs:} By Lemma \ref{lem_rate_of_arrivals_inst1_eqpref}, $\mu_1\gamma$ is the maximum rate at which the entire population of mass $\Lambda^{(1)}+\Lambda^{(2)}$ can arrive within the time interval $[T_a,T_f]$ of length $\frac{\Lambda^{(1)}+\Lambda^{(2)}}{\mu_1\gamma}$. Therefore, we must have $(F^{(1)})^\prime(t)+(F^{(2)})^\prime(t)=\mu_1\gamma$ in $[T_a,T_f]$. Also by Lemma \ref{lem_rate_of_arrivals_inst1_eqpref}, we have $(F^{(1)})^\prime(t)+(F^{(2)})^\prime(t)=\mu_1\gamma$ only in $\mathcal{S}(F^{(1)})=[T^{(1)}_a,T^{(1)}_f]\subseteq[T_a,T_f]$, which implies $T^{(1)}_a=T_a$. By Lemma \ref{lem_inst1_eqpreF^{(2)}rate_less_than_mu2}, class 2 users arrive at queue 2 from queue 1 in $[0,T_f]$ at a maximum rate of $\mu_2$. Hence queue 2 stays empty and as a result, by Lemma \ref{lem_rate_of_arrivals_inst1_eqpref}, $(F^{(2)})^\prime(t)\leq\mu_2\gamma$ for every $t\in[T_a,T_f]$. Therefore, our reduced set of candidates with $T_f=T^{(1)}_f$ will be the set of all joint arrival profiles $\mathbf{F}=\{F^{(1)},F^{(2)}\}$ satisfying:~~\textbf{1)}~$\mathcal{S}(F^{(1)})=[T_a,T_f]$ and $\mathcal{S}(F^{(2)})\subseteq[T_a,T_f]$,~for $i=1,2$ $F^{(i)}(T_f)=\Lambda^{(i)}$,~~\textbf{2)}~$(F^{(1)})^\prime(t)+(F^{(2)})^\prime(t)=\mu_1\gamma$, and ~~\textbf{3)}~$(F^{(2)})^\prime(t)\leq\mu_2\gamma$ for every $t\in[T_a,T_f]$ where $T_a=-\frac{1-\gamma}{\gamma}\frac{\Lambda^{(1)}+\Lambda^{(2)}}{\mu_1}$ and $T_f=\frac{\Lambda^{(1)}+\Lambda^{(2)}}{\mu_1}$, which is exactly the set mentioned under the case 2 of Theorem \ref{mainthm_inst1_eqpref}.   

\textbf{Proving sufficiency of the obtained necessary condition:}~It is easy to verify that, if the necessary condition $\Lambda^{(1)}\geq\left(\frac{\mu_1}{\mu_2}-1\right)\Lambda^{(2)}$ holds, the set of candidates with $T_f=T^{(1)}_f$ mentioned earlier is non-empty. Two elements of the set are the limits of the EAPs in cases 2a and 3a of Theorem \ref{mainthm_inst1}, respectively, when $\gamma^{(1)}=\gamma,~\gamma^{(2)}\to\gamma+$ and $\gamma^{(2)}=\gamma,~\gamma^{(1)}\to\gamma+$. After this, we prove that, if $\Lambda^{(1)}\geq\left(\frac{\mu_1}{\mu_2}-1\right)\cdot\Lambda^{(2)}$, every joint arrival profile $\mathbf{F}=\{F^{(1)},F^{(2)}\}$ in the obtained set of candidates is an EAP. This step requires exploiting Lemma \ref{lem_cost2_const_inst1_eqpref} and the fact that queue 2 stays empty when users are arriving by such a candidate and the details of it are in  \ref{appndx:thm_inst1_eqpref}. Hence, the statement of Lemma \ref{lem:inst1eqprefcase2} stands proved. \hfill\qedsymbol
\section{Heterogeneous Arrival System (HAS)}\label{sec_hetarrivals}

In this section, we consider the two groups arrive at different queues and depart from the same queue as illustrated in Table \ref{tab:all_2queue_instances}. Namely, the $i$-th group arrive at $i$-th queue and both the groups depart from the 2nd queue.  We consider the case $\gamma^{(1)}\neq\gamma^{(2)}$ and $\gamma^{(1)}=\gamma^{(2)}$ separately, since the later displays different behaviors. 

\subsection{Unequal Preferences $\gamma^{(1)}\neq\gamma^{(2)}$}\label{sec_inst2_uneqpref}

Theorem \ref{thm:HAS_uneqpref_brief} contains the main result of this section. Theorems \ref{mainthm_inst2_reg1} and \ref{mainthm_inst2_reg2} contain the detailed version of the main results, explicitly disclosing the closed forms of the EAP as a function of the instance parameters $\{(\Lambda^{(i)},\gamma^{(i)},\mu_i)\}_{i=1,2}$. 

\begin{theorem}\label{thm:HAS_uneqpref_brief}
    HAS with unequal preference, \textit{i.e.} $\gamma^{(1)}\neq\gamma^{(2)}$, has a unique EAP. Furthermore, there can be eight possible structures of the EAP depending on the instance parameters. 
\end{theorem}

As for HDS, here too we narrow down our search for EAP by exploiting its structural properties. The  detailed account of these structural properties as well as detailed proofs of results in this section are  given in  \ref{appndx:inst2_uneqpref}. The main results in this section are proved by arguments similar to those  used for proving Theorem \ref{mainthm_inst1} for HDS,
except that we now exploit a different set of structural properties and the no. of different regimes is significantly larger and diverse for HAS. 
In this section, we state the key results and after Theorem \ref{mainthm_inst2_reg2} in Remark \ref{rem:disj_support}, we discuss two interesting
cases:~~\textbf{1)}~where the EAP corresponds to arrivals in a stream coming in disjoint intervals 
and~~\textbf{2)}~where all the arrivals for class 2 arrive before time zero.

The following lemma identifies an important threshold property of  an EAP. Proof of Lemma \ref{lem_threshold_behav_inst2} is in  \ref{appndx:inst2_uneqpref} and follows a sequence of arguments similar to the one used for proving Lemma \ref{lem_threshold_behav_inst1}. 

\begin{lemma}[\textbf{Threshold Behavior}]\label{lem_threshold_behav_inst2}
    If $\gamma^{(1)}\neq\gamma^{(2)}$, in EAP, queue 2 serves the two classes over disjoint sets of time if and only if $\mu_1\geq\mu_2\gamma^{(2)}$.    
\end{lemma}

\noindent\textit{\textbf{Proof sketch:}}\hspace{0.05in}Proof of Lemma \ref{lem_threshold_behav_inst2} is similar to that of Lemma \ref{lem_threshold_behav_inst1}. First we argue via a contradiction that if $\mu_1\geq\mu_2\gamma^{(2)}$, then queue 2 cannot serve the two classes together in an EAP. If queue 2 is serving the two classes together, we argue that, queue 1 must stay engaged in a neighbourhood of that time for class 1 to be iso cost (by Lemma \ref{lem_appndx_inst2_uneqpref_queu1_engaged_mixedarrival} in \ref{appndx:inst2_uneqpref}). As a result, class 1 users arrive from queue 1 to 2 at rate $\mu_1$. For class 2 have constant cost in that neighbourhood, we argue that, users of the two classes arrive at queue 2 at a combined rate of $\mu_2\gamma^{(2)}$. Since queue 2 serves a positive mass of both the classes together, the combined arrival rate must be strictly larger than the arrival rate of class 1 to queue 2, implying $\mu_2\gamma^{(2)}>\mu_1$ and contradicting $\mu_1\geq\mu_2\gamma^{(2)}$. The other direction is proved via contradiction by exploiting the structures of $\mathcal{S}(F^{(1)}),~\mathcal{S}(F^{(2)})$ in an EAP and showing that, if $\mu_1<\mu_2\gamma^{(2)}$ and queue 2 is serving the two classes over disjoint sets of times, we can find a user who can decrease her cost by arriving at a different time. \hfill\qedsymbol \\

\noindent\textbf{Specification of the EAP.}\hspace{0.05in}The two regimes $\mu_1<\mu_2\gamma^{(2)}$ and $\mu_1\geq\mu_2\gamma^{(2)}$ exhibit substantially different EAP structures, owing to the threshold behavior given by Lemma \ref{lem_threshold_behav_inst2}. For conciseness, we specify the unique EAP in these two regimes separately in Theorem \ref{mainthm_inst2_reg1} (for $\mu_1<\mu_2\gamma^{(2)}$) and \ref{mainthm_inst2_reg2} (for $\mu_1\geq\mu_2\gamma^{(2)}$). Proofs of these theorems (in  \ref{appndx:inst2_uneqpref}) have a structure similar to the proof of Theorem \ref{mainthm_inst1}.    

For the two classes $i=1,2$, we define the quantities $T_{i,a}=\inf \mathcal{S}(F_i)$ and $T_{i,f}=\sup \mathcal{S}(F_i)$. In an EAP, both the sets $\mathcal{S}(F^{(1)}),~\mathcal{S}(F^{(2)})$ are compact and therefore $T^{(1)}_a,T^{(1)}_f,T^{(2)}_a,T^{(2)}_f$ are all finite.\vspace{0.05in}

\noindent\textbf{\textit{Regime I}~~$\mu_1<\mu_2\gamma^{(2)}$}.~~We mention below the support boundaries of the unique EAP, which we will refer later in Theorem \ref{mainthm_inst2_reg1}: 
\begin{enumerate}
    \item[\textbf{1.}] If $\mu_1<\mu_2\gamma^{(2)}$ and $\Lambda^{(1)}\geq\max\left\{\frac{1-\gamma^{(2)}}{1-\gamma^{(1)}},1\right\}\cdot\frac{\mu_1}{\mu_2-\mu_1}\Lambda^{(2)}$, then 
    \begin{align}\label{eq_inst2_uneqpref_reg1_bdary_case1a}
        T^{(1)}_a&=-\frac{1-\gamma^{(1)}}{\gamma^{(1)}}\frac{\Lambda^{(1)}}{\mu_1}+\frac{1-\gamma^{(2)}}{\gamma^{(1)}}\frac{\Lambda^{(2)}}{\mu_2-\mu_1},~T^{(1)}_f=\frac{\Lambda^{(1)}}{\mu_1},~T^{(2)}_a=-\frac{1-\gamma^{(2)}}{\gamma^{(2)}}\frac{\Lambda^{(2)}}{\mu_2-\mu_1},~\text{and}~T^{(2)}_f=\frac{\Lambda^{(2)}}{\mu_2-\mu_1}.
    \end{align}
    \vspace{-0.05in}
    \item[\textbf{2.}] If $\mu_1<\mu_2\gamma^{(2)}$ and $\Lambda^{(1)}<\max\left\{\frac{1-\gamma^{(2)}}{1-\gamma^{(1)}},1\right\}\cdot\frac{\mu_1}{\mu_2-\mu_1}\Lambda^{(2)}$, then 
    \begin{enumerate}
        \item[\textbf{a)}] if $\gamma^{(1)}>\gamma^{(2)}$, then
        \begin{align}\label{eq_inst2_uneqpref_reg1_bdary_case1b}
            T^{(1)}_a&=\frac{1}{\mu_2}\left(\Lambda^{(2)}-\frac{1-\gamma^{(1)}}{1-\gamma^{(2)}}\frac{\mu_2-\mu_1}{\mu_1}\Lambda^{(1)}\right),~T^{(1)}_f=\frac{1}{\mu_2}\left(\Lambda^{(2)}+\frac{(\gamma^{(1)}-\gamma^{(2)})\mu_2+(1-\gamma^{(1)})\mu_1}{(1-\gamma^{(2)})\mu_1}\Lambda^{(1)}\right),\nonumber \\
        T^{(2)}_a&=-\frac{1}{\mu_2\gamma^{(2)}}\left((1-\gamma^{(1)})\Lambda^{(1)}+(1-\gamma^{(2)})\Lambda^{(2)}\right),~\text{and}~T^{(2)}_f=\frac{1}{\mu_2}\left(\Lambda^{(2)}+\frac{1-\gamma^{(1)}}{1-\gamma^{(2)}}\Lambda^{(1)}\right).
        \end{align}
        
        \item[\textbf{b)}] if $\gamma^{(1)}<\gamma^{(2)}$, then
        \begin{align}\label{eq_inst2_uneqpref_reg1_bdary_case2b}
            T^{(1)}_a&=-\left(\frac{\gamma^{(2)}}{\gamma^{(1)}}-1\right)\frac{\Lambda^{(1)}}{\mu_1},~T^{(1)}_f=\frac{\Lambda^{(1)}}{\mu_1},~T^{(2)}_a=-\frac{1-\gamma^{(2)}}{\gamma^{(2)}}\frac{\Lambda^{(1)}+\Lambda^{(2)}}{\mu_2},~\text{and}~T^{(2)}_f=\frac{\Lambda^{(1)}+\Lambda^{(2)}}{\mu_2}.
        \end{align}
    \end{enumerate}
\end{enumerate}

\begin{theorem}\label{mainthm_inst2_reg1}
    If $\gamma^{(1)}\neq\gamma^{(2)}$ and $\mu_1<\mu_2\gamma^{(2)}$, HAS has a unique EAP which has the following arrival rates and support boundaries:
    \begin{enumerate}
        \item[\textbf{1.}] If $\gamma^{(1)}>\gamma^{(2)}$: 
        \begin{enumerate}
            \item[\textbf{1a.}] when $\Lambda^{(1)}\geq\frac{1-\gamma^{(2)}}{1-\gamma^{(1)}}\frac{\mu_1}{\mu_2-\mu_1}\Lambda^{(2)}$: 
            \begin{align*}
                (F^{(1)})^\prime(t)&=\begin{cases} \frac{\mu_1\gamma^{(1)}}{\gamma^{(2)}}~~&\text{for}~t\in[T^{(1)}_a,\tau_1^{-1}(T^{(2)}_f)],\\ \mu_1\gamma^{(1)}~~&\text{for}~t\in[\tau_1^{-1}(T^{(2)}_f),T^{(1)}_f],
                \end{cases}~\text{and}~ (F^{(2)})^\prime(t)=\begin{cases} \mu_2\gamma^{(2)}~~&\text{for}~t\in[T^{(2)}_a,0],\\ \mu_2\gamma^{(2)}-\mu_1~~&\text{for}~t\in[0,T^{(2)}_f],\\ \end{cases}
            \end{align*}
            where $T^{(1)}_a,T^{(1)}_f,T^{(2)}_a$ and $T^{(2)}_f$ are given by (\ref{eq_inst2_uneqpref_reg1_bdary_case1a}).    
            \item[\textbf{1b.}] when $\Lambda^{(1)}<\frac{1-\gamma^{(2)}}{1-\gamma^{(1)}}\frac{\mu_1}{\mu_2-\mu_1}\Lambda^{(2)}$: 
            \begin{align*}
                (F^{(1)})^\prime(t)&=\begin{cases} \frac{\mu_1\gamma^{(1)}}{\gamma^{(2)}}~~&\text{for}~t\in[T^{(1)}_a,\tau_1^{-1}(T^{(2)}_f)],~\\ \mu_1\gamma^{(1)}~~&\text{for}~t\in[\tau_1^{-1}(T^{(2)}_f),T^{(1)}_f],\\ \end{cases}~\text{and}~(F^{(2)})^\prime(t)=\begin{cases} \mu_2\gamma^{(2)}~~&\text{for}~t\in[T^{(2)}_a,T^{(1)}_a],\\ \mu_2\gamma^{(2)}-\mu_1~~&\text{for}~t\in[T^{(1)}_a,T^{(2)}_f],\\ \end{cases}    
            \end{align*}
            where $T^{(1)}_a,T^{(1)}_f,T^{(2)}_a$ and $T^{(2)}_f$ are given by (\ref{eq_inst2_uneqpref_reg1_bdary_case1b}).     
        \end{enumerate}
        
        \item[\textbf{2.}] If $\gamma^{(1)}<\gamma^{(2)}$:
        \begin{enumerate}
             \item[\textbf{2a.}] when $\Lambda^{(1)}\geq\frac{\mu_1}{\mu_2-\mu_1}\Lambda^{(2)}$, the EAP has closed form same as case 1a. 
            \item[\textbf{2b.}] when $\Lambda^{(1)}<\frac{\mu_1}{\mu_2-\mu_1}\Lambda^{(2)}$:
            \begin{align*}
                 (F^{(1)})^\prime(t)&=\frac{\mu_1\gamma^{(1)}}{\gamma^{(2)}}~~\text{for}~t\in[T^{(1)}_a,T^{(1)}_f],~\text{and}~(F^{(2)})^\prime(t)=\begin{cases} \mu_2\gamma^{(2)}~~&\text{for}~t\in[T^{(2)}_a,0]\cup[T^{(1)}_f,T^{(2)}_f],\\ \mu_2\gamma^{(2)}-\mu_1~~&\text{for}~t\in[0,T^{(1)}_f], \end{cases}   
            \end{align*}
            where $T^{(1)}_a,T^{(1)}_f,T^{(2)}_a$ and $T^{(2)}_f$ are given by (\ref{eq_inst2_uneqpref_reg1_bdary_case2b}).
        \end{enumerate}
    \end{enumerate}
\end{theorem}

\begin{remark}
    \emph{Illustrative EAPs and resulting queue length processes of cases 1a, 1b, 2a and 2b of Theorem \ref{mainthm_inst2_reg1} are illustrated, respectively, in Figures \ref{fig:inst2reg1case11}, \ref{fig:inst2reg1case12}, \ref{fig:inst2reg1case21} and \ref{fig:inst2reg1case22}. Structure of the queue length processes in some regimes might vary depending on the support boundaries and related quantities. We have illustrated only one possible structure of the queue length process in those regimes and mentioned the conditions on the support boundaries for the attainment of that structure in the caption. \textbf{\textcolor{red}{Red}} and \textbf{\textcolor{blue}{blue}} respectively denote class 1 and 2 populations. In the plot on right-top, the \textbf{black} dashed line represents the total waiting mass of the two classes in queue 2. The same is true for the illustrative EAPs and resulting queue length processes referred to in Theorem \ref{mainthm_inst2_reg2}.}
\end{remark}

\begin{figure}[h]
    \centering
    \includegraphics[width=12cm]{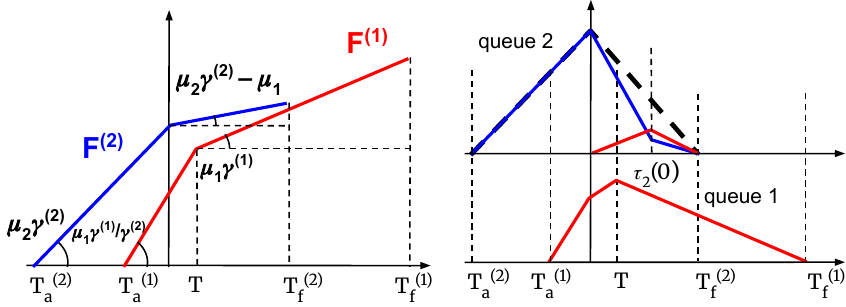}
    \caption{Illustrative EAP (left) and queue length process (right) for HAS with $\mu_1<\mu_2\gamma^{(2)}$ under case 1a of Theorem \ref{mainthm_inst2_reg1} and the assumption $T=T^{(1)}_a+\frac{\gamma^{(2)}}{\gamma^{(1)}}T^{(2)}_f>0$ on (\ref{eq_inst2_uneqpref_reg1_bdary_case1a}). $T\overset{def.}{=}\tau_1^{-1}(T^{(2)}_f)$. }
    \label{fig:inst2reg1case11}
\end{figure}

\begin{figure}[h]
    \centering
    \includegraphics[width=12cm]{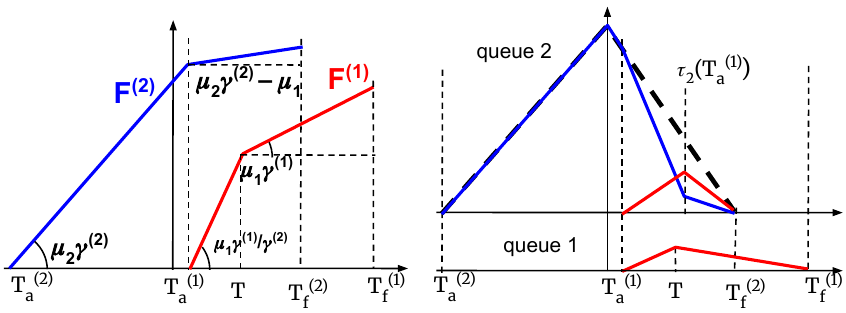}
    \caption{Illustrative EAP (left) and queue length process (right) for HAS with $\mu_1<\mu_2\gamma^{(2)}$ under case 1b of Theorem \ref{mainthm_inst2_reg1}. $T\overset{def.}{=}\tau_1^{-1}(T^{(2)}_f)$.  }
    \label{fig:inst2reg1case12}
\end{figure}

\begin{figure}[h]
    \centering
    \includegraphics[width=12cm]{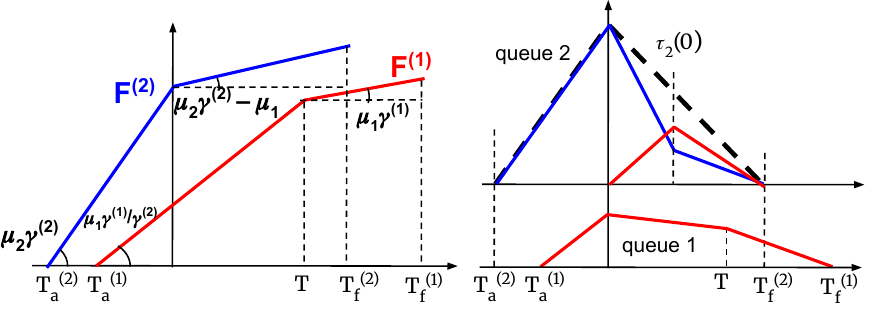}
    \caption{Illustrative EAP (left) and queue length process (right) for HAS with $\mu_1<\mu_2\gamma^{(2)}$ under case 2a of Theorem \ref{mainthm_inst2_reg2} and the assumption  $T=T^{(1)}_a+\frac{\gamma^{(2)}}{\gamma^{(1)}}T^{(2)}_f>0$ on (\ref{eq_inst2_uneqpref_reg1_bdary_case1a}). $T\overset{def.}{=}\tau_1^{-1}(T^{(2)}_f)$. }
    \label{fig:inst2reg1case21}
\end{figure}

\begin{figure}[h]
    \centering
    \includegraphics[width=12cm]{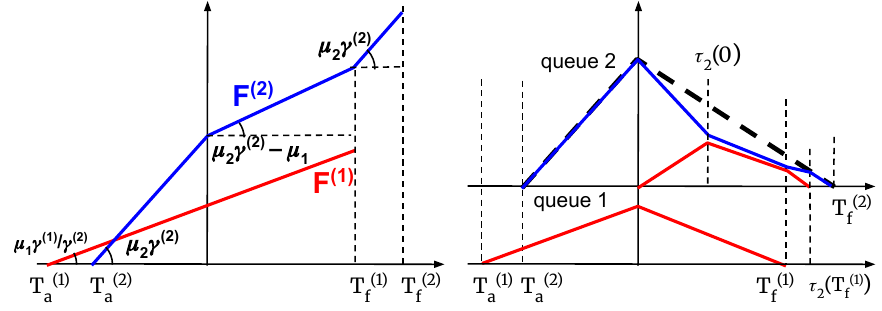}
    \caption{Illustrative EAP for HAS with $\mu_1<\mu_2\gamma^{(2)}$ under case 2b  of Theorem \ref{mainthm_inst2_reg2} and the assumption $\tau_2(0)=-\gamma^{(2)} T^{(2)}_a<T^{(1)}_f$ on (\ref{eq_inst2_uneqpref_reg1_bdary_case2b}). Mass of class 2 waiting in queue 2 can increase in $[T^{(1)}_f,\tau_2(T^{(1)}_f)]$ if $\mu_2\gamma^{(2)}>\mu_2-\mu_1/\gamma^{(2)}$. }
    \label{fig:inst2reg1case22}
\end{figure}

\noindent\textbf{\textit{Regime II}~~$\mu_1\geq\mu_2\gamma^{(2)}$}.~~The quantity $T\overset{def}{=}\inf\{t>0~\vert~Q_2(t)=0\}$, representing the time at which queue 2 empties for the first time, is necessary to describe the unique EAP in Theorem \ref{mainthm_inst2_reg2}. 

We mention below the support boundaries and $T$ of the unique EAP, which we will refer later in the statement of Theorem \ref{mainthm_inst2_reg2}: 
\begin{enumerate}
    \item[\textbf{1.}] If $\mu_2\gamma^{(1)}>\mu_1\geq\mu_2\gamma^{(2)}$ and $\Lambda^{(1)}\geq \frac{\mu_1}{(1-\gamma^{(1)})\mu_2}\Lambda^{(2)}$; or $\mu_1\geq\mu_2\cdot\max\{\gamma^{(1)},\gamma^{(2)}\}$ and $\left(\frac{\mu_2}{\mu_1}-1\right)\Lambda^{(1)}>\Lambda^{(2)}$, then
    \begin{align}\label{eq_inst2_uneqpref_reg2_bdary_case1a}
            T^{(1)}_a&=\frac{\Lambda^{(2)}}{\mu_2\gamma^{(1)}}-\left(\frac{1}{\gamma^{(1)}}-1\right)\frac{\Lambda^{(1)}}{\mu_1},~T^{(1)}_f=\frac{\Lambda^{(1)}}{\mu_1},~T=\frac{\Lambda^{(2)}}{\mu_2-\mu_1},~T^{(2)}_a=-\frac{\Lambda^{(2)}}{\mu_2\gamma^{(2)}},~\text{and}~T^{(2)}_f=0.
    \end{align}
    \item[\textbf{2.}] If $\mu_2\gamma^{(1)}>\mu_1\geq\mu_2\gamma^{(2)}$ and $\Lambda^{(1)}<\frac{\mu_1}{(1-\gamma^{(1)})\mu_2}\Lambda^{(2)}$, then
    \begin{align}\label{eq_inst2_uneqpref_reg2_bdary_case1b}
            T^{(1)}_a=T^{(2)}_f&=\frac{1}{\mu_2}\left(\Lambda^{(2)}-\frac{(1-\gamma^{(1)})\mu_2}{\mu_1}\Lambda^{(1)}\right),~T^{(1)}_f=\gamma^{(1)}\frac{\Lambda^{(1)}}{\mu_1}+\frac{\Lambda^{(2)}}{\mu_2},\nonumber\\
            T&=\frac{\Lambda^{(2)}}{\mu_2}+\frac{(1-\gamma^{(1)})\Lambda^{(1)}}{\mu_2-\mu_1},~\text{and}~T^{(2)}_a=-\left(\frac{1}{\gamma^{(2)}}-1\right)\frac{\Lambda^{(2)}}{\mu_2}-(1-\gamma^{(1)})\frac{\Lambda^{(1)}}{\mu_1}.
    \end{align}
    \item[\textbf{3.}] If $\mu_1\geq\mu_2\cdot\max\{\gamma^{(1)},\gamma^{(2)}\}$ and $\left(\frac{\mu_2}{\mu_1}-1\right)\Lambda^{(1)}\leq\Lambda^{(2)}\leq\left(\frac{1}{\max\{\gamma^{(1)},\gamma^{(2)}\}}-1\right)\Lambda^{(1)}$:
    \begin{align}\label{eq_inst2_uneqpref_reg2_bdary_case2b}
            T^{(1)}_a&=\frac{1}{\mu_2}\left(\Lambda^{(2)}-\left(\frac{1}{\gamma^{(1)}}-1\right)\Lambda^{(1)}\right),~T=T^{(1)}_f=\frac{\Lambda^{(1)}+\Lambda^{(2)}}{\mu_2},~T^{(2)}_a=-\frac{\Lambda^{(2)}}{\mu_2\gamma^{(2)}},~\text{and}~T^{(2)}_f=0.
    \end{align}
    \item[\textbf{4.}] If $\mu_1\geq\mu_2\gamma^{(1)}>\mu_2\gamma^{(2)}$ and $\Lambda^{(2)}>\left(\frac{1}{\gamma^{(1)}}-1\right)\Lambda^{(1)}$:
    \begin{align}\label{eq_inst2_uneqpref_reg2_bdary_case2c}
        T^{(1)}_a=T^{(2)}_f&=-\left(\frac{1}{\gamma^{(1)}}-1\right)\frac{\Lambda^{(1)}}{\mu_2}+\frac{\Lambda^{(2)}}{\mu_2},~T^{(1)}_f=T=\frac{\Lambda^{(1)}+\Lambda^{(2)}}{\mu_2},~\text{and} \nonumber\\
        T^{(2)}_a&=-\left(\frac{1}{\gamma^{(1)}}-1\right)\frac{\Lambda^{(1)}}{\mu_2}-\left(\frac{1}{\gamma^{(2)}}-1\right)\frac{\Lambda^{(2)}}{\mu_2}.
    \end{align}
    \item[\textbf{5.}] If $\mu_1\geq\mu_2\gamma^{(2)}>\mu_2\gamma^{(1)}$ and $\Lambda^{(2)}>\left(\frac{1}{\gamma^{(2)}}-1\right)\Lambda^{(1)}$:
    \begin{align}\label{eq_inst2_uneqpref_reg2_bdary_case3c}
        T^{(1)}_a&=-\left(\frac{1}{\gamma^{(1)}}-\frac{1}{\gamma^{(2)}}\right)\frac{\Lambda^{(1)}}{\mu_2},~T^{(1)}_f=\frac{\Lambda^{(1)}}{\mu_2\gamma^{(2)}},~T=T^{(2)}_f=\frac{\Lambda^{(1)}+\Lambda^{(2)}}{\mu_2},~\text{and}~T^{(2)}_a=-\frac{1-\gamma^{(2)}}{\gamma^{(2)}}\frac{\Lambda^{(1)}+\Lambda^{(2)}}{\mu_2}.
    \end{align}
\end{enumerate}
\begin{theorem}\label{mainthm_inst2_reg2}
    If $\mu_1\geq\mu_2\gamma^{(2)}$ and $\gamma^{(1)}\neq\gamma^{(2)}$, HAS has a unique EAP. Moreover, the EAP has the following arrival rates and support boundaries:
    \begin{enumerate}
        \item[\textbf{1.}] If $\mu_2\gamma^{(1)}>\mu_1\geq\mu_2\gamma^{(2)}$:
        \begin{enumerate}
            \item[\textbf{1a.}] when $\Lambda^{(1)}\geq \frac{\mu_1}{(1-\gamma^{(1)})\mu_2}\Lambda^{(2)}$: 
            \begin{align*}
             (F^{(1)})^\prime(t)&=\begin{cases} \mu_2\gamma^{(1)}~~&\text{for}~t\in[T^{(1)}_a,\tau_1^{-1}(T)],\\ \mu_1\gamma^{(1)}~~&\text{for}~t\in[\tau_1^{-1}(T),T^{(1)}_f], \end{cases}~\text{and}~(F^{(2)})^\prime(t)=\mu_2\gamma^{(2)}~~\text{for}~t\in\left[T^{(2)}_a,T^{(2)}_f\right]   
            \end{align*}
            where $T^{(1)}_a,T^{(1)}_f,T^{(2)}_a,T^{(2)}_f$ and $T$ are in (\ref{eq_inst2_uneqpref_reg2_bdary_case1a}).
            
            \item[\textbf{1b.}] when $\Lambda^{(1)}< \frac{\mu_1}{(1-\gamma^{(1)})\mu_2}\Lambda^{(2)}$, 
            \begin{align*}
             (F^{(1)})^\prime(t)&=\begin{cases} \mu_2\gamma^{(1)}~~&\text{for}~t\in[T^{(1)}_a,\tau_1^{-1}(T)],\\ \mu_1\gamma^{(1)}~~&\text{for}~t\in[\tau_1^{-1}(T),T^{(1)}_f], \end{cases}~\text{and}~(F^{(2)})^\prime(t)=\mu_2\gamma^{(2)}~~~\text{for}~t\in[T^{(2)}_a,T^{(2)}_f]   
            \end{align*}
            where $T^{(1)}_a,T^{(1)}_f,T^{(2)}_a,T^{(2)}_f$ and $T$ are in (\ref{eq_inst2_uneqpref_reg2_bdary_case1b}).
        \end{enumerate}
        
        \item[\textbf{2.}] If $\mu_1\geq\mu_2\gamma^{(1)}>\mu_2\gamma^{(2)}$:
        \begin{enumerate}
            \item[\textbf{2a.}] when $\left(\frac{\mu_2}{\mu_1}-1\right)\Lambda^{(1)}>\Lambda^{(2)}$, the EAP has a closed form same as case 1a. 
            
            \item[\textbf{2b.}] when $\left(\frac{1}{\gamma^{(1)}}-1\right)\Lambda^{(1)}\geq\Lambda^{(2)}\geq\left(\frac{\mu_2}{\mu_1}-1\right)\Lambda^{(1)}$, \begin{align*}
            (F^{(1)})^\prime(t)&=\mu_2\gamma^{(1)}~~\text{for}~t\in[T^{(1)}_a,T^{(1)}_f]~\text{and}~(F^{(2)})^\prime(t)=\mu_2\gamma^{(2)}~~\text{for}~t\in\left[T^{(2)}_a,T^{(2)}_f\right]    
            \end{align*} 
            where $T^{(1)}_a,T^{(1)}_f,T^{(2)}_a,T^{(2)}_f$ and $T$ are in (\ref{eq_inst2_uneqpref_reg2_bdary_case2b}). \vspace{0.05in} 
            
            \item[\textbf{2c.}] when $\Lambda^{(2)}>\left(\frac{1}{\gamma^{(1)}}-1\right)\Lambda^{(1)}$, 
            \begin{align*}
             (F^{(1)})^\prime(t)&=\mu_2\gamma^{(1)}~~~\text{for}~t\in[T^{(1)}_a,T^{(1)}_f]~\text{and}~(F^{(2)})^\prime(t)=\mu_2\gamma^{(2)}~~~\text{for}~t\in[T^{(2)}_a,T^{(2)}_f]   
            \end{align*}
            where $T^{(1)}_a,T^{(1)}_f,T^{(2)}_a,T^{(2)}_f$ and $T$ are in (\ref{eq_inst2_uneqpref_reg2_bdary_case2c}).
        \end{enumerate} 
        
        \item[\textbf{3.}] If $\mu_1\geq\mu_2\gamma^{(2)}>\mu_2\gamma^{(1)}$:
        \begin{enumerate}
            \item[\textbf{3a.}] when $\left(\frac{\mu_2}{\mu_1}-1\right)\Lambda^{(1)}>\Lambda^{(2)}$, the EAP has a closed form same as case 1a.\vspace{0.05in}
        
            \item[\textbf{3b.}] when $\left(\frac{1}{\gamma^{(2)}}-1\right)\Lambda^{(1)}\geq\Lambda^{(2)}\geq\left(\frac{\mu_2}{\mu_1}-1\right)\Lambda^{(1)}$, the EAP has a closed form same as case 2b.\vspace{0.05in}
            
            \item[\textbf{3c.}] when $\Lambda^{(2)}>\left(\frac{1}{\gamma^{(2)}}-1\right)\Lambda^{(1)}$, 
            \begin{align*}
            (F^{(1)})^\prime(t)&=\mu_2\gamma^{(1)}~~\text{for}~t\in[T^{(1)}_a,T^{(1)}_f]~\text{and}~(F^{(2)})^\prime(t)=\mu_2\gamma^{(2)}~~\text{for}~t\in[T^{(2)}_a,0]\cup[T^{(1)}_f,T^{(2)}_f]    
            \end{align*}
            where $T^{(1)}_a,T^{(1)}_f,T^{(2)}_a,T^{(2)}_f$ and $T$ are in (\ref{eq_inst2_uneqpref_reg2_bdary_case3c}).
        \end{enumerate}
    \end{enumerate}
\end{theorem}
 
\begin{remark}
    \emph{Illustrative EAPs and resulting queue length processes of cases 1a, 1b, 2a, 2b, 2c, 3a, 3b,  and 3c of Theorem \ref{mainthm_inst2_reg2} are illustrated, respectively, in Figures \ref{fig:inst2reg2case11}, \ref{fig:inst2reg2case12}, \ref{fig:inst2reg2case21}, \ref{fig:inst2reg2case22}, \ref{fig:inst2reg2case23}, \ref{fig:inst2reg2case31}, \ref{fig:inst2reg2case32} and \ref{fig:inst2reg2case33}.}
\end{remark}  
\begin{remark}\label{rem:disj_support}
    \emph{\textbf{\textit{(Intuition behind broken support in EAPs for instances under case 3c)}}~We show by Lemma \ref{lem_inst2_uneqpref_queue1idle} and \ref{lem_inst2_uneqpref_reg2_case12_sign} in \ref{appndx:inst2_uneqpref} that in every EAP of  HAS under case 3 ($\mu_1\geq\mu_2\gamma^{(2)}>\mu_2\gamma^{(1)}$) of Theorem \ref{mainthm_inst2_reg2}, class 1 users arrive from queue 1 to queue 2 in $[0,T^{(1)}_f]$ at rate $\mu_1$. As a result, by Lemma \ref{lem_threshold_behav_inst2}, class 2 users cannot arrive in $[0,T^{(1)}_f]$. However, in every EAP, class 2 users start arriving before time zero (otherwise queue 2 stays empty at time zero and every class 2 user can be strictly better off arriving at time zero). Therefore there are two possible ways in which class 2 users can arrive in any EAP: \textbf{1)} the whole class 2 population arrives before time zero, or, \textbf{2)} a fraction of the class 2 population arrives after $T^{(1)}_f$ and the rest arrives before time zero. While proving the existence of a unique EAP in case 3 of Theorem \ref{mainthm_inst2_reg2}, we argue that, if $\Lambda^{(2)}$ is under an identified threshold of $\left(\frac{1}{\gamma^{(2)}}-1\right)\Lambda^{(1)}$, the whole class 2 population arrives before time zero and waits in queue 2 at time zero, as can be seen for cases 3a and 3b, respectively, in Figure \ref{fig:inst2reg2case31} and \ref{fig:inst2reg2case32}. After $\Lambda^{(2)}$ crosses that threshold, a fraction of the class 2 population of mass $\gamma^{(2)}\left[\Lambda^{(2)}-\left(\frac{1}{\gamma^{(2)}}-1\right)\Lambda^{(1)}\right]$ arrives after $T^{(1)}_f$ and the rest arrives before time zero, as can be seen for case 3c in Figure \ref{fig:inst2reg2case33}.}
\end{remark}

\begin{figure}[h]
    \centering
    \includegraphics[width=12cm]{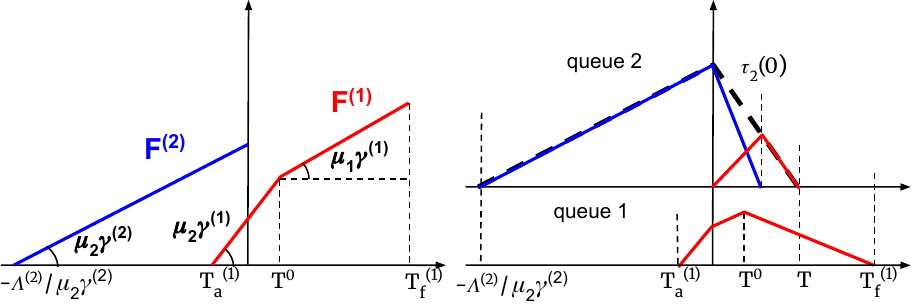}
    \caption{Illustrative EAP (left) and resulting queue length process (right) of HAS with $\mu_1\geq\mu_2\gamma^{(2)}$ under case 1a of Theorem \ref{mainthm_inst2_reg2} with the assumption $T^0>0$, where $T^0\overset{def.}{=}\tau_1^{-1}(T)$.  }
    \label{fig:inst2reg2case11}
\end{figure}

\begin{figure}[h]
    \centering
    \includegraphics[width=12cm]{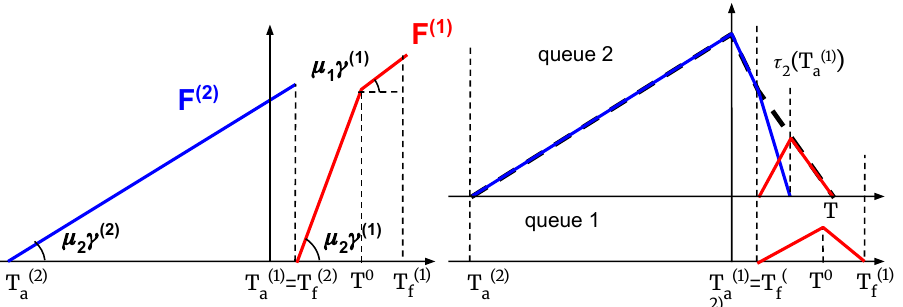}
    \caption{Illustrative EAP (left) and resulting queue length process (right) of HAS with $\mu_2\gamma^{(1)}>\mu_1\geq\mu_2\gamma^{(2)}$ under case 1b of Theorem \ref{mainthm_inst2_reg2}. $T^0\overset{def.}{=}\tau_1^{-1}(T)$. }
    
    \label{fig:inst2reg2case12}.
\end{figure}

\begin{figure}[h]
    \centering
    \includegraphics[width=12cm]{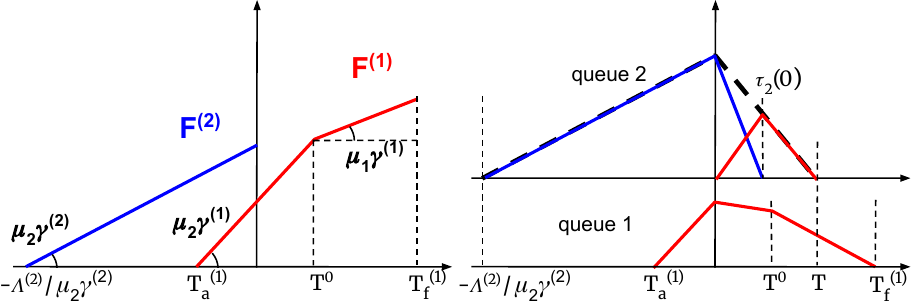}
    \caption{Illustrative EAP (left) and resulting queue length process (right) of HAS with $\mu_1\geq\mu_2\gamma^{(1)}>\mu_2\gamma^{(2)}$ under case 2a of Theorem \ref{mainthm_inst2_reg2} with the assumption $T^0>0$ where $T^0\overset{def.}{=}\tau_1^{-1}(T)$. }
    \label{fig:inst2reg2case21}
\end{figure}

\begin{figure}[h]
    \centering
    \includegraphics[width=12cm]{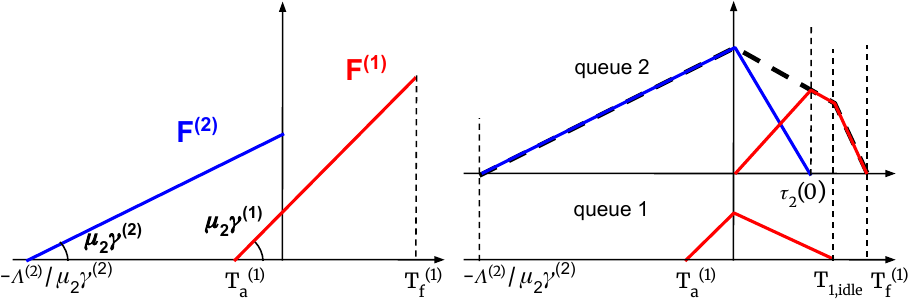}
    \caption{Illustrative EAP (left) and queue length process (right) of HAS with $\mu_1\geq\mu_2\gamma^{(1)}>\mu_2\gamma^{(2)}$ under case 2b of Theorem \ref{mainthm_inst2_reg2} with assumption $\tau_2(0)<T_{1,idle}$, where $T_{1,idle}$ is the time at which queue 1 empties after the network starts at time zero. In the illustration, we assumed $\mu_1<\mu_2$, causing mass of class 1 users waiting in queue 2 to decrease in $[\tau_2(0),T_{1,idle}]$, but it can be non-decreasing if $\mu_1\geq\mu_2$. For the same assumption, the overall mass of the two classes is decreasing in $[0,T_{1,idle}]$ and can be non-decreasing if  $\mu_1\geq\mu_2$. }
    \vspace{-0.2in}
    \label{fig:inst2reg2case22}
\end{figure}

\begin{figure}[h]
    \centering
    \includegraphics[width=12cm]{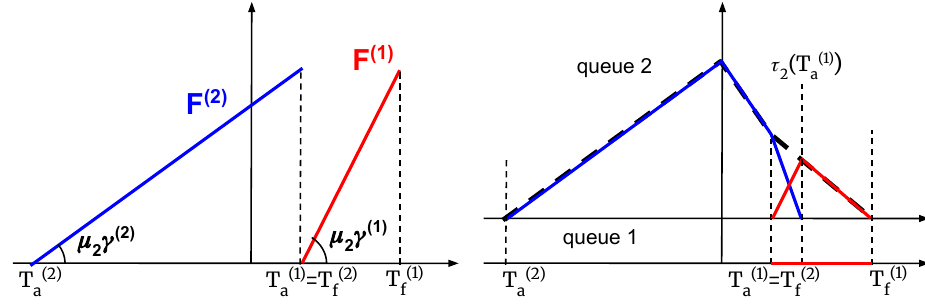}
    \caption{Illustrative EAP (left) and resulting queue length process (right) of HAS with $\mu_1\geq\mu_2\gamma^{(1)}>\mu_2\gamma^{(2)}$ under case 2c of Theorem \ref{mainthm_inst2_reg2}. }
    \label{fig:inst2reg2case23}
\end{figure}

\begin{figure}[h]
    \centering
    \includegraphics[width=12cm]{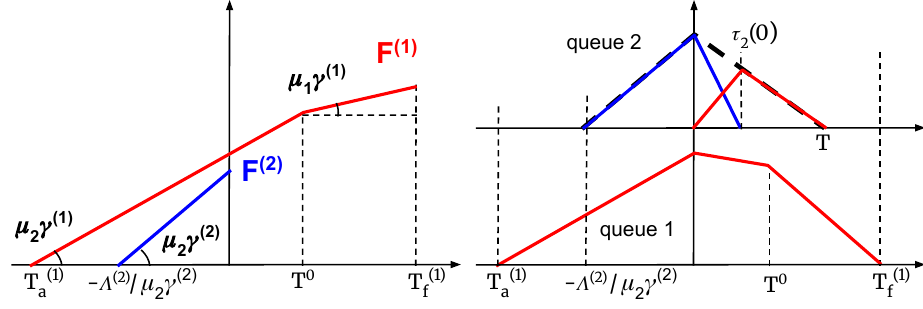}
    \caption{Illustrative EAP of HAS (left) and queue length process (right) with $\mu_1\geq\mu_2\gamma^{(2)}>\mu_2\gamma^{(1)}$ under case 3a of Theorem \ref{mainthm_inst2_reg2} with the assumption $T^0=\tau_1^{-1}(T)>0$. }
    \label{fig:inst2reg2case31}
\end{figure}

\begin{figure}[h]
    \centering
    \includegraphics[width=12cm]{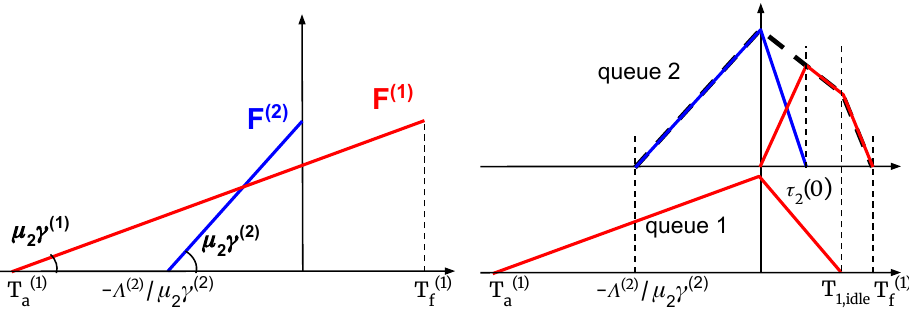}
    \caption{Illustrative EAP (left) and resulting queue length process (right) of HAS with $\mu_1\geq\mu_2\gamma^{(2)}>\mu_2\gamma^{(1)}$ under case 3b of Theorem \ref{mainthm_inst2_reg2} with the assumption $\tau_2(0)<T_{1,idle}$, where $T_{1,idle}$ is the time at which queue 1 empties after the network starts at time zero. }
    \label{fig:inst2reg2case32}
\end{figure}

\begin{figure}[h]
    \centering
    \includegraphics[width=12cm]{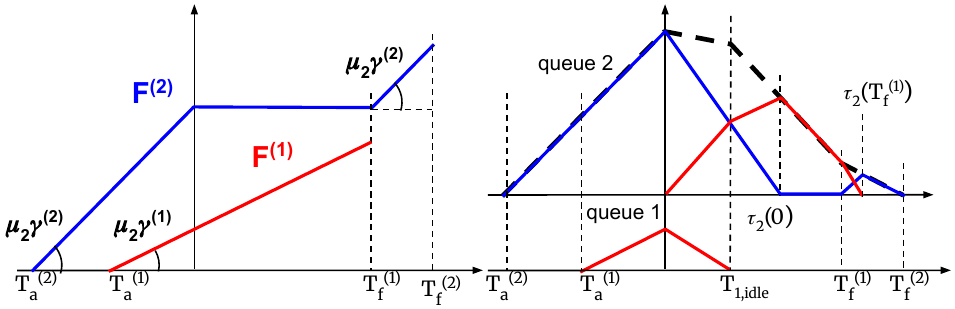}
    \caption{Illustrative EAP (left) and queue length process (right) of HAS with $\mu_1\geq\mu_2\gamma^{(2)}>\mu_2\gamma^{(1)}$ under case 3c of Theorem \ref{mainthm_inst2_reg2} with the assumption $T^{(1)}_f>\tau_2(0)>T_{1,idle}$, where $T_{1,idle}$ is the time at which queue 1 empties after time zero.}
    \label{fig:inst2reg2case33}
\end{figure}

\subsection{Equal Preferences $\gamma^{(1)}=\gamma^{(2)}=\gamma$}\label{sec_inst2_eqpref}
We state the main result of this section separately for two regimes as Theorem \ref{mainthm_inst2_eqpref_reg1} (for $\mu_1<\mu_2\gamma$), and \ref{mainthm_inst2_eqpref_reg2} (for $\mu_1\geq\mu_2\gamma$), since they have different EAP structures. Proofs of these two theorems (in  \ref{appndx:inst2_eqpref}) are similar in structure to that of Theorem \ref{mainthm_inst1_eqpref}. 

$T^{(1)}_a,T^{(1)}_f,T^{(2)}_a,T^{(2)}_f$ are as in the unequal preference case. Let $T_a=\inf \mathcal{S}(F^{(1)})\cup\mathcal{S}(F^{(2)})$ and $T_f=\sup\mathcal{S}(F^{(1)})\cup\mathcal{S}(F^{(2)})$. 

\begin{theorem}\label{mainthm_inst2_eqpref_reg1}
    If $\mu_1<\mu_2\gamma$,\begin{enumerate}
        \item when $\Lambda^{(1)}>\frac{\mu_1}{\mu_2-\mu_1}\Lambda^{(2)}$, the EAP is unique and has the following arrival rates and support boundaries:
        \begin{align*}
            (F^{(1)})^\prime(t)&=\begin{cases}
                \mu_1~~&\text{if}~~t\in[T^{(1)}_a,\tau_1^{-1}(T^{(2)}_f)],\\
                \mu_1\gamma~~&\text{if}~~t\in[\tau_1^{-1}(T^{(2)}_f),T^{(1)}_f].
            \end{cases},~\text{and}~(F^{(2)})^\prime(t)=\begin{cases}
                \mu_2\gamma~~&\text{if}~t\in[T^{(2)}_a,0],\\
                \mu_2\gamma-\mu_1~~&\text{if}~t\in[0,T^{(2)}_f],
            \end{cases}\\
             \text{where}~T^{(1)}_a&=-\frac{1-\gamma}{\gamma}\left(\frac{\Lambda^{(1)}}{\mu_1}-\frac{\Lambda^{(2)}}{\mu_2-\mu_1}\right),~T^{(1)}_f=\frac{\Lambda^{(1)}}{\mu_1},~T^{(2)}_a=-\frac{1-\gamma}{\gamma}\frac{\Lambda^{(2)}}{\mu_2-\mu_1},~\text{and}~T^{(2)}_f=\frac{\Lambda^{(2)}}{\mu_2-\mu_1}.
        \end{align*}
        \item when $\Lambda^{(1)}\leq\frac{\mu_1}{\mu_2-\mu_1}\Lambda^{(2)}$, set of EAPs is given by the convex set of joint customer arrival profiles $\{F^{(1)},F^{(2)}\}$ satisfying: \textbf{1)} $(F^{(1)})^\prime(t)=0,~(F^{(2)})^\prime(t)=\mu_2\gamma$ for all $t\in[T_a,0]$, \textbf{2)} $(F^{(1)})^\prime(t)\leq\mu_1,~(F^{(1)})^\prime(t)+(F^{(2)})^\prime(t)=\mu_2\gamma$ for all $t\in[0,T_f]$, and \textbf{3)} $F^{(1)}(T_f)=\Lambda^{(1)},~F^{(2)}(T_f)=\Lambda^{(2)}$, where $T_a=-\left(\frac{1}{\gamma}-1\right)\frac{\Lambda^{(1)}+\Lambda^{(2)}}{\mu_2}$ and $T_f=\frac{\Lambda^{(1)}+\Lambda^{(2)}}{\mu_2}$.
    \end{enumerate}
\end{theorem}
\begin{remark}
    \emph{If $\Lambda^{(1)}>\frac{\mu_1}{\mu_2-\mu_1}\Lambda^{(2)}$, the EAPs in case 1a and 2a of Theorem \ref{mainthm_inst2_reg1}, respectively, upon taking limits $\gamma^{(1)}\to\gamma+,~\gamma^{(2)}=\gamma$ and $\gamma^{(1)}=\gamma,~\gamma^{(2)}\to\gamma+$ converges to the EAP under case 1 of Theorem \ref{mainthm_inst2_eqpref_reg1}. On the other hand, if $\Lambda^{(1)}\leq\frac{\mu_1}{\mu_2-\mu_1}\Lambda^{(2)}$, the EAPs in case 1b and 2b of Theorem \ref{mainthm_inst2_reg1}, respectively, upon taking limits $\gamma^{(1)}\to\gamma+,~\gamma^{(2)}=\gamma$ and $\gamma^{(1)}=\gamma,~\gamma^{(2)}\to\gamma+$ might converge to different limits, but both the limits will be contained in the convex set of EAPs under case 2 of \ref{mainthm_inst2_reg1}.}
\end{remark}

The quantity $T\overset{def.}{=}\inf\left\{t>0~\vert~Q_2(t)=0\right\}$  will be necessary to describe the structure of the EAP in the regime $\mu_1\leq\mu_2\gamma$.

\begin{theorem}\label{mainthm_inst2_eqpref_reg2}
    If $\mu_1\geq\mu_2\gamma$,\begin{enumerate}
        \item when $\left(\frac{\mu_2}{\mu_1}-1\right)\Lambda^{(1)}>\Lambda^{(2)}$, the EAP is unique and has the following arrival rates and support boundaries:
        \begin{align*}
            (F^{(1)})^\prime(t)&=\begin{cases}
                \mu_2\gamma~~&\text{if}~t\in[T^{(1)}_a,\tau_1^{-1}(T)],\\
                \mu_1\gamma~~&\text{if}~t\in[\tau_1^{-1}(T),T^{(1)}_f],
            \end{cases},~\text{and}~(F^{(2)})^\prime(t)=\mu_2\gamma~~\text{if}~t\in\left[-\frac{\Lambda^{(2)}}{\mu_2\gamma},0\right],\\
            \text{where}~T^{(1)}_a&=\frac{\Lambda^{(2)}}{\mu_2\gamma}-\left(\frac{1}{\gamma}-1\right)\frac{\Lambda^{(1)}}{\mu_1},~T=\frac{\Lambda^{(2)}}{\mu_2-\mu_1},~\text{and}~T^{(1)}_f=\frac{\Lambda^{(1)}}{\mu_1}.    
        \end{align*}
        
        \item when $\left(\frac{1}{\gamma}-1\right)\Lambda^{(1)}\geq\Lambda^{(2)}\geq\left(\frac{\mu_2}{\mu_1}-1\right)\Lambda^{(1)}$, the EAP is unique and has the following arrival rates and support boundaries:
        \begin{align*}
            (F^{(1)})^\prime(t)&=\mu_2\gamma~~\text{if}~t\in[T^{(1)}_a,T^{(1)}_f],~\text{and}~(F^{(2)})^\prime(t)=\mu_2\gamma~~\text{if}~t\in\left[-\frac{\Lambda^{(2)}}{\mu_2\gamma},0\right],\\
            \text{where}~T^{(1)}_a&=\frac{1}{\mu_2}\left(\Lambda^{(2)}-\left(\frac{1}{\gamma}-1\right)\Lambda^{(1)}\right),~\text{and}~T=T^{(1)}_f=\frac{\Lambda^{(1)}+\Lambda^{(2)}}{\mu_2}.
        \end{align*}
        
        \item If $\Lambda^{(2)}>\left(\frac{1}{\gamma}-1\right)\Lambda^{(1)}$, set of EAPs is given by the convex set of joint customer arrival profiles $\{F^{(1)},F^{(2)}\}$ satisfying: \textbf{1)} $(F^{(1)})^\prime(t)=0$ and $(F^{(2)})^\prime(t)=\mu_2\gamma$ for all $t\in[T_a,0]$, \textbf{2)} $(F^{(1)})^\prime(t)+(F^{(2)})^\prime(t)=\mu_2\gamma$ for all $t\in[0,T_f]$, and \textbf{3)} $F^{(1)}(T_f)=\Lambda^{(1)},~F^{(2)}(T_f)=\Lambda^{(2)}$, where $T_a=-\left(\frac{1}{\gamma}-1\right)\frac{\Lambda^{(1)}+\Lambda^{(2)}}{\mu_2}$ and $T_f=\frac{\Lambda^{(1)}+\Lambda^{(2)}}{\mu_2}$.
    \end{enumerate}
\end{theorem}

\begin{remark}
\emph{If $\left(\frac{\mu_2}{\mu_1}-1\right)\Lambda^{(1)}>\Lambda^{(2)}$, the EAPs in case 2a and 3a of Theorem \ref{mainthm_inst2_reg2}, respectively upon taking limits $\gamma^{(1)}\to\gamma+, \gamma^{(2)}=\gamma$ and $\gamma^{(1)}=\gamma, \gamma^{(2)}\to\gamma+$, converges to the EAP under case 1 of Theorem \ref{mainthm_inst2_eqpref_reg2}. If $\left(\frac{1}{\gamma}-1\right)\Lambda^{(1)}\geq\Lambda^{(2)}\geq\left(\frac{\mu_2}{\mu_1}-1\right)\Lambda^{(1)}$, upon taking the same limits, the EAPs in case 2b and 3b of Theorem \ref{mainthm_inst2_reg2} converges to the EAP under case 2 of Theorem \ref{mainthm_inst2_eqpref_reg2}. If $\Lambda^{(2)}>\left(\frac{1}{\gamma}-1\right)\Lambda^{(1)}$, upon taking the same limits, the EAPs in case 2c and 3c of Theorem \ref{mainthm_inst2_reg2} converges to different limits and both of them are contained in the convex set of EAPs under case 3 of Theorem \ref{mainthm_inst2_eqpref_reg2}.}
\end{remark}

\section{Conclusion}\label{sec:conclusion}
Arrival games to a single queue have been well studied in the literature, but the important area of arrival games to a general
queuing network, has limited literature. As we discovered, this may be because the problem complexity may increase substantially with the network size. 
We studied a simple two queue, two class network with a linear cost structure, where we analyzed
in detail two heterogeneous routing configurations: (1) where the customer classes
arrived at the same queue but departed in different queues, and (2) 
where the customer classes arrived at different queues but departed from the same queue. We discovered non-trivial customer behaviour
not apparent in a single queue or in networks studied thus far in the literature. In a specific setting we found that the parameter 
space had to be partitioned into eight distinct regions, where  each region  had its own closed form parametric representation of the
arrival equilibrium profile. We found that although for most
set of parameters, the equilibrium profile is unique,
there exists settings where the collection of equilibrium profiles
is not unique but a convex set.
While in a single queue, multi-class customers arrive in contiguous,
non-overlapping intervals, in our two queue setting
there are regions where, in equilibrium, different class arrival times may overlap. Further, there exist regions, where a single class customer may arrive in disjoint intervals. 
The broad message of the paper is mixed and motivates further research - it suggests that 
even for more complex networks one may expect a unique equilibrium profile for most parameter settings. It also suggests that 
the number of different solution structures may blow up with the network size, so that learning the structure in any instance may be a difficult task.

\section*{Acknowledgement}

Our work is supported by Department of Atomic Energy, Government of India, under project no. RTI4001.

\bibliographystyle{plain} 
\bibliography{references}


\newpage
\appendix
\section*{Appendix}

For the ease of analysis, we define the set $E_j\overset{def.}{=}\{t~\vert~t<0~\text{or}~Q_j(t)>0\}$ to be the set of all times when queue $j$ has a positive waiting time for $j=1,2$. For any set $S\subseteq\mathbb{R}$, we use the notations $\overline{S}$ and $S^o$ to respectively denote the closure and interior of $S$. 

By Lemma \ref{lem_eq_has_no_jump}, we consider only absolutely continuous arrival profiles $F^{(1)}$ and $F^{(2)}$ and as a result, $\mathcal{S}(F^{(1)})$ and $\mathcal{S}(F^{(2)})$ cannot have isolated points.  

\section{Proofs of Lemmas in Section \ref{sec_hetdepartures}}

As argued in Section \ref{sec_hetdepartures}, we assume $\mu_1>\mu_2$ for the following discussion. 

\subsection{Proofs of Lemmas in Section \ref{sec_inst1_uneqpref} (HDS with $\gamma^{(1)}\neq\gamma^{(2)}$)}\label{appndx:inst1_uneqpref}
Several of the lemmas proved below will extend to the case where $\gamma^{(1)}=\gamma^{(2)}$. So, whenever some lemma is applicable only for the case of unequal preferences, we mention it explicitly in the lemma statement. 

Lemma \ref{lem_inst1_uneqpref_mixedarrival} specifies the state of queue 2 when users of both classes are arriving together in queue 1 and is a special case of Lemma \ref{lem_inst1_uneqpref_queue2_busy}. This lemma helps us prove Lemma \ref{lem_threshold_behav_inst1} and together with Lemma \ref{lem_inst1_uneqpref_derv_of_tau2}, it will help us calculate the arrival rate of the two classes in equilibrium. 

\begin{lemma}\label{lem_inst1_uneqpref_mixedarrival}
If $\gamma^{(1)}\neq\gamma^{(2)}$ and $\mu_1>\mu_2$, in the EAP, $t\in (\mathcal{S}(F^{(1)}))^o\cap (\mathcal{S}(F^{(2)}))^o$ implies $\tau_1(t)\in\overline{E_2}$.
\end{lemma}
\begin{proof}
    Assuming contradiction, \textit{i.e.}, $\tau_1(t)\notin \overline{E_2}$, there is a neighbourhood $[\tau_1(t)-\delta,\tau_1(t)+\delta]$ of $\tau_1(t)$, where queue 2 is empty. Since the arrival profiles $F^{(1)},~F^{(2)}$ are absolutely continuous, $\tau_1(\cdot)$ will be continuous. As a result, we have $[t-\epsilon,t+\epsilon]\subseteq (\mathcal{S}(F^{(1)}))^o\cap (\mathcal{S}(F^{(2)}))^o$ such that, $Q_2(\tau_1(s))=0$ for all $s\in[t-\epsilon,t+\epsilon]$. Hence every class 2 user arriving in $[t-\epsilon,t+\epsilon]$ waits only at queue 1. For both the classes $i=1,2$ departure time $\tau_{\mathbf{F}}^{(i)}(s)=\tau_1(s)$ in $[t-\epsilon,t+\epsilon]$ and hence cost $C_{\mathbf{F}}^{(i)}(s)=\tau_1(s)-\gamma^{(i)} s$. Since, $[t-\epsilon,t+\epsilon]\subseteq (\mathcal{S}(F^{(1)}))^o\cap (\mathcal{S}(F^{(2)}))^o$, by definition of EAP, for $i=1,2$ $C_{\mathbf{F}}^{(i)}(\cdot)$ is constant in $[t-\epsilon,t+\epsilon]$. Therefore, $(C_{\mathbf{F}}^{(i)})^\prime(s)=\tau_1^\prime(s)-\gamma^{(i)}=0$ in $[t-\epsilon,t+\epsilon]$, giving us $\tau_1^\prime(s)=\gamma^{(1)}=\gamma^{(2)}$, contradicting the assumption $\gamma^{(1)}\neq\gamma^{(2)}$. 
\end{proof}

Lemma \ref{lem_inst1_uneqpref_derv_of_tau2} helps us relate the arrival rate of class 2 users and will be referred to ubiquitously while proving the structural properties. 

\begin{lemma}\label{lem_inst1_uneqpref_derv_of_tau2}
    For every absolutely continuous joint user arrival profile $\{F^{(1)},F^{(2)}\}$, $\tau_{\mathbf{F}}^{(2)}(\cdot)$ is differentiable a.e. in the set of times $\tau_1^{-1}(\overline{E_2})\overset{def.}{=}\{t~\vert~\tau_1(t)\in\overline{E_2}\}$ with derivative $(\tau_{\mathbf{F}}^{(2)})^\prime(t)=\frac{(F^{(2)})^\prime(t)}{\mu_2}$ a.e. As a consequence, $(C_{\mathbf{F}}^{(2)})^\prime(t)=\frac{(F^{(2)})^\prime(t)}{\mu_2}-\gamma^{(2)}$ a.e. in  $\tau_1^{-1}(\overline{E_2})$. 
\end{lemma}
\begin{proof}
    We divide the proof into two separate cases: 
    \begin{enumerate}[leftmargin=*]
        \item  In $\overline{E_1}^c$, we will have $\tau_1(t)=t$ and as a result $\overline{E_1}^c\cap \tau_1^{-1}(\overline{E_2})=\overline{E_1}^c\cap \overline{E_2}$. Therefore, using (\ref{eq:derv_of_tau}) $(\tau_{\mathbf{F}}^{(2)})^\prime(t)=\tau_2^\prime(t)=\frac{(F^{(2)})^\prime(t)}{\mu_2}$ a.e. in $\overline{E_1}^c\cap \tau_1^{-1}(\overline{E_2})$. 

        \item  In $\overline{E_1}$, by (\ref{eq:derv_of_tau}), $\tau_1^\prime(\cdot)$ exists a.e. Again using (\ref{eq:derv_of_tau}), since $\tau_1(t)\in \overline{E_2}$, we have $\tau_2^\prime(\tau_1(t))=\frac{A_2^\prime(\tau_1(t))}{\mu_2}$. Therefore, $(\tau_{\mathbf{F}}^{(2)})^\prime(t)=\frac{A_2^\prime(\tau_1(t))\tau_1^\prime(t)}{\mu_2}$ a.e. in $\overline{E_1}\cap\tau_1^{-1}(\overline{E_2})$. Since $A_2(\tau_1(t))=F^{(2)}(t)$, by chain rule, $(\tau_{\mathbf{F}}^{(2)})^\prime(t)=\frac{(F^{(2)})^\prime(t)}{\mu_2}$ a.e. in $\overline{E_1}\cap\tau_1^{-1}(\overline{E_2})$. 
    \end{enumerate}

    Since $C_{\mathbf{F}}^{(2)}(t)=\tau_{\mathbf{F}}^{(2)}(t)-\gamma^{(2)} t$, the second statement of the lemma follows trivially. 
\end{proof}

Lemma \ref{lem_appndx_inst2_uneqpref_threshold_behav_suff_cond} implies that $\mu_1\leq\mu_2\cdot\max\left\{1,\frac{\gamma^{(2)}}{\gamma^{(1)}}\right\}$ is sufficient for $\mathcal{S}(F^{(1)})$ and $(\mathcal{S}(F^{(2)}))^o$ to be disjoint, and is helpful for proving sufficiency of the condition in Lemma \ref{lem_threshold_behav_inst1}.

\begin{lemma}\label{lem_inst1_uneqpref_suff_cond_for_disj_arrival}
    If $\mu_1\gamma^{(1)}\leq\mu_2\gamma^{(2)}$, in every EAP, $\mathcal{S}(F^{(1)})$ and $(\mathcal{S}(F^{(2)}))^o$ are disjoint. 
\end{lemma}

\begin{proof}
    Proof of Lemma \ref{lem_inst1_uneqpref_suff_cond_for_disj_arrival} follows the argument mentioned after Lemma \ref{lem_threshold_behav_inst1} for proving sufficiency of $\mu_1\leq\mu_2\cdot\max\left\{1,\frac{\gamma^{(2)}}{\gamma^{(1)}}\right\}$ for $\mathcal{S}(F^{(1)}), \mathbf{S}(F^{(2)})$ to be non-overlapping, except while proving $(C_{\mathbf{F}}^{(2)})^\prime(t)=\frac{(F^{(2)})^\prime(t)}{\mu_2}-\gamma^{(2)}$ in $[t_1,t_2]$, we use Lemma \ref{lem_appndx_inst2_uneqpref_queu1_engaged_mixedarrival} to argue queue 2 stays engaged at $\tau_1(t)$ and then use Lemma \ref{lem_inst1_uneqpref_derv_of_tau2} to have $(\tau_{\mathbf{F}}^{(2)})^\prime(t)=\frac{(F^{(2)})^\prime(t)}{\mu_2}$ in $[t_1,t_2]$.
    
\end{proof}

As we stated in the main body, $T_{i,a}=\inf \mathcal{S}(F^{(i)})$ and $T_{i,f}=\sup\mathcal{S}(F^{(i)})$ for the two classes $i=1,2$. Note that $T^{(1)}_a,T^{(1)}_f,T^{(2)}_a,T^{(2)}_f$ are all finite if $\mathbf{F}=\{F^{(1)},F^{(2)}\}$ is an EAP. Lemma \ref{lem_appndx_inst1_uneqpref_queue2idle} and \ref{lem_appndx_inst1_uneqpref_queue1idle} are respectively the first and second statements of Lemma \ref{lem_queue1_and_2_idle}.  

\begin{lemma}{\textbf{(First statement of Lemma \ref{lem_queue1_and_2_idle})}}\label{lem_appndx_inst1_uneqpref_queue2idle}
    In the EAP, if $\mathcal{S}(F^{(2)})=[T^{(2)}_a,T^{(2)}_f]$, then queue 2 will have zero waiting time at $\tau_1(T^{(2)}_f)$.
\end{lemma}

\begin{proof}
    Assume contradiction, \textit{i.e.}, $\tau_1(T^{(2)}_f)\in E_2$. For $\delta>0$ chosen sufficiently small, image of $[T^{(2)}_f,T^{(2)}_f+\delta]$ under $\tau_1(\cdot)$ is contained in $E_2$. Therefore by Lemma \ref{lem_inst1_uneqpref_derv_of_tau2}, $(\tau_{\mathbf{F}}^{(2)})^\prime(t)=\frac{(F^{(2)})^\prime(t)}{\mu_2}$ in $[T^{(2)}_f,T^{(2)}_f+\delta]$. Since $(F^{(2)})^\prime(t)=0$ in $(T^{(2)}_f,T^{(2)}_f+\delta]$, we get $\tau_{\mathbf{F}}^{(2)}(T^{(2)}_f+\delta)=\tau_{\mathbf{F}}^{(2)}(T^{(2)}_f)$. As a result, the class 2 user arriving at $T^{(2)}_f$ will be strictly better off arriving at $T^{(2)}_f+\delta$, contradicting that $\{F^{(1)},F^{(2)}\}$ is an EAP.
\end{proof}

We now prove below Lemma \ref{lem_supp_are_intervals_inst1} about the structure of the supports of the two classes in EAP.\\

\noindent\textit{\textbf{Proof of Lemma \ref{lem_supp_are_intervals_inst1}:}}~~\textbf{(1)}~On assuming contradiction, there exists $t_1,t_2\in\mathcal{S}(F^{(1)})\cup\mathcal{S}(F^{(2)})$ such that $t_2>t_1$ and $F^{(1)}(t_1)+F^{(2)}(t_1)=F^{(1)}(t_2)+F^{(2)}(t_2)$. Note that $t_2<\max\{T^{(1)}_f,T^{(2)}_f\}$. Otherwise, $\max\{T^{(1)}_f,T^{(2)}_f\}$ will be an isolated point of $\mathcal{S}(F^{(1)})\cup\mathcal{S}(F^{(2)})$, contradicting absolute continuity of $F^{(1)}$ and $F^{(2)}$. 

With $t_2<\max\{T^{(1)}_f,T^{(2)}_f\}$, a positive mass of users must arrive after $t_2$. Therefore, none of the classes can have zero waiting time in $[t_1,t_2]$, implying $[t_1,t_2]\subseteq E_1\cup E_2$. There can be two situations now:

\begin{itemize}[leftmargin=*]
    \item \textbf{If $t_1\in E_1$:}~~For $\delta>0$ picked sufficiently small, $[t_1,t_1+\delta]\subseteq E_1$. Since $A_1^\prime(s)=0$ for every $s\in[t_1,t_1+\delta]$, by (\ref{eq:derv_of_tau}), $\tau_1(t_1)=\tau_1(t_1+\delta)$. As a result, for both classes $i=1,2$~~~$\tau_{\mathbf{F}}^{(i)}(t_1)=\tau_{\mathbf{F}}^{(i)}(t_1+\delta)$. Therefore, the user arriving at $t_1$ will be strictly better off arriving at $t_1+\delta$.
    \vspace{0.05in}
    \item \textbf{If $t_1\notin E_1$:}~~Since no users arrive in $[t_1,t_2]$, queue 1 has zero waiting time in $[t_1,t_2]$ (by (\ref{eq:derv_of_tau})) implying, $[t_1,t_2]\subseteq E_2$. Therefore some positive mass of class 2 users must have arrived before $t_1$. Let $\tilde{t}=\sup\{t\leq t_1~\vert~t\in \mathcal{S}(F^{(2)})\}=$ the last time before $t_1$ some class 2 user has arrived. Since no class 2 users arrives in $[\tilde{t},t_2]$, $Q_2(\cdot)$ cannot increase in $[\tau_1(\tilde{t}),\tau_1(t_2)]$. Therefore, $t_2=\tau_1(t_2)\in E_2$ implies $[\tau_1(\tilde{t}),\tau_1(t_2)]\subseteq E_2$. With this observation, since $(F^{(2)})^\prime(s)=0$ in $(\tilde{t},t_2)$, we conclude $\tau_{\mathbf{F}}^{(2)}(\tilde{t})=\tau_{\mathbf{F}}^{(2)}(t_2)$ (by Lemma \ref{lem_inst1_uneqpref_derv_of_tau2}). Therefore, the class 2 user arriving at $\tilde{t}$ will be strictly better off arriving at $t_2$ instead.
\end{itemize}
In both cases, $\mathbf{F}=\{F^{(1)},F^{(2)}\}$ will not be an EAP, contradicting our assumption. \vspace{0.05in}

\textbf{(2)}~On assuming contradiction, there exists $t_1,t_2\in\mathcal{S}(F^{(2)})$ such that $F^{(2)}(t_1)=F^{(2)}(t_2)$. We must have $t_2<T^{(2)}_f$, otherwise, $T^{(2)}_f$ will be an isolated point in $\mathcal{S}(F^{(2)})$. Since $\mathcal{S}(F^{(1)})\cup \mathcal{S}(F^{(2)})$ is an interval, we must have $[t_1,t_2]\subseteq \mathcal{S}(F^{(1)})$. Two situations might arise:
\begin{itemize}[leftmargin=*]
    \item \textbf{If $\tau_1(t_1)\in E_2$:}~~By continuity of $\tau_1(\cdot)$, we can find $\delta \in (0,t_2-t_1)$ sufficiently small such that image of $[t_1,t_1+\delta]$ under $\tau_1(\cdot)$ is contained in $E_2$. Since $(F^{(2)})^\prime(t)=0$ in $[t_1,t_2]$, by Lemma \ref{lem_inst1_uneqpref_derv_of_tau2}, the previous statement implies $\tau_{\mathbf{F}}^{(2)}(t_1)=\tau_{\mathbf{F}}^{(2)}(t_1+\delta)$. Therefore the class 2 user arriving at $t_1$ will be strictly better off arriving at $t_1+\delta$.
    \vspace{0.05in}
    \item \textbf{If $\tau_1(t_1)\notin E_2$:}~With no class 2 user arriving at queue 2 in $[\tau_1(t_1),\tau_1(t_2)]$, we have \newline $[\tau_1(t_1),\tau_1(t_2)]\subseteq E_{2}^{c}$. Therefore for $t\in[t_1,t_2]$, $\tau_{\mathbf{F}}^{(1)}(t)=\tau_{\mathbf{F}}^{(2)}(t)=\tau_1(t)$. Now $[t_1,t_2]\subseteq \mathcal{S}(F^{(1)})$ implies $C_{\mathbf{F}}^{(1)}(\cdot)$ is constant in $[t_1,t_2]$. Hence $(C_{\mathbf{F}}^{(1)})^\prime(t)=\tau_1^\prime(t)-\gamma^{(1)}=0$ in $(t_1,t_2)$, implying $\tau_1^\prime(t)=\gamma^{(1)}$. As a result, in $(t_1,t_2)$, $(C_{\mathbf{F}}^{(2)})^\prime(t)=\tau_1^\prime(t)-\gamma^{(2)}=\gamma^{(1)}-\gamma^{(2)}\neq 0$. Now if $\gamma^{(1)}>\gamma^{(2)}$, class 2 user arriving at time $t_2$ will be better off arriving at time $t_1$ and vice-versa when $\gamma^{(1)}<\gamma^{(2)}$. 
\end{itemize}
$\{F^{(1)},F^{(2)}\}$ cannot be an EAP in both the situations. \vspace{0.05in}

\textbf{(3)}~We consider two separate situations for proving $\mathcal{S}(F^{(1)})$ is an interval:
\begin{itemize}[leftmargin=*]
    \item \textbf{If $\mu_1\gamma^{(1)}\neq\mu_2\gamma^{(2)}$:}~~On assuming contradiction, there exists $t_1,t_2\in \mathcal{S}(F^{(1)})$ such that, $F^{(1)}(t_1)=F^{(1)}(t_2)$. We must have $t_2<T^{(1)}_f$, otherwise, $T^{(1)}_f$ will be an isolated point in $\mathcal{S}(F^{(1)})$. Since $\mathcal{S}(F^{(1)})\cup\mathcal{S}(F^{(2)})$ is an interval, we must have $[t_1,t_2]\subseteq\mathcal{S}(F^{(2)})$. Also, we must have $[t_1,t_2]\subseteq E_1$, since a positive mass of class 1 users arrive after $t_2$. This implies, class 2 users will be arriving at queue 2 at a rate $\mu_1>\mu_2$ in $[\tau_1(t_1),\tau_1(t_2)]$ and therefore $[\tau_1(t_1),\tau_1(t_2)]\subseteq\overline{E_2}$. Using Lemma \ref{lem_inst1_uneqpref_derv_of_tau2} and the definition of EAP, in $[t_1,t_2]$, $(C_{\mathbf{F}}^{(2)})^\prime(t)=(\tau_{\mathbf{F}}^{(2)})^\prime(t)-\gamma^{(2)}=\frac{(F^{(2)})^\prime(t)}{\mu_2}-\gamma^{(2)}=0$ a.e. implying $(F^{(2)})^\prime(t)=\mu_2\gamma^{(2)}$ a.e. Since $[t_1,t_2]\subseteq E_1$, by (\ref{eq:derv_of_tau}), $\tau_1^\prime(t)=\frac{(F^{(2)})^\prime(t)}{\mu_1}=\frac{\mu_2\gamma^{(2)}}{\mu_1}$ a.e. in $[t_1,t_2]$. As a result, $(C_{\mathbf{F}}^{(1)})^\prime(t)=\tau_1^\prime(t)-\gamma^{(1)}=\frac{\mu_2\gamma^{(2)}}{\mu_1}-\gamma^{(1)}$ a.e. in $[t_1,t_2]$. Now if $\mu_1\gamma^{(1)}>\mu_2\gamma^{(2)}$, class 1 user arriving at time $t_1$ will be strictly better off arriving at time $t_2$ and vice-versa if $\mu_1\gamma^{(1)}<\mu_2\gamma^{(2)}$. Therefore $\{F^{(1)},F^{(2)}\}$ will not be an EAP.
    \vspace{0.05in}
    \item \textbf{If~$\mu_1\gamma^{(1)}=\mu_2\gamma^{(2)}$:}~~Since $\mu_1>\mu_2$, $\mu_1\gamma^{(1)}=\mu_2\gamma^{(2)}$ implies $\gamma^{(1)}<\gamma^{(2)}$. By 2nd statement of the lemma, $\mathcal{S}(F^{(2)})$ must be the interval $[T^{(2)}_a,T^{(2)}_f]$. By Lemma \ref{lem_inst1_uneqpref_suff_cond_for_disj_arrival}, $\mathcal{S}(F^{(1)})$ and $(T^{(2)}_a,T^{(2)}_f)$ are disjoint and their union is an interval. As a result, $\mathcal{S}(F^{(1)})$ is union of two intervals, one ending at $T^{(2)}_a$ and the other one starting from $T^{(2)}_f$. We now show that $\mathcal{S}(F^{(1)})$ has no element larger than $T^{(2)}_f$. This implies that $\mathcal{S}(F^{(1)})$ is an interval ending at $T^{(2)}_a$. For this, we again assume a contradiction and suppose that there exists $T>T^{(2)}_f$ such that $[T^{(2)}_f,T]\subseteq \mathcal{S}(F^{(1)})$.  By Lemma \ref{lem_appndx_inst1_uneqpref_queue2idle}, queue 2 must be empty at $\tau_1(T^{(2)}_f)$. Therefore for $t\in[T^{(2)}_f,T]$, $\tau_{\mathbf{F}}^{(1)}(t)=\tau_{\mathbf{F}}^{(2)}(t)=\tau_1(t)$. Since $[T^{(2)}_f,T]\subseteq \mathcal{S}(F^{(1)})$, we will have $(C_{\mathbf{F}}^{(1)})^\prime(t)=\tau_1^\prime(t)-\gamma^{(1)}=0$ implying $\tau_1^\prime(t)=\gamma^{(1)}$ in $[T^{(2)}_f,T]$. Using this, in $[T^{(2)}_f,T]$, $(C_{\mathbf{F}}^{(2)})^\prime(t)=\tau_1^\prime(t)-\gamma^{(2)}=\gamma^{(1)}-\gamma^{(2)}<0$.  Hence, the class 2 user arriving at $T^{(2)}_f$ is strictly better off arriving at $T$, implying that  $\{F^{(1)},F^{(2)}\}$ is not an EAP. \hfill\qedsymbol
\end{itemize}

By Lemma \ref{lem_supp_are_intervals_inst1}, for every EAP $\{F^{(1)},F^{(2)}\}$, we must have $\mathcal{S}(F^{(1)})=[T^{(1)}_a,T^{(1)}_f]$, $\mathcal{S}(F^{(2)})=[T^{(2)}_a,T^{(2)}_f]$  and their union must be an interval. 

\noindent\textbf{\textit{Proof of Lemma \ref{lem_inst1_uneqpref_queue2_busy}:}}~If $t\in(T^{(1)}_a,T^{(1)}_f)\cap(T^{(2)}_a,T^{(2)}_f)$, the statement follows from Lemma \ref{lem_inst1_uneqpref_mixedarrival}. Otherwise, we can obtain $\delta>0$ such that either $[t,t+\delta]$ or $[t-\delta,t]\subseteq(T^{(2)}_a,T^{(2)}_f)/(T^{(1)}_a,T^{(1)}_f)$. Now if $t\in E_1$, we can choose $\delta>0$ small enough such that $[t-\delta,t+\delta]\subseteq E_1$ As a result, $\tau_1(t-\delta)<\tau_1(t)<\tau_1(t+\delta)$ (by (\ref{eq:derv_of_tau})) and class 2 users arrive in queue 2 at a rate $\mu_1>\mu_2$ in $[\tau_1(t),\tau_1(t+\delta)]$ or $[\tau_1(t-\delta),\tau_1(t)]$ implying that  $\tau_1(t)\in\overline{E_2}$.  Otherwise,  if $t\in E_{1}^c$, since a positive mass of class 2 users arrive in $(t,T^{(2)}_f]$, we must have $t=\tau_1(t)\in E_2$. Therefore, the lemma stands proved.  \hfill\qedsymbol \\

\noindent\textit{\textbf{Proof of Lemma \ref{lem_arrival_rates_inst1}:}} ~~By Lemma \ref{lem_inst1_uneqpref_queue2_busy}, if $t\in \mathcal{S}(F^{(2)})$, $\tau_1(t)\in\overline{E_2}$, which by Lemma \ref{lem_inst1_uneqpref_derv_of_tau2} implies $(C_{\mathbf{F}}^{(2)})^\prime(t)=\frac{(F^{(2)})^\prime(t)}{\mu_2}-\gamma^{(2)}$ a.e. in $\mathcal{S}(F^{(2)})$. Since $C_{\mathbf{F}}^{(2)}(\cdot)$ is constant in $\mathcal{S}(F^{(2)})$, we must have $(F^{(2)})^\prime(t)=\mu_2\gamma^{(2)}$ a.e. in $\mathcal{S}(F^{(2)})$.

Note that $t\in\mathcal{S}(F^{(1)})=[T^{(1)}_a,T^{(1)}_f]$ implies $t\in \overline{E_1}$. Therefore, by (\ref{eq:derv_of_tau}),  $(C_{\mathbf{F}}^{(1)})^\prime(t)=(\tau_{\mathbf{F}}^{(1)})^\prime(t)-\gamma^{(1)}=\frac{(F^{(1)})^\prime(t)+(F^{(2)})^\prime(t)}{\mu_1}-\gamma^{(1)}$ a.e. in $\mathcal{S}(F^{(1)})$. Since $C_{\mathbf{F}}^{(1)}(\cdot)$ is constant in $\mathcal{S}(F^{(1)})$, we must have $(F^{(1)})^\prime(t)+(F^{(2)})^\prime(t)=\mu_1\gamma^{(1)}$ a.e. in $\mathcal{S}(F^{(1)})$. As $(F^{(2)})^\prime(t)=\mu_2\gamma^{(2)}$ a.e. in $\mathcal{S}(F^{(2)})$, we must have,
\begin{align*}
    (F^{(1)})^\prime(t)&=\begin{cases}
        \mu_1\gamma^{(1)}~~&\text{if}~t\in \mathcal{S}(F^{(1)})/ \mathcal{S}(F^{(2)}),~\text{and},\\
        \mu_1\gamma^{(1)}-\mu_2\gamma^{(2)}~~&\text{if}~t\in \mathcal{S}(F^{(1)})\cap \mathcal{S}(F^{(2)}).
    \end{cases}
\end{align*}
Note that, the above arrival rates are positive a.e. because, by Lemma \ref{lem_inst1_uneqpref_suff_cond_for_disj_arrival} and \ref{lem_supp_are_intervals_inst1}, if $\mu_1\gamma^{(1)}\leq\mu_2\gamma^{(2)}$, $\mathcal{S}(F^{(1)})\cap \mathcal{S}(F^{(2)})$ has zero Lebesgue measure.\hfill\qedsymbol \\

\begin{lemma}{(\textbf{Second statement of Lemma \ref{lem_queue1_and_2_idle}})}\label{lem_appndx_inst1_uneqpref_queue1idle}
    In the EAP, if $\mu_1\gamma^{(1)}>\mu_2\gamma^{(2)}$ and $\mathcal{S}(F^{(1)})=[T^{(1)}_a,T^{(1)}_f]$, then queue 1 will have zero waiting time at $T^{(1)}_f$.
\end{lemma}
\begin{proof}
    Assume a contradiction, \textit{i.e.}, $T^{(1)}_f\in E_1$ with $\mu_1\gamma^{(1)}>\mu_2\gamma^{(2)}$. Two situations are possible: \begin{itemize}[leftmargin=*]
    \item If $(T^{(1)}_f,\infty)\cap \mathcal{S}(F^{(2)})=\emptyset$, by (\ref{eq:derv_of_tau}), $\tau_{\mathbf{F}}^{(1)}(T^{(1)}_f)=\tau_{\mathbf{F}}^{(1)}(T_{idle})$, where $T_{idle}>T^{(1)}_f$ is the time queue 1 empties after $T^{(1)}_f$. Hence, the class 1 user arriving at $T^{(1)}_f$ is strictly better off arriving at $T_{idle}$.
    
    \item If a positive mass of class 2 users arrive after $T^{(1)}_f$, by Lemma \ref{lem_supp_are_intervals_inst1} and \ref{lem_arrival_rates_inst1}, they will arrive over an interval $[T^{(1)}_f,T]$ at rate $\mu_2\gamma^{(2)}$, where $T>T^{(1)}_f$. Therefore, picking $\delta\in(0,T-T^{(1)}_f)$ sufficiently small, we can have $[T^{(1)}_f,T^{(1)}_f+\delta]\subseteq E_1$. By (\ref{eq:derv_of_tau}), for $t\in[T^{(1)}_f,T^{(1)}_f+\delta]$, $(C_{\mathbf{F}}^{(1)})^\prime(t)=\tau_1^\prime(t)-\gamma^{(1)}=\frac{(F^{(2)})^\prime(t)}{\mu_1}-\gamma^{(1)}=\frac{\mu_2\gamma^{(2)}}{\mu_1}-\gamma^{(1)}<0$ a.e. Hence the class 1 user arriving at $T^{(1)}_f$ will be strictly better off arriving at $T^{(1)}_f+\delta$.
\end{itemize} 
For both the situations above, we get $\{F^{(1)},F^{(2)}\}$ will not be an EAP.    
\end{proof}

\noindent\textbf{\textit{Proof of Lemma \ref{lem_queue1_and_2_idle}:}}~~Lemma \ref{lem_queue1_and_2_idle} follows from Lemma \ref{lem_appndx_inst1_uneqpref_queue2idle} and \ref{lem_appndx_inst1_uneqpref_queue1idle}. \hfill\qedsymbol  \\

\noindent\textit{\textbf{Proof of Lemma \ref{lem_threshold_behav_inst1}:}} 
By Lemma \ref{lem_inst1_uneqpref_suff_cond_for_disj_arrival} and \ref{lem_supp_are_intervals_inst1}, we get $\mu_1\leq\mu_2\cdot\max\{1,\frac{\gamma^{(2)}}{\gamma^{(1)}}\}$ is sufficient for $[T^{(1)}_a,T^{(1)}_f]$ and $[T^{(2)}_a,T^{(2)}_f]$ to have disjoint interiors. Proving that the condition is necessary follows from the argument after Lemma \ref{lem_queue1_and_2_idle} in Section \ref{sec_inst1_uneqpref}. Hence Lemma \ref{lem_threshold_behav_inst1} stands proved.  \hfill\qedsymbol \\

\noindent\textbf{\textit{Proof of Lemma \ref{lem_sign_inst1_uneqpref}}:}~\textbf{(1)}~If $T^{(1)}_f=T^{(2)}_a$, queue 1 is empty at $T^{(1)}_f$ (by Lemma \ref{lem_queue1_and_2_idle}), making the network empty at $T^{(1)}_f$. As a result, every class 2 user will be strictly better off arriving at $T^{(1)}_f$ and $\mathbf{F}$ will not be an EAP. Therefore, the only other possibility is $T^{(2)}_f=T^{(1)}_a$. 

\textbf{(2)}~On assuming a contradiction, if $T^{(2)}_f<T^{(1)}_f$, by Lemma \ref{lem_threshold_behav_inst1}, we must have $T^{(2)}_f\in(T^{(1)}_a, T^{(1)}_f)$. Since group 1 users will arrive in $[T^{(2)}_f,T^{(1)}_f]$, we must have $[T^{(2)}_f,T^{(1)}_f]\subseteq \overline{E_1}$. Using Lemma \ref{lem_appndx_inst1_uneqpref_queue2idle}, queue 2 has a zero waiting time for $t\geq\tau_1(T^{(2)}_f)$ and as a consequence $\tau_{\mathbf{F}}^{(2)}(t)=\tau_1(t)$ for $t\geq T^{(2)}_f$. Therefore, for $t\in[T^{(2)}_f,T^{(1)}_f]$,
\begin{align*}
        (C^{(2)}_{\mathbf{F}})^\prime(t)&=(\tau_F^{(2)})^\prime(t)-\gamma^{(2)}=\tau_1^\prime(t)-\gamma^{(2)}\\
        &=\frac{(F^{(1)})^\prime(t)}{\mu_1}-\gamma^{(2)}~~&\text{(using (\ref{eq:derv_of_tau}) and the fact that $[T^{(2)}_f,T^{(1)}_f]\in \overline{E_1}$)}\\
        &=\gamma^{(1)}-\gamma^{(2)}<0.~~&\text{(using Lemma \ref{lem_arrival_rates_inst1}).}
\end{align*}
Thus, the group 2 user arriving at $T^{(2)}_f$ will be strictly better off by arriving at $T^{(1)}_f$. Hence $\mathbf{F}$ will not be an EAP.

\textbf{(3)}~Let $T=\min\left\{T^{(1)}_f,T^{(2)}_f\right\}$. By Lemma \ref{lem_threshold_behav_inst1}, $\max\{T^{(1)}_a,T^{(2)}_a\}<T$. If we assume contradiction to the statement, we must have $T^{(1)}_a\leq T^{(2)}_a < T$. As a result, class 2 users are arriving in $[T^{(2)}_a, T]$ alongside class 1 users. Using Lemma \ref{lem_arrival_rates_inst1}, for every $t\in[T^{(2)}_a,T]$,
\begin{align*}
    (F^{(2)})^\prime(t)&=\mu_2\gamma^{(2)},~\text{and},\\
    (F^{(1)})^\prime(t)&=\mu_1\gamma^{(1)}-\mu_2\gamma^{(2)}.
\end{align*}
Moreover, for $t\in[T^{(2)}_a, T]$, we have $A_2(\tau_1(t))=F^{(2)}(t)$. Again $[T^{(2)}_a,T]\subseteq \mathcal{S}(F^{(1)})$ implies $[T^{(2)}_a,T]\subseteq \overline{E_1}$ and therefore by (\ref{eq:derv_of_tau}) $\tau_1^\prime(t)=\frac{(F^{(1)})^\prime(t)+(F^{(2)})^\prime(t)}{\mu_1}=\gamma^{(1)}$ a.e. in $[T^{(2)}_a,T]$. Hence, $A_2^\prime(\tau_1(t))=\frac{(F^{(2)})^\prime(t)}{\tau_1^\prime(t)}=\mu_2\frac{\gamma^{(2)}}{\gamma^{(1)}}<\mu_2$ a.e. in $[T^{(2)}_a,T]$. This implies that class 2 users are arriving at queue 2 at a rate $<\mu_2$ from the time when queue 2 starts serving them, \textit{i.e.}, $\tau_1(T^{(2)}_a)$. Therefore $Q_2(t)=0$ for all $t\in[\tau_1(T^{(2)}_a),\tau_1(T)]$, which contradicts Lemma \ref{lem_inst1_uneqpref_mixedarrival}. 
\hfill\qedsymbol

\subsubsection{Proof of Theorem \ref{mainthm_inst1}} \label{appndx:thm_inst1_uneqpref}
Below we mention those parts of the proof of Theorem \ref{mainthm_inst1} which were not covered in the proof sketch mentioned in Section \ref{sec_inst1_uneqpref}. Of the skipped portion is the proof of existence and uniqueness of EAP for case 3,  $\gamma^{(1)}>\gamma^{(2)}$, owing to its similarity in arguments with case 2 $\gamma^{(2)}>\gamma^{(1)}>\frac{\mu_2}{\mu_1}\cdot\gamma^{(2)}$. Here, we cover those portions of the proof for each of the three cases in which the theorem statement is divided. 

\noindent\textbf{Case 1}~~$\gamma^{(1)}\leq\frac{\mu_2}{\mu_1}\gamma^{(2)}$:~~
For convenience of the reader, we mention the unique candidate we obtained by Lemma \ref{lem_arrival_rates_inst1} and \ref{lem:bdaryinst1case1}: 
\begin{align}
    (F^{(1)})^\prime(t)&=\mu_1\gamma^{(1)}~\text{if}~t\in[T^{(1)}_a,T^{(1)}_f]~\text{and}~(F^{(2)})^\prime(t)=\mu_2\gamma^{(2)}~\text{if}~t\in[T^{(1)}_f,T^{(2)}_f].
\end{align}

We first verify that, when the two classes arrive by the unique remaining candidate profile, for both the classes $i=1,2$ cost derivative satisfies~\textbf{1)}~$(C_{\mathbf{F}}^{(i)})^\prime(t)\leq 0$ in $(-\infty,T_{i,a})$,~\textbf{2)}~$(C_{\mathbf{F}}^{(i)})^\prime(t)=0$ in $[T_{i,a},T_{i,f}]$,~and~\textbf{3)}~$(C_{\mathbf{F}}^{(i)})^\prime(t)\geq 0$ in $(T_{i,f},\infty)$. The preceding statement implies both the classes have their cost constant on the support interval and higher outside of it, implying the candidate is an EAP. Following our agenda, we go by the following sequence of arguments. 

\begin{itemize}[leftmargin=*]
    \item In $(-\infty,T^{(1)}_a]$:~~$\tau_{\mathbf{F}}^{(i)}(t)=0$ for both classes $i\in\{1,2\}$ in $(-\infty,T^{(1)}_a]$. As a result, $(C_{\mathbf{F}}^{(i)})^\prime(t)=-\gamma^{(i)}<0$ in $(-\infty,T^{(1)}_a]$.
    
    \item  In $[T^{(1)}_a,T^{(1)}_f]$:~~Queue 1 stays engaged in $[T^{(1)}_a,T^{(1)}_f]$, since $A_1(t)=F^{(1)}(t)>\mu_1\cdot \max\{t,0\}$ in $[T^{(1)}_a,T^{(1)}_f]$. Hence by (\ref{eq:derv_of_tau}), for both the classes $i\in\{1,2\}$, $(C_{\mathbf{F}}^{(i)})^\prime(t)=\frac{(F^{(1)})^\prime(t)}{\mu_1}-\gamma^{(i)}=\gamma^{(1)}-\gamma^{(i)}$. As a result, for $i=1$, $(C_{\mathbf{F}}^{(1)})^\prime(t)=0$ and for $i=2$, $(C_{\mathbf{F}}^{(2)})^\prime(t)=\gamma^{(1)}-\gamma^{(2)}<0$. 

    \item In $[T^{(1)}_f,T^{(2)}_f]$:~~Since $\mu_1\cdot T^{(2)}_f=\Lambda^{(1)}+\frac{\mu_1}{\mu_2}\Lambda^{(2)}>\Lambda^{(1)}+\Lambda^{(2)}=A_1(T^{(2)}_f)$, queue 1 empties at some time $T\in(T^{(1)}_f,T^{(2)}_f]$. As class 2 users arrive at a constant rate of $\mu_2\gamma^{(2)}$ in $[T^{(1)}_f,T^{(2)}_f]$, queue 1 stays engaged in $[T^{(1)}_f,T]$ and empty in $[T,T^{(2)}_f]$. Therefore, in $[T^{(1)}_f,T]$, by (\ref{eq:derv_of_tau}), $(C_{\mathbf{F}}^{(1)})^\prime(t)=\frac{(F^{(2)})^\prime(t)}{\mu_1}-\gamma^{(1)}=\frac{\mu_2\gamma^{(2)}-\mu_1\gamma^{(1)}}{\mu_1}\geq 0$ and in $(T,T^{(2)}_f]$, since queue 1 is empty, $(C_{\mathbf{F}}^{(1)})^\prime(t)=1-\gamma^{(1)}\geq 0$. Now we argue that queue 2 remains engaged in $[\tau_1(T^{(1)}_f),T^{(2)}_f]$ by considering two parts separately:\begin{itemize}
        \item In $[\tau_1(T^{(1)}_f),T]$, class 2 users arrive at queue 2 from queue 1 at rate $\mu_1>\mu_2$ in that period, making queue 2 engaged.
        \item In $[T,T^{(2)}_f]$, since $A_2^\prime(t)=\mu_2\gamma^{(2)}<\mu_2 $, $A_2(t)-\mu_2\cdot(t-\tau_1(T^{(1)}_f))$ is a decreasing function. Combining this with the fourth equation $A_2(T^{(2)}_f)=\mu_2\cdot(T^{(2)}_f-\tau_1(T^{(1)}_f))$ in the identified linear system, we get $A_2(t)>\mu_2\cdot(t-\tau_1(T^{(1)}_f))$ in $[T,T^{(2)}_f)$, implying queue 2 stays engaged. 
    \end{itemize} 
    Therefore, by (\ref{eq:derv_of_tau}), for $t\in [T^{(1)}_f,T^{(2)}_f]$, $\tau_{\mathbf{F}}^{(2)}(t)=\tau_2(\tau_1(t))=\tau_1(T^{(1)}_f)+\frac{A_2(\tau_1(t))}{\mu_2}=\frac{\Lambda^{(1)}}{\mu_1}+\frac{F^{(2)}(t)}{\mu_2}$. Hence, $(C_{\mathbf{F}}^{(2)})^\prime(t)=(\tau_{\mathbf{F}}^{(2)})^\prime(t)-\gamma^{(2)}=\frac{(F^{(2)})^\prime(t)}{\mu_2}-\gamma^{(2)}=0$ in $[T^{(1)}_f,T^{(2)}_f]$. Upon summarizing, for $t\in[T^{(1)}_f,T^{(2)}_f]$, we get $(C_{\mathbf{F}}^{(1)})^\prime(t)\geq 0$ and $(C_{\mathbf{F}}^{(2)})^\prime(t)=0$. 

    \item In $[T^{(2)}_f,\infty)$:~~By the last step, both the queues stay empty after $T^{(2)}_f$. As a result, for both the classes $i\in\{1,2\}$, $(C_{\mathbf{F}}^{(i)})^\prime(t)=1-\gamma^{(i)}>0$ in $[T^{(2)}_f,\infty)$.
\end{itemize}

\noindent\textbf{Case 2}~~$\frac{\mu_2}{\mu_1}\gamma^{(2)}<\gamma^{(1)}<\gamma^{(2)}$:

\noindent\textbf{Skipped parts of the proof of Lemma \ref{lem:inst1case2}:}

\noindent\textbf{Identifying the unique candidate for Type II:}~For every EAP under Type II, we obtain the following system of equations to be satisfied by the support boundaries $T^{(1)}_a,T^{(1)}_f,T^{(2)}_a,T^{(2)}_f$: 
\begin{enumerate}[leftmargin=*]
    \item By Lemma \ref{lem_arrival_rates_inst1}, $(F^{(1)})^\prime(t)=\mu_1\gamma^{(1)}-\mu_2\gamma^{(2)}$ in $[T^{(1)}_a,T^{(1)}_f]$, giving us: $T^{(1)}_f=T^{(1)}_a+\frac{\Lambda^{(1)}}{\mu_1\gamma^{(1)}-\mu_2\gamma^{(2)}}$.
    \item By Lemma \ref{lem_arrival_rates_inst1},  $(F^{(2)})^\prime(t)=\mu_2\gamma^{(2)}$ in $[T^{(2)}_a,T^{(2)}_f]$, giving us: $T^{(2)}_f=T^{(2)}_a+\frac{\Lambda^{(2)}}{\mu_2\gamma^{(2)}}$.
    \item Applying the argument used for getting the third equation of Type I, we must have $F^{(1)}(T^{(1)}_f)+F^{(2)}(T^{(1)}_f)=\mu_1 T^{(1)}_f$. Plugging in $F^{(1)}(T^{(1)}_f)=\Lambda^{(1)}$ and $F^{(2)}(T^{(1)}_f)=\mu_2\gamma^{(2)}(T^{(1)}_f-T^{(2)}_a)$ (by Lemma \ref{lem_arrival_rates_inst1}), we get: $\Lambda^{(1)}+\mu_2\gamma^{(2)}(T^{(1)}_f-T^{(2)}_a)=\mu_1 T^{(1)}_f$.
    \item By an argument similar to the one used for getting the fourth equation of Type I, we have $\mu_2\cdot(\tau_1(T^{(2)}_f)-\tau_1(T^{(2)}_a))=\Lambda^{(2)}$. Since queue 1 is empty at $T^{(1)}_f$ (by Lemma \ref{lem_queue1_and_2_idle}) and class 2 users arrive at rate $\mu_2\gamma^{(2)}<\mu_1$ (by Lemma \ref{lem_arrival_rates_inst1}), queue 1 stays empty at $T^{(2)}_f$, causing $\tau_1(T^{(2)}_f)=T^{(2)}_f$. Queue 1 serves the first class 2 user at time zero, causing $\tau_1(T^{(2)}_a)=0$. Plugging these in, we get: $T^{(2)}_f=\frac{\Lambda^{(2)}}{\mu_2}$. \hfill\qedsymbol
\end{enumerate}
The solution to the above system of equations is in (\ref{eq_bdary_inst1_uneqpref_case2b}). The support boundaries in (\ref{eq_bdary_inst1_uneqpref_case2b}) must satisfy $T^{(2)}_a< T^{(1)}_a<T^{(1)}_f\leq T^{(2)}_f$ to represent an EAP under Type II. Imposing $T^{(1)}_a>T^{(2)}_a$ on (\ref{eq_bdary_inst1_uneqpref_case2b}), we get the condition $\Lambda^{(1)}<\frac{1-\gamma^{(2)}}{1-\gamma^{(1)}}\left(\frac{\mu_1\gamma^{(1)}}{\mu_2\gamma^{(2)}}-1\right)\Lambda^{(2)}$.

Moreover, if $\Lambda^{(1)}<\frac{1-\gamma^{(2)}}{1-\gamma^{(1)}}\left(\frac{\mu_1\gamma^{(1)}}{\mu_2\gamma^{(2)}}-1\right)\Lambda^{(2)}$, (\ref{eq_bdary_inst1_uneqpref_case2b}) satisfies $T^{(2)}_a<T^{(1)}_a<T^{(1)}_f\leq T^{(2)}_f$ and upon plugging in arrival rates from Lemma \ref{lem_arrival_rates_inst1}, we get a candidate having arrival rates and support boundaries same as the joint arrival profile mentioned under case 2b of Theorem \ref{mainthm_inst1}. This candidate satisfies $F^{(1)}(T^{(1)}_f)=\Lambda^{(1)},~F^{(2)}(T^{(2)}_f)=\Lambda^{(2)}$ and will be  the only Type II candidate to qualify as an EAP. Upon proving that this identified Type II candidate is an EAP, the second statement of Lemma \ref{lem:inst1case2} will follow.  

\textbf{Proving that the unique candidate obtained is an EAP:}~~For convenience of the reader we mention the candidates we were left with for both the types: 
\begin{enumerate}[leftmargin=*]
    \item If $\Lambda^{(1)}\geq\frac{1-\gamma^{(2)}}{1-\gamma^{(1)}}\left(\frac{\mu_1\gamma^{(1)}}{\mu_2\gamma^{(2)}}-1\right)\Lambda^{(2)}$, the unique Type I candidate is: 
    \begin{align}\label{eq:rate_inst1_uneqpref_case21}
        (F^{(1)})^\prime(t)&=\begin{cases}
            \mu_1\gamma^{(1)}~~&\text{if}~t\in[T^{(1)}_a,T^{(2)}_a),\\
            \mu_1\gamma^{(1)}-\mu_2\gamma^{(2)}~~&\text{if}~t\in[T^{(2)}_a,T^{(1)}_f],
        \end{cases}~\text{and}~(F^{(2)})^\prime(t)=\mu_2\gamma^{(2)}~~\text{if}~t\in[T^{(2)}_a,T^{(2)}_f] 
    \end{align}
    where $T^{(1)}_a,T^{(1)}_f,T^{(2)}_a,T^{(2)}_f$ are in (\ref{eq_bdary_inst1_uneqpref_case2a}). 
    
    \item If $\Lambda^{(1)}<\frac{1-\gamma^{(2)}}{1-\gamma^{(1)}}\left(\frac{\mu_1\gamma^{(1)}}{\mu_2\gamma^{(2)}}-1\right)\Lambda^{(2)}$, the unique Type II candidate is:
    \begin{align}\label{eq:rate_inst1_uneqpref_case22}
        (F^{(1)})^\prime(t)&=\mu_1\gamma^{(1)}-\mu_2\gamma^{(2)}~~\text{if}~t\in[T^{(2)}_a,T^{(1)}_f],~\text{and}~ 
        (F^{(2)})^\prime(t)=\mu_2\gamma^{(2)}~~\text{if}~t\in[T^{(2)}_a,T^{(2)}_f] 
    \end{align}
    where $T^{(1)}_a,T^{(1)}_f,T^{(2)}_a,T^{(2)}_f$ are in (\ref{eq_bdary_inst1_uneqpref_case2b}).  
\end{enumerate}

By the following sequence of arguments, for both the types, we argue that if users of the two classes arrive by the unique remaining candidate, both classes $i=1,2$ have their cost derivative satisfying: $(C_{\mathbf{F}}^{(i)})^\prime(t)\leq 0$ in $(-\infty,T_{i,a})$, $(C_{\mathbf{F}}^{(i)})^\prime(t)=0$ in $[T_{i,a},T_{i,f}]$, and $(C_{\mathbf{F}}^{(i)})^\prime(t)\geq 0$ in $(T_{i,f},\infty)$. This will imply, for both the classes, cost is constant in the support interval and higher outside and hence the candidate is an EAP.
\begin{enumerate}[leftmargin=*]
    \item In $(-\infty,\min\{T^{(1)}_a,T^{(2)}_a\}]$:~~Note that for $t\in(-\infty,\min\{T^{(1)}_a,T^{(2)}_a\}]$, $\tau_{\mathbf{F}}^{(i)}(t)=0$ for both the classes $i\in\{1,2\}$. As a result, $(C_{\mathbf{F}}^{(i)})^\prime(t)=-\gamma^{(i)}<0$ for $i\in\{1,2\}$. 

    \item In $[\min\{T^{(1)}_a,T^{(2)}_a\},T^{(1)}_f]$:~~Obtained candidates of both the types satisfy $A_1^\prime(t)=(F^{(1)})^\prime(t)+(F^{(2)})^\prime(t)<\mu_1$ in $[\min\{T^{(1)}_a,T^{(2)}_a\},T^{(1)}_f]$, making $A_1(t)-\mu_1 t$ decreasing in $[0,T^{(1)}_f]$. Since both the types satisfy $A_1(T^{(1)}_f)=\mu_1 T^{(1)}_f$, the preceding statement imply $A_1(t)>\mu_1 t$ in $[0,T^{(1)}_f)$ and therefore, queue 1 is engaged in $[\min\{T^{(1)}_a,T^{(2)}_a\},T^{(1)}_f]$ and empty after $T^{(1)}_f$. Using (\ref{eq:derv_of_tau}) and the arrival rates from (\ref{eq:rate_inst1_uneqpref_case21}) and (\ref{eq:rate_inst1_uneqpref_case22}), we get $(C_{\mathbf{F}}^{(1)})^\prime(t)=\tau_1^\prime(t)-\gamma^{(1)}=\frac{(F^{(1)})^\prime(t)+(F^{(2)})^\prime(t)}{\mu_1}-\gamma^{(1)}=0$ for $t\in[\min\{T^{(1)}_a,T^{(2)}_a\},T^{(1)}_f]$. For analyzing the cost of the second class, we consider the two types separately: 
    \begin{itemize}
         \item In Type I ($T^{(1)}_a\leq T^{(2)}_a<T^{(1)}_f$): In $ [T^{(1)}_a,T^{(2)}_a]$, $(C_{\mathbf{F}}^{(2)})^\prime(t)=\tau_1^\prime(t)-\gamma^{(2)}=\frac{(F^{(1)})^\prime(t)+(F^{(2)})^\prime(t)}{\mu_1}-\gamma^{(2)}=\gamma^{(1)}-\gamma^{(2)}<0$. Upon observing $A_2(\tau_1(t))=F^{(2)}(t)$, we can write using chain rule that $A_2^\prime(\tau_1(t))=\frac{(A_2\circ\tau_1)^\prime(t)}{\tau_1^\prime(t)}=\frac{(F^{(2)})^\prime(t)}{\tau_1^\prime(t)}$, whenever the derivatives exists. Now for $t\in[T^{(2)}_a,T^{(1)}_f]$, using (\ref{eq:derv_of_tau}), rate of arrival of class 2 users in queue 2 at $\tau_1(t)$ is $A_2^\prime(\tau_1(t))=\frac{(F^{(2)})^\prime(t)}{\tau_1^\prime(t)}=\mu_1\cdot\frac{(F^{(2)})^\prime(t)}{(F^{(1)})^\prime(t)+(F^{(2)})^\prime(t)}=\frac{\mu_2\gamma^{(2)}}{\gamma^{(1)}}>\mu_2$. As a result, queue 2 remains engaged in $[\tau_1(T^{(2)}_a),T^{(1)}_f]$. With this, the time of service of class 2 users arriving in $[T^{(2)}_a,T^{(1)}_f]$ is $\tau_{\mathbf{F}}^{(2)}(t)=\tau_2(\tau_1(t))=\tau_1(T^{(2)}_a)+\frac{A_2(\tau_1(t))}{\mu_2}=\tau_1(T^{(2)}_a)+\frac{F^{(2)}(t)}{\mu_2}$. Using this and the arrival rates from (\ref{eq:rate_inst1_uneqpref_case21}), $(C_{\mathbf{F}}^{(2)})^\prime(t)=(\tau_{\mathbf{F}}^{(2)})^\prime(t)-\gamma^{(2)}=\frac{(F^{(2)})^\prime(t)}{\mu_2}-\gamma^{(2)}=0$ in $[T^{(2)}_a,T^{(1)}_f]$.  \\  

         \item In Type II ($T^{(2)}_a<T^{(1)}_a<T^{(1)}_f$):  In $[T^{(2)}_a,T^{(1)}_a]$, rate of arrival  of class 2 users at queue 2 at time $\tau_1(t)$ is $\mu_1>\mu_2$, and $\tau_1(T^{(2)}_a)=0$. As a result, queue 2 has a positive waiting queue at $\tau_1(T^{(1)}_a)$. Following a calculation similar to the one used in the previous step for Type I, rate of arrival of class 2 users at queue 2 at time $\tau_1(t)$ is $\frac{\mu_2\gamma^{(2)}}{\gamma^{(1)}}\geq\mu_2$ when $t\in[T^{(1)}_a,T^{(1)}_f]$. The last two statements imply queue 2 has positive waiting time in $(0,T^{(1)}_f]$. Using this, time of service of class 2 users arriving in $[T^{(2)}_a,T^{(1)}_f]$ is $\tau_{\mathbf{F}}^{(2)}(t)=\tau_2(\tau_1(t))=\tau_1(T^{(2)}_a)+\frac{A_2(\tau_1(t))}{\mu_2}=\frac{F^{(2)}(t)}{\mu_2}$ and as a result, $(C_{\mathbf{F}}^{(2)})^\prime(t)=\frac{(F^{(2)})^\prime(t)}{\mu_2}-\gamma^{(2)}=0$.
    \end{itemize} 
    
    \item In $[T^{(1)}_f,T^{(2)}_f]$:~~Queue 1 remains empty in this region. On the other hand, queue 2 remains engaged in $[T^{(1)}_f,T^{(2)}_f]$, since $F^{(2)}(t)-\mu_2\cdot(t-\tau_1(T^{(2)}_a))$ is decreasing in $[T^{(1)}_f,T^{(2)}_f)$ with $F^{(2)}(T^{(2)}_f)=\Lambda^{(2)}=\mu_2\cdot(T^{(2)}_f-\tau_1(T^{(2)}_a))$. As a result, $(C_{\mathbf{F}}^{(1)})^\prime(t)=1-\gamma^{(1)}>0$ and using (\ref{eq:derv_of_tau}), $(C_{\mathbf{F}}^{(2)})^\prime(t)=\tau_2^\prime(t)-\gamma^{(2)}=\frac{(F^{(2)})^\prime(t)}{\mu_2}-\gamma^{(2)}=0$. 

    \item In $[T^{(2)}_f,\infty)$:~~Since both the queues are idle, for both the classes $i\in\{1,2\}$, $(C_{\mathbf{F}}^{(i)})^\prime(t)=1-\gamma^{(i)}>0$.
\end{enumerate}

\noindent\textbf{Case 3}~~$\gamma^{(2)}<\gamma^{(1)}$:

\noindent\textbf{\textit{Proof of Lemma \ref{lem:inst1case3}}:}

\noindent\textbf{Identifying the unique Type I candidate and necessary condition:}\hspace{0.05in}For every EAP under Type I, the support boundaries must satisfy the following system of equations: 
    \begin{enumerate}[leftmargin=*]
            \item By Lemma \ref{lem_queue1_and_2_idle}, queue 1 serves both the groups between $[0,T^{(1)}_f]$ and empties at $T^{(1)}_f$, giving us, $T^{(1)}_f=\frac{\Lambda^{(1)}+\Lambda^{(2)}}{\mu_1}$.
            
            \item By Lemma \ref{lem_arrival_rates_inst1}, group 2 users arrive at rate $\mu_2\gamma^{(2)}$ between $[T^{(2)}_a,T^{(2)}_f]$, giving us, $T^{(2)}_f=T^{(2)}_a+\frac{\Lambda^{(2)}}{\mu_2\gamma^{(2)}}$.  
            
            \item By Lemma \ref{lem_queue1_and_2_idle}, queue 2 starts serving at time zero, empties at $\tau_1(T^{(2)}_f)$ and serves all the group 2 users between $[0,\tau_1(T^{(2)}_f)]$. This gives us, $\mu_2\tau_1(T^{(2)}_f)=\Lambda^{(2)}$, where $\tau_1(T^{(2)}_f)=\frac{F^{(1)}(T^{(2)}_f)+F^{(2)}(T^{(2)}_f)}{\mu_1}=\frac{(\mu_1\gamma^{(1)}-\mu_2\gamma^{(2)})(T^{(2)}_f-T^{(1)}_a)+\Lambda^{(2)}}{\mu_1}$. 
            
            \item By definition of EAP $C_{\mathbf{F}}^{(1)}(T^{(1)}_a)=C_{\mathbf{F}}^{(1)}(T^{(1)}_f)$, giving us, $(1-\gamma^{(1)})T^{(1)}_f=\tau_1(T^{(1)}_a)-\gamma^{(1)} T^{(1)}_a$, where $\tau_1(T^{(1)}_a)=\frac{F^{(2)}(T^{(1)}_a)}{\mu_1}=\frac{\mu_2\gamma^{(2)}(T^{(1)}_a-T^{(2)}_a)}{\mu_1}$.
        \end{enumerate}

Solution to the above linear system is in (\ref{eq_bdary_inst1_uneqpref_case3a}) and every EAP under Type I has (\ref{eq_bdary_inst1_uneqpref_case3a}) as support boundaries. (\ref{eq_bdary_inst1_uneqpref_case3a}) must satisfy $T^{(2)}_a<T^{(1)}_a<T^{(2)}_f\leq T^{(1)}_f$ to represent a Type I EAP. Imposing $T^{(1)}_f\geq T^{(2)}_f$ on \ref{eq_bdary_inst1_uneqpref_case3a}, we get $\Lambda^{(1)}\geq\left(\frac{\mu_1}{\mu_2}-1\right)\Lambda^{(2)}$ and therefore, this is a necessary condition for existence of an EAP under Type I. It is easy to verify that, if $\Lambda^{(1)}\geq\left(\frac{\mu_1}{\mu_2}-1\right)\Lambda^{(2)}$, (\ref{eq_bdary_inst1_uneqpref_case3a}) satisfies $T^{(2)}_a<T^{(1)}_a<T^{(2)}_f\leq T^{(1)}_f$ and upon plugging in arrival rates from Lemma \ref{lem_arrival_rates_inst1}, we get a candidate having arrival rates and support boundaries same as the joint arrival profile under case 3a of Theorem \ref{mainthm_inst1}. This candidate satisfies $F^{(1)}(T^{(1)}_f)=\Lambda^{(1)},~F^{(2)}(T^{(2)}_f)=\Lambda^{(2)}$ and will be the only candidate under Type I to qualify as an EAP. 

\textbf{Identifying the unique Type II candidate and necessary condition:}\hspace{0.05in}The support boundaries of any Type II EAP satisfies the same system of equations as was obtained for Type II in case 2 ($\gamma^{(2)}>\gamma^{(1)}>\frac{\mu_2}{\mu_1}\gamma^{(2)}$) of Theorem \ref{mainthm_inst1}. As a result, every EAP under Type II has (\ref{eq_bdary_inst1_uneqpref_case2b}) as support boundaries. Now (\ref{eq_bdary_inst1_uneqpref_case2b}) must satisfy $T^{(2)}_a<T^{(1)}_a<T^{(1)}_f<T^{(2)}_f$ to represent a Type II EAP. Imposing $T^{(2)}_f>T^{(1)}_f$ on (\ref{eq_bdary_inst1_uneqpref_case2b}), we get $\Lambda^{(1)}<\left(\frac{\mu_1}{\mu_2}-1\right)\Lambda^{(2)}$, and this is a necessary condition for existence of a Type II EAP. Once the necessary condition is satisfied, it is easy to verify that (\ref{eq_bdary_inst1_uneqpref_case2b}) satisfies $T^{(2)}_a<T^{(1)}_a<T^{(1)}_f<T^{(2)}_f$ and plugging in arrival rates from Lemma \ref{lem_arrival_rates_inst1}, we get a candidate with arrival rates and support boundaries same as the joint arrival profile under case 3b of Theorem \ref{mainthm_inst1}. This candidate satisfies $F^{(1)}(T^{(1)}_f)=\Lambda^{(1)},~F^{(2)}(T^{(2)}_f)=\Lambda^{(2)}$ and is the only Type II candidate to qualify as an EAP. 

Now we show that, the unique candidates obtained in both the types are EAP. Upon proving this, the necessary conditions obtained for both the types will also be sufficient and therefore, Lemma \ref{lem:inst1case3} will follow. 

\textbf{Proving that the unique remaining candidate of both classes are EAP:}~By the following sequence of arguments, we prove that, the unique remaining candidates of the two types satisfy: $(C_{\mathbf{F}}^{(i)})^\prime(t)\leq 0$ in $(-\infty,T_{i,a})$, $(C_{\mathbf{F}}^{(i)})^\prime(t)=0$ in $[T_{i,a},T_{i,f}]$, and $(C_{\mathbf{F}}^{(i)})^\prime(t)\geq 0$ in $(T_{i,f},\infty)$, for both the classes $i=1,2$. This will imply that the unique remaining candidates of the two types are EAP. Note that, for both the types $\tau_1(T^{(2)}_a)=0$. 
\begin{itemize}[leftmargin=*]
    \item In $(-\infty,T^{(2)}_a]$:~~For both the classes $i\in\{1,2\}$, $\tau_{\mathbf{F}}^{(i)}(t)=0$ in $(-\infty,T^{(2)}_a]$. As a result, for both $i\in\{1,2\}$, $(C_{\mathbf{F}}^{(i)})^\prime(t)=-\gamma^{(i)}<0$. 
    
    \item In $[T^{(2)}_a,T^{(1)}_a]$:~~Queue 1 remains engaged in $[T^{(2)}_a,T^{(1)}_a]$ since $A_1(t)=F^{(2)}(t)>\mu_1\cdot\max\{0,t\}$ for every $t\in(T^{(2)}_a,T^{(1)}_a]$. As a result, class 2 users arrive at rate $\mu_1>\mu_2$ from queue 1 to 2 in $[0,\tau_1(T^{(1)}_a)]$, making queue 2 engaged in that interval. Therefore, using (\ref{eq:derv_of_tau}), $\tau_{\mathbf{F}}^{(2)}(t)=\tau_2(\tau_1(t))=\tau_1(T^{(2)}_a)+\frac{A_2(\tau_1(t))}{\mu_2}=\frac{F^{(2)}(t)}{\mu_2}$. Hence, in $[T^{(2)}_a,T^{(1)}_a]$, $(C_{\mathbf{F}}^{(2)})^\prime(t)=(\tau_{\mathbf{F}}^{(2)})^\prime(t)-\gamma^{(2)}=\frac{(F^{(2)})^\prime(t)}{\mu_2}-\gamma^{(2)}=0$ and using (\ref{eq:derv_of_tau}), $(C_{\mathbf{F}}^{(1)})^\prime(t)=\tau_1^\prime(t)-\gamma^{(2)}=\frac{(F^{(2)})^\prime(t)}{\mu_1}-\gamma^{(2)}=\left(\frac{\mu_2}{\mu_1}-1\right)\cdot\gamma^{(2)}<0$. 
    
    \item In $[T^{(1)}_a,\max\{T^{(1)}_f,T^{(2)}_f\}]$:~~Note that, $\mu_1\cdot t<F^{(1)}(t)+F^{(2)}(t)$ in $[T^{(1)}_a,T^{(1)}_f)$ and $\mu_1\cdot T^{(1)}_f=F^{(1)}(T^{(1)}_f)+F^{(2)}(T^{(1)}_f)$ for both the types. As a result, queue 1 stays engaged in $[T^{(1)}_a,T^{(1)}_f]$ and empties at $T^{(1)}_f$ in both the types. Therefore, using (\ref{eq:derv_of_tau}), $(C_{\mathbf{F}}^{(1)})^\prime(t)=\tau_1^\prime(t)-\gamma^{(1)}=\frac{(F^{(1)})^\prime(t)+(F^{(2)})^\prime(t)}{\mu_1}-\gamma^{(1)}=0$ for every $t\in[T^{(1)}_a,T^{(1)}_f]$. For the case $T^{(2)}_f>T^{(1)}_f$, since queue 1 stays empty in $[T^{(1)}_f,T^{(2)}_f]$, $(C_{\mathbf{F}}^{(1)})^\prime(t)=1-\gamma^{(1)}>0$ in $[T^{(1)}_f,T^{(2)}_f]$.  For analyzing the cost of the second class, we consider the two types separately:
 
    \begin{itemize}
        \item In Type I ($T^{(1)}_a<T^{(2)}_f\leq T^{(1)}_f$), for every $t\in [T^{(1)}_a,T^{(2)}_f]$, using (\ref{eq:derv_of_tau}), rate of arrival of class 2 users from queue 1 to 2 at time $\tau_1(t)$ is $A_2^\prime(\tau_1(t))=\frac{(F^{(2)})^\prime(t)}{\tau_1^\prime(t)}=\mu_1\cdot\frac{(F^{(2)})^\prime(t)}{(F^{(1)})^\prime(t)+(F^{(2)})^\prime(t)}=\frac{\mu_2\gamma^{(2)}}{\gamma^{(1)}}<\mu_2$. As a result, $A_2(s)-\mu_2\cdot s$ is strictly decreasing in $[\tau_1(T^{(1)}_a),\tau_1(T^{(2)}_f)]$ .
        Also, the support boundaries satisfy $\mu_2 \cdot \tau_1(T^{(2)}_f)=\Lambda^{(2)}=A_2(\tau_1(T^{(2)}_f))$ implying $A_2(\tau_1(T^{(1)}_f))-\mu_2\cdot\tau_1(T^{(2)}_f)=0$. As a result, for every $t\in[T^{(1)}_a,T^{(2)}_f]$, $A_2(\tau_1(t))>\mu_2\cdot\tau_1(t)$, implying queue 2 will be engaged between $[\tau_1(T^{(1)}_a),\tau_1(T^{(2)}_f)]$. Using this, for every $t\in[T^{(1)}_a,T^{(2)}_f]$, by (\ref{eq:derv_of_tau}), $\tau_{\mathbf{F}}^{(2)}(t)=\tau_2(\tau_1(t))=\tau_1(T^{(2)}_a)+\frac{A_2(\tau_1(t))}{\mu_2}=\frac{F^{(2)}(t)}{\mu_2}$. Hence, $(C_{\mathbf{F}}^{(2)})^\prime(t)=\frac{(F^{(2)})^\prime(t)}{\mu_2}-\gamma^{(2)}=0$ in $[T^{(1)}_a,T^{(2)}_f]$. Since queue 2 empties at $\tau_1(T^{(2)}_f)$, we have $\tau_{\mathbf{F}}^{(2)}(t)=\tau_1(t)$ for every $t\in[T^{(2)}_f,T^{(1)}_f]$. As a result,  by (\ref{eq:derv_of_tau}),  $(C_{\mathbf{F}}^{(2)})^\prime(t)=\tau_1^\prime(t)-\gamma^{(2)}=\frac{(F^{(1)})^\prime(t)}{\mu_1}-\gamma^{(2)}=\gamma^{(1)}-\gamma^{(2)}>0$ in $[T^{(2)}_f,T^{(1)}_f]$. 

        \item In Type II ($T^{(1)}_a<T^{(1)}_f<T^{(2)}_f$), for every $t\in[T^{(1)}_a,T^{(1)}_f]$, using a method similar to the one used in previous step, rate of arrival of class 2 users from queue 1 to 2 at time $\tau_1(t)$ is $A_2^\prime(\tau_1(t))=\frac{\mu_2\gamma^{(2)}}{\gamma^{(1)}}<\mu_2$. Since queue 1 empties at $T^{(1)}_f$, class 2 users arrive at queue 2 at rate $A_2^\prime(\tau_1(t))=(F^{(2)})^\prime(t)=\mu_2\gamma^{(2)}<\mu_2$ in $[T^{(1)}_f,T^{(2)}_f]$. Therefore, $A_2(s)-\mu_2\cdot s$ is a strictly decreasing function in $[\tau_1(T^{(1)}_a),T^{(2)}_f]$. Since $\mu_2\cdot T^{(2)}_f=\Lambda^{(2)}=A_2(\tau_1(T^{(2)}_f))$, we have $A_2(\tau_1(t))>\mu_2\cdot\tau_1(t)$ for every $t\in[T^{(1)}_a,T^{(2)}_f)$, making queue 2 engaged in $[\tau_1(T^{(1)}_a),T^{(2)}_f)$. As a result, for every $t\in[T^{(1)}_a,T^{(2)}_f]$, using (\ref{eq:derv_of_tau}), $\tau_{\mathbf{F}}^{(2)}(t)=\tau_2(\tau_1(t))=\frac{A_2(\tau_1(t))}{\mu_2}=\frac{F^{(2)}(t)}{\mu_2}$ and therefore $(C_{\mathbf{F}}^{(2)})^\prime(t)=\frac{(F^{(2)})^\prime(t)}{\mu_2}-\gamma^{(2)}=0$.    
    \end{itemize}    
    
    \item In $[\max\{T^{(1)}_f,T^{(2)}_f\},\infty)$:~~Since both the queues are empty in $[\max\{T^{(1)}_f,T^{(2)}_f\},\infty)$, we have $(C_{\mathbf{F}}^{(i)})^\prime(t)=1-\gamma^{(i)}>0$ for both $i\in\{1,2\}$. \hfill\qedsymbol
\end{itemize}

\subsection{Proofs of Lemmas in Section \ref{sec_inst1_eqpref} (HDS with $\gamma^{(1)}=\gamma^{(2)}$)}\label{appndx:inst1_eqpref}

\noindent\textbf{\textit{Proof of Lemma \ref{lem_rate_of_arrivals_inst1_eqpref}:}}~~Note that $\mathcal{S}(F^{(1)})\subseteq \overline{E_1}$. As a result, $(\tau_{\mathbf{F}}^{(1)})^\prime(t)=\tau_1^\prime(t)=\frac{(F^{(1)})^\prime(t)+(F^{(2)})^\prime(t)}{\mu_1}$ a.e. in $\mathcal{S}(F^{(1)})$. By definition of EAP, $C_{\mathbf{F}}^{(1)}(\cdot)$ must be constant in $\mathcal{S}(F^{(1)})$, giving us, $(C_{\mathbf{F}}^{(1)})^\prime(t)=\frac{(F^{(1)})^\prime(t)+(F^{(2)})^\prime(t)}{\mu_1}-\gamma^{(1)}=0$ in $\mathcal{S}(F^{(1)})$. This implies $(F^{(1)})^\prime(t)+(F^{(2)})^\prime(t)=\mu_1\gamma$.

If $\tau_1(t)\in\overline{E_2}$, Lemma \ref{lem_inst1_uneqpref_derv_of_tau2} together with the fact that $C_{\mathbf{F}}^{(2)}(\cdot)$ is constant in $\mathcal{S}(F^{(2)})$ implies $(F^{(2)})^\prime(t)=\mu_2\gamma$. Otherwise, queue 2 will have zero waiting time in some neighbourhood of $\tau_1(t)$. This can only happen in $\mathcal{S}(F^{(2)})$ if $t\in \mathcal{S}(F^{(1)})$. In that case, since $A_2(\tau_1(t))=F^{(2)}(t)$, we must have $A_2^\prime(\tau_1(t))=\frac{(F^{(2)})^\prime(t)}{\tau_1^\prime(t)}=\frac{(F^{(2)})^\prime(t)}{\gamma}$. Now for queue 2 to remain empty in a neighbourhood of $\tau_1(t)$, we will need $\frac{(F^{(2)})^\prime(t)}{\gamma}\leq\mu_2$ implying $(F^{(2)})^\prime(t)\leq\mu_2\gamma$. \hfill\qedsymbol 

\bigskip

\noindent\textbf{\textit{Proof of Lemma \ref{lem_queue1_and_2_idle_inst1_eqpref}:}}
\begin{enumerate}[leftmargin=*]
        \item  Proof of this follows the same argument used to prove Lemma \ref{lem_appndx_inst1_uneqpref_queue2idle}.
        \item Assume contradiction, \textit{i.e.}, queue 1 has a positive waiting time at $T^{(1)}_f$. By Lemma \ref{lem_rate_of_arrivals_inst1_eqpref} class 2 customers can arrive at a maximum rate of $\mu_2\gamma$ after $T^{(1)}_f$. Therefore using (\ref{eq:derv_of_tau}), till queue 1 remains engaged, $(\tau_{\mathbf{F}}^{(1)})^\prime(t)=\tau_1^\prime(t)\leq\frac{\mu_2\gamma}{\mu_1}$. Hence, $(C_{\mathbf{F}}^{(1)})^\prime(t)\leq\frac{\mu_2\gamma}{\mu_1}-\gamma<0$. Therefore, the class 1 customer arriving at $T^{(1)}_f$ will be better off arriving when queue 1 empties, making $\mathbf{F}$ unstable. \hfill\qedsymbol 
\end{enumerate}

\noindent\textbf{\textit{Proof of Lemma \ref{lem_cost2_const_inst1_eqpref}:}}~~By definition of EAP, if $t\in \mathcal{S}(F^{(2)})$, $(C_{\mathbf{F}}^{(2)})^\prime(t)=0$. Otherwise, if $t\in[T_a,T_f]/\mathcal{S}(F^{(2)})$, by Lemma \ref{lem_queue1_and_2_idle_inst1_eqpref}, queue 2 must be empty at some neighbourhood of $\tau_1(t)$. As a result, $\tau_{\mathbf{F}}^{(2)}(s)=\tau_1(s)$ for every $s\in[t-\delta,t+\delta]$ for $\delta>0$ chosen sufficiently small, and therefore $(\tau_{\mathbf{F}}^{(2)})^\prime(t)=\tau_1^\prime(t)$. By Lemma \ref{lem_rate_of_arrivals_inst1_eqpref}, class 1 customers arrive at rate $\mu_1\gamma$ in $[T_a,T_f]/\mathcal{S}(F^{(2)})$. So, using (\ref{eq:derv_of_tau}), $(C_{\mathbf{F}}^{(2)})^\prime(t)=\tau_1^\prime(t)-\gamma=\frac{(F^{(1)})^\prime(t)}{\mu_1}-\gamma=0$. Hence, $(C_{\mathbf{F}}^{(2)})^\prime(t)=0$ for every $t\in[T_a,T_f]$ and therefore the cost remains constant. \hfill\qedsymbol \\

\noindent\textbf{\textit{Proof of Lemma \ref{lem_inst1_eqpreF^{(2)}rate_less_than_mu2}:}}~~If $t\in\mathcal{S}(F^{(1)})\cap \mathcal{S}(F^{(2)})$, we must have $t\in\overline{E_1}$. By (\ref{eq:derv_of_tau}) and Lemma \ref{lem_rate_of_arrivals_inst1_eqpref}, we have $\tau_1^\prime(t)=\gamma$. Now, we also know $A_2(\tau_1(t))=F^{(2)}(t)$. By chain rule and using Lemma \ref{lem_rate_of_arrivals_inst1_eqpref}, we get $A_2^\prime(\tau_1(t))=\frac{(F^{(2)})^\prime(t)}{\tau_1^\prime(t)}\leq\frac{\mu_2\gamma}{\gamma}=\mu_2$, which is exactly the statement of the lemma.\hfill\qedsymbol

\subsubsection{Proof of Theorem \ref{mainthm_inst1_eqpref}}\label{appndx:thm_inst1_eqpref}
We only provide here the arguments which were skipped while presenting the key steps of the proof of Theorem \ref{mainthm_inst1_eqpref} in Section \ref{sec_inst1_eqpref}.  

\noindent\textit{\textbf{Skipped details in the proof of Lemma \ref{lem:inst1eqprefcase1}:}}

\textbf{Queue 2 has positive waiting time in $(0,T_f)$:}~By Lemma \ref{lem_queue1_and_2_idle_inst1_eqpref}, queue 1 must be empty at time $T^{(1)}_f$ and will stay empty in $[T^{(1)}_f,T_f]$ since maximum arrival rate of class 2 users in $[T^{(1)}_f,T_f]$ is $\mu_2\gamma<\mu_1$. Since a positive mass of class 2 users arrive in $(T^{(1)}_f,T_{f}]$, queue 2 must have a positive waiting time at $T^{(1)}_f$. Otherwise the network will be empty at $T^{(1)}_f$ and every class 2 user arriving after $T^{(1)}_f$ will have a strict incentive to arrive at $T^{(1)}_f$. Also as queue 2 is the only queue serving in the network in $[T^{(1)}_f,T_f]$, it must have a positive waiting time in $[T^{(1)}_f,T_f)$.

By Lemma \ref{lem_inst1_eqpreF^{(2)}rate_less_than_mu2}, in equilibrium, class 2 users arrive at queue 2 after passing through queue 1 at a maximum rate of $\mu_2$ in $[\tau_1(T^{(1)}_a),T^{(1)}_f]$. As a result, for queue 2 to have positive waiting queue at $T^{(1)}_f$, queue 2 must have positive waiting queue in $[\tau_1(T^{(1)}_a),T^{(1)}_f]$. 

Now if $T^{(1)}_a=T_a<0$, then $\tau_1(T^{(1)}_a)=0$ and the preceding argument implies queue 2 has positive waiting queue in $[0,T^{(1)}_f]$. On the other hand, if $T_a<T^{(1)}_a$, queue 1 must have a positive waiting time at $T^{(1)}_a$. As a result, class 2 users will arrive at queue 2 from queue 1 at rate $\mu_1>\mu_2$ in $(0,\tau_1(T^{(1)}_a)]$ causing queue 2 to have a positive queue length in $(0,\tau_1(T^{(1)}_a)]$. 

The preceding argument implies queue 2 has a positive waiting time in the intervals $(0,\tau_1(T^{(1)}_a)]$, $[\tau_1(T^{(1)}_a),T^{(1)}_f]$ and $[T^{(1)}_f,T_f)$ and therefore the same is true in $(0,T_f)$. \hfill\qedsymbol

\noindent\textbf{\textit{Skipped details in the proof of Lemma \ref{lem:inst1eqprefcase2}:}}

\textbf{If the necessary condition holds, every joint arrival profile in the set of candidates is an EAP:} We will prove that, if $\Lambda^{(1)}\geq\left(\frac{\mu_1}{\mu_2}-1\right)\cdot\Lambda^{(2)}$, when users are arriving by any joint arrival profile in the obtained set of Type II candidates (which we argued will be non-empty): every class have their cost constant in $[T_a,T_f]$ and higher outside. Towards that, it will be sufficient to show, if $\Lambda^{(1)}\geq\left(\frac{\mu_1}{\mu_2}-1\right)\cdot\Lambda^{(2)}$, for every joint arrival profile in the obtained set of Type II candidates: both the groups $i\in\{1,2\}$ have their derivative of the cost function satisfying $(C_{\mathbf{F}}^{(i)})^\prime(t)\leq 0$ in $(-\infty,T_a)$, $(C_{\mathbf{F}}^{(i)})^\prime(t)=0$ in $[T_a,T_f]$ and $(C_{\mathbf{F}}^{(i)})^\prime(t)\geq 0$ in $(T_f,\infty)$. We now analyze below the derivative of the cost function of the two classes to prove our mentioned objective: 

\begin{itemize}[leftmargin=*]
    \item In $(-\infty,T_a)$:~~Since for both classes $i\in\{1,2\}$, $\tau_{\mathbf{F}}^{(i)}(t)=0$ in $(-\infty,T_{a})$, we have $(C_{\mathbf{F}}^{(i)})^\prime(t)-\gamma<0$ for every $t\in(-\infty,T_a)$. 
    
    \item In $[T_a,T_f]$:~The function $F^{(1)}(t)+F^{(2)}(t)-\mu_1 t=\mu_1\gamma\cdot(t-T_a)-\mu_1 t$ is decreasing in $[0,T_f]$. We also have $F^{(1)}(T_f)+F^{(2)}(T_f)-\mu_1 T_f=\Lambda^{(1)}+\Lambda^{(2)}-\mu_1 T_f=0$, implying queue 1 empties at $T_f$. Combining the previous two statements $F^{(1)}(t)+F^{(2)}(t)>\mu_1\cdot\max\{0,t\}$ in $[T_a,T_f]$, implying queue 1 has a positive waiting time. Therefore using (\ref{eq:derv_of_tau}) and the fact $(F^{(1)})^\prime(t)+(F^{(2)})^\prime(t)=\gamma$ for every $t\in[T_a,T_f]$, we get $(C_{\mathbf{F}}^{(1)})^\prime(t)=\tau_1^\prime(t)-\gamma=\frac{(F^{(1)})^\prime(t)+(F^{(2)})^\prime(t)}{\mu_1}-\gamma=0$ for every $t\in[T_a,T_f]$. By Lemma \ref{lem_cost2_const_inst1_eqpref}, $C_{\mathbf{F}}^{(2)}(\cdot)$ stays constant between $[T_a,T_f]$ making $(C_{\mathbf{F}}^{(2)})^\prime(t)=0$ for every $t\in[T_a,T_f]$.   
    
    \item In $(T_f,\infty)$:~~By the previous step queue 1 became empty at $T_f$. We also argued before queue 2 stays empty after time zero. As a result, the whole network is empty in $(T_f,\infty)$, and for both classes $i\in\{1,2\}$, $(C_{\mathbf{F}}^{(i)})^\prime(t)=1-\gamma>0$ for every $t\in(T_f,\infty)$. \hfill\qedsymbol
\end{itemize}

\section{Proofs of the main results in Section \ref{sec_hetarrivals}}

Note that $\mathcal{S}(F^{(2)})\subseteq\overline{E_2}$ in every EAP. As a result, by (\ref{eq:derv_of_tau}), $\tau_2(\cdot)$ is differentiable in $\mathcal{S}(F^{(2)})$. Moreover since $C_{\mathbf{F}}^{(2)}(\cdot)$ is constant over $\mathcal{S}(F^{(2)})$, we will have $(C_{\mathbf{F}}^{(2)})^\prime(t)=\tau_2^\prime(t)-\gamma^{(2)}=0$ in $(\mathcal{S}(F^{(2)}))^o$ implying:
\begin{align}\label{eq:isocost_inst2_grp2}
    \forall~t\in(\mathcal{S}(F^{(2)}))^o,~\tau_2^\prime(t)&=\gamma^{(2)}.   
\end{align}

Similarly, since $C_{\mathbf{F}}^{(1)}(\cdot)$ is constant in $\mathcal{S}(F^{(1)})$, we must have $(C_{\mathbf{F}}^{(1)})^\prime(t)=(\tau_{\mathbf{F}}^{(1)})^\prime(t)-\gamma^{(1)}=0$ in $(\mathcal{S}(F^{(1)}))^o$,
\begin{align}\label{eq:isocost_inst2_grp1}
    \forall~t\in(\mathcal{S}(F^{(1)}))^o,~(\tau_{\mathbf{F}}^{(1)})^\prime(t)&=\gamma^{(1)}. 
\end{align}

\subsection{Structural properties of EAP for HAS with $\gamma^{(1)}\neq\gamma^{(2)}$}\label{appndx:inst2_uneqpref}

Several of the structural properties will be true in the case $\gamma^{(1)}=\gamma^{(2)}$. So, whenever some lemma is applicable only for the case of unequal preferences, we mention it explicitly in the lemma statement. 

Lemma \ref{lem_appndx_inst2_uneqpref_queu1_engaged_mixedarrival} helps us to identify the structure of the supports and also to prove Lemma \ref{lem_threshold_behav_inst2} on threshold behavior in HAS. 
\begin{lemma}\label{lem_appndx_inst2_uneqpref_queu1_engaged_mixedarrival}
    If $\gamma^{(1)}\neq\gamma^{(2)}$, in the EAP, $t\in(\mathcal{S}(F^{(1)}))^o\cap(\mathcal{S}(F^{(2)}))^o$ implies $t\in\overline{E_1}$. 
\end{lemma}
\begin{proof} 
Assume the contradiction, \textit{i.e.}, $t\in (\mathcal{S}(F^{(1)}))^o\cap(\mathcal{S}(F^{(2)}))^o$ but $t\notin \overline{E_1}$. Then we can have $\delta>0$ sufficiently small such that $[t-\delta,t+\delta]\subseteq (\mathcal{S}(F^{(1)}))^o\cap(\mathcal{S}(F^{(2)}))^o\cap\overline{E_1}^c$. Since queue 1 has zero waiting time in $[t-\delta,t+\delta]$, we must have $[t-\delta,t+\delta]\subseteq E_2$ and for both the groups $i=1,2$ $\tau_{\mathbf{F}}^{(i)}(s)=\tau_2(s)$ at every $s\in[t-\delta,t+\delta]$. Therefore, using (\ref{eq:isocost_inst2_grp1}) and (\ref{eq:isocost_inst2_grp2}), $\tau_1^\prime(s)=\gamma^{(1)}=\gamma^{(2)}$ at every $s\in[t-\delta,t+\delta]$. This contradicts with the fact that $\gamma^{(1)}\neq\gamma^{(2)}$. 
\end{proof}

We define the arrival profile of class 1 users from queue 1 to 2 as $Y_1(t)=F^{(1)}(\tau_1^{-1}(t))$, where \newline$\tau_1^{-1}(t)=\sup\{s~\vert~\tau_1(s)\leq t\}$. Since $F^{(1)}(\cdot)$ is absolutely continuous,  $Y_1(\cdot)$ is also absolutely continuous and $\mathcal{S}(Y_1)$ has no isolated point. Lemma \ref{lem_inst2_uneqpref_mixedarr_nec_cond} partially proves the sufficiency of $\mu_1\geq\mu_2\gamma^{(2)}$ for queue 2 to serve the two classes over disjoint sets of times.

\begin{lemma}\label{lem_inst2_uneqpref_mixedarr_nec_cond}
    If $\mu_1\geq\mu_2\gamma^{(2)}$ and $\gamma^{(1)}\neq\gamma^{(2)}$, then $\mathcal{S}(Y_1)$ and $(\mathcal{S}(F^{(2)}))^o$ must be disjoint.
\end{lemma}

\begin{proof}
    Assuming contradiction and the fact that $\mathcal{S}(Y_1)$ cannot have an isolated point, we will have $t_1,t_2\in\mathcal{S}(Y_1)$ such that $Y_1(t_2)>Y_1(t_1)$ and $(t_1,t_2)\subseteq(\mathcal{S}(F^{(2)}))^o$. By Lemma \ref{lem_appndx_inst2_uneqpref_queu1_engaged_mixedarrival}, we can find $S\subseteq(t_1,t_2)\cap\mathcal{S}(Y_1)$, such that, $\lambda(S)>0$ (where $\lambda(\cdot)$ is Lebesgue measure) and queue 1 stays engaged in  $S$. As a result, mass of class 1 users who have arrived in queue 2 from queue 1 in $S$ is $\mu_1\cdot\lambda(S)$. By (\ref{eq:isocost_inst2_grp2}) and (\ref{eq:derv_of_tau}), since $(t_1,t_2)\subseteq\mathcal{S}(F^{(2)})$, we have $Y_1^\prime(t)+(F^{(2)})^\prime(t)=\mu_2\gamma^{(2)}$ almost everywhere in $(t_1,t_2)$. Therefore, total mass of users who have arrived in queue 2  in the set of times $S$ will be $\mu_2\gamma^{(2)}\cdot\lambda(S)$. This gives us, $\mu_1\cdot\lambda(S)<\mu_2\gamma^{(2)}\cdot\lambda(S)$, which implies $\mu_1<\mu_2\gamma^{(2)}$, a contradiction to our assumption. 
\end{proof}

\noindent\textbf{A better proof:}~Observe that $E_1\subseteq\mathcal{S}(Y_1)$. There are two possibilities now. If $E_1\cap(\mathcal{S}(F^{(2)}))^o\neq\emptyset$, the old proof works nicely. Otherwise, if $(\mathcal{S}(Y_1)/E_1)\cap(\mathcal{S}(F^{(2)}))^o\neq\emptyset$, we have $t_1,t_2\subseteq\mathcal{S}(Y_1)\cap\mathcal{S}(F^{(2)})$ with $[t_1,t_2]\subseteq(\mathcal{S}(F^{(2)}))^o$, such that queue 1 has zero waiting in $[t_1,t_2]$. Since $[t_1,t_2]\subseteq(\mathcal{S}(F^{(2)}))^o$, we have $(C_{\mathbf{F}}^{(2)})^\prime(t)=0$ in $[t_1,t_2]$, which implies $\tau_2^\prime(t)=\gamma^{(2)}$. Again we have $\tau_{F}^{(1)}(t)=\tau_2(t)$ in $[t_1,t_2]$, since queue 1 is idle in $[t_1,t_2]$. As a result, $(C_{\mathbf{F}}^{(1)})^\prime(t)=\tau_2^\prime(t)-\gamma^{(1)}=\gamma^{(2)}-\gamma^{(1)}$ in $[t_1,t_2]$, implying $C_{\mathbf{F}}^{(1)}(t_1)\neq C_{\mathbf{F}}^{(2)}(t_2)$, which contradicts our assumption that $\mathbf{F}$ is an EAP.\hfill\qed\\ 

Lemma \ref{lem_supp_are_intervals_inst2} together with Lemma \ref{lem_arrival_rates_inst2} help us reduce our search of EAP to piece-wise linear joint arrival profiles. 

\begin{lemma}[\textbf{Support structure}]\label{lem_supp_are_intervals_inst2}
In every EAP of HAS, the following statements are true:
\begin{enumerate}
    \item $\mathcal{S}(F^{(1)})\cup \mathcal{S}(F^{(2)})$ is an interval.
    \item  If $\mu_1<\mu_2\gamma^{(2)}$, $\mathcal{S}(F^{(2)})$ is an interval.
    \item  If $\gamma^{(1)}\neq\gamma^{(2)}$, $\mathcal{S}(F^{(1)})$ is an interval.
\end{enumerate}
\end{lemma}
\begin{proof}
\noindent\textbf{(1)} Assume that there is a gap $(t_1,t_2)$ in the support such that $F^{(1)}(t_1)+F^{(2)}(t_1)=F^{(1)}(t_2)+F^{(2)}(t_2)$. The following scenarios might happen:
\begin{itemize}[leftmargin=*]
    \item If $\mathcal{S}(F^{(1)})\cap (-\infty,t_1]$ has zero Lebesgue measure, then, all class $1$ users will arrive after time $t_2$. Till time $t_1$ only Class 2 players have arrived and $W_{\mathbf{F}}^{(1)}(\cdot)$ must be positive between $[t_1,t_2]$. Therefore $[t_1,t_2]\subseteq E_2$. Applying (\ref{eq:derv_of_tau}), since no class 2 user arrives between $[t_1,t_2]$, $\tau_2(t_1)=\tau_2(t_2)$. Therefore, the class 2 user arriving at time $t_1$ will be strictly better off by arriving at time $t_2$.
    
    \item Let $\mathcal{S}(F^{(1)})\cap (-\infty,t_1)$ has a positive Lebesgue measure. Then consider the time\\ $\tilde{t}=\inf\left\{t\leq t_1~\text{s.t.}~t\in \mathcal{S}(F^{(1)})\right\}$. 
    \begin{itemize}
        \item If $\tilde{t}\in E_1$, we can find $\delta\in(0,t_2-\tilde{t})$ such that $[\tilde{t},\tilde{t}+\delta]\subseteq E_1$. Since no class 1 user arrives in $[\tilde{t},\tilde{t}+\delta]$, by (\ref{eq:derv_of_tau}), $\tau_1(\tilde{t})=\tau_1(\tilde{t}+\delta)$. The previous statement along with $\tau_{\mathbf{F}}^{(1)}(t)=\tau_2(\tau_1(t))$ implies $\tau_{\mathbf{F}}^{(1)}(\tilde{t})=\tau_{\mathbf{F}}^{(1)}(\tilde{t}+\delta)$. Therefore, the class 1 user arriving at $\tilde{t}$ will be strictly better off arriving at $\tilde{t}+\delta$. 
        
        \item Otherwise if $\tilde{t}\notin E_1$, then, $[t_1,t_2]\subseteq E_2$. Again using (\ref{eq:derv_of_tau}), $\tau_2(t_1)=\tau_2(t_2)$. Since queue 1 has zero waiting time between $[t_1,t_2]$, $\tau_{\mathbf{F}}^{(i)}(s)=\tau_2(s)$ for every $s\in[t_1,t_2]$. Therefore the user arriving at $t_1$, irrespective of her group, will be strictly better off arriving at $t_2$.
    \end{itemize}
\end{itemize}
In all the situations above we can conclude $\{F^{(1)},F^{(2)}\}$ cannot be an EAP if there is a gap in $\mathcal{S}(F^{(1)})\cup\mathcal{S}(F^{(2)})$. 

\noindent\textbf{(2)} Assume contradiction, \textit{i.e.}, $\mu_1<\mu_2\gamma^{(2)}$ and $\mathcal{S}(F^{(2)})$ has a gap $(t_1,t_2)$ such that $t_2>t_1$, $t_1,t_2\in \mathcal{S}(F^{(2)})$ and $F^{(2)}(t_1)=F^{(2)}(t_2)$. As a positive mass of class 2 users will arrive after time $t_2$, we must have $(t_1,t_2)\subseteq E_2$. Therefore by (\ref{eq:derv_of_tau}), $(\tau_{\mathbf{F}}^{(2)})^\prime(t)=\frac{Y_1^\prime(t)}{\mu_2}\leq\frac{\mu_1}{\mu_2}$ a.e. in $(t_1,t_2)$ . Using this, $(C_{\mathbf{F}}^{(2)})^\prime(t)=\tau_2^\prime(t)-\gamma^{(2)}\leq\frac{\mu_1}{\mu_2}-\gamma^{(2)}<0$ a.e. in $(t_1,t_2)$ . This implies $C_{\mathbf{F}}^{(2)}(t_2)<C_{\mathbf{F}}^{(2)}(t_1)$ and as a result, the class 2 user arriving at $t_1$ will be strictly better off arriving at $t_2$ which contradicts the fact that $\{F^{(1)},F^{(2)}\}$ is an EAP.

\noindent\textbf{(3)} Assume contradiction, \textit{i.e.}, $\gamma^{(1)}\neq\gamma^{(2)}$ and $\mathcal{S}(F^{(1)})$ has a gap $(t_1,t_2)$ such that $t_2>t_1$, $t_1,t_2\in \mathcal{S}(F^{(1)})$ and $F^{(1)}(t_1)=F^{(1)}(t_2)$. Since $\mathcal{S}(F^{(1)})\cup\mathcal{S}(F^{(2)})$ must be an interval, we must have $[t_1,t_2]\subseteq\mathcal{S}(F^{(2)})$. 

Note that queue 1 must have zero waiting time at $t_1$. Otherwise, we can find $\delta>0$ sufficiently small such that $[t_1,t_1+\delta]\subseteq[t_1,t_2]\cap E_1$. Using (\ref{eq:derv_of_tau}), $\tau_1(t_1+\delta)=\tau_1(t_1)$ and as a consequence $\tau_{\mathbf{F}}^{(1)}(t_1+\delta)=\tau_{\mathbf{F}}^{(1)}(t_1)$, implying the class 1 user arriving at $t_1$ will be strictly better off arriving at $t_1+\delta$. 

Now if queue 1 has zero waiting time between $[t_1,t_2]$, we must have \newline $[t_1,t_2]\subseteq\overline{E_2}$ since $[t_1,t_2]\subseteq\mathcal{S}(F^{(2)})$. As a result, for $t\in[t_1,t_2]$, using (\ref{eq:derv_of_tau}), $\tau_2^\prime(t)=\frac{(F^{(2)})^\prime(t)}{\mu_2}$ and after applying (\ref{eq:isocost_inst2_grp2}), $(F^{(2)})^\prime(t)=\mu_2\gamma^{(2)}$. Therefore, in $[t_1,t_2]$, $(C_{\mathbf{F}}^{(1)})^\prime(t)=\tau_2^\prime(t)-\gamma^{(1)}=\gamma^{(2)}-\gamma^{(1)}\neq 0$. This implies, if $\gamma^{(1)}>\gamma^{(2)}$, class 1 user arriving at $t_1$ will be strictly better off arriving at $t_2$ and vice-versa. Therefore $\{F^{(1)},F^{(2)}\}$ will not be an EAP.
\end{proof}

\begin{lemma}[\textbf{Rate of arrival}]\label{lem_arrival_rates_inst2}
    If $\gamma^{(1)}\neq\gamma^{(2)}$, except over a set of times having zero Lebesgue measure, arrival rates in any EAP are,\newline $\forall t\in\mathcal{S}(F^{(1)})~~(F^{(1)})^\prime(t)=\begin{cases} \mu_1\gamma^{(1)}, &\text{for}~\tau_1(t)\in\overline{E_2}^c,\\ \mu_2\gamma^{(1)}, &\text{for $\tau_1(t)\notin \mathcal{S}(F^{(2)})\cup\overline{E_2}^c$},\\ \frac{\mu_1\gamma^{(1)}}{\gamma^{(2)}}, &\text{for $\tau_1(t)\in \mathcal{S}(F^{(2)})$}, \end{cases}$ and \newline $\forall t\in\mathcal{S}(F^{(2)}),~(F^{(2)})^\prime(t)=\begin{cases} \mu_2\gamma^{(2)},~&\text{for}~t\in(-\infty,0)\cup(\mathcal{S}(F^{(1)}))^c,\\ \mu_2\gamma^{(2)}-\mu_1~&\text{for}~t\in [0,\infty)\cap\mathcal{S}(F^{(1)}),\end{cases}$ \newline where $[0,\infty)\cap \mathcal{S}(F^{(1)})\cap\mathcal{S}(F^{(2)})$ has zero Lebesgue measure if $\mu_1\geq\mu_2\gamma^{(2)}$.
\end{lemma}
\begin{proof} For class 1 users, if $\tau_1(t)\in\overline{E_2}^c$, we can find $\delta>0$ sufficiently small such that, $\tau_{\mathbf{F}}^{(1)}(s)=\tau_1(s)$ for every $s\in[t-\delta,t+\delta]$. Moreover we must have $[t-\delta,t+\delta]\subseteq\overline{E_1}$. Therefore, $(\tau_{\mathbf{F}}^{(1)})^\prime(s)=\tau_1^\prime(s)=\frac{(F^{(1)})^\prime(s)}{\mu_1}$ a.e. in $[t-\delta,t+\delta]$. Using (\ref{eq:isocost_inst2_grp1}), this implies $(F^{(1)})^\prime(s)=\mu_1\gamma^{(1)}$ a.e. in $[t-\delta,t+\delta]$.

Now if $\tau_1(t)\in\mathcal{S}(F^{(2)})$, by (\ref{eq:isocost_inst2_grp2}), we will have $\tau_2^\prime(\tau_1(t))=\gamma^{(2)}$. By Lemma \ref{lem_appndx_inst2_uneqpref_queu1_engaged_mixedarrival}, we have $t\in\overline{E_1}$. Therefore by (\ref{eq:derv_of_tau}), $\tau_1^\prime(t)=\frac{(F^{(1)})^\prime(t)}{\mu_1}$. Now applying (\ref{eq:isocost_inst2_grp1}), $(\tau_{\mathbf{F}}^{(1)})^\prime(t)=\frac{(F^{(1)})^\prime(t)}{\mu_1}\gamma^{(2)}=\gamma^{(1)}$, giving us $(F^{(1)})^\prime(t)=\frac{\mu_1\gamma^{(1)}}{\gamma^{(2)}}$. 

If $\tau_1(t)\in (\mathcal{S}(F^{(2)})\cup\overline{E_2}^c)^c$, we must have $\tau_1(t)\in \overline{E_2}$. Using an argument similar to the one used for proving Lemma \ref{lem_inst1_uneqpref_derv_of_tau2}, we will have $(\tau_{\mathbf{F}}^{(1)})^\prime(t)=\frac{(F^{(1)})^\prime(t)}{\mu_2}$. Using (\ref{eq:isocost_inst2_grp1}), we have $(F^{(1)})^\prime(t)=\mu_2\gamma^{(1)}$. 

Now $\mathcal{S}(F^{(2)})\subseteq \overline{E_2}$. Therefore, for $t\in\mathcal{S}(F^{(2)})$, $\tau_2^\prime(t)=\frac{(F^{(2)})^\prime(t)+Y_1^\prime(t)}{\mu_2}$. Using (\ref{eq:isocost_inst2_grp2}), $\tau_2^\prime(t)=\gamma^{(2)}$ implies, $(F^{(2)})^\prime(t)=\mu_2\gamma^{(2)}-Y_1^\prime(t)$. Now if $t\in(-\infty,0]$, $Y_1^\prime(t)=0$ and as a result, $(F^{(2)})^\prime(t)=\mu_2\gamma^{(2)}$. If $t\in(0,\infty)\cap\mathcal{S}(F^{(1)})$, by Lemma \ref{lem_appndx_inst2_uneqpref_queu1_engaged_mixedarrival}, we have $t\in \overline{E_1}$. As a result $Y_1^\prime(t)=\mu_1$. Therefore, if $t\in(0,\infty)\cap\mathcal{S}(F^{(1)})$, $(F^{(2)})^\prime(t)=\mu_2\gamma^{(2)}-\mu_1$. 

Now if $t\in(0,\infty)/\mathcal{S}(F^{(1)})$, queue 1 must be empty at $t$. This follows trivially if no class 1 user has arrived before $t$. Otherwise if some positive mass of class 1 user has arrived before $t$, let $\tilde{t}=\sup\{s<t~\vert~s\in\mathcal{S}(F^{(1)})\}$. Then the  class 1 user arriving at $\tilde{t}$ must see zero waiting time in queue 1. Otherwise she can arrive slightly late and improve her cost. Therefore queue 1 will have zero waiting time at $t$. As a result $Y_1^\prime(t)=0$ and this implies $(F^{(2)})^\prime(t)=\mu_2\gamma^{(2)}$ for $t\in (0,\infty)/\mathcal{S}(F^{(1)})$.
\end{proof}

Lemma \ref{lem_appndx_inst2_uneqpref_threshold_behav_suff_cond} partially implies that the condition in Lemma \ref{lem_threshold_behav_inst2} is necessary for queue 2 to serve the two classes over non overlapping sets of times. 

\begin{lemma}\label{lem_appndx_inst2_uneqpref_threshold_behav_suff_cond}
    If $\mu_1<\mu_2\gamma^{(2)}$ and $\gamma^{(1)}\neq\gamma^{(2)}$, $\mathcal{S}(Y_1)$ and $\mathcal{S}(F^{(2)})$ cannot have disjoint interiors.
\end{lemma}
\begin{proof}
    By Lemma \ref{lem_supp_are_intervals_inst2}, we can assume $\mathcal{S}(F^{(1)})=[T^{(1)}_a,T^{(1)}_f]$ and $\mathcal{S}(F^{(2)})=[T^{(2)}_a,T^{(2)}_f]$, such that their union is an interval. Note that under all situations, in EAP $T^{(2)}_a<0$ and $T^{(1)}_f>0$. Moreover, we have $\mathcal{S}(Y_1)=[\tau_1(T^{(1)}_a),\tau_1(T^{(1)}_f)]$, where $\tau_1(T^{(1)}_a)=\max\{0,T^{(1)}_a\}$. The only possible way $\mathcal{S}(Y_1)$ and $\mathcal{S}(F^{(2)})$ can have disjoint interiors is by having $T^{(2)}_f=T^{(1),+}_a$. In that case, queue 2 must have positive waiting time at $T^{(2)}_f$. Therefore till queue 2 becomes idle, by (\ref{eq:derv_of_tau}), $\tau_2^\prime(t)=\frac{Y_1^\prime(t)}{\mu_2}\leq\frac{\mu_1}{\mu_2}$ and as a result $(C_{\mathbf{F}}^{(2)})^\prime(t)\leq\frac{\mu_1}{\mu_2}-\gamma^{(2)}<0$. Hence the class 2 user arriving at $T^{(2)}_f$ will be better off arriving at the moment queue 2 becomes idle. This implies $\{F^{(1)},F^{(2)}\}$ will not be an EAP and we arrive at a contradiction. 
\end{proof}

\noindent\textit{\textbf{Proof of Lemma \ref{lem_threshold_behav_inst2}:}} Lemma \ref{lem_inst2_uneqpref_mixedarr_nec_cond} and \ref{lem_appndx_inst2_uneqpref_threshold_behav_suff_cond} together with Lemma \ref{lem_supp_are_intervals_inst2} implies the statement of Lemma \ref{lem_threshold_behav_inst2}. \hfill\qedsymbol \\

Using Lemma \ref{lem_supp_are_intervals_inst2}, for every EAP $\mathbf{F}=\{F^{(1)},F^{(2)}\}$, we can assume $\mathcal{S}(F^{(1)})=[T^{(1)}_a,T^{(1)}_f]$ for some $T^{(1)}_f>T^{(1)}_a$. Lemma \ref{lem_inst2_uneqpref_queue1idle} and \ref{lem_inst2_uneqpref_queue2idle} are about the state of the two queues in $\mathcal{S}(F^{(1)})$ and $\mathcal{S}(F^{(2)})$, and they will help us identify the support boundaries of the EAP.

\begin{lemma}\label{lem_inst2_uneqpref_queue1idle}
    If $\gamma^{(1)}\neq\gamma^{(2)}$, in the EAP, Queue 1 must be empty at $T^{(1)}_f$.
\end{lemma}
\begin{proof} 
Otherwise, if $T^{(1)}_f\in E_1$, by (\ref{eq:derv_of_tau}), $\tau_1(\cdot)$ remains constant till queue 1 empties. As a result $\tau_{\mathbf{F}}^{(1)}(\cdot)$ remains constant till queue 1 empties. Hence the class 1 user arriving at $T^{(1)}_f$ will be strictly better off arriving at the moment queue 1 empties.  
\end{proof}

Again using Lemma \ref{lem_supp_are_intervals_inst2}, if $\mu_1<\mu_2\gamma^{(2)}$, we can assume $\mathcal{S}(F^{(2)})=[T^{(2)}_a,T^{(2)}_f]$ for some $T^{(2)}_f>T^{(2)}_a$.  

\begin{lemma}\label{lem_inst2_uneqpref_queue2idle}
    If $\mu_1<\mu_2\gamma^{(2)}$, in the EAP, Queue 2 must be empty at $T^{(2)}_f$.
\end{lemma}
\begin{proof} 
Assuming contradiction, if queue 2 has a positive waiting time at $T^{(2)}_f$, using (\ref{eq:derv_of_tau}) $\tau_2^\prime(t)=\frac{Y_1^\prime(t)}{\mu_2}\leq\frac{\mu_1}{\mu_2}$ till queue 2 empties. Therefore, $(C_{\mathbf{F}}^{(2)})^\prime(t)\leq\frac{\mu_1}{\mu_2}-\gamma^{(2)}<0$ till queue 2 empties, implying the class 2 user arriving at $T^{(2)}_f$ will be strictly better off arriving at the moment queue 2 empties. This implies  $\{F^{(1)},F^{(2)}\}$ will not be an EAP.
\end{proof}

\subsubsection{\textbf{Supporting Lemmas for proving Theorem \ref{mainthm_inst2_reg1}}}
In the proof of Theorem \ref{mainthm_inst2_reg1}, we use the following two lemmas to characterize the EAP in the case: \textbf{1)}~$\gamma^{(1)}>\gamma^{(2)}$ and \textbf{2)}~$\gamma^{(1)}<\gamma^{(2)}$.

\begin{lemma}\label{lem_inst2_uneqpref_sign_reg1a}
    If $\gamma^{(1)}>\gamma^{(2)}$ and $\mu_1<\mu_2\gamma^{(2)}$, the support boundaries must satisfy $T^{(1)}_f\geq T^{(2)}_f$.
\end{lemma}

\begin{proof} 
Otherwise if $T^{(1)}_f<T^{(2)}_f$, by Lemma \ref{lem_inst2_uneqpref_queue1idle}, queue 1 is empty at $T^{(1)}_f$ and $(T^{(1)}_f,T^{(2)}_f)\subseteq E_2$. As a result, between $[T^{(1)}_f,T^{(2)}_f]$ $\tau_{\mathbf{F}}^{(1)}(t)=\tau_2(t)$. Also,  by Lemma \ref{lem_arrival_rates_inst2} class 2 users are arriving at rate $\mu_2\gamma^{(2)}$ between $[T^{(1)}_f,T^{(2)}_f]$. So, using (\ref{eq:derv_of_tau}) $(C_{\mathbf{F}}^{(1)})^\prime(t)=\tau_2^\prime(t)-\gamma^{(1)}=\frac{(F^{(2)})^\prime(t)}{\mu_2}-\gamma^{(1)}=\gamma^{(2)}-\gamma^{(1)}<0$. As a result, the class 1 user arriving at $T^{(1)}_f$ will be strictly better off arriving at $T^{(2)}_f$.  
\end{proof}

\begin{lemma}\label{lem_inst2_uneqpref_sign_reg1b}
    If $\gamma^{(1)}<\gamma^{(2)}$, the support boundaries must satisfy $T^{(1)}_a\leq 0$.
\end{lemma}

\begin{proof} 
Assuming contradiction, if $T^{(1)}_a>0$, by Lemma \ref{lem_arrival_rates_inst2} class 2 users arrive at rate $\mu_2\gamma^{(2)}$ in the interval $[0,T^{(1)}_a]$. As a result, for $t\in[0,T^{(1)}_a]$, $(C_{\mathbf{F}}^{(1)})^\prime(t)=(\tau_{\mathbf{F}}^{(1)})^\prime(t)-\gamma^{(1)}=\tau_2^\prime(t)-\gamma^{(1)}=\frac{(F^{(2)})^\prime(t)}{\mu_2}-\gamma^{(1)}=\gamma^{(2)}-\gamma^{(1)}>0$. So the class 1 user arriving at $T^{(1)}_a$ will be better off arriving at time $0$ and hence $\{F^{(1)},F^{(2)}\}$ will not be an EAP. 
\end{proof} 

\subsubsection{Proof of Theorem \ref{mainthm_inst2_reg1}}

We consider candidates which are absolutely continuous with $\mathcal{S}(F^{(1)})=[T^{(1)}_a,T^{(1)}_f]$, $\mathcal{S}(F^{(2)})=[T^{(2)}_a,T^{(2)}_f]$  and have arrival rates following the property in Lemma \ref{lem_arrival_rates_inst2}. As a result, if an EAP exists, it will be contained in the mentioned set of candidates. Now we consider the two cases: \textbf{1)}~$\gamma^{(1)}>\gamma^{(2)}$ and \textbf{2)}~$\gamma^{(1)}<\gamma^{(2)}$, separately. 

\noindent\textbf{Case 1 $\gamma^{(1)}>\gamma^{(2)}$}:~~Queue 2 cannot be empty at time zero, so $T^{(2)}_a<0$. By Lemma \ref{lem_inst2_uneqpref_sign_reg1a}, the support boundaries must satisfy one of the two properties: \textbf{1)}~Type I: $T^{(1)}_a\leq 0$, and \textbf{2)}~Type II: $0<T^{(1)}_a<T^{(2)}_f$. The following lemma implies existence of unique EAP for case 1 of Theorem \ref{mainthm_inst2_reg1}. 

\begin{lemma}\label{lem:inst2reg1case1}
    If $\mu_1<\mu_2\gamma^{(2)}$ and $\gamma^{(1)}>\gamma^{(2)}$, the following statements are true: 
    \begin{enumerate}
        \item[\textbf{1.}] An EAP of Type I exists if and only if $\Lambda^{(1)}\geq\frac{1-\gamma^{(2)}}{1-\gamma^{(1)}}\frac{\mu_1}{\mu_2-\mu_1}\Lambda^{(2)}$, and if it exists, it will be unique with arrival rates and support boundaries same as the joint arrival profile mentioned under case 1a of Theorem \ref{mainthm_inst2_reg1}. 
        \item[\textbf{2.}]  An EAP of Type II exists if and only if $\Lambda^{(1)}<\frac{1-\gamma^{(2)}}{1-\gamma^{(1)}}\frac{\mu_1}{\mu_2-\mu_1}\Lambda^{(2)}$, and if it exists, it will be unique with arrival rates and support boundaries same as the joint arrival profile mentioned under case 1b of Theorem \ref{mainthm_inst2_reg1}.
    \end{enumerate}
\end{lemma}

\begin{proof}
\noindent\textbf{Getting the arrival rates:}\hspace{0.05in}For both the types queue 2 must have positive queue length in $[0,T^{(2)}_f)$ and by Lemma \ref{lem_inst2_uneqpref_queue2idle}, queue 2 must become empty at $T^{(2)}_f$. Combining this observation with Lemma \ref{lem_arrival_rates_inst2}, arrival rates of candidates under both the types must satisfy, 
\begin{align}\label{eq_rate_inst2_uneqpref_reg1a}
    (F^{(1)})^\prime(t)&=\begin{cases}
        \frac{\mu_1\gamma^{(1)}}{\gamma^{(2)}}~~&\text{for}~t\in[T^{(1)}_a,\tau_1^{-1}(T^{(2)}_f)], \\
        \mu_1\gamma^{(1)}~~&\text{for}~t\in[\tau_1^{-1}(T^{(2)}_f), T^{(1)}_f], 
    \end{cases}~\text{and}~(F^{(2)})^\prime(t)=\begin{cases}
        \mu_2\gamma^{(2)}~~&\text{for}~t\in[T^{(2)}_a,T^{(1),+}_a],\\
        \mu_2\gamma^{(2)}-\mu_1~~&\text{for}~t\in[T^{(1),+}_a,T^{(2)}_f],
    \end{cases}
\end{align}
where $T^{(1),+}_a=\max\{T^{(1)}_a,0\}$. Applying Lemma \ref{lem_appndx_inst2_uneqpref_queu1_engaged_mixedarrival}, it is easy to observe that queue 1 must be engaged in $(T^{(1),+}_a,T^{(1)}_f)$ and starts serving from time $T^{(1),+}_a$. Therefore using (\ref{eq:derv_of_tau}), $\tau_1(t)=\frac{F^{(1)}(t)}{\mu_1}+T^{(1),+}_a$ for $t\in[T^{(1)}_a,T^{(1)}_f]$. Therefore $\tau_1(\cdot)$ is increasing in $[T^{(1)}_a,T^{(1)}_f]$ and $\tau_1^{-1}(T^{(2)}_f)$ is well-defined. Using (\ref{eq_rate_inst2_uneqpref_reg1a}) on the expression obtained for $\tau_1(t)$, we get, 
\begin{align*}
    T^{(2)}_f=\tau_1(\tau_1^{-1}(T^{(2)}_f))&=T^{(1),+}_a + \frac{F^{(1)}(\tau_1^{-1}(T^{(2)}_f))}{\mu_1}\\
    &=T^{(1),+}_a + \frac{\gamma^{(1)}}{\gamma^{(2)}}\left(\tau_1^{-1}(T^{(2)}_f)-T^{(1)}_a\right).
\end{align*}

After some manipulation, we get
\begin{align}\label{eq:inst2_uneqpref_reg1a_inv_of_tau2_at_T2f}
    \tau_1^{-1}(T^{(2)}_f)=T^{(1)}_a+\frac{\gamma^{(2)}}{\gamma^{(1)}}(T^{(2)}_f-T^{(1),+}_a).    
\end{align}

\noindent\textbf{Identifying the support boundaries:}\hspace{0.05in}Exploiting the structural properties proved before, we obtain the following system of equations to be satisfied by the support boundaries $T^{(1)}_a,T^{(1)}_f,T^{(2)}_a,T^{(2)}_f$: 
\begin{enumerate}[leftmargin=*]
    \item By Lemma \ref{lem_inst2_uneqpref_queue1idle} and \ref{lem_appndx_inst2_uneqpref_queu1_engaged_mixedarrival}, queue 1 starts at $T^{(1),+}_a$, empties at $T^{(1)}_f$ and has positive waiting time in $(T^{(1),+}_a,T^{(1)}_f)$. This gives us, $T^{(1)}_f=T^{(1),+}_a +\frac{\Lambda^{(1)}}{\mu_1}$. 
    
    \item By Lemma \ref{lem_inst2_uneqpref_queue2idle}, queue 2 starts at time zero, empties at $T^{(2)}_f$ and has positive waiting time in $(0,T^{(2)}_f)$. This gives us, $\mu_2 T^{(2)}_f=\Lambda^{(2)}+F^{(1)}(\tau_1^{-1}(T^{(2)}_f))$.
    
    \item All class 2 users arrive in $[T^{(2)}_a,T^{(2)}_f]$ at rates given by (\ref{eq_rate_inst2_uneqpref_reg1a}). This gives us,\\ $\Lambda^{(2)}=(\mu_2\gamma^{(2)}-\mu_1)(T^{(2)}_f-T^{(1),+}_a)+\mu_2\gamma^{(2)}(T^{(1),+}_a -T^{(2)}_a)$.
    
    \item All class 1 users arrive in $[T^{(1)}_a,T^{(1)}_f]$ at rates given by (\ref{eq_rate_inst2_uneqpref_reg1a}). This gives us,\\ $\Lambda^{(1)}=\frac{\mu_1\gamma^{(1)}}{\gamma^{(2)}}\left(\tau_1^{-1}(T^{(2)}_f)-T^{(1)}_a\right)+\mu_1\gamma^{(1)}\left(T^{(1)}_f-\tau_1^{-1}(T^{(2)}_f)\right)$.
\end{enumerate}

In the above system of equations  $F^{(1)},F^{(2)}$ are obtained using (\ref{eq_rate_inst2_uneqpref_reg1a}) and $\tau_1^{-1}(T^{(2)}_f)$ obtained using (\ref{eq:inst2_uneqpref_reg1a_inv_of_tau2_at_T2f}). Upon plugging in $T^{(1),+}_a = 0$ for Type I and $T^{(1)}_a$ for Type II, we obtain a system of linear equations for both the types. Their solutions for Type I and II are, respectively, in (\ref{eq_inst2_uneqpref_reg1_bdary_case1a}) and (\ref{eq_inst2_uneqpref_reg1_bdary_case1b}). Therefore, every EAP under Type I and II must, respectively, have (\ref{eq_inst2_uneqpref_reg1_bdary_case1a}) and (\ref{eq_inst2_uneqpref_reg1_bdary_case1b}) as support boundaries.

\noindent\textbf{Identifying the necessary conditions:}\hspace{0.05in}The support boundaries in (\ref{eq_inst2_uneqpref_reg1_bdary_case1a}) must satisfy $T^{(1)}_a\leq 0$ and upon imposing that, we get the condition $\Lambda^{(1)}\geq\frac{1-\gamma^{(2)}}{1-\gamma^{(1)}}\frac{\mu_1}{\mu_2-\mu_1}\Lambda^{(2)}$, which is necessary for existence of a Type I EAP. It is easy to verify that, if $\Lambda^{(1)}\geq\frac{1-\gamma^{(2)}}{1-\gamma^{(1)}}\frac{\mu_1}{\mu_2-\mu_1}\Lambda^{(2)}$, (\ref{eq_inst2_uneqpref_reg1_bdary_case1a}) satisfies $T^{(1)}_a\leq 0, T^{(2)}_a<0<T^{(2)}_f, T^{(1)}_a<\tau_1^{-1}(T^{(2)}_f)\leq T^{(1)}_f$. By (\ref{eq_rate_inst2_uneqpref_reg1a}), the obtained Type I candidate has a closed form same as the joint arrival profile under case 1a of Theorem \ref{mainthm_inst2_reg1} and is the only Type I candidate which can qualify as an EAP. \\
Similarly imposing $T^{(1)}_a>0$ on the support boundaries in (\ref{eq_inst2_uneqpref_reg1_bdary_case1b}), we get  $\Lambda^{(1)}<\frac{1-\gamma^{(2)}}{1-\gamma^{(1)}}\frac{\mu_1}{\mu_2-\mu_1}\Lambda^{(2)}$ and this is a necessary condition for existence of a Type II candidate. Again, it is easy to verify, if  $\Lambda^{(1)}<\frac{1-\gamma^{(2)}}{1-\gamma^{(1)}}\frac{\mu_1}{\mu_2-\mu_1}\Lambda^{(2)}$,  (\ref{eq_inst2_uneqpref_reg1_bdary_case1b}) satisfies $T^{(2)}_f>T^{(1)}_a>0>T^{(2)}_a, T^{(1)}_f\geq\tau_1^{-1}(T^{(2)}_f)>T^{(1)}_a$. By (\ref{eq_rate_inst2_uneqpref_reg1a}),, with the necessary condition satisfies, the obtained Type II candidate has closed form same as the joint arrival profile mentioned under case 1b of Theorem \ref{mainthm_inst2_reg1} and is the only Type II candidate which can qualify as an EAP.

\noindent\textbf{Proving sufficiency of the obtained necessary condition:}\hspace{0.05in}Now the following sequence of arguments imply that the obtained candidate for both the types satisfies: $(C_{\mathbf{F}}^{(i)})^\prime(t)\leq  0$ in $(-\infty,T_{i,a})$, $(C_{\mathbf{F}}^{(i)})^\prime(t)=0$ in $[T_{i,a},T_{i,f}]$ and $(C_{\mathbf{F}}^{(i)})^\prime(t)\leq  0$ in $(T_{i,f},\infty)$ for both the classes $i=1,2$. As a result, the candidate obtained under both the types is an EAP if the necessary condition is satisfied. 
\begin{itemize}[leftmargin=*]
    \item \textbf{For class 1:} Every class 1 user arriving in $(-\infty,T^{(1)}_a)$ will depart at $\tau_2(0)$. As a result, $(C_{\mathbf{F}}^{(1)})^\prime(t)=-\gamma^{(1)}<0$ for $t\in(-\infty,T^{(1)}_a)$. For both the types, the obtained candidate satisfies  $F^{(1)}(t)>\mu_1\cdot(t-T^{(1),+}_a)\cdot\mathbb{I}(t\geq T^{(1),+}_a)$ for every $t\in[T^{(1)}_a,T^{(1)}_f)$, making queue 1 engaged in $[T^{(1)}_a,T^{(1)}_f]$. Also, since $\Lambda^{(1)}=\mu_1\cdot(T^{(1)}_f-T^{(1),+}_a)$, queue 1 becomes empty at $T^{(1)}_f$. As a result, class 1 users arrive at queue 2 at rate $\mu_1$ in $[0,T^{(2)}_f]$ and by \ref{eq_rate_inst2_uneqpref_reg1a}, class 2 users arrive at rate $\mu_2\gamma^{(2)}-\mu_1$ in $[0,T^{(2)}_f]$ making $A_2^\prime(t)=\mu_2\gamma^{(2)}<\mu_2$ in $[0,T^{(2)}_f]$. As a result, $A_2(t)-\mu_2 t$ is decreasing in $[0,T^{(2)}_f]$. Now, the previous section and $A_2(T^{(2)}_f)=\Lambda^{(2)}+F^{(1)}(\tau_1^{-1}(T^{(2)}_f))=\mu_2\cdot T^{(2)}_f$ imply $A_2(t)>\mu_2 t$ for $t\in[0,T^{(2)}_f)$. Hence queue 2 stays engaged in $[0,T^{(2)}_f]$ and empties at $T^{(2)}_f$. With this, for $t\in[T^{(1)}_a,\tau_1^{-1}(T^{(2)}_f)]$, using (\ref{eq:derv_of_tau}) and (\ref{eq_rate_inst2_uneqpref_reg1a}), $(\tau_{\mathbf{F}}^{(1)})^\prime(t)=(\tau_2\circ\tau_1)^\prime(t)=\tau_2^\prime(\tau_1(t))\cdot\tau_1^\prime(t)=\frac{A_2^\prime(\tau_1(t))}{\mu_2}\cdot\frac{(F^{(2)})^\prime(t)}{\mu_1}=\gamma^{(1)}$ and as a result, $(C_{\mathbf{F}}^{(1)})^\prime(t)=(\tau_{\mathbf{F}}^{(1)})^\prime(t)-\gamma^{(1)}=0$ in $[T^{(1)}_a,\tau_1^{-1}(T^{(2)}_f)]$. Note that, class 1 users arrive from queue 1 to 2 at a maximum rate of $\mu_1<\mu_2$ in $[T^{(2)}_f,T^{(1)}_f]$. As a result, queue 2 stays empty in $[T^{(2)}_f,T^{(1)}_f]$ and for every $t\in[\tau_1^{-1}(T^{(2)}_f),T^{(1)}_f]$, using (\ref{eq:derv_of_tau}) and (\ref{eq_rate_inst2_uneqpref_reg1a}), $(\tau_{\mathbf{F}}^{(1)})^\prime(t)=\tau_1^\prime(t)=\frac{(F^{(1)})^\prime(t)}{\mu_1}=\gamma^{(1)}$. Hence, $(C_{\mathbf{F}}^{(1)})^\prime(t)=0$ in $[\tau_1^{-1}(T^{(2)}_f),T^{(1)}_f]$. Since queue 1 becomes empty at $T^{(1)}_f$, $(C_{\mathbf{F}}^{(1)})^\prime(t)=1-\gamma^{(1)}>0$ for $t\in(T^{(1)}_f,\infty)$. 

    \item \textbf{For class 2:} Every class 2 user arriving before $T^{(2)}_a$ will get served at time zero and hence, $(C_{\mathbf{F}}^{(2)})^\prime(t)=-\gamma^{(2)}<0$ for $t\in(-\infty,T^{(2)}_a)$. By the argument for class 1, queue 2 stays engaged in $(T^{(2)}_a,T^{(2)}_f)$ and empties at $T^{(2)}_f$. As a result, using (\ref{eq:derv_of_tau}) and (\ref{eq_rate_inst2_uneqpref_reg1a}), for every $t\in[T^{(2)}_a,T^{(2)}_f]$, $(\tau_{\mathbf{F}}^{(2)})^\prime(t)=\tau_2^\prime(t)=\frac{A_2^\prime(t)}{\mu_2}=\frac{\mu_2\gamma^{(2)}}{\mu_2}=\gamma^{(2)}$ and as a result, $(C_{\mathbf{F}}^{(2)})^\prime(t)=(\tau_{\mathbf{F}}^{(2)})^\prime(t)-\gamma^{(2)}=0$. Since queue 2 stays empty in $(T^{(2)}_f,\infty)$, $(C_{\mathbf{F}}^{(2)})^\prime(t)=1-\gamma^{(2)}>0$ in $(T^{(2)}_f,\infty)$.  \vspace{-0.3in}
\end{itemize}
\end{proof}

\noindent\textbf{Case 2 $\gamma^{(1)}<\gamma^{(2)}$}:~~By Lemma \ref{lem_inst2_uneqpref_sign_reg1b}, the support boundaries in this case must satisfy $T^{(1)}_a\leq 0$ and one of the two properties: \textbf{1)} Type I: $T^{(1)}_f\geq T^{(2)}_f$, and \textbf{2)} Type II: $T^{(2)}_f>T^{(1)}_f$. Also we must have $T^{(2)}_a<0$, since queue 2 cannot be empty at time zero. Now the following lemma implies existence of a unique EAP under case 2 of Theorem \ref{mainthm_inst2_reg2}.

\begin{lemma}\label{lem:inst2reg1case2}
    If $\mu_1<\mu_2\gamma^{(2)}$ and $\gamma^{(1)}<\gamma^{(2)}$, the following statements are true: 
        \begin{enumerate}
        \item[\textbf{1.}] An EAP of Type I exists if and only if $\Lambda^{(1)}\geq\frac{\mu_1}{\mu_2-\mu_1}\Lambda^{(2)}$, and if it exists, it will be unique with arrival rates and support boundaries same as the joint arrival profile mentioned under case 2a of Theorem \ref{mainthm_inst2_reg1}. 
        \item[\textbf{2.}]  An EAP of Type II exists if and only if $\Lambda^{(1)}<\frac{\mu_1}{\mu_2-\mu_1}\Lambda^{(2)}$, and if it exists, it will be unique with arrival rates and support boundaries same as the joint arrival profile mentioned under case 2b of Theorem \ref{mainthm_inst2_reg1}.
    \end{enumerate}
\end{lemma}

\begin{proof}
\noindent\textbf{Getting the arrival rates:}\hspace{0.05in}By Lemma \ref{lem_arrival_rates_inst2}, arrival rates of the two classes in candidates of the two types must be: 
\begin{enumerate}
    \item[Type I:] Same as mentioned in (\ref{eq_rate_inst2_uneqpref_reg1a}) by putting $T^{(1),+}_a = 0$. 
    \item[Type II:] \begin{align}\label{eq_rate_inst2_uneqpref_reg2b}
        (F^{(1)})^\prime(t)&=\frac{\mu_1\gamma^{(1)}}{\gamma^{(2)}}~~\text{for}~t\in[T^{(1)}_a,T^{(1)}_f]\nonumber \\ 
        \text{and}~(F^{(2)})^\prime(t)&=\begin{cases} \mu_2\gamma^{(2)}~~&\text{for}~t\in[T^{(2)}_a,0]\cup[T^{(1)}_f,T^{(2)}_f],\\ \mu_2\gamma^{(2)}-\mu_1~~&\text{for}~t\in[0,T^{(1)}_f]. \end{cases}   
    \end{align} 
\end{enumerate}

\noindent\textbf{Identifying the support boundaries:}\hspace{0.05in}For Type I, since $T^{(1)}_a\leq 0$ and $T^{(1)}_f\geq T^{(2)}_f$, the support boundaries satisfy the same linear system obtained for Type I in the proof of Lemma \ref{lem:inst2reg1case1}. Hence every Type I EAP has (\ref{eq_inst2_uneqpref_reg1_bdary_case1a}) as support boundaries. Similarly we obtain the following system of equations to be satisfied by the support boundaries for any Type II EAP:

\begin{itemize}[leftmargin=*]
    \item By Lemma \ref{lem_inst2_uneqpref_queue1idle}, queue 1 empties at $T^{(1)}_f$ starting at time zero and has a positive waiting time between $(0,T^{(1)}_f)$. This gives us $T^{(1)}_f=\frac{\Lambda^{(1)}}{\mu_1}$.
    \item By Lemma \ref{lem_inst2_uneqpref_queue2idle}, queue 2 empties at $T^{(2)}_f$ starting at time zero, and has a positive waiting time in $(0,T^{(2)}_f)$ giving us, $T^{(2)}_f=\frac{\Lambda^{(1)}+\Lambda^{(2)}}{\mu_2}$.
    \item By definition of EAP $C_{\mathbf{F}}^{(2)}(T^{(2)}_a)=C_{\mathbf{F}}^{(2)}(T^{(2)}_f)$, giving us, $T^{(2)}_a=-\left(\frac{1}{\gamma^{(2)}}-1\right)T^{(2)}_f$.
    \item By Lemma \ref{lem_arrival_rates_inst2}, all class 1 users arrive between $[T^{(1)}_a,T^{(1)}_f]$ at rate $\frac{\mu_1\gamma^{(1)}}{\gamma^{(2)}}$. This gives us $\Lambda^{(1)}=\frac{\mu_1\gamma^{(1)}}{\gamma^{(2)}}(T^{(1)}_f-T^{(1)}_a)$. 
\end{itemize}

Solution to the above system is in (\ref{eq_inst2_uneqpref_reg1_bdary_case2b}) and hence every EAP under Type II has (\ref{eq_inst2_uneqpref_reg1_bdary_case2b}) as its support boundaries.

\noindent\textbf{Obtaining the necessary conditions:}\hspace{0.05in}The support boundaries in (\ref{eq_inst2_uneqpref_reg1_bdary_case1a}) must satisfy $T^{(1)}_f\geq T^{(2)}_f$, which gives us $\Lambda^{(1)}\geq\frac{\mu_1}{\mu_2-\mu_1}\Lambda^{(2)}$ and this is a necessary condition for existence of a Type I EAP. It is easy to verify that, if $\Lambda^{(1)}\geq\frac{\mu_1}{\mu_2-\mu_1}\Lambda^{(2)}$, (\ref{eq_inst2_uneqpref_reg1_bdary_case1a}) satisfies $T^{(1)}_a<0, T^{(1)}_f\geq T^{(2)}_f$ and the obtained Type I candidate has a closed form same as the joint arrival profile mentioned under case 1a of Theorem \ref{mainthm_inst2_reg1}.

On the other hand, imposing $T^{(2)}_f>T^{(1)}_f$ on (\ref{eq_inst2_uneqpref_reg1_bdary_case2b}), we obtain  $\Lambda^{(1)}<\frac{\mu_1}{\mu_2-\mu_1}\Lambda^{(2)}$ and this is a necessary condition for existence of an EAP under Type II. It is easy to verify, if $\Lambda^{(1)}<\frac{\mu_1}{\mu_2-\mu_1}\Lambda^{(2)}$, (\ref{eq_inst2_uneqpref_reg1_bdary_case2b}) satisfies $T^{(1)}_a\leq 0, T^{(2)}_f>T^{(1)}_f$ and the obtained Type II candidate has a closed form same as the joint arrival profile mentioned under case 2b of Theorem \ref{mainthm_inst2_reg1}. 

\noindent\textbf{Proving sufficiency of the obtained necessary condition:}\hspace{0.05in}Now the following sequence of arguments imply that the obtained candidate under the two types satisfies: $(C_{\mathbf{F}}^{(i)})^\prime(t)\leq 0$ in $(-\infty,T_{i,a})$, $(C_{\mathbf{F}}^{(i)})^\prime(t)=0$ in $[T^{(1)}_a,T^{(1)}_f]$, and $(C_{\mathbf{F}}^{(i)})^\prime(t)\geq 0$ in $(T^{(1)}_f,\infty)$ for both the classes $i=1,2$. As a result, the unique candidate obtained under the two types are EAPs. Hence, the statement of Lemma \ref{lem:inst2reg1case2} stands proved. 

\begin{itemize}[leftmargin=*]
    \item \textbf{For Type I candidate:}~It follows by the same argument as was used for Type I of case 1 ($\gamma^{(1)}>\gamma^{(2)}$). 
    \item \textbf{For Type II candidate:}~Note that, $F^{(1)}(t)>\mu_1\cdot\max\{t,0\}$ for $t\in[T^{(1)}_a,T^{(1)}_f)$ and $\mu_1\cdot T^{(1)}_f=\Lambda^{(1)}$. As a result, queue 1 stays engaged in $[T^{(1)}_a,T^{(1)}_f]$ and empties at $T^{(1)}_f$. Therefore, class 1 users arrive at rate $\mu_1$ from queue 1 to 2 in $[0,T^{(1)}_f]$. Hence, using (\ref{eq_rate_inst2_uneqpref_reg2b}), $A_2^\prime(t)=\mu_2\gamma^{(2)}<\mu_2$ between $[0,T^{(2)}_f]$, making $A_2(t)-\mu_2 t$ a decreasing function in $[0,T^{(2)}_f]$. Since $\mu_2 T^{(2)}_f=\Lambda^{(1)}+\Lambda^{(2)}=A_2(T^{(2)}_f)$, the previous statement implies $A_2(t)>\mu_2 t$ in $[0,T^{(2)}_f]$ and as a result, queue 2 stays engaged in $[T^{(2)}_a,T^{(2)}_f]$ and empties at $T^{(2)}_f$. Now every class 1 user arriving before $T^{(1)}_a$ gets served at $\tau_2(0)$, making $(C_{\mathbf{F}}^{(1)})^\prime(t)=-\gamma^{(1)}<0$ in $(-\infty,T^{(1)}_a)$. Since queue 1 stays engaged in $[T^{(1)}_a,T^{(1)}_f]$ and queue 2 stays engaged in $[0,T^{(1)}_f]$, using (\ref{eq:derv_of_tau}) and (\ref{eq_rate_inst2_uneqpref_reg2b}), for every $t\in[T^{(1)}_a,T^{(1)}_f]$, $(\tau_{\mathbf{F}}^{(1)})^\prime(t)=(\tau_2\circ\tau_1)^\prime(t)=\tau_2^\prime(\tau_1(t))\tau_1^\prime(t)=\frac{A_2^\prime(\tau_1(t))}{\mu_2}\cdot\frac{(F^{(1)})^\prime(t)}{\mu_1}=\frac{\mu_2\gamma^{(2)}}{\mu_2}\cdot\frac{\mu_1\gamma^{(1)}/\gamma^{(2)}}{\mu_1}=\gamma^{(1)}$. As a result, $(C_{\mathbf{F}}^{(1)})^\prime(t)=0$ in $[T^{(1)}_a,T^{(1)}_f]$. For $t\in(T^{(1)}_f,\infty)$, since queue 1 stays empty , $(C_{\mathbf{F}}^{(1)})^\prime(t)=1-\gamma^{(1)}>0$. Now every class 2 user arriving before $T^{(2)}_a$ gets served at time zero, making $(C_{\mathbf{F}}^{(2)})^\prime(t)=-\gamma^{(2)}$ in $(-\infty,T^{(2)}_a)$. In $[T^{(2)}_a,T^{(2)}_f]$, queue 2 remains engaged with arrival rate $A_2^\prime(t)=\mu_2\gamma^{(2)}$ and therefore, using (\ref{eq:derv_of_tau}), $(C_{\mathbf{F}}^{(2)})^\prime(t)=\frac{A_2^\prime(t)}{\mu_2}-\gamma^{(2)}=0$ in $[T^{(2)}_a,T^{(2)}_f]$. Since queue 2 stays empty after $T^{(2)}_f$, $(C_{\mathbf{F}}^{(2)})^\prime(t)=1-\gamma^{(2)}>0$ in $(T^{(2)}_f,\infty)$.   
\end{itemize}    
\end{proof}

\subsubsection{Supporting lemmas for proving Theorem \ref{mainthm_inst2_reg2}}

Before proving Theorem \ref{mainthm_inst2_reg2}, which specifies the EAP for the regime $\mu_1\geq\mu_2\gamma^{(2)}$, we need the following two lemmas to characterize the EAP in different parametric regions. In the statement of these two lemmas, the quantities $T^{(1)}_a,T^{(1)}_f,T^{(2)}_a,T^{(2)}_f$ are same as they were introduced in Section \ref{sec_inst2_uneqpref} before specifying the EAPs. By Lemma \ref{lem_supp_are_intervals_inst1}, we must have $\mathcal{S}(F^{(1)})=[T^{(1)}_a,T^{(1)}_f]$. 

\begin{lemma}\label{lem_inst2_uneqpref_reg2_case12_sign}
    If $\mu_1\geq\mu_2\gamma^{(2)}$ and $\gamma^{(1)}>\gamma^{(2)}$, $T^{(2)}_f\leq T^{(1)}_f$.
\end{lemma}
\begin{proof} 
By Lemma \ref{lem_inst2_uneqpref_queue1idle}, queue 1 will be empty at time $T^{(1)}_f$. Assuming contradiction, by (\ref{eq:isocost_inst2_grp2}), $(C_{\mathbf{F}}^{(1)})^\prime(t)=\tau_2^\prime(t)-\gamma^{(1)}=\gamma^{(2)}-\gamma^{(1)}<0$ for $t\in[T^{(1)}_f,T^{(2)}_f]$. Therefore, the class 1 user arriving at $T^{(1)}_f$ will be better off arriving at $T^{(2)}_f$. 
\end{proof}

\begin{lemma}\label{lem_inst2_uneqpref_reg2_case12_sign2}
    If $\mu_1\geq\mu_2\gamma^{(2)}$ and $\gamma^{(1)}>\gamma^{(2)}$, then for every EAP, the only possible arrival profile of class 2 users is $(F^{(2)})^\prime(t)=\mu_2\gamma^{(2)}\cdot\mathbb{I}\left(t\in\left[T^{(1),+}_a -\frac{\Lambda^{(2)}}{\mu_2\gamma^{(2)}},T^{(1),+}_a\right]\right)$.
\end{lemma}
\begin{proof} 
By Lemma \ref{lem_inst2_uneqpref_reg2_case12_sign} and \ref{lem_threshold_behav_inst2}, no class 2 user can arrive after $T^{(1),+}_a$. Moreover, $\mathcal{S}(F^{(2)})$ cannot have a gap before $T^{(1),+}_a$. Otherwise, if there is a gap $[t,t+\delta]$ for some $t\in \mathcal{S}(F^{(2)})$ and $\delta>0$ sufficiently small, then queue 2 must have positive waiting time in $[t,t+\delta]$ with no new class 2 user arriving. As a result, by (\ref{eq:derv_of_tau}) $\tau_{\mathbf{F}}^{(2)}(\cdot)$ remains constant in $[t,t+\delta]$. Hence the class 2 user arriving at $t$ can improve her cost arriving at $t+\delta$, giving us a contradiction. Therefore, class 2 users will  arrive over a contiguous interval ending at $T^{(1),+}_a$ at rate $\mu_2\gamma^{(2)}$ given by Lemma \ref{lem_arrival_rates_inst2}. The only arrival profile satisfying this property is the one with arrival rate $(F^{(2)})^\prime(t)=\mu_2\gamma^{(2)}\cdot\mathbb{I}\left(t\in\left[T^{(1),+}_a -\frac{\Lambda^{(2)}}{\mu_2\gamma^{(2)}},T^{(1),+}_a\right]\right)$. 
\end{proof}

\subsubsection{Proof of Theorem \ref{mainthm_inst2_reg2}}

We will prove the existence and uniqueness of the EAP separately for the three cases by exploiting different structural properties:~\textbf{1)}~$\mu_2\gamma^{(1)}>\mu_1\geq\mu_2\gamma^{(2)}$;~\textbf{2)}~$\mu_1\geq\mu_2\gamma^{(1)}>\mu_2\gamma^{(2)}$;~and~\textbf{3)}~$\mu_1\geq\mu_2\gamma^{(2)}>\mu_2\gamma^{(1)}$. The structure of the proof is similar to the one used in proving Theorem \ref{mainthm_inst2_reg1}. We define the quantity $T\overset{def.}{=}\inf\{t>0~\vert~Q_2(t)=0\}$. 

\noindent\textbf{Case 1}~$\mu_2\gamma^{(1)}>\mu_1\geq\mu_2\gamma^{(2)}$:1\hspace{0.05in}By Lemma \ref{lem_supp_are_intervals_inst2} and \ref{lem_inst2_uneqpref_reg2_case12_sign2}, we only consider absolutely continuous candidates $\mathbf{F}=\{F^{(1)},F^{(2)}\}$ such that, $\mathcal{S}(F^{(1)})=[T^{(1)}_a,T^{(1)}_f]$ for some $T^{(1)}_f>T^{(1),+}_a=\max\{T^{(1)}_a,0\}$, $(F^{(2)})^\prime(t)=\mu_2\gamma^{(2)}$ in $\left[T^{(1),+}_a-\frac{\Lambda^{(2)}}{\mu_2\gamma^{(2)}}, T^{(1),+}_a\right]$, and the arrival rate $(F^{(1)})^\prime(\cdot)$ satisfies the properties in Lemma \ref{lem_arrival_rates_inst2}. The set of EAPs (if non-empty) will be contained in this set. 

For every candidate in the above set and therefore for every EAP, there are two possibilities:~\textbf{1)}~Type I:~$T^{(1)}_a\leq 0$, and~\textbf{2)}~Type II:~$T^{(1)}_a>0$. The following lemma gives the necessary and sufficient condition for existence of EAP under both the types. The existence of unique EAP under case 1 follows from this lemma. 

\begin{lemma}\label{lem:inst2reg2case1}
    If $\mu_2\gamma^{(1)}>\mu_1\geq\mu_2\gamma^{(2)}$, then the following properties are true about the EAP, 
    \begin{enumerate}
        \item There exists an EAP under Type I if and only if $\Lambda^{(1)}\geq\frac{\mu_1}{(1-\gamma^{(1)})\mu_2}\Lambda^{(2)}$, and if it exists, then there is a unique EAP under Type I which has closed form same as the joint arrival profile mentioned under case 1a in the theorem statement.  
        \item There exists an EAP under Type II if and only if $\Lambda^{(1)}<\frac{\mu_1}{(1-\gamma^{(1)})\mu_2}\Lambda^{(2)}$, and if it exists, then there is a unique EAP under Type II which has closed form same as the joint arrival profile mentioned under case 1b in the theorem statement.
    \end{enumerate}
\end{lemma} 
\begin{proof}
\textbf{Getting the arrival rates:}\hspace{0.05in}In both the types, class 1 users start arriving at queue 2 from time $T^{(1),+}_a=\max\{T^{(1)}_a,0\}$. Note that, in both the types queue 2 must have a positive waiting time at $\tau_1(T^{(1)}_a)=T^{(1),+}_a$:~\textbf{1)}~In Type I, the entire class 2 population waits at queue 2 at time zero,~\textbf{2)}~In Type II, queue 2 must have a positive waiting time at $T^{(1)}_a$. Otherwise, the whole network must be empty at $T^{(1)}_a$, making the class 1 users arriving after $T^{(1)}_a$ strictly better off arriving at $T^{(1)}_a$. 

Therefore, for both the types queue 2 has a positive waiting time at $\tau_1(T^{(1)}_a)=T^{(1),+}_a$. As a result, by Lemma \ref{lem_arrival_rates_inst2}, class 1 users start arriving at rate $\mu_2\gamma^{(1)}>\mu_1$ from $T^{(1)}_a$, causing a queue to form in queue 1. Also by Lemma \ref{lem_inst2_uneqpref_queue1idle}, both the queues must be empty at $T^{(1)}_f$. Therefore we can conclude that queue 2 will empty strictly before queue 1 empties and hence $T<T^{(1)}_f$. Otherwise, if $T\geq T^{(1)}_f$, queue 2 will have a positive waiting time in $(T^{(1),+}_a,T^{(1)}_f)$ and queue 1 will be empty at $T^{(1)}_f$ (by Lemma \ref{lem_inst2_uneqpref_queue1idle}), making $\tau_1(T^{(1)}_f)=T^{(1)}_f$. As a result, for every $t\in[T^{(1)}_a,T^{(1)}_f]$, $\tau_1(t)\leq \tau_1(T^{(1)}_f)=T^{(1)}_f\leq T$. Hence class 1 users will arrive at queue 1 at rate $\mu_2\gamma^{(1)}>\mu_1$ in $[T^{(1)}_a,T^{(1)}_f]$ causing queue 1 to have a positive waiting time at $T^{(1)}_f$, contradicting Lemma \ref{lem_inst2_uneqpref_queue1idle}. 

Therefore, for both the types, queue 1 stays engaged in $[0,T^{(1)}_f]$ and queue 2 stays engaged in $[0,T]$ with $T<T^{(1)}_f$. Applying Lemma \ref{lem_arrival_rates_inst2}, arrival rates of the two classes will be: 
\begin{align}\label{eq_rates_inst2_uneqpref_reg2_case1}
    (F^{(1)})^\prime(t)&=\begin{cases}
        \mu_2\gamma^{(1)}~~&\text{if}~t\in[T^{(1)}_a,\tau_1^{-1}(T)], \\
        \mu_1\gamma^{(1)}~~&\text{if}~t\in[\tau_1^{-1}(T),T^{(1)}_f],
    \end{cases}~\text{and}~(F^{(2)})^\prime(t)=\mu_2\gamma^{(2)}~\text{if}~t\in\left[T^{(1),+}_a -\frac{\Lambda^{(2)}}{\mu_2\gamma^{(2)}},T^{(1),+}_a\right],
\end{align}

where $\tau_1^{-1}(T)=\inf\{t~\vert~\tau_1(t)\geq T\}$. Since queue 1 stays engaged in $(T^{(1)}_a,T^{(1)}_f)$, using (\ref{eq:derv_of_tau}), $\tau_1(t)$ is strictly increasing in $(T^{(1)}_a,T^{(1)}_f)$, making $\tau_1^{-1}(T)$ well defined. Therefore, $\tau_{1}(\tau_1^{-1}(T))=T$. Now, in $(T^{(1)}_a,T^{(1)}_f)$, using (\ref{eq:derv_of_tau}), $\tau_1(t)=\frac{F^{(1)}(t)}{\mu_1}+T^{(1),+}_a$. Using this, and (\ref{eq_rates_inst2_uneqpref_reg2_case1}), $T=\tau_1 (\tau_1^{-1}(T))=\frac{F^{(1)}(\tau_1^{-1}(T))}{\mu_1}+T^{(1),+}_a=\frac{\mu_2\gamma^{(1)}}{\mu_1}\cdot(\tau_1^{-1}(T)-T^{(1)}_a)+T^{(1),+}_a$. After some manipulation, this gives us, 
\begin{align}\label{eq:inst2_uneqpref_reg2_case1_inv_of_tau2_at_T}
    \tau_1^{-1}(T)&=T^{(1)}_a+\frac{\mu_1}{\mu_2\gamma^{(1)}}\cdot(T-T^{(1),+}_a)
\end{align}

\textbf{Identifying the support boundaries:}\hspace{0.05in}We can obtain the following system of equations to be satisfied by $T^{(1)}_a,T^{(1)}_f,T$:
\begin{enumerate}[leftmargin=*]
    \item Class 1 users arrive in $[T^{(1)}_a,T^{(1)}_f]$ at rates given by (\ref{eq_rates_inst2_uneqpref_reg2_case1}), implying $\mu_1\gamma^{(1)}(T^{(1)}_f-\tau_1^{-1}(T))+\mu_2\gamma^{(1)}(\tau_1^{-1}(T)-T^{(1)}_a)=\Lambda^{(1)}$.
    
    \item By Lemma \ref{lem_inst2_uneqpref_queue1idle}, queue 1 starts serving from $T^{(1),+}_a$, empties at $T^{(1)}_f$ and has positive waiting time in $(T^{(1),+}_a,T^{(1)}_f)$. This gives us, $T^{(1)}_f=T^{(1),+}_a +\frac{\Lambda^{(1)}}{\mu_1}$.

    \item  By assumption queue 2 empties at $T$ and starts serving from time zero. This gives us, $\mu_2 T=\Lambda^{(2)} + \mu_2\gamma^{(1)} (\tau_1^{-1}(T)-T^{(1)}_a)$.
\end{enumerate}

Applying $T^{(1),+}_a = 0$ for Type I and $=T^{(1)}_a$ for Type II, and plugging in $\tau_1^{-1}(T)$ from (\ref{eq:inst2_uneqpref_reg2_case1_inv_of_tau2_at_T}), we obtain the linear systems for Types I and II whose solutions are respectively in (\ref{eq_inst2_uneqpref_reg2_bdary_case1a}) and (\ref{eq_inst2_uneqpref_reg2_bdary_case1b}). Therefore, every EAP under Type I and II must, respectively, have support boundaries (\ref{eq_inst2_uneqpref_reg2_bdary_case1a}) and (\ref{eq_inst2_uneqpref_reg2_bdary_case1b}), and arrival rates given by (\ref{eq_rates_inst2_uneqpref_reg2_case1}). 

\textbf{Obtaining the necessary conditions:}\hspace{0.05in}The support boundaries in (\ref{eq_inst2_uneqpref_reg2_bdary_case1a}) must satisfy $T^{(1)}_a\leq 0$, for existence of an EAP under Type I. Imposing $T^{(1)}_a\leq 0$ on (\ref{eq_inst2_uneqpref_reg2_bdary_case1a}) gives us $\Lambda^{(1)}\geq \frac{\mu_1}{(1-\gamma^{(1)})\mu_2}\Lambda^{(2)}$, and this is a necessary condition for existence of an EAP under Type I. It is easy to verify that, once the necessary condition is satisfied, the support boundaries in (\ref{eq_inst2_uneqpref_reg2_bdary_case1a}) satisfy $T^{(1)}_a\geq 0, T^{(1)}_f>T>T^{(1)}_a$, and upon plugging them into (\ref{eq_rates_inst2_uneqpref_reg2_case1}), we get the only Type I candidate which qualifies to be an EAP. The obtained Type I candidate has a closed form same as the joint arrival profile mentioned under case 1a of Theorem \ref{mainthm_inst2_reg2}. 

Similarly plugging in $T^{(1)}_a>0$ in (\ref{eq_inst2_uneqpref_reg2_bdary_case1b}), we get $\Lambda^{(1)}<\frac{\mu_1}{(1-\gamma^{(1)})\mu_2}\Lambda^{(2)}$, which is a necessary condition for existence of an EAP under Type II. It is easy to verify, once the necessary condition holds, the support boundaries in (\ref{eq_inst2_uneqpref_reg2_bdary_case1b}) satisfies $T^{(1)}_a>0, T^{(1)}_f>T>0$, and upon plugging them into (\ref{eq_rates_inst2_uneqpref_reg2_case1}), we get the only Type II candidate which qualifies to be an EAP. The obtained Type II candidate has a closed form same as the joint arrival profile mentioned under case 1b of Theorem \ref{mainthm_inst2_reg2}.

\textbf{Proving sufficiency of the obtained necessary conditions:}\hspace{0.05in}Now by the next sequence of argument, we prove that, the only remaining candidate under both the types satisfy: 
\begin{align*}
    (C^{(1)}_{\mathbf{F}})^\prime(t)&\begin{cases}
        \leq 0~~\text{if}~t\in(-\infty,T^{(1)}_a) \\
        =0~~\text{if}~t\in[T^{(1)}_a,T^{(1)}_f] \\
        \geq 0~~\text{if}~t\in(T^{(1)}_f,\infty),
    \end{cases}~\text{and}~(C^{(2)}_{\mathbf{F}})^\prime(t)\begin{cases}
        \leq 0~~\text{if}~t\in(-\infty,T^{(1),+}_a-\frac{\Lambda^{(2)}}{\mu_2\gamma^{(2)}}) \\
        =0~~\text{if}~t\in\left[T^{(1),+}_a-\frac{\Lambda^{(2)}}{\mu_2\gamma^{(2)}},T^{(1),+}_a\right]\\
        \geq 0~~\text{if}~t\in(T^{(1),+}_a,\infty).
    \end{cases}
\end{align*} 
As a result, the following argument will imply, if the necessary condition of the corresponding type is true, the only remaining candidate under it is an EAP. Hence, the obtained necessary conditions are sufficient for existence of a unique EAP under those types, and the statement of Lemma \ref{lem:inst2reg2case1} follows.  
\begin{itemize}[leftmargin=*]
    \item \textbf{State of the queues:}~Obtained candidates of both the types satisfy $F^{(1)}(t)>\mu_1\cdot \max\{t-T^{(1),+}_a,0\}$ for $t\in(T^{(1)}_a,T^{(1)}_f)$ and $F^{(1)}(T^{(1)}_f)=\mu_1\cdot (T^{(1)}_f-T^{(1),+}_a)$. As a result, queue 1 stays engaged in $[T^{(1)}_a,T^{(1)}_f]$ and empties at $T^{(1)}_f$. Since the obtained candidates of both the types satisfy $T^{(1),+}_a<T<T^{(1)}_f$ and class 1 users arrive from queue 1 to 2 at rate $\mu_1$ in $[T^{(1),+}_a,T]$, we have $A_2(t)=F^{(2)}(t)+\mu_1\cdot(t-T^{(1),+}_a)^+$ for $t\in [0,T]$. For both the types, $A_2(T)=\Lambda^{(2)}+\mu_1\cdot (T-T^{(1),+}_a)=\mu_2 T$ and by (\ref{eq_rates_inst2_uneqpref_reg2_case1}), $A_2^\prime(t)<\mu_2$ in $[0,T)$, making $A_2(t)-\mu_2 t$ decreasing in $[0,T]$. As a result, $A_2(t)>\mu_2\cdot t$ in $[0,T)$ causing queue 2 to be engaged in $[T^{(1),+}_a-\frac{\Lambda^{(2)}}{\mu_2\gamma^{(2)}},T]$ and empty at $T$. Once queue 2 empties at $T$, it never engages again since class 1 users can arrive at a maximum rate of $\mu_1<\mu_2$. 

    \item \textbf{For class 1 users:}~For every $t<T^{(1)}_a$, $\tau_1(t)=t^+$ and $\tau^{(1)}_{\mathbf{F}}(t)=\tau_2(\tau_1(t))=\tau_2(t^+)$. Now for the Type I candidate, $\tau^{(1)}_{\mathbf{F}}(t)=\tau_2(0)$ in $(-\infty,T^{(1)}_a)$, causing $(C^{(1)}_{\mathbf{F}})^\prime(t)=-\gamma^{(1)}<0$ in $(-\infty,T^{(1)}_a)$. For the Type II candidate, following an argument similar to Type I, $(C^{(1)}_{\mathbf{F}})^\prime(t)=-\gamma^{(1)}<0$ in $(-\infty,0)$ and in $[0,T^{(1)}_a)$, $\tau^{(1)}_{\mathbf{F}}(t)=\tau_2(t)$. As a result, for Type II, using (\ref{eq:derv_of_tau}) and (\ref{eq_rates_inst2_uneqpref_reg2_case1}), $(\tau^{(1)}_{\mathbf{F}})^\prime(t)=\tau_2^\prime(t)=\frac{(F^{(2)})^\prime(t)}{\mu_2}=\gamma^{(2)}$ in $[0,T^{(1)}_a)$, causing $(C^{(1)}_{\mathbf{F}})^\prime(t)=\gamma^{(2)}-\gamma^{(1)}\leq 0$ in $[0,T^{(1)}_a)$. For $t\in[T^{(1)}_a,\tau_1^{-1}(T)]$, since queue 1 is engaged and queue 2 is engaged at $\tau_1(t)$, using (\ref{eq:derv_of_tau}), $(\tau^{(1)}_{\mathbf{F}})^\prime(t)=(\tau_2\circ\tau_1)^\prime(t)=\tau_2^\prime(\tau_1(t))\tau_1^\prime(t)=\frac{\mu_1}{\mu_2}\cdot\frac{\mu_2\gamma^{(1)}}{\mu_1}=\gamma^{(1)}$ causing $(C^{(1)}_{\mathbf{F}})^\prime(t)=0$ in $[T^{(1)}_a,\tau_1^{-1}(T)]$. For $t\in[\tau_1^{-1}(T),T^{(1)}_f]$, since queue 1 is engaged and queue 2 is empty at $\tau_1(t)\geq T$, $(\tau^{(1)}_{\mathbf{F}})^\prime(t)=\tau_1^\prime(t)=\frac{(F^{(1)})^\prime(t)}{\mu_1}=\gamma^{(1)}$, causing $(C^{(1)}_{\mathbf{F}})^\prime(t)=0$ in $[\tau_1^{-1}(T),T^{(1)}_f]$. Since queue 1 stays empty after $T^{(1)}_f$, $(C^{(1)}_{\mathbf{F}})^\prime(t)=1-\gamma^{(1)}>0$ in $[T^{(1)}_f,\infty)$. 

    \item \textbf{For class 2 users:}~Candidates obtained for both the types satisfy $T^{(1)}_a < \frac{\Lambda^{(2)}}{\mu_2\gamma^{(2)}}$. As a result, $\tau^{(2)}_{\mathbf{F}}(t)=\tau_2(t)=0$ for $t<T^{(1),+}_a-\frac{\Lambda^{(2)}}{\mu_2\gamma^{(2)}}$ making $(C_{\mathbf{F}}^{(2)})^\prime(t)=-\gamma^{(2)}<0$ in $\left(-\infty,T^{(1),+}_a-\frac{\Lambda^{(2)}}{\mu_2\gamma^{(2)}}\right)$. For $t\in\left[T^{(1),+}_a-\frac{\Lambda^{(2)}}{\mu_2\gamma^{(2)}},T^{(1),+}_a\right]$, queue 2 stays engaged and using (\ref{eq:derv_of_tau}), $(\tau_{\mathbf{F}}^{(2)})^\prime(t)=\tau_2^\prime(t)=\frac{(F^{(2)})^\prime(t)}{\mu_2}=\gamma^{(2)}$ causing $(C^{(2)}_{\mathbf{F}})^\prime(t)=0$ in $\left[T^{(1),+}_a-\frac{\Lambda^{(2)}}{\mu_2\gamma^{(2)}},T^{(1),+}_a\right]$. For $t\in(T^{(1),+}_a,T]$, queue 1 remains engaged and as a result, using (\ref{eq:derv_of_tau}), $(\tau_{\mathbf{F}}^{(2)})^\prime(t)=\frac{\mu_1}{\mu_2}$, causing $(C_{\mathbf{F}}^{(2)})^\prime(t)=\frac{\mu_1}{\mu_2}-\gamma^{(2)}\geq 0$ in $(T^{(1),+}_a,T]$. Since queue 2 remains empty after $T$, $(C_{\mathbf{F}}^{(2)})^\prime(t)=1-\gamma^{(2)}>0$ in $(T,\infty)$.
\end{itemize}
\vspace{-0.3in}
\end{proof}

\noindent\textbf{Case 2}~$\mu_1\geq\mu_2\gamma^{(1)}>\mu_2\gamma^{(2)}$:\hspace{0.05in}Like in case 1, by Lemma \ref{lem_supp_are_intervals_inst2} and \ref{lem_inst2_uneqpref_reg2_case12_sign2}, we only consider absolutely continuous candidates $\mathbf{F}=\{F^{(1)},F^{(2)}\}$ such that, $\mathcal{S}(F^{(1)})=[T^{(1)}_a,T^{(1)}_f]$ for some $T^{(1)}_f>T^{(1),+}_a=\max\{T^{(1)}_a,0\}$, $(F^{(2)})^\prime(t)=\mu_2\gamma^{(2)}$ in $\left[T^{(1),+}_a-\frac{\Lambda^{(2)}}{\mu_2\gamma^{(2)}}, T^{(1),+}_a\right]$, and the arrival rate $(F^{(1)})^\prime(\cdot)$ satisfies the properties in Lemma \ref{lem_arrival_rates_inst2}. The set of EAPs (if non-empty) will be contained in this set.

By Lemma \ref{lem_inst2_uneqpref_queue1idle}, queue 1 must be empty at $T^{(1)}_f$. Now if queue 2 has a positive length at $T^{(1)}_f$, the last arriving class 1 user will be better-off arriving at the time queue 2 empties. As a result, every EAP must have $T\leq T^{(1)}_f$.

If $T^{(1)}_a>0$, queue 2 must have a positive length at $T^{(1)}_a$ and by Lemma \ref{lem_arrival_rates_inst2}, class 1 users will arrive at rate $\mu_2\gamma^{(1)}\leq \mu_1$. So, for candidates with $T^{(1)}_a>0$, queue 1 will have no waiting queue and as a result, queue 2 will stay engaged in $[0,T^{(1)}_f]$, and empty at $T=T^{(1)}_f$. Now based on Lemma \ref{lem_inst2_uneqpref_reg2_case12_sign2} and the observation we just made, there are three possibilities: \textbf{1)}~Type I $T^{(1)}_a\leq 0$ and $T<T^{(1)}_f$; \textbf{2)}~Type II $T^{(1)}_a\leq 0$ and $T=T^{(1)}_f$; and \textbf{3)}~Type III $T^{(1)}_a>0$ and $T=T^{(1)}_f$. The following lemma states the necessary and sufficient conditions for existence of EAPs under the three types. Existence of unique EAP under case 2 of Theorem \ref{mainthm_inst2_reg2} follows from this lemma. 

\begin{lemma}\label{lem:inst2reg2case2}
    If $\mu_1\geq\mu_2\gamma^{(1)}>\mu_2\gamma^{(2)}$, the following statements are true about the EAP, 
    \begin{enumerate}
        \item There exists an EAP under Type I if and only if $\Lambda^{(2)}< \left(\frac{\mu_2}{\mu_1}-1\right) \Lambda^{(1)}$, and if it exists, it will be unique with a closed form same as the joint arrival profile mentioned under case 2a of Theorem \ref{mainthm_inst2_reg2}.   
        
        \item There exists an EAP under Type II if and only if $\left(\frac{\mu_2}{\mu_1}-1\right)\Lambda^{(1)}\leq\Lambda^{(2)}\leq \left(\frac{1}{\gamma^{(1)}}-1\right)\Lambda^{(1)}$, and if it exists, it must be unique with a closed form same as the joint arrival profile mentioned under case 2b of Theorem \ref{mainthm_inst2_reg2}.
        
        \item There exists an EAP under Type III if and only if $\left(\frac{1}{\gamma^{(1)}}-1\right)\Lambda^{(1)}<\Lambda^{(2)}$, and if it exists, it must be unique with a closed form same as the joint arrival profile mentioned under case 2c in Theorem \ref{mainthm_inst2_reg2}.
    \end{enumerate}
\end{lemma}

\begin{proof}
\noindent\textbf{Getting the arrival rates:}\hspace{0.05in}For all the types, queue 2 stays engaged in $[T^{(1),+}_a,T]$ and queue 1 empties at $T^{(1)}_f$ (by Lemma \ref{lem_inst2_uneqpref_queue1idle}). Therefore, by Lemma \ref{lem_arrival_rates_inst2}, we restrict to Types II and III candidates with arrival rates: 
\begin{align}\label{eq_rates_inst2_uneqpref_reg2_case2}
    (F^{(1)})^\prime(t)&=\mu_2\gamma^{(1)}~\text{for}~t\in[T^{(1)}_a,T^{(1)}_f]~\text{and}~(F^{(2)})^\prime(t)=\mu_2\gamma^{(2)}\cdot\mathbb{I}\left(t\in\left[T^{(1),+}_a-\frac{\Lambda^{(2)}}{\mu_2\gamma^{(2)}},T^{(1),+}_a\right]\right).
\end{align}
Note that, under Type I, queue 1 must remain engaged in $[T^{(1)}_a,T^{(1)}_f]$, otherwise, if queue 1 empties at some $\tilde{T}\in[0,T^{(1)}_f]$, by Lemma \ref{lem_arrival_rates_inst2} class 1 users arrive at a rate $\leq\mu_1$ after time $\tilde{T}$ and the networks stays empty in $[\max\{T,\tilde{T}\},T^{(1)}_f]$, which is not possible in an EAP. As a result, throughout $[0,T^{(1)}_f]$ class 1 users arrive from queue 1 to 2 at rate $\mu_1$. Therefore, once queue 2 empties at $T<T^{(1)}_f$, it stays empty in  $(T,\infty)$ and has a positive waiting time in $[0,T)$. Combining these observations with Lemma \ref{lem_arrival_rates_inst2}, we will only consider Type I candidates where arrival rates of the two classes are given by (\ref{eq_rates_inst2_uneqpref_reg2_case1}) upon putting $T^{(1),+}_a=0$.

\noindent\textbf{Identifying the support boundaries:}\hspace{0.05in}The support boundaries and $T$ of any Type I EAP satisfies a linear system same as the one we obtained for Type I in the proof of Lemma \ref{lem:inst2reg2case1}, and therefore has solution (\ref{eq_inst2_uneqpref_reg2_bdary_case1a}). 

The support boundaries of any Type II EAP satisfies the following linear system: \begin{itemize}[leftmargin=*]
    \item Queue 2 must empty at $T=T^{(1)}_f$, start at time zero and serve all the users in $[0,T^{(1)}_f]$, giving us $T^{(1)}_f=\frac{\Lambda^{(1)}+\Lambda^{(2)}}{\mu_2}$.
    \item By Lemma \ref{lem_arrival_rates_inst2}, class 1 users arrive in $[T^{(1)}_a,T^{(1)}_f]$ at rate $\mu_2\gamma^{(1)}$, giving us, $T^{(1)}_a=T^{(1)}_f-\frac{\Lambda^{(1)}}{\mu_2\gamma^{(1)}}$
\end{itemize}
Solution to the above system is in (\ref{eq_inst2_uneqpref_reg2_bdary_case2b}). 

The support boundaries of  any Type III EAP satisfies the following linear system: \begin{itemize}[leftmargin=*]
    \item We argued that $T=T^{(1)}_f$ if $T^{(1)}_a>0$. 
    \item Queue 2 must serve the whole population in $[0,T^{(1)}_f]$ and empty at $T^{(1)}_f$, giving us, $T^{(1)}_f=\frac{\Lambda^{(1)}+\Lambda^{(2)}}{\mu_2}$.
    \item By Lemma \ref{lem_arrival_rates_inst2}, class 1 users arrive at rate $\mu_2\gamma^{(1)}$ between $[T^{(1)}_a,T^{(1)}_f]$. This implies $T^{(1)}_a=T^{(1)}_f-\frac{\Lambda^{(1)}}{\mu_2\gamma^{(1)}}$.
\end{itemize}
Solution to the above system is in (\ref{eq_inst2_uneqpref_reg2_bdary_case2c}). 

\noindent\textbf{Obtaining the necessary conditions:}\hspace{0.05in}Support boundaries and $T$ in (\ref{eq_inst2_uneqpref_reg2_bdary_case1a}) must satisfy $T^{(1)}_f>T$ and $T^{(1)}_a\leq 0$, to represent a Type I EAP. This gives us $\Lambda^{(2)}<\left(\frac{\mu_2}{\mu_1}-1\right)\Lambda^{(1)}$, which is a necessary condition for existence of a Type I EAP. Once the necessary condition is satisfied, (\ref{eq_inst2_uneqpref_reg2_bdary_case1a}) satisfies $T^{(1)}_f>T$ and $T^{(1)}_a\leq 0$. Upon plugging in support boundaries from (\ref{eq_inst2_uneqpref_reg2_bdary_case1a}) into (\ref{eq_rates_inst2_uneqpref_reg2_case1}), we get the only Type I candidate which qualifies to be an EAP, and it has a closed form same as the joint arrival profile mentioned under case 2a of Theorem \ref{mainthm_inst2_reg2}.

The support boundaries in (\ref{eq_inst2_uneqpref_reg2_bdary_case2b}) must satisfy $T^{(1)}_a\leq 0$ and $\mu_1 T^{(1)}_f\geq\Lambda^{(1)}$ (for queue 1 to have zero waiting time at $T^{(1)}_f$), and this implies $\left(\frac{1}{\gamma^{(1)}}-1\right)\Lambda^{(1)}\geq\Lambda^{(2)}\geq\left(\frac{\mu_2}{\mu_1}-1\right)\Lambda^{(1)}$, which is a necessary condition for existence of an EAP under Type II. Once the necessary condition holds, (\ref{eq_inst2_uneqpref_reg2_bdary_case2b}) satisfies $T^{(1)}_a\leq 0$ and $\mu_1 T^{(1)}_f\geq\Lambda^{(1)}$. Upon plugging in (\ref{eq_inst2_uneqpref_reg2_bdary_case2b}) into (\ref{eq_rates_inst2_uneqpref_reg2_case2}), we get the only Type II candidate which qualifies to be an EAP, and it has a closed form same as the joint arrival profile mentioned under case 2b of Theorem \ref{mainthm_inst2_reg2} 

Similarly, the support boundaries in (\ref{eq_inst2_uneqpref_reg2_bdary_case2c}) must satisfy $T^{(1)}_a>0$. This gives us $\Lambda^{(2)}>\left(\frac{1}{\gamma^{(1)}}-1\right)\Lambda^{(1)}$, which is a necessary condition for existence of an EAP under Type III. Upon plugging in (\ref{eq_inst2_uneqpref_reg2_bdary_case2c}) into (\ref{eq_rates_inst2_uneqpref_reg2_case2}), we get the only Type III candidate which qualifies to be an EAP, and it has a closed form same as the joint arrival profile under case 3 of Theorem \ref{mainthm_inst2_reg2}.

\noindent\textbf{Proving sufficiency of the obtained necessary conditions:}\hspace{0.05in}Now by the next sequence of argument, we prove that, for every type,  if the corresponding necessary condition holds, the obtained candidate satisfies: 
\begin{align*}
    (C^{(1)}_{\mathbf{F}})^\prime(t)&\begin{cases}
        \leq 0~~\text{if}~t\in(-\infty,T^{(1)}_a) \\
        =0~~\text{if}~t\in[T^{(1)}_a,T^{(1)}_f] \\
        \geq 0~~\text{if}~t\in(T^{(1)}_f,\infty),
    \end{cases}~\text{and}~(C^{(2)}_{\mathbf{F}})^\prime(t)\begin{cases}
        \leq 0~~\text{if}~t\in(-\infty,T^{(1),+}_a-\frac{\Lambda^{(2)}}{\mu_2\gamma^{(2)}}) \\
        =0~~\text{if}~t\in\left[T^{(1),+}_a-\frac{\Lambda^{(2)}}{\mu_2\gamma^{(2)}},T^{(1),+}_a\right]\\
        \geq 0~~\text{if}~t\in(T^{(1),+}_a,\infty).
    \end{cases}
\end{align*} 
As a result, under every type, the obtained candidate is an EAP and the statement of Lemma \ref{lem:inst2reg2case2} follows. When $\Lambda^{(2)}<\left(\frac{\mu_2}{\mu_1}-1\right)\Lambda^{(1)}$, proving the above property for the only remaining Type I candidate follows the same argument as was used in the proof of Lemma \ref{lem:inst2reg2case1} for proving that, the unique remaining Type I candidate is an EAP. So, here we only prove the above property for the only remaining candidates under Type II and III.
\begin{enumerate}[leftmargin=*]
    \item If $\left(\frac{\mu_2}{\mu_1}-1\right)\Lambda^{(1)}\leq\Lambda^{(2)}\leq \left(\frac{1}{\gamma^{(1)}}-1\right)\Lambda^{(1)}$ the unique Type II candidate satisfies $\mu_1 T^{(1)}_f\geq\Lambda^{(1)}=A_1(T^{(1)}_f)$. As a result,  queue 1 empties in $[0,T^{(1)}_f]$ and let $\tilde{T}=\inf\{t\geq 0~\vert~Q_1(t)=0\}$, \textit{i.e.} the time at which queue 1 empties. Since class 1 users arrive at a constant rate of $\mu_2\gamma^{(1)}$ in $[T^{(1)}_a,T^{(1)}_f]$, queue 1 stays engaged in $[T^{(1)}_a,\tilde{T}]$ and empty in $[\tilde{T},T^{(1)}_f]$. In the case $\mu_1=\mu_2\gamma^{(1)}$, since the only way the candidate can be of Type II is by having $\Lambda^{(1)}=\left(\frac{1}{\gamma^{(1)}}-1\right)\Lambda^{(2)}$, from (\ref{eq_inst2_uneqpref_reg2_bdary_case2b}), we have $T^{(1)}_a=0$ and as a result, no queue develops in queue 1 implying $\tilde{T}=0$. In all cases, class 1 users arrive from queue 1 to 2 at rate $\mu_1<\mu_2$ in $[0,\tilde{T}]$ and at rate $\mu_2\gamma^{(1)}<\mu_2$ in $[\tilde{T},T^{(1)}_f]$. As a result, $A_2(t)-\mu_2\cdot t$ is strictly decreasing in $[\tilde{T},T^{(1)}_f]$ with $A_2(T^{(1)}_f)=\mu_2\cdot T^{(1)}_f$, implying $A_2(t)>\mu_2 t$ in $[\tilde{T},T^{(1)}_f)$. Note that $A_2(t)=\Lambda^{(2)}+\mu_1 t$ in $[0,\tilde{T}]$ with $A_2(0)=\Lambda^{(2)}>0$, and $A_2(\tilde{T})>\mu_2 \tilde{T}$ implies $A_2(t)\geq\mu_2 t$ in $[0,\tilde{T}]$. Combining the preceding two statements, we get $A_2(t)>\mu_2\cdot t$ in $[0,T^{(1)}_f)$ which implies queue 2 stays engaged with a positive queue length in $[0,T^{(1)}_f)$ and empties at $T^{(1)}_f$. We now argue for the two classes separately:
    \begin{enumerate}
        \item[\textbf{a.}] \textbf{For class 1 users:}~Users arriving before $T^{(1)}_a$, gets served at time $\tau_2(0)$, causing $(C_{\mathbf{F}}^{(1)})^\prime(t)=-\gamma^{(1)}$ in $(-\infty,T^{(1)}_a)$. In $[T^{(1)}_a,T^{(1)}_f]$, queue 2 stays engaged at $\tau_1(t)\in[0,T^{(1)}_f]$.  Now in $[T^{(1)}_a,\tilde{T}]$, since queue 1 stays engaged, using (\ref{eq:derv_of_tau}), $(\tau_{\mathbf{F}}^{(1)})^\prime(t)=(\tau_2\circ\tau_1)^\prime(t)=\tau_2^\prime(\tau_1(t))\cdot\tau_1^\prime(t)=\frac{\mu_1}{\mu_2}\cdot\frac{\mu_2\gamma^{(1)}}{\mu_1}=\gamma^{(1)}$. In $[\tilde{T},T^{(1)}_f]$, since queue 1 empties, using (\ref{eq:derv_of_tau}), $(\tau_{\mathbf{F}}^{(1)})^\prime(t)=\tau_2^\prime(t)=\gamma^{(1)}$ in $[\tilde{T},T^{(1)}_f]$. As a result, $(\tau_{\mathbf{F}}^{(1)})^\prime(t)=\gamma^{(1)}$ in $[T^{(1)}_a,T^{(1)}_f]$, causing $(C_{\mathbf{F}}^{(1)})^\prime(t)=0$ there. Since both the queues become empty after $T^{(1)}_f$, we have $(C_{\mathbf{F}}^{(1)})^\prime(t)=1-\gamma^{(1)}>0$ in $(T^{(1)}_f,\infty)$.
        
        \item[\textbf{b.}] \textbf{For class 2 users:}~Since $\tau_{\mathbf{F}}^{(2)}(t)=0$ in $(-\infty,-\frac{\Lambda^{(2)}}{\mu_2\gamma^{(2)}}$, $(C^{(2)}_{\mathbf{F}})^\prime(t)=-\gamma^{(2)}<0$ there. Using (\ref{eq:derv_of_tau}), since queue 2 stays engaged in $\left[-\frac{\Lambda^{(2)}}{\mu_2\gamma^{(2)}},0\right]$, $(\tau^{(2)}_{\mathbf{F}})^\prime(t)=\tau_2^\prime(t)=\frac{(F^{(2)})^\prime(t)}{\mu_2}=\gamma^{(2)}$, causing $(C^{(2)}_{\mathbf{F}})^\prime(t)=0$ there. In $[0,\tilde{T}]$, since queue 2 stays engaged with class 1 users arriving at rate $\mu_1$, using (\ref{eq:derv_of_tau}), $(\tau^{(2)}_{\mathbf{F}})^\prime(t)=\frac{\mu_1}{\mu_2}$, causing $(C_{\mathbf{F}}^{(2)})^\prime(t)=\frac{\mu_1}{\mu_2}-\gamma^{(2)}\geq 0$ there. In $[\tilde{T},T^{(1)}_f]$, queue 2 stays engaged and queue 1 idle, with class 1 users arriving at rate $\mu_2\gamma^{(1)}$, causing $(\tau_{\mathbf{F}}^{(2)})^\prime(t)=\gamma^{(1)}$ and $(C_{\mathbf{F}}^{(2)})^\prime(t)=\gamma^{(1)}-\gamma^{(2)}\geq 0$ there. In $(T^{(1)}_f,\infty)$, since queue 2 is idle, we have $(C_{\mathbf{F}}^{(2)})^\prime(t)=1-\gamma^{(2)}$ there. 
    \end{enumerate}
    
    \item If $\left(\frac{1}{\gamma^{(1)}}-1\right)\Lambda^{(1)}<\Lambda^{(2)}$, in the unique remaining Type III candidate, class 1 users arrive at queue 1 at rate $\mu_2\gamma^{(1)}\leq\mu_1$ from $T^{(1)}_a>0$. As a result, no queue develops at queue 1. Hence, both the classes $i=1,2$ have $\tau^{(i)}_{\mathbf{F}}(t)=\tau_2(t)$. Since $A_2^\prime(t)<\mu_2$ in $[0,T^{(1)}_f]$, $A_2(t)-\mu_2 t$ is strictly decreasing in $[0,T^{(1)}_f]$. This along with $A_2(T^{(1)}_f)=\Lambda^{(1)}+\Lambda^{(2)}=\mu_2 T^{(1)}_f$ implies $A_2(t)>\mu_2 t$, which causes queue 2 to have a positive waiting time in $[0,T^{(1)}_f)$ and empty at $T^{(1)}_f$. Since every user arriving before $T^{(2)}_a$ departs at time zero, we have $(C_{\mathbf{F}}^{(i)})^\prime(t)=-\gamma^{(i)}$ in $(-\infty,T^{(2)}_a)$ for $i\in\{1,2\}$. In $[T^{(2)}_a,T^{(2)}_f]$, $A_2^\prime(t)=\mu_2\gamma^{(2)}$ implies $\tau_2^\prime(t)=\gamma^{(2)}$ and $(C_{\mathbf{F}}^{(1)})^\prime(t)=\gamma^{(2)}-\gamma^{(1)}<0,~(C_{\mathbf{F}}^{(2)})^\prime(t)=0$ in $[T^{(2)}_a,T^{(2)}_f]$. In $[T^{(2)}_f,T^{(1)}_f]$, $A_2^\prime(t)=\mu_2\gamma^{(1)}$ implies $\tau_2^\prime(t)=\gamma^{(1)}$ and $(C_{\mathbf{F}}^{(1)})^\prime(t)==0,~(C_{\mathbf{F}}^{(2)})^\prime(t)=\gamma^{(1)}-\gamma^{(2)}>0$ in $[T^{(2)}_f,T^{(1)}_f]$. In $(T^{(1)}_f,\infty)$, queue 2 is empty, causing $(C_{\mathbf{F}}^{(i)})^\prime(t)=1-\gamma^{(i)}>0$ for $i\in\{1,2\}$.  \vspace{-0.2in}
\end{enumerate}
\end{proof}
\vspace{-0.2in}
\noindent\textbf{Case 3}~~$\mu_1\geq\mu_2\gamma^{(2)}>\mu_2\gamma^{(1)}$:\hspace{0.05in}By Lemma \ref{lem_inst2_uneqpref_sign_reg1b}, every EAP will have $T^{(1)}_a\leq 0$. 
Using Lemma \ref{lem_supp_are_intervals_inst2} and \ref{lem_inst2_uneqpref_sign_reg1b}, we consider absolutely continuous candidates $\mathbf{F}=\{F^{(1)},F^{(2)}\}$, such that $\mathcal{S}(F^{(1)})=[T^{(1)}_a,T^{(1)}_f]$ for some $T^{(1)}_f>0\geq T^{(1)}_a$, and $\mathcal{S}(F^{(1)})\cup\mathcal{S}(F^{(2)})$ is an interval, and the arrival rates $(F^{(1)})^\prime(\cdot),~(F^{(2)})^\prime(\cdot)$ satisfies the property in Lemma \ref{lem_arrival_rates_inst2}. We restrict to candidates with $T^{(2)}_a$ and $T^{(2)}_f$ finite. For every candidate there are two possibilities, either the whole class 2 population arrives before time zero, or a fraction of them arrives after $T^{(1)}_f$. Based on this observation, we divide the set of candidates into three types: \textbf{1)}  Type I~~$T^{(1)}_a\leq 0$ and all class 2 users arrive before time zero and $T<T^{(1)}_f$,~\textbf{2)}  Type II~~$T^{(1)}_a\leq 0$, all class 2 users arrive before time zero and $T=T^{(1)}_f$, and \textbf{3)}  Type III~~$T^{(1)}_a\leq 0$ and a positive mass of class 2 users arrive after $T^{(1)}_f$. The following lemma gives the necessary and sufficient condition for existence of an EAP under the three types. Existence of unique EAP under case 3 of Theorem \ref{mainthm_inst2_reg2} follows from this lemma. 

\begin{lemma}\label{lem:inst2reg2case3}
    If $\mu_1\geq\mu_2\gamma^{(2)}>\mu_2\gamma^{(1)}$, the following statements are true about the EAP,  
    \begin{enumerate}
        \item There exists an EAP under Type I if and only if $\Lambda^{(2)}< \left(\frac{\mu_2}{\mu_1}-1\right) \Lambda^{(1)}$, and if it exists, it will be unique with a closed form same as the joint arrival profile mentioned under case 3a of Theorem \ref{mainthm_inst2_reg2}.   
        
        \item There exists an EAP under Type II if and only if $\left(\frac{\mu_2}{\mu_1}-1\right)\Lambda^{(1)}\leq\Lambda^{(2)}\leq \left(\frac{1}{\gamma^{(2)}}-1\right)\Lambda^{(1)}$, and if it exists, it must be unique with a closed form same as the joint arrival profile mentioned under case 3b of Theorem \ref{mainthm_inst2_reg2}.
        
        \item There exists an EAP under Type III if and only if $\left(\frac{1}{\gamma^{(2)}}-1\right)\Lambda^{(1)}<\Lambda^{(2)}$, and if it exists, it must be unique with a closed form same as the joint arrival profile mentioned under case 3c in Theorem \ref{mainthm_inst2_reg2}.
    \end{enumerate}
\end{lemma}

\begin{proof}
\noindent\textbf{Getting the arrival rates:}\hspace{0.05in}For Type I and II, we can restrict ourselves to candidates for which $\mathcal{S}(F^{(2)})$ must be an interval. Otherwise, if $\mathcal{S}(F^{(2)})$ has any gap $[t_1,t_2]$, such that $t_1<t_2\leq 0$ and $F^{(2)}(t_1)=F^{(2)}(t_2)$, then by (\ref{eq:derv_of_tau}), $\tau_{\mathbf{F}}^{(2)}(\cdot)=\tau_2(\cdot)$ remains constant in $[t_1,t_2]$ and the class 2 user arriving at $t_1$ will be strictly better off arriving at $t_2$ and such candidates cannot be EAP. Also, following the same argument, we must have $T^{(2)}_f=0$. Otherwise, if $T^{(2)}_f<0$, the class 2 user arriving at $T^{(2)}_f$ will be strictly better off arriving at time zero. Now by Lemma \ref{lem_arrival_rates_inst2}, class 2 users must arrive over an interval ending at time zero at a constant rate of $\mu_2\gamma^{(2)}$. As a result, under Type I and II, we consider only candidates having $(F^{(2)})^\prime(t)=\mu_2\gamma^{(2)}\cdot\mathbb{I}\left(t\in\left[-\frac{\Lambda^{(2)}}{\mu_2\gamma^{(2)}},0\right]\right)$. Using Lemma \ref{lem_arrival_rates_inst2}, arrival rates of the two classes for Type I and II candidates must, respectively, be the ones in (\ref{eq_rates_inst2_uneqpref_reg2_case1}) and (\ref{eq_rates_inst2_uneqpref_reg2_case2}) with $T^{(1),+}_a = 0$. 

For Type III candidates, using an argument similar to types I and II, the portion of $\mathcal{S}(F^{(2)})$ before time zero must be an interval ending at time zero and therefore $\mathcal{S}(F^{(2)})\cap(-\infty,0]=[T^{(2)}_a,0]$. By Lemma \ref{lem_inst2_uneqpref_queue1idle}, queue 1 empties at time $T^{(1)}_f$ and therefore, as a result, no class 1 user arrives at queue 2 after time $T^{(1)}_f$. Since $\mathcal{S}(F^{(1)})\cup\mathcal{S}(F^{(2)})$ must be an interval, the portion of $\mathcal{S}(F^{(2)})$ after $T^{(1)}_f$ must be an interval ending at $T^{(2)}_f>T^{(1)}_f$. Hence, using Lemma \ref{lem_arrival_rates_inst2}, we can restrict ourselves to Type III candidates for whom $(F^{(2)})^\prime(t)=\mu_2\gamma^{(2)}\cdot\mathbb{I}\left(t\in[T^{(2)}_a,0]\cup[T^{(1)}_f,T^{(2)}_f]\right)$ for some $T^{(2)}_f>T^{(1)}_f>0>T^{(2)}_a$. We restrict ourselves to Type III candidates where, queue 2 has a positive waiting time in $[0,T^{(1)}_f]$, otherwise, the class 2 users arriving after $T^{(1)}_f$ will be strictly better off arriving at the time queue 2 empties in $[0,T^{(1)}_f]$. Therefore, for every $t\in[T^{(1)}_a,T^{(1)}_f]$, since queue 1 empties at $T^{(1)}_f$ and $\tau_1(T^{(1)}_a)=0$, we have $\tau_1(t)\in[0,T^{(1)}_f]$ and as a result, queue 2 has a positive waiting time at $\tau_1(t)$. As a result, by Lemma \ref{lem_arrival_rates_inst2}, we can restrict ourselves to candidates where class 1 users arrive at rate $(F^{(1)})^\prime(t)=\mu_2\gamma^{(1)}\cdot\mathbb{I}\left(t\in[T^{(1)}_a,T^{(1)}_f]\right)$. Hence, we are left with Type III candidates where arrival rates of the two classes are: 
\begin{align}\label{eq_rates_inst2_uneqpref_reg2_case3}
    (F^{(1)})^\prime(t)&=\mu_2\gamma^{(1)}\cdot\mathbb{I}(t\in[T^{(1)}_a,T^{(1)}_f])~\text{and}~(F^{(2)})^\prime(t)=\mu_2\gamma^{(2)}\cdot\mathbb{I}(t\in[T^{(2)}_a,0]\cup[T^{(1)}_f,T^{(2)}_f]).
\end{align}

\noindent\textbf{Identifying the support boundaries:}\hspace{0.05in}For Type I, the support boundaries satisfy a linear system same as the one obtained in the proof of Lemma \ref{lem:inst2reg2case1} for Type I candidates, and therefore has solution (\ref{eq_inst2_uneqpref_reg2_bdary_case1a}). 

For Type II, the support boundaries satisfy a linear system same as the one obtained in the proof of Lemma \ref{lem:inst2reg2case2} for Type II candidates, and therefore has solution (\ref{eq_inst2_uneqpref_reg2_bdary_case2b}).

For Type III, we obtain the following linear system to be satisfied by the support boundaries: 
\begin{itemize}[leftmargin=*]
    \item By (\ref{eq_rates_inst2_uneqpref_reg2_case3}), class 1 users arrive at $\mu_2\gamma^{(1)}$ between $[T^{(1)}_a,T^{(1)}_f]$. This gives us $T^{(1)}_f-T^{(1)}_a=\frac{\Lambda^{(1)}}{\mu_2\gamma^{(1)}}$.
    \item By (\ref{eq_rates_inst2_uneqpref_reg2_case3}), class 2 users arrive at $\mu_2\gamma^{(2)}$ between $[T^{(2)}_a,0]\cup[T^{(1)}_f,T^{(2)}_f]$. This gives us $T^{(2)}_f-T^{(1)}_f-T^{(2)}_a=\frac{\Lambda^{(2)}}{\mu_2\gamma^{(2)}}$.
    \item Queue 2 must serve the whole population between $[0,T^{(2)}_f]$ and empty at $T=T^{(2)}_f$, giving us, $T^{(2)}_f=\frac{\Lambda^{(1)}+\Lambda^{(2)}}{\mu_2}$.
    \item By definition of EAP $C_{\mathbf{F}}^{(2)}(T^{(2)}_a)=C_{\mathbf{F}}^{(2)}(T^{(2)}_f)$. This gives us $T^{(2)}_a=-\left(\frac{1}{\gamma^{(2)}}-1\right)T^{(2)}_f$. \vspace{-0.1in}
\end{itemize} 
Solution to the above linear system is (\ref{eq_inst2_uneqpref_reg2_bdary_case3c}). 

\noindent\textbf{Obtaining the necessary conditions:}\hspace{0.05in}To represent a Type I candidate (\ref{eq_inst2_uneqpref_reg2_bdary_case1a}) must satisfy $T^{(1)}_a\leq 0$ and $T^{(1)}_f>T$. This gives us $\Lambda^{(2)}<\left(\frac{\mu_2}{\mu_1}-1\right)\Lambda^{(1)}$, which is a necessary condition for existence of an EAP under Type I. Once the necessary condition is satisfied, (\ref{eq_inst2_uneqpref_reg2_bdary_case1a}) satisfies $T^{(1)}_a\leq 0$ and $T^{(1)}_f>T$. Upon plugging in (\ref{eq_inst2_uneqpref_reg2_bdary_case1a}) into (\ref{eq_rate_inst2_uneqpref_reg1a}), we get the only Type I candidate which qualifies to be an EAP, and it has closed form same as the joint arrival profile mentioned under case 3a of Theorem \ref{mainthm_inst2_reg2}. 

To represent a Type II candidate, the solution in (\ref{eq_inst2_uneqpref_reg2_bdary_case2b}) must satisfy $T^{(1)}_a\leq 0$, $\mu_1 T^{(1)}_f\geq\Lambda^{(1)}$ (by Lemma \ref{lem_inst2_uneqpref_queue1idle} queue 1 must be empty by time $T^{(1)}_f$) and $C_\mathbf{F}^{(2)}(T^{(1)}_f)=\gamma^{(2)} T^{(1)}_f\geq C_\mathbf{F}^{(2)}(T^{(2)}_a)=-(1-\gamma^{(2)})T^{(2)}_a$, which implies $\left(\frac{1}{\gamma^{(2)}}-1\right)\Lambda^{(1)}\geq\Lambda^{(2)}\geq\left(\frac{\mu_2}{\mu_1}-1\right)\Lambda^{(1)}$ and this gives us a necessary condition for existence of an EAP under Type II. It is easy to verify that, if $\left(\frac{1}{\gamma^{(2)}}-1\right)\Lambda^{(1)}\geq\Lambda^{(2)}\geq\left(\frac{\mu_2}{\mu_1}-1\right)\Lambda^{(1)}$ , (\ref{eq_inst2_uneqpref_reg2_bdary_case2b}) satisfies the desired conditions. Upon plugging in (\ref{eq_inst2_uneqpref_reg2_bdary_case2b}) into (\ref{eq_rate_inst2_uneqpref_reg2b}), we get the only Type II candidate which qualifies to be an EAP, and it has a closed form same as the joint arrival profile mentioned under case 3b in Theorem \ref{mainthm_inst2_reg2}.

Repeating the same procedure with the solution in (\ref{eq_inst2_uneqpref_reg2_bdary_case3c}) but instead imposing $T^{(2)}_f>T^{(1)}_f$, we obtain $\left(\frac{1}{\gamma^{(2)}}-1\right)\Lambda^{(1)}>\Lambda^{(2)}$, which is a necessary condition for existence of an EAP under Type III. With the necessary condition satisfied, (\ref{eq_inst2_uneqpref_reg2_bdary_case3c}) satisfies the desired conditions. Upon plugging in  (\ref{eq_inst2_uneqpref_reg2_bdary_case3c}) into (\ref{eq_rates_inst2_uneqpref_reg2_case3}), we get the only Type III candidate which qualifies to be an EAP, and it has closed form same as the joint arrival profile mentioned under case 3 of Theorem \ref{mainthm_inst2_reg2}. 

\noindent\textbf{Proving sufficiency of the obtained necessary conditions:}\hspace{0.05in}When $\Lambda^{(2)}< \left(\frac{\mu_2}{\mu_1}-1\right)\Lambda^{(1)}$, the only remaining Type I candidate is an EAP following the same argument as was used for proving that the only remaining Type I candidate is an EAP in the proof of Lemma \ref{lem:inst2reg2case1}. Below we argue that, in the other two types as well, if the necessary condition is satisfied, the only remaining candidate is an EAP. As a result, statement of Lemma \ref{lem:inst2reg2case3} follows.
\begin{itemize}[leftmargin=*]
    \item If $\left(\frac{\mu_2}{\mu_1}-1\right)\Lambda^{(1)}\leq\Lambda^{(2)}\leq \left(\frac{1}{\gamma^{(2)}}-1\right) \Lambda^{(1)}$: The unique remaining Type II candidate satisfies $\mu_1 T^{(1)}_f\geq A_1(T^{(1)}_f)=\Lambda^{(1)}$ causing queue 1 to empty at some $\tilde{T}\in[0,T^{(1)}_f]$. 
    Following the argument used in the proof of Lemma \ref{lem:inst2reg2case2} to prove the unique remaining Type II candidate is an EAP, queue 1 stays engaged in $[T^{(1)}_a,\tilde{T}]$ and empty in $[\tilde{T},\infty)$. On the other hand, queue 2 stays engaged in $[-\Lambda^{(2)}/\mu_2\gamma^{(2)}, T^{(1)}_f]$ and empty in $[T^{(1)}_f,\infty)$. Now for class 1 users, using the same argument, 
    \begin{align*}
        (C_{\mathbf{F}}^{(1)})^\prime(t)&\begin{cases}
            \leq 0~~\text{for}~t\in(-\infty,T^{(1)}_a)\\
            =0~~\text{for}~t\in[T^{(1)}_a,T^{(1)}_f]\\
            \geq 0~~\text{for}~t\in(T^{(1)}_f,\infty),
        \end{cases}
    \end{align*} 
   and similarly for class 2 users, we have, 
    \begin{align*}
        (C_{\mathbf{F}}^{(2)})^\prime(t)&=\begin{cases}
            -\gamma^{(2)}<0~~\text{for}~t\in(-\infty,-\Lambda^{(2)}/\mu_2\gamma^{(2)}) \\
            =0~~\text{for}~t\in[-\Lambda^{(2)}/\mu_2\gamma^{(2)},0) \\
            =\frac{\mu_1}{\mu_2}-\gamma^{(2)}\geq 0~~\text{for}~t\in[0,\tilde{T}) \\
            =\gamma^{(1)}-\gamma^{(2)} < 0~~\text{for}~t\in[\tilde{T},T^{(1)}_f] \\
            =1-\gamma^{(2)} >0 ~~ \text{for}~t\in(T^{(1)}_f,\infty)
        \end{cases}
    \end{align*}
    Now $C_{\mathbf{F}}^{(2)}(\cdot)$ is constant in $[-\Lambda^{(2)}/\mu_2\gamma^{(2)},0]$, non-decreasing in $[0,\tilde{T}]$, decreasing in $[\tilde{T},T^{(1)}_f]$ and then increasing in $(T^{(1)}_f,\infty)$. Since the candidate satisfies, $C_{\mathbf{F}}^{(2)}(T^{(1)}_f)\geq C_{\mathbf{F}}^{(2)}(-\Lambda^{(2)}/\mu_2\gamma^{(2)})=C_{\mathbf{F}}^{(2)}(0)$, by the previous statement we have $C_{\mathbf{F}}^{(2)}(t)\geq C_{\mathbf{F}}^{(2)}(0)$ for every $t\geq 0$. This is sufficient to prove that, both the classes have constant cost on their support and higher cost outside of it and hence the candidate is an EAP. 

    \item If $\left(\frac{1}{\gamma^{(2)}}-1\right)\Lambda^{(1)}<\Lambda^{(2)}$: The unique remaining Type III candidate satisfies, $\mu_1 T^{(1)}_f\geq \Lambda^{(1)}=A_1(T^{(1)}_f)$, causing queue 1 to empty at some $\tilde{T}\in[0,T^{(1)}_f]$. Single class 1 users arrive at a constant rate of $\mu_2\gamma^{(1)}$, queue 1 stays engaged in $[T^{(1)}_a,\tilde{T}]$ and empty in $(\tilde{T},\infty)$. Using this, arrival rate of class 2 users in queue 2 is: $\mu_1$ in $[0,\tilde{T}]$, $\mu_2\gamma^{(1)}<\mu_2$ in $[\tilde{T},T^{(1)}_f]$ and $\mu_2\gamma^{(2)}<\mu_2$ in $[T^{(1)}_f,T^{(2)}_f]$. As a result, $A_2(t)-\mu_2 t$ is strictly decreasing in $[\tilde{T},T^{(2)}_f]$ with $A_2(T^{(2)}_f)=\Lambda^{(1)}+\Lambda^{(2)}=\mu_2 T^{(2)}_f$, causing $A_2(t)>\mu_2 t$ in $[\tilde{T},T^{(2)}_f)$. Again, $A_2(t)-\mu_2 t$ is linear in $[0,\tilde{T}]$ with $A_2(0)=F^{(2)}(0)>0$ and $A_2(\tilde{T})>\mu_2 \tilde{T}$, implying $A_2(t)>\mu_2 t$ in $[0,\tilde{T}]$. As a result, queue 2 has a positive waiting time in $[0,T^{(2)}_f)$ and is empty in $[T^{(2)}_f,\infty)$. Now we consider the two classes separately and prove that, for both the classes cost is constant on their support and higher outside:
    
    \begin{itemize}
    \item \textbf{For class 1:} $\tau_{\mathbf{F}}^{(1)}(t)=\tau_2(0)$ in $(-\infty,T^{(1)}_a)$, causing $(C_{\mathbf{F}}^{(1)})^\prime(t)=-\gamma^{(1)}<0$ in $(-\infty,T^{(1)}_a)$. In $[T^{(1)}_a,\tilde{T}]$, using (\ref{eq:derv_of_tau}), $(\tau_{\mathbf{F}}^{(1)})^\prime(t)=(\tau_2\circ\tau_1)^\prime(t)=\tau_2^\prime(\tau_1(t))\cdot\tau_1^\prime(t)=\frac{\mu_1}{\mu_2}\cdot\frac{\mu_2\gamma^{(1)}}{\mu_1}=\gamma^{(1)}$, causing $(C_{\mathbf{F}}^{(1)})^\prime(t)=0$ in $[T^{(1)}_a,\tilde{T}]$. In $[\tilde{T},T^{(1)}_f]$, queue 1 is empty and queue 2 is busy. As a result, using (\ref{eq:derv_of_tau}), $(\tau_{\mathbf{F}}^{(1)})^\prime(t)=\tau_2^\prime(t)=\frac{(F^{(1)})^\prime(t)}{\mu_2}=\gamma^{(1)}$ and $(C_{\mathbf{F}}^{(1)})^\prime(t)=0$ in $[\tilde{T},T^{(1)}_f]$. In $[T^{(1)}_f,T^{(2)}_f]$, queue 2 stays busy and $(\tau_{\mathbf{F}}^{(1)})^\prime(t)=\tau_2^\prime(t)=\frac{(F^{(2)})^\prime(t)}{\mu_2}=\gamma^{(2)}>\gamma^{(1)}$, causing $(C_{\mathbf{F}}^{(1)})^\prime(t)=\gamma^{(2)}-\gamma^{(1)}>0$ in $[T^{(1)}_f,T^{(2)}_f]$. In $(T^{(2)}_f,\infty)$, both the queues stay idle, causing $(C_{\mathbf{F}}^{(1)})^\prime(t)=1-\gamma^{(1)}>0$. Hence for the class 1 users, cost is constant in $[T^{(1)}_a,T^{(1)}_f]$ and higher outside. 
    
    \item \textbf{For class 2:} $\tau_{\mathbf{F}}^{(2)}(t)=0$ in $(-\infty,-\Lambda^{(2)}/\mu_2\gamma^{(2)})$, causing $(C_{\mathbf{F}}^{(2)})^\prime(t)=-\gamma^{(2)}<0$. In $[-\Lambda^{(2)}/\mu_2\gamma^{(2)},0]$, since queue 2 is engaged, using (\ref{eq:derv_of_tau}), we have $(\tau_{\mathbf{F}}^{(2)})^\prime(t)=\frac{(F^{(2)})^\prime(t)}{\mu_2}=\gamma^{(2)}$, which causes $(C_{\mathbf{F}}^{(2)})^\prime(t)=0$. In $[0,T^{(2)}_f]$, since queue 2 stays engaged, using (\ref{eq:derv_of_tau}),    \begin{align*}
        (C_{\mathbf{F}}^{(2)})^\prime(t)=\tau_2^\prime(t)-\gamma^{(2)}&=\begin{cases}
            \frac{\mu_1}{\mu_2}-\gamma^{(2)}\geq 0~~\text{for}~t\in(0,\tilde{T}], \\
            \gamma^{(1)}-\gamma^{(2)}<0~~\text{for}~t\in[\tilde{T},T^{(1)}_f], \\
            0~~\text{for}~t\in(T^{(1)}_f,T^{(2)}_f].
        \end{cases}
    \end{align*}    
    Now, $(C_{\mathbf{F}}^{(2)})^\prime(t)=0$ in $[T^{(2)}_a,0]\cup[T^{(1)}_f,T^{(2)}_f]$ and $C_{\mathbf{F}}^{(2)}(T^{(2)}_f)=\gamma^{(2)} T^{(2)}_f=-(1-\gamma^{(2)})T^{(2)}_a=C_{\mathbf{F}}^{(2)}(T^{(2)}_a)$ implies $C_{\mathbf{F}}^{(2)}(0)=C_{\mathbf{F}}^{(2)}(T^{(1)}_f)$. Since $C_{\mathbf{F}}^{(2)}(\cdot)$ is first non-decreasing in $[0,\tilde{T}]$ and then decreasing in $[\tilde{T},T^{(1)}_f]$, with $C_{\mathbf{F}}^{(2)}(0)=C_{\mathbf{F}}^{(2)}(T^{(1)}_f)$, we must have $C_{\mathbf{F}}^{(2)}(t)\geq C_{\mathbf{F}}^{(2)}(0)=C_{\mathbf{F}}^{(2)}(T^{(1)}_f)$ for every $t\in[0,T^{(1)}_f]$. Since queue 2 empties at $T^{(2)}_f$, we have $(C_{\mathbf{F}}^{(2)})^\prime(t)=1-\gamma^{(2)}>0$ in $(T^{(2)}_f,\infty)$. Therefore, for the class 2 users, the cost is constant on $[T^{(2)}_a,0]\cup[T^{(1)}_f,T^{(2)}_f]$ and higher outside.   
    \end{itemize}
    \vspace{-0.4in}
\end{itemize}
\end{proof}

\subsection{Structural properties of EAP for HAS with $\gamma^{(1)}=\gamma^{(2)}$}\label{appndx:inst2_eqpref}

\begin{lemma}[\textbf{Rate of arrival}]\label{lem_arrivalrates_inst2_eqpref}
    In every EAP the rate of arrivals will satisfy the following properties: 
    \begin{align*}
    \forall t\in \mathcal{S}(F^{(1)}),~~(F^{(1)})^\prime(t)&\begin{cases}
            =\mu_1\gamma~~&~\text{if}~\tau_1(t)\in\overline{E_2}^c,\\
            =\mu_2\gamma~~&\text{if}~\tau_1(t)\notin\mathcal{S}(F^{(2)})\cup\overline{E_2}^c,\\
            \leq\mu_1~~&\text{if}~\tau_1(t)\in\mathcal{S}(F^{(2)})~\text{and}~t\notin\overline{E_1},\\
            =\mu_1~~&\text{if}~\tau_1(t)\in\mathcal{S}(F^{(2)})~\text{and}~t\in\overline{E_1},
        \end{cases}\\
    \text{and}~\forall t\in \mathcal{S}(F^{(2)}),~~(F^{(2)})^\prime(t)&=\begin{cases}
    \mu_2\gamma-\mu_1~~&\text{if}~t\in E_1,~\text{and},\\
    \mu_2\gamma-(F^{(1)})^\prime(t)~~~&\text{otherwise.}
    \end{cases}
    \end{align*}
\end{lemma}

\begin{proof} 
For class 2 users, by (\ref{eq:isocost_inst2_grp2}), if $t\in\mathcal{S}(F^{(2)})$, $\tau_2^\prime(t)=\gamma$. Since $\mathcal{S}(F^{(2)})\subseteq \overline{E_2}$, using (\ref{eq:derv_of_tau}), $\tau_2^\prime(t)=\frac{Y_1^\prime(t)+(F^{(2)})^\prime(t)}{\mu_2}$. Now if $t\in E_1$, $y_1^\prime(t)=\mu_1$ and as a result, $(F^{(2)})^\prime(t)=\mu_2\gamma-\mu_1$. Otherwise if $t\notin \overline{E_1}$, $Y_1^\prime(t)=(F^{(1)})^\prime(t)$ and as a result, $(F^{(2)})^\prime(t)=\mu_2\gamma-(F^{(1)})^\prime(t)$.

For class 1 users, the argument behind arrival rates in the first two cases is similar to the one used for first two cases of class 1 arrival rate in Lemma \ref{lem_arrival_rates_inst2}. Therefore in the rest of the proof, we argue for the rest two cases. If $\tau_1(t)\in\mathcal{S}(F^{(2)})$, by (\ref{eq:isocost_inst2_grp2}), $\tau_2^\prime(\tau_1(t))=\gamma$. Now if $t\in \overline{E_1}^c$ and $\tau_1(t)\in\mathcal{S}(F^{(2)})$, $\tau_1(t)=t$ and $\tau_{\mathbf{F}}^{(1)}(t)=\tau_2(t)$ and as a result, $(C_{\mathbf{F}}^{(1)})^\prime(t)=\tau_2^\prime(t)-\gamma=0$.  But the condition $t\in \overline{E_1}^c$ puts the constraint $(F^{(1)})^\prime(t)\leq\mu_1$. On the other hand if $t\in E_1$ and $\tau_1(t)\in\mathcal{S}(F^{(2)})$, we simultaneously have $\tau_1^\prime(t)=\frac{(F^{(1)})^\prime(t)}{\mu_1}$ and $\tau_2^\prime(\tau_1(t))=\gamma$. Applying (\ref{eq:isocost_inst2_grp1}) we will have $(\tau_{\mathbf{F}}^{(1)})^\prime(t)=\frac{(F^{(1)})^\prime(t)}{\mu_1}\cdot\gamma=\gamma$, implying $(F^{(1)})^\prime(t)=\mu_1$.  
\end{proof}

Note that $\overline{E_2}^c$ and $\mathcal{S}(F^{(2)})$ must be disjoint in EAP. Also, Lemma \ref{lem_arrivalrates_inst2_eqpref} implies, if $\mu_1\geq\mu_2\gamma$ class 2 users cannot arrive when queue 1 has a positive waiting time and the set $\mathcal{S}(F^{(2)})\cap\overline{E_1}$ has zero Lebesgue measure. 

\subsubsection{Supporting lemmas for proving Theorem \ref{mainthm_inst2_eqpref_reg1}}
The following two lemmas are helpful for proving Theorem \ref{mainthm_inst2_eqpref_reg1}.

\begin{lemma}\label{lem_inst2_eqpref_reg1_case1sign}
    If $\mu_1<\mu_2\gamma$ and $T_f>T^{(2)}_f$, in the EAP, class 1 users should start arriving before time zero.
\end{lemma}

\begin{proof} 
Otherwise, if class 1 users start arriving after time zero, since queue 2 has positive queue length between $(0,T^{(2)}_f)$, class 1 users arrive at a maximum rate of $\mu_1$ by Lemma \ref{lem_arrivalrates_inst2_eqpref}. So, no waiting queue develops in queue 1 and it remains empty at $T^{(2)}_f$. By Lemma \ref{lem_inst2_uneqpref_queue2idle}, queue 2 has zero waiting time at $T^{(2)}_f$. Hence, the network will become empty at $T^{(2)}_f$ while there is a positive mass of class 1 users yet to arrive between $(T^{(2)}_f,T_f)$, which cannot happen in EAP. 
\end{proof}

\begin{lemma}\label{lem_inst2_eqpref_reg1_case2sign}
    If $\mu_1<\mu_2\gamma$ and $T_f=T^{(2)}_f$, in the EAP, class 1 users should start arriving after time zero.
\end{lemma}
\begin{proof} 
By Lemma \ref{lem_inst2_uneqpref_queue2idle}, queue 2 will be empty at $T_f$. Queue 1 must also be empty at $T_f$, otherwise the last arriving class 1 user can arrive at the moment queue 1 empties and be strictly better off.

Now if class 1 users start arriving before time zero, since $[0,T^{(2)}_f]\subseteq \overline{E_2}$ and queue 1 has a positive waiting time at time zero, by Lemma \ref{lem_arrivalrates_inst2_eqpref}, class 1 users will arrive at rate $\mu_1$ between $[0,T^{(2)}_f]$. As a result, $Q_1(t)$ remains constant in $[0,T^{(2)}_f]$ and queue 1 will have a positive waiting queue at $T^{(2)}_f=T_f$, contradicting our observation that queue 1 must be empty at $T_f$. 
\end{proof}

\subsubsection{Proof of Theorem \ref{mainthm_inst2_eqpref_reg1}}

With $\gamma^{(1)}=\gamma^{(2)}=\gamma$ and $\mu_1<\mu_2\gamma$, there are two possibilities for every EAP:\hspace{0.05in} \textbf{1)}$T^{(2)}_f<T_f$, and \textbf{2)}$T^{(2)}_f=T_f$. The statement of Theorem \ref{mainthm_inst2_eqpref_reg1} follows from the following two lemma. 
\begin{lemma}\label{lem:inst2reg1eqpref1}
    If $\mu_1<\mu_2\gamma$ and $\gamma^{(1)}=\gamma^{(2)}=\gamma$, there exists an EAP with $T^{(2)}_f<T_f$ if and only if $\Lambda^{(1)}>\frac{\mu_1}{\mu_2-\mu_1}\Lambda^{(2)}$, and if it exists, it will be unique with a closed form same as the joint arrival profile mentioned under case 1 of Theorem \ref{mainthm_inst2_eqpref_reg1}.
\end{lemma}

\begin{proof}
\textbf{Proving $\mathcal{S}(F^{(1)})$ is an interval:}\hspace{0.05in}Since $T^{(2)}_f$ is not an isolated point of $\mathcal{S}(F^{(2)})$, queue 2 stays engaged in $[0,T^{(2)}_f]$ and must be empty at $T^{(2)}_f$ (by Lemma \ref{lem_inst2_uneqpref_queue2idle}). As a result, queue 1 should have a positive waiting time between $[T^{(2)}_f,T_f)$, since a positive mass of class 1 users arrive in $[T^{(2)}_f,T_f]$. 

Queue 1 must also have a positive waiting time between $[0,T^{(2)}_f]$. Otherwise, if queue 1 empties somewhere between $[0,T^{(2)}_f]$, since queue 2 stays engaged in $[0,T^{(2)}_f]$, class 1 users will arrive at a maximum rate of $\mu_1$ (by Lemma \ref{lem_arrivalrates_inst2_eqpref}). As a result, queue 1 becomes empty at $T^{(2)}_f$, contradicting our previous observation.

Therefore, if queue 1 stays engaged in $[0,T^{(2)}_f]$ and $[T^{(2)}_f,T_f)$, by Lemma \ref{lem_arrivalrates_inst2_eqpref}, $(F^{(1)})^\prime(t)>0$ in $[0,T_f]$, making $[0,T_f]\subseteq\mathcal{S}(F^{(1)})$. Hence, $\mathcal{S}(F^{(1)})\cap[0,\infty)=[0,T_f]$

By Lemma \ref{lem_inst2_eqpref_reg1_case1sign}, class 1 users start arriving before time zero. It is easy to observe that, before time zero, $F^{(1)}$ will be supported on an interval ending at $0$. Otherwise, if there is a gap $(t_1,t_2)$ in $\mathcal{S}(F^{(1)})\cap(-\infty,0]$ such that $F^{(1)}(t_1)=F^{(1)}(t_2)$, the class 1 user arriving at $t_1$ can be strictly better off arriving at $t_2$.  

Combining all our observations, we conclude $\mathcal{S}(F^{(1)})$ will be an interval and we can assume it to be $\mathcal{S}(F^{(1)})=[T^{(1)}_a,T^{(1)}_f]$ for some $T^{(1)}_a<0$ (by Lemma \ref{lem_inst2_eqpref_reg1_case1sign}) and $T^{(1)}_f>T^{(1)}_a$. Note that $T_f=T^{(1)}_f>T^{(2)}_f$.

\textbf{Getting the arrival rates and support boundaries:}\hspace{0.05in}We observed that, queue 2 stays engaged in $[0,T^{(2)}_f]$, queue 2 empties at $T^{(2)}_f$, queue 1 stays engaged in $[0,T^{(1)}_f]$ and $T^{(1)}_f>T^{(2)}_f$. Moreover, queue 1 must be empty at $T^{(1)}_f$ by an argument similar to the one used for proving Lemma \ref{lem_inst2_uneqpref_queue1idle}. As a result, by Lemma \ref{lem_arrivalrates_inst2_eqpref}, arrival rates in any EAP with $T_f>T^{(2)}_f$ must be:
\begin{align}\label{eq:rates_inst2reg1eqpref}
    (F^{(1)})^\prime(t)&=\begin{cases}
        \mu_1~~&\text{if}~~t\in[T^{(1)}_a,\tau_1^{-1}(T^{(2)}_f)],\\
        \mu_1\gamma~~&\text{if}~~t\in[\tau_1^{-1}(T^{(2)}_f),T^{(1)}_f].
    \end{cases},~\text{and}~(F^{(2)})^\prime(t)=\begin{cases}
        \mu_2\gamma~~&\text{if}~t\in[T^{(2)}_a,0],\\
        \mu_2\gamma-\mu_1~~&\text{if}~t\in[0,T^{(2)}_f].
    \end{cases}
\end{align}
where $\tau_1^{-1}(T^{(2)}_f)$ is well-defined because, queue 1 is engaged in $[0,T^{(1)}_f]$ and as a result, (by \ref{eq:derv_of_tau}), $\tau_1(\cdot)$ is strictly increasing in $[0,T^{(1)}_f]$. Therefore, $T^{(1)}_a,T^{(1)}_f,T^{(2)}_a,T^{(2)}_f$ satisfy the system of equations which was satisfied by the support boundaries of any Type I candidate in the proof of Lemma \ref{lem:inst2reg1case1}, except replacing $\gamma^{(1)}=\gamma^{(2)}=\gamma$. Solution to that system of equations is, $T^{(1)}_a=-\frac{1-\gamma}{\gamma}\left(\frac{\Lambda^{(1)}}{\mu_1}-\frac{\Lambda^{(2)}}{\mu_2-\mu_1}\right)$, $T^{(1)}_f=\frac{\Lambda^{(1)}}{\mu_1}$, $T^{(2)}_a=-\frac{1-\gamma}{\gamma}\frac{\Lambda^{(2)}}{\mu_2-\mu_1}$, and $T^{(2)}_f=\frac{\Lambda^{(2)}}{\mu_2-\mu_1}$.

\textbf{Obtaining the necessary condition and proving its sufficiency:}~The support boundaries obtained must satisfy $T^{(1)}_a<0$ and $T^{(1)}_f>T^{(2)}_f$. Upon imposing them, we get $\Lambda^{(1)}>\frac{\mu_1}{\mu_2-\mu_1}\Lambda^{(2)}$ and this is the necessary condition for existence of an EAP with $T_f>T^{(2)}_f$. Moreover, if the necessary condition is true, it is easy to verify that, the support boundaries satisfy $T^{(1)}_a<0$ and $T^{(1)}_f>T^{(2)}_f$. Putting these support boundaries in (\ref{eq:rates_inst2reg1eqpref}), we get only one candidate with $T_f>T^{(2)}_f$, which qualifies to be an EAP, and it has a closed form same as the joint arrival profile mentioned under case 1 of Theorem \ref{mainthm_inst2_eqpref_reg1}. Proving that the unique remaining candidate is an EAP follows by the same argument used in the proof of Lemma \ref{lem:inst2reg1case1} for proving that the unique remaining Type I candidate is an EAP. Therefore, Lemma \ref{lem:inst2reg1eqpref1} stands proved.
\end{proof}

\begin{lemma}\label{lem:inst2reg1eqpref2}
    If $\mu_1<\mu_2\gamma$ and $\gamma^{(1)}=\gamma^{(2)}=\gamma$, there exists an EAP with $T^{(2)}_f=T_f$ if and only if $\Lambda^{(1)}\leq\frac{\mu_1}{\mu_2-\mu_1}\Lambda^{(2)}$ and it it exists, the set of all such EAPs is given by the convex set mentioned under case 2 of Theorem \ref{mainthm_inst2_eqpref_reg1}.
\end{lemma}
\begin{proof}
\textbf{Identifying $\mathbf{T_a}$ and $\mathbf{T_f}$:}\hspace{0.05in}By Lemma \ref{lem_inst2_eqpref_reg1_case2sign}, class 1 users start arriving after time zero, implying $T_a=T^{(2)}_a$. Queue 2 will be empty at $T^{(2)}_f$ (by Lemma \ref{lem_inst2_uneqpref_queue2idle}) and has positive waiting time in $[0,T^{(2)}_f)$. So, by Lemma \ref{lem_arrivalrates_inst2_eqpref}, class 1 users arrive between $[0,T^{(2)}_f]$ at a maximum rate of $\mu_1$ and as a result queue 1 stays empty in $[0,T^{(2)}_f]$. Since queue 2 serves all the users between $[0,T^{(2)}_f]$, we have $T_f=T^{(2)}_f=\frac{\Lambda^{(1)}+\Lambda^{(2)}}{\mu_2}$. By Lemma \ref{lem_arrivalrates_inst2_eqpref}, since $(F^{(1)})^\prime(t)\leq \mu_1$ in $[0,T^{(2)}_f]$, we have $(F^{(2)})^\prime(t)\geq \mu_2\gamma-\mu_1$ in $[0,T^{(2)}_f]$, implying $[0,T^{(2)}_f]\subseteq\mathcal{S}(F^{(2)})$. Therefore $[T_{a},T^{(2)}_f]=\mathcal{S}(F^{(2)})$ and as a result, $C_{\mathbf{F}}^{(2)}(T_a)=C_{\mathbf{F}}^{(2)}(T^{(2)}_f)$. This gives us $T_{a}=-\left(\frac{1}{\gamma}-1\right)\frac{\Lambda^{(1)}+\Lambda^{(2)}}{\mu_2}$. 

\textbf{Getting the necessary condition:}\hspace{0.05in}The class 1 users have to arrive between $[0,T_{f}]$ at a maximum rate of $\mu_1$. Therefore, we must have $\mu_1 T_f \geq \Lambda^{(1)}$, which implies $\Lambda^{(1)}\leq\frac{\mu_1}{\mu_2-\mu_1}\Lambda^{(2)}$ and this is the necessary condition for existence of an EAP with $T_f=T^{(2)}_f$.

\noindent\textbf{Identifying the convex set:}\hspace{0.05in}Since queue 1 stays empty and queue 2 stays engaged in $[0,T_f]$, by Lemma \ref{lem_arrivalrates_inst2_eqpref}, every EAP with $T_f=T^{(2)}_f$ must satisfy: 
\begin{enumerate}[leftmargin=*]
    \item[\textbf{1.}] $(F^{(1)})^\prime(t)=0$ and $(F^{(2)})^\prime(t)=\mu_2\gamma$ for $t\in[T_a,0]$, and
    \item[\textbf{2.}] $(F^{(1)})^\prime(t)\leq\mu_1$ and $(F^{(1)})^\prime(t)+(F^{(2)})^\prime(t)=\mu_2\gamma$ for $t\in[0,T_f]$.
    \item[\textbf{3.}] $F^{(1)}(T_f)=\Lambda^{(1)},~F^{(2)}(T_f)=\Lambda^{(2)}$. 
\end{enumerate}
The above set of candidates is same as the set mentioned under case 2 of Theorem \ref{mainthm_inst2_eqpref_reg1}. If the necessary condition is satisfied, the set of candidates defined by the above two properties is non-empty. It is easy to verify that, two elements of the set are the limit of the EAPs in cases 1b (in Figure \ref{fig:inst2reg1case12}) and 2b (in Figure \ref{fig:inst2reg1case22}) of Theorem \ref{mainthm_inst2_reg1}, respectively, upon taking limits $\gamma^{(2)}=\gamma$,~$\gamma^{(1)}\to\gamma+$ and $\gamma^{(1)}=\gamma$,~$\gamma^{(2)}\to\gamma+$. This implies, the obtained set of candidates is non-empty if and only if the necessary condition holds.

\textbf{Proving sufficiency of the obtained necessary condition:}\hspace{0.05in}For every candidate in the obtained set, queue 1 remains empty in $[0,T_f]$ and class 1 users start arriving after time zero. As a result, $\tau_{\mathbf{F}}^{(1)}(\cdot)=\tau_{\mathbf{F}}^{(2)}(\cdot)=\tau_2(\cdot)$ in $[0,T_f]$ and therefore, both classes have same cost function. Note that every candidate satisfies $F^{(1)}(t)+F^{(2)}(t)>\mu_2\cdot\max\{0,t\}$ in $(T_a,T_f)$, causing queue 2 to have a positive waiting time in $[0,T_f)$. Therefore, by (\ref{eq:derv_of_tau}), $\tau_2^\prime(t)=\frac{(F^{(1)})^\prime(t)+(F^{(2)})^\prime(t)}{\mu_2}=\gamma$ in $[T_a,T_f]$, and as a result,  $(C_{\mathbf{F}}^{(1)})^\prime(t)=0$ in $[0,T_f]$ and $(C_{\mathbf{F}}^{(2)})^\prime(t)=0$ in $[T_a,T_f]$. Every class 1 user arriving before time zero gets served at $\tau_2(0)$, causing $(C_{\mathbf{F}}^{(1)})^\prime(t)=-\gamma$ in $(-\infty,0)$. Similarly every class 2 user arriving before time $T_a$, gets served at time zero, causing $(C_{\mathbf{F}}^{(2)})^\prime(t)=-\gamma$ in $(-\infty,T_a)$. After $T_f$, since queue 2 stays idle, we have $(C_{\mathbf{F}}^{(i)})^\prime(t)=1-\gamma$ in $(T_f,\infty)$ for $i=1,2$.  Therefore, for every candidate in the obtained set, every class have their cost constant in their support and higher outside. Hence, every candidate in the set  is an EAP. Therefore, if the necessary condition holds, the obtained set of candidates is the set of all EAPs with $T_f=T^{(2)}_f$, and hence Lemma \ref{lem:inst2reg1eqpref2} stands proved.
\end{proof}

\subsubsection{Supporting lemmas for proving Theorem \ref{mainthm_inst2_eqpref_reg2}}
The following two lemmas will be helpful for proving Theorem \ref{mainthm_inst2_eqpref_reg2}. 
\begin{lemma}\label{lem_inst2_eqpref_reg2_sign}
    If $\mu_1\geq\mu_2\gamma$ and all the class 2 users have arrived before time zero by the arrival profile \newline $(F^{(2)})^\prime(t)=\mu_2\gamma\cdot\mathbb{I}\left(t\in\left[-\frac{\Lambda^{(2)}}{\mu_2\gamma},0\right]\right)$, then class 1 users must start arriving before time zero and $\mathcal{S}(F^{(1)})$ must be an interval.
\end{lemma}
\begin{proof} 
If class 1 users starts to arrive from some positive time, there will be a gap in $\mathcal{S}(F^{(1)})\cup \mathcal{S}(F^{(2)})$ starting from time zero, contradicting Lemma \ref{lem_supp_are_intervals_inst2}. 

Now we assume a contradiction, \textit{i.e.}, there is a gap $(t_1,t_2)$ in $\mathcal{S}(F^{(1)})$ with $t_1,t_2\in \mathcal{S}(F^{(1)})$, $t_2>t_1$ and $F^{(1)}(t_2)=F^{(1)}(t_1)$. If $t_1<0$, by (\ref{eq:derv_of_tau}), $\tau_1(t_1)=\tau_1(\min\{t_2,0\})$, and as a result, the class 1 user arriving at $t_1$ will be strictly better off arriving at $\min\{t_2,0\}$.

On the other hand, if $t_1\geq 0$, $(t_1,t_2)$ is a gap in $\mathcal{S}(F^{(1)})\cup \mathcal{S}(F^{(2)})$ contradicting Lemma \ref{lem_supp_are_intervals_inst2}.
\end{proof}

\begin{lemma}\label{lem_inst2_eqpref_reg2_sign2}
    If $\mu_1\geq\mu_2\gamma$ and a positive mass of class 2 users arrive after time zero in the EAP, class 1 users cannot start arriving before time zero.
\end{lemma}

\begin{proof} 
On assuming a contradiction, there will be a positive waiting queue in queue 1 at time zero. So class 1 users will be arriving from queue 1 to 2 at a rate $\mu_1$ from time zero and the remaining mass of class 2 players will start arriving from some time after queue 1 has idled. Let $\tilde{T}$ be the first time after time zero at which queue 1 empties. Then $[0,\tau_1(\tilde{T})]\subseteq E_2$, since the remaining mass of class 2 users will start arriving after $\tau_1(\tilde{T})$. Therefore, class 1 users will arrive between $[0,\tau_1(\tilde{T})]$ at rate $\mu_2\gamma$ by Lemma \ref{lem_arrivalrates_inst2_eqpref}. Two situations might happen:  

\begin{itemize}[leftmargin=*]
    \item If $\mu_1>\mu_2\gamma$, till queue 1 empties, $\tau_2^\prime(t)=\frac{\mu_1}{\mu_2}$. As a result, $(C_{\mathbf{F}}^{(2)})^\prime(t)=\frac{\mu_1}{\mu_2}-\gamma>0$ till $T$. After that, till class 2 players start arriving, cost of class 2 users will not change since class 1 users will be arriving at queue 2 at a rate $\mu_2\gamma$ by Lemma \ref{lem_arrivalrates_inst2_eqpref}. Therefore, the class 2 users arriving after time zero will be strictly better off arriving before time zero.
    
    \item If $\mu_1=\mu_2\gamma$, waiting time in queue 1 will  not change between $[0,T]$. Therefore $Q_1(T)=Q_1(0)>0$, which contradicts with the fact that queue 1 empties at $T$. \vspace{-0.2in}
\end{itemize}
\end{proof}
\vspace{-0.15in}
\subsubsection{Proof of Theorem \ref{mainthm_inst2_eqpref_reg2}}
In any EAP, class 2 users start arriving before time zero. There are two possibilities: \textbf{1)} the whole class 2 population arrive before time zero to queue 2 at rate $\mu_2\gamma$ given by Lemma \ref{lem_arrivalrates_inst2_eqpref} or, \textbf{2)} a positive mass of class 2 users arrive after time zero. In the first case, we can assume by Lemma \ref{lem_inst2_eqpref_reg2_sign}, $\mathcal{S}(F^{(1)})=[T^{(1)}_a,T^{(1)}_f]$ with $T^{(1)}_a\leq 0$. By an argument similar to the one used while proving Lemma \ref{lem_inst2_uneqpref_queue1idle}, queue 1 must be empty at $T^{(1)}_f$. Let $T=\inf~[0,\infty)\cap\overline{E_2}^c$. 

For every EAP under the first case, we must have $0<T\leq T^{(1)}_f$. Otherwise, if $T>T^{(1)}_f$, since queue 1 stays empty, and queue 2 stays engaged  in $[T^{(1)}_f,T]$, by (\ref{eq:derv_of_tau}), $\tau_{\mathbf{F}}^{(1)}(T^{(1)}_f)=\tau_{\mathbf{F}}^{(1)}(T)$. As a result, the last class 1 user arriving at $T^{(1)}_f$ becomes strictly better off arriving at $T$. Now every EAP can have three possible structures:
\begin{enumerate}
    \item[\textbf{1)}] Type I: All class 2 users arrive before time zero and $T<T^{(1)}_f$. 
    \item[\textbf{2)}] Type II: All class 2 users arrive before time zero and $T=T^{(1)}_f$. 
    \item[\textbf{3)}] Type III: A positive mass of class 2 users arrive after time zero.
\end{enumerate}

The statement of Theorem \ref{mainthm_inst2_eqpref_reg2} follows from Lemma \ref{lem:inst2reg2eqpref1}, \ref{lem:inst2reg2eqpref2}, and \ref{lem:inst2reg2eqpref3} below.

\begin{lemma}\label{lem:inst2reg2eqpref1}
    If $\gamma^{(1)}=\gamma^{(2)}=\gamma$ and $\mu_1\geq\mu_2\gamma$, there exists an EAP of Type I if and only if $\Lambda^{(2)}<\left(\frac{\mu_2}{\mu_1}-1\right)\Lambda^{(1)}$, and if it exists, it will be unique with a closed form same as the joint arrival profile mentioned under case 1 of Theorem \ref{mainthm_inst2_eqpref_reg2}.
\end{lemma}

\begin{proof}
\textbf{Identifying $T^{(1)}_a,T^{(1)}_f,T$ and getting the arrival rates:}\hspace{0.05in}By Lemma \ref{lem_arrivalrates_inst2_eqpref}, for every EAP under Type I, class 1 users arrive at a rate $\mu_2\gamma$ in $[T^{(1)}_a,\tau_1^{-1}(T)]$. Since queue 2 empties at $T$ and $T<T^{(1)}_f$, queue 1 must have a positive waiting time at $\tau_1^{-1}(T)$. Since $\mu_2\gamma\leq\mu_1$, the previous statement implies queue 1 has a positive waiting time in $[T^{(1)}_a,\tau_1^{-1}(T)]$. Therefore, in $[0,T]$, class 1 users arrived at queue 2 from queue 1 at rate $\mu_1$. Since queue 2 empties at $T$, queue 2 stays empty after $T$ till the time class 1 users keep arriving at a rate $\mu_1$ at queue 2 from queue 1. Now consider $T_{1,idle}=\min{t\geq \tau_1^{-1}(T)~\vert~Q_1(t)=0}$. Then we must have $t>\tau_1^{-1}(T)$ and till $T_{1,idle}$, class 1 users arrive at queue 2 from queue 1 at rate $\mu_1$. As a result, queue 2 stays empty at $T_{1,idle}$. Note that we must have $T_{1,idle}\leq T^{(1)}_f$. Now if $T_{1,idle}<T^{(1)}_f$, the network stays empty at $T_{1,idle}$ with a positive mass of class 1 users arriving after $T_{1,idle}$, which is not possible in EAP. As a result, we must have $T_{1,idle}=T^{(1)}_f$. Hence queue 1 has a positive waiting time in $[0,T^{(1)}_f)$ and it empties at $T^{(1)}_f$. On the other hand, queue 2 has a positive waiting time in $[0,T)$ and it stays empty after $T$. 

By Lemma \ref{lem_arrivalrates_inst2_eqpref}, the arrival rates of the two classes will be:
\begin{align}\label{eq_rates_inst2_eqpref_reg2_case1}
    \forall t\in[T^{(1)}_a,T^{(1)}_f]~(F^{(1)})^\prime(t)&=\begin{cases}
        \mu_2\gamma~~&\text{if}~t\in[T^{(1)}_a,\tau_1^{-1}(T)], \\
        \mu_1\gamma~~&\text{if}~t\in[\tau_1^{-1}(T),T^{(1)}_f]
    \end{cases}~\text{and} \nonumber \\
     (F^{(2)})^\prime(t)&=\mu_2\gamma\cdot\mathbb{I}\left(t\in\left[-\frac{\Lambda^{(2)}}{\mu_2,\gamma},0\right]\right).
\end{align}

As a result, $T^{(1)}_a,T^{(1)}_f,T$ satisfies the system of equations obtained for Type I under case 1 ($\mu_2\gamma^{(1)}>\mu_1\geq\mu_2\gamma^{(2)}$) in the proof of Theorem \ref{mainthm_inst2_reg2}, except replacing $\gamma^{(1)}=\gamma^{(2)}=\gamma$. The solution to that system is $T^{(1)}_a=\frac{\Lambda^{(2)}}{\mu_2\gamma}-\left(\frac{1}{\gamma}-1\right)\frac{\Lambda^{(1)}}{\mu_1},~T=\frac{\Lambda^{(2)}}{\mu_2-\mu_1}$ and $T^{(1)}_f=\frac{\Lambda^{(1)}}{\mu_1}$. 
 
\noindent\textbf{Identifying the necessary condition and proving its sufficiency:}\hspace{0.05in}The obtained solution must satisfy $T^{(1)}_f>T$ and $T^{(1)}_a\leq 0$, to represent a Type I EAP. For those conditions to be satisfied, we need $\left(\frac{\mu_2}{\mu_1}-1\right)\Lambda^{(1)}>\Lambda^{(2)}$ and this is a necessary condition for existence of a Type I EAP. With the necessary condition satisfied, it is easy to check, that the obtained values of $T^{(1)}_a,T^{(1)}_f,T$ satisfies $T^{(1)}_f>T$ and $T^{(1)}_a\leq 0$. As a result, upon plugging in the arrival rates from (\ref{eq_rates_inst2_eqpref_reg2_case1}), we get the only candidate under Type I which can qualify to be an EAP. The candidate has a closed form same as the joint arrival profile under case 1 of Theorem \ref{mainthm_inst2_eqpref_reg2}. Proving that this candidate is an EAP follows by the argument used for proving that the Type I candidate under case 2 ($\mu_1\geq\mu_2\gamma^{(1)}>\mu_2\gamma^{(2)}$) in the proof of Theorem \ref{mainthm_inst2_reg2} is an EAP, by replacing $\gamma^{(1)}=\gamma^{(2)}=\gamma$. As a result, Lemma \ref{lem:inst2reg2eqpref1} stands proved.
\end{proof}

\begin{lemma}\label{lem:inst2reg2eqpref2}
    If $\gamma^{(1)}=\gamma^{(2)}=\gamma$ and $\mu_1\geq\mu_2\gamma$, there exists an EAP of Type II if and only if $\left(\frac{\mu_2}{\mu_1}-1\right)\Lambda^{(1)}\leq\Lambda^{(2)}\leq\left(\frac{1}{\gamma}-1\right)\Lambda^{(1)}$, and if it exists, it will be unique with a closed form same as the joint arrival profile mentioned under case 2 of Theorem \ref{mainthm_inst2_eqpref_reg2}.
\end{lemma}

\begin{proof}
\textbf{Identifying $\mathbf{T^{(1)}_a,T^{(1)}_f,T}$:}\hspace{0.05in}Since queue 2 must be engaged in $[0,T^{(1)}_f]$, class 1 users arrive at a rate $\mu_2\gamma$ in $[T^{(1)}_a,T^{(1)}_f]$ by Lemma \ref{lem_arrivalrates_inst2_eqpref}. As a result, $T^{(1)}_a,T^{(1)}_f,T$ follows the linear system followed by $T^{(1)}_a,T^{(1)}_f,T$ in Type II under case 2 ($\mu_1\geq\mu_2\gamma^{(1)}>\mu_2\gamma^{(2)}$) in the proof of Theorem \ref{mainthm_inst2_reg2}, except by replacing $\gamma^{(1)}=\gamma^{(2)}=\gamma$. Solution to that linear system is $T^{(1)}_a=\frac{1}{\mu_2}\left(\Lambda^{(2)}-\left(\frac{1}{\gamma}-1\right)\Lambda^{(1)}\right)$, $T=T^{(1)}_f=\frac{\Lambda^{(1)}+\Lambda^{(2)}}{\mu_2}$.

\textbf{Getting the necessary condition and proving its sufficiency:}\hspace{0.05in}The obtained solution must satisfy $T^{(1)}_a\leq 0$ and $\mu_1 T^{(1)}_f\geq \Lambda^{(1)}$ (for queue 1 to be empty at $T^{(1)}_f$). Upon imposing them, we get the necessary condition $\left(\frac{\mu_2}{\mu_1}-1\right)\Lambda^{(1)}\leq\Lambda^{(2)}\leq\left(\frac{1}{\gamma}-1\right)\Lambda^{(1)}$ for existence of an EAP under Type II. It is easy to check that, if the necessary conditions satisfied, the desired conditions on the support boundaries are also satisfied. Also, by Lemma \ref{lem_arrivalrates_inst2_eqpref}, arrival rates of the two classes must be $(F^{(1)})^\prime(t)=\mu_2\gamma$ in $[T^{(1)}_a,T^{(1)}_f]$ and $(F^{(2)})^\prime(t)=\mu_2\gamma$ in $[-\Lambda^{(2)}/\mu_2\gamma, 0]$, which gives us exactly one candidate, which has closed form same as the joint arrival profile mentioned under case 2 of Theorem \ref{mainthm_inst2_eqpref_reg2}. Proving that this candidate is an EAP follows by the argument used for proving that the unique Type II candidate under case 2 in the proof of Theorem \ref{mainthm_inst2_reg2} is an EAP. 
\end{proof}

\begin{lemma}\label{lem:inst2reg2eqpref3}
    If $\gamma^{(1)}=\gamma^{(2)}=\gamma$ and $\mu_1\geq\mu_2\gamma$, there exists and EAP of Type III if and only if $\left(\frac{1}{\gamma}-1\right)\Lambda^{(1)}<\Lambda^{(2)}$, and if it exists, the set of all EAPs under Type III is given by the set of joint arrival profiles mentioned under case 3 of Theorem \ref{mainthm_inst2_eqpref_reg2}. 
\end{lemma}

\begin{proof}
\textbf{Identifying $T_a$ and $T_f$:}\hspace{0.05in}By Lemma \ref{lem_inst2_eqpref_reg2_sign2}, in Type III, class 1 users will start arriving after time zero. Therefore before time zero, by Lemma \ref{lem_arrivalrates_inst2_eqpref}, $(F^{(1)})^\prime(t)=0$ and $(F^{(2)})^\prime(t)=\mu_2\gamma$ for $t\in[T_a,0]$.

By Lemma \ref{lem_arrivalrates_inst2_eqpref}, class 1 users will arrive after time zero at a maximum rate of $\mu_1$. As a result, queue 1 will have zero waiting time at all times. Hence, by Lemma \ref{lem_arrivalrates_inst2_eqpref}, the EAP after time zero will satisfy $(F^{(1)})^\prime(t)+(F^{(2)})^\prime(t)=\mu_2\gamma$. 

Queue 2 must have a positive waiting time in $(T_a,T_f)$ and it will empty at $T_f$ after serving all users. Therefore, we have $T_f=\frac{\Lambda^{(1)}+\Lambda^{(2)}}{\mu_2}$. Since users of both the groups arrive at rate $\mu_2\gamma$ in $[0,T_f]$ and queue 2 stays engaged, using (\ref{eq:derv_of_tau}), we must have $(C_{\mathbf{F}}^{(2)})^\prime(t)=0$ in $[0,T_f]$.  Combining the previous statement with the fact $[T_a,0]\subseteq\mathcal{S}(F^{(2)})$, we must have $C_{\mathbf{F}}^{(2)}(\cdot)$ constant in $[T_a,T_f]$. Therefore, we have $C_{\mathbf{F}}^{(2)}(T_a)=C_{\mathbf{F}}^{(2)}(T_f)$, which gives us $T_a=-\left(\frac{1}{\gamma}-1\right)\frac{\Lambda^{(1)}+\Lambda^{(2)}}{\mu_2}$.

\textbf{Getting the necessary condition:}\hspace{0.05in}In every Type III EAP, by Lemma \ref{lem_arrivalrates_inst2_eqpref}, all class 1 users must arrive between $[0,T_f]$ at a maximum rate of $\mu_2\gamma$ and a positive mass of class 2 users arrive after time zero. Since $\mu_2\gamma\cdot T_f-\Lambda^{(1)}$ is larger than the mass of class 2 users arriving after time zero, we must have $\mu_2\gamma\cdot T_f>\Lambda^{(1)}$. This gives us the necessary condition $\Lambda^{(2)}>\left(\frac{1}{\gamma}-1\right)\Lambda^{(1)}$ for existence of a Type III EAP. With the necessary condition satisfied, the convex set described in the third case in Theorem \ref{mainthm_inst2_eqpref_reg2} will be non empty. Two elements of this set will be the limit of the EAPs in cases 2c and 3c of Theorem \ref{mainthm_inst2_reg2}, respectively, when $\gamma^{(2)}=\gamma$, $\gamma^{(1)}\to\gamma+$ and $\gamma^{(1)}=\gamma$, $\gamma^{(2)}\to\gamma+$. 

\textbf{Identifying the set of EAPs:} By the argument for identifying $T_a,T_f$, any Type III EAP must be contained in the set of candidates satisfying: \begin{enumerate}[leftmargin=*]
    \item[\textbf{1.}] $(F^{(1)})^\prime(t)=0$ and $(F^{(2)})^\prime(t)=\mu_2\gamma$ in $[T_a,0]$, 
    \item[\textbf{2.}] $(F^{(1)})^\prime(t)+(F^{(2)})^\prime(t)=\mu_2\gamma$, and 
    \item[\textbf{3.}] $F^{(1)}(T_f)=\Lambda^{(1)},~(F^{(2)})^\prime(T_f)=\Lambda^{(2)}$, 
\end{enumerate}
where $T_a=-\left(\frac{1}{\gamma}-1\right)\frac{\Lambda^{(1)}+\Lambda^{(2)}}{\mu_2}$ and $T_f=\frac{\Lambda^{(1)}+\Lambda^{(2)}}{\mu_2}$. The obtained set of candidates is same as the set of joint arrival profiles mentioned under case 3 of Theorem \ref{mainthm_inst2_eqpref_reg2}.

\textbf{Proving sufficiency of the necessary condition:}\hspace{0.05in}Following the same argument used in the proof of Lemma \ref{lem:inst2reg1eqpref2} for proving sufficiency of the obtained necessary condition, it follows that, with the necessary condition satisfied, every joint arrival profile in the obtained set of Type III candidates is an EAP. Therefore, the lemma stands proved.  
\end{proof}

\end{document}